%% file: AKM.tex
\documentclass{memo-l}     
  \usepackage[normalem]{ulem}   % to strike out text
  \usepackage{amsmath}
\usepackage{amssymb} 
\usepackage{mathtools}   % loads »amsmath«
\usepackage{graphicx}
\usepackage{graphics}  
\usepackage{color}  
\usepackage{verbatim}  
\usepackage[normalem]{ulem} 
\usepackage[mathscr]{eucal}
\usepackage{upgreek}
\usepackage{enumerate}   
 
 \usepackage[refpage,noprefix]{nomencl}
\usepackage{nomencl} 
\def\lsm{\nomenclature} 

\makenomenclature   

\newtheorem{theorem}{Theorem}[chapter]
\newtheorem{lemma}[theorem]{Lemma}
\newtheorem{prop}[theorem]{Proposition}
\newtheorem{cor}[theorem]{Corollary}

\theoremstyle{definition}
\newtheorem{definition}[theorem]{Definition}

\theoremstyle{remark}
\newtheorem{remark}[theorem]{Remark}
 
\numberwithin{section}{chapter}
\numberwithin{equation}{chapter}

\DeclareMathAlphabet{\mathpzc}{OT1}{pzc}{m}{it}

\newcommand{\babs}[1]{{\bigl\lvert #1\bigr\rvert}}
\newcommand{\Babs}[1]{{\Bigl\lvert #1\Bigr\rvert}}
\newcommand{\bnorm}[1]{{\boldsymbol{\lvert} #1\boldsymbol{\rvert}}}
\newcommand{\tnorm}[1]{{|\hspace{-0.35mm}\lVert #1\rVert\hspace{-0.35mm}|}}

\newcommand{\ttt}{|\hspace{-0.35mm}\lVert}

\DeclarePairedDelimiter{\abs}{\lvert}{\rvert}
\DeclarePairedDelimiter{\norm}{\lVert}{\rVert}

\newcommand{\bA} {\boldsymbol{A}} 
\newcommand{\bB} {\boldsymbol{B}} 
\newcommand{\bC} {\boldsymbol{C}} 
\newcommand{\bE} {\boldsymbol{E}}
\newcommand{\bF} {\boldsymbol{F}}
\newcommand{\bG} {\boldsymbol{G}}

\newcommand{\bL} {\boldsymbol{L}} 
\newcommand{\bM} {\boldsymbol{M}} 
\newcommand{\bP} {\boldsymbol{P}} 
\newcommand{\bQ} {\boldsymbol{Q}} 
\newcommand{\bR} {\boldsymbol{R}} 
\newcommand{\bT} {\boldsymbol{T}} 
\newcommand{\bU} {\boldsymbol{U}} 
\newcommand{\bV} {\boldsymbol{V}} 

\newcommand{\bX} {\boldsymbol{X}} 
\newcommand{\bY} {\boldsymbol{Y}}
\newcommand{\bZ} {\boldsymbol{Z}}

\newcommand{\bc} {\boldsymbol{c}}
\newcommand{\bd} {\boldsymbol{d}}

\newcommand{\bp} {\boldsymbol{p}}
\newcommand{\bq} {\boldsymbol{q}}

\newcommand{\bx} {\boldsymbol{x}}
\newcommand{\by} {\boldsymbol{y}}

\newcommand{\cbX} {\boldsymbol{\mathcal X}}

\newcommand{\cbV} {\boldsymbol{\mathcal V}}
\newcommand{\cbT} {\boldsymbol{\mathcal T}}

\newcommand{\bbX} {\mathbb X}
\newcommand{\bbY} {\mathbb Y}
\newcommand{\bbZ} {\mathbb Z}

\newcommand{\bgamma} {\boldsymbol{\gamma}}

\newcommand{\zero} {\boldsymbol{0}}

\newcommand{\G} {\Gamma} 
\renewcommand{\L} {\Lambda}

\def\a{\alpha}

\def\d{\delta} 
 
\newcommand{\eps}{\varepsilon}

\newcommand{\g} {\gamma}
\def\l{\lambda} 
 
\def\o{\omega} 

\newcommand{\s} {\sigma}

\newfam\Bbbfam 
\font\tenBbb=msbm10 
\font\sevenBbb=msbm7 
\font\fiveBbb=msbm5 
\textfont\Bbbfam=\tenBbb 
\scriptfont\Bbbfam=\sevenBbb 
\scriptscriptfont\Bbbfam=\fiveBbb

 \newcommand{\C}     {\mathbb{C}} 
\newcommand{\R}     {\mathbb{R}} 
\newcommand{\N}     {\mathbb{N}}

\newcommand{\E}     {\mathbb{E}} 
 
\newcommand{\T}     {\mathbb{T}} 
\newcommand{\X}     {\mathbb{X}} 
\newcommand{\Y}     {\mathbb{Y}} 
\newcommand{\Z}     {\mathbb{Z}}

\def\1{{\mathchoice {1\mskip-4mu\mathrm l}      % Blackboard bold 1 
{1\mskip-4mu\mathrm l} 
{1\mskip-4.5mu\mathrm l} {1\mskip-5mu\mathrm l}}} 
\newcommand{\ssup}[1] {{\scriptscriptstyle{({#1}})}} 
\newtheoremstyle{thm}{2ex}{2ex}{\itshape\rmfamily}{} 
{\bfseries\rmfamily}{}{1.7ex}{} 
 
\newtheoremstyle{rem}{1.3ex}{1.3ex}{\rmfamily}{} 
{\itshape\rmfamily}{}{1.5ex}{} 
 
\newenvironment{proofsect}[1] 
{\vskip0.1cm\noindent{\scshape #1.}\hskip0.5cm}

\newcommand{\cS} {{\mathcal S}}

\newcommand{\Acal}   {{\mathcal A }}
\newcommand{\Bcal}   {{\mathcal B }}
\newcommand{\Ccal}   {{\mathcal C }} 
 
\newcommand{\Ecal}   {{\mathcal E }} 
\newcommand{\Fcal}   {{\mathcal F }} 
\newcommand{\Gcal}   {{\mathcal G }} 
\newcommand{\Hcal}   {{\mathcal H }} 
 
\newcommand{\Jcal}   {{\mathcal J }} 
\newcommand{\Kcal}   {{\mathcal K }} 
\newcommand{\Lcal}   {{\mathcal L }} 
\newcommand{\Mcal}   {{\mathcal M }}  
 
\newcommand{\Ocal}   {{\mathcal O }} 
\newcommand{\Pcal}   {{\mathcal P }} 
\newcommand{\Qcal}   {{\mathcal Q }} 
 
\newcommand{\Scal}   {{\mathcal S }} 
 
\newcommand{\Ucal}   {{\mathcal U }} 
\newcommand{\Vcal}   {{\mathcal V }} 
\newcommand{\Wcal}   {{\mathcal W }} 
 
\newcommand{\Ycal}   {{\mathcal Y }}

\newcommand{\Ascr} {\mathscr{A}}
\newcommand{\Bscr} {\mathscr{B}}
\newcommand{\Hscr} {\mathscr{H}}

\newcommand{\Cscr} {\mathscr{C}}
\newcommand{\Lscr} {\mathscr{L}}

\newcommand{\ip} {{\mathit p}}

 \newcommand{\ex}{{\rm e}} 
\newcommand{\com}{{\rm c}} 
\renewcommand{\d}{{\rm d}}

\newcommand{\sym}{{\rm Sym}}
\newcommand{\id}{{\sf{I}}} 
 
\newcommand{\supp}{{\operatorname {supp}}} 

\newcommand{\Hess}{{\operatorname {Hess}\,}}

\newcommand{\dist}{{\operatorname {dist}}} 
\newcommand{\diam}{{\operatorname {diam}}} 
\newcommand{\dime}{{\operatorname {dim}\,}} 
 
\newcommand{\re}{{\operatorname {Re}\,}}
\newcommand{\tr}{{\operatorname {Tr}}}
\newcommand{\Exp}{\mathscr{E}\kern-0.2mm{\operatorname{xp}}}
\newcommand{\Log}{\mathscr{L}\kern-0.2mm{\operatorname{og}}}
\newcommand{\Id}{{\operatorname {Id}}}
\newcommand{\heap}[2]{\genfrac{}{}{0pt}{}{#1}{#2}} 

\renewcommand{\emptyset} {\varnothing} 
 
\newcommand{\p} {\partial} 
\newcommand\embed{\hookrightarrow}
\newcommand\poly{{\mathrm{poly}}}

\newcommand{\dprime}{^{\prime\prime}\mkern-1.2mu}

\definecolor{weakgray}{gray}{.7}

\newcommand{\oldfour}[1]{{\color{weakgray}#1}}

\definecolor{forestgreen}   {cmyk}{0.91, 0   , 0.88, 0.12}

\makeindex

\begin{document}

\frontmatter

\title[\hfill Strict Convexity of the Surface Tension for Non-convex Potentials\hfill]
{Strict Convexity of the Surface Tension \\ for Non-convex Potentials}

%\fbox{\textbf{Version: 6th March 2016}}} 

%    Remove any unused author tags.

%    author one information
\author{Stefan Adams}
\address{Mathematics Institute, University of Warwick, Coventry CV4 7AL, United Kingdom}
%\curraddr{}
\email{S.Adams@warwick.ac.uk}
%\thanks{}

\author{Roman Koteck\'{y}}
\address{Mathematics Institute, University of Warwick, Coventry CV4 7AL, United Kingdom 
and Center for Theoretical Study, Charles University, Jilsk\'a 1, Prague, Czech Republic}
\curraddr{}
\email{R.Kotecky@warwick.ac.uk}
\thanks{}
 
\author{Stefan M\"uller}
\address{Universit\"at Bonn, Endenicher Allee 60, D-53115 Bonn, Germany}
\curraddr{}
\email{stefan.mueller@hcm.uni-bonn.de}
\thanks{}
 
%    \date is required; it is the date received by the editor.
%\date{19th July 2015}

%\subjclass[2000]{Primary }
%    The 2010 edition of the Mathematics Subject Classification is
%    now available.  If you are citing a classification from the
%    new scheme, use the following input coding instead.
 

%\dedication{Dedication text (use \\[2pt] for line break if necessary)} 

\begin{abstract}
We study gradient models on the  lattice $\Z^d$ with non-convex interactions. 
These Gibbs fields (lattice models with continuous spin) emerge in various branches of physics and mathematics.
In quantum field theory they appear as massless field theories.  Even though our motivation stems from considering vector valued fields as displacements for atoms of crystal structures and the study of the Cauchy-Born rule for these models, our attention here is mostly devoted to interfaces, with the gradient field as an \emph{effective} interface interaction.  In this case we prove the strict convexity of the surface tension (interface free energy)  for low temperatures and sufficiently small interface tilts
 using muli-scale (renormalisation group analysis) techniques following the approach of Brydges and coworkers \cite{B07}. 
This is a complement to the study of the high temperature regime in 
\cite{CDM09} and it is an extension of Funaki and Spohn's result \cite{FS97} valid for strictly convex interactions.

\end{abstract}

 \maketitle

 \tableofcontents

\mainmatter

\include{AKM-acknowledgments-30thJune-2016}

\include{AKM-ch1-introduction-30thJune-2016}

\include{AKM-ch2-setting-30thJune-2016}

\include{AKM-ch3-strategy-30thJune-2016}

\include{AKM-ch4-detailed-30thJune-2016}

\include{AKM-ch5-norms-30thJune-2016}

\include{AKM-ch6-smoothness-30thJune-2016}

\include{AKM-ch7-contraction-30thJune-2016}

\include{AKM-ch8-tuning-30thJune-2016}

\appendix
\include{AKM-appA-sobolev-30thJune-2016}
\include{AKM-appB-parts-30thJune-2016}
\include{AKM-appC-gaussian-30thJune-2016}

\include{AKM-appD-chain-30thJune-2016}

\include{AKM-appE-implicit-30thJune-2016}

\include{AKM-appF-geometry-30thJune-2016}

\backmatter
\bibliographystyle{amsalpha}
\include{AKM-biblio-30thJune-2016}

%    See note above about multiple indexes.
\printindex

\newpage
\renewcommand{\nomname}{List of Symbols}
\printnomenclature[1in]

\end{document}

%% file: AKM-acknowledgments-30thJune-2016.tex
\chapter*{Acknowledgment}

%We are grateful to David Brydges for generously sharing
%his ideas on renormalisation group methods with us
%and for many interesting discussions.
%%
%%
%We thank David Preiss for  inspiring discussion on 
%differentiability properties and for 
%providing notes on which Appendix D is based.
%%
%We  also thank  S. Buchholz, S. Hilger, G. Menz, F. Otto, E. Runa for many helpful suggestions and comments.
%%% Add further names??
%%
%The research of S. Adams was supported by EPSRC grant number EP/I003746/1 and by the Royal Society  Exchange grant  IE130438   'The Challenge of Different Scales in Nature'. S.~Adams thanks the mathematics department at UBC for the warm hospitality during his sabbatical stay 2013-2014. S. M\"uller was supported by the DFG Research group FOR 718 {\it Analysis and stochastics in complex physical systems} (2006--2013), by the Hausdorff Center for Mathematics (since 2008)
%and by the CRC 1060 {\it The mathematics of emergent effects} (since 2013).
%% Which grants for Roman ????

%NEW 

We are grateful to David Brydges for generously sharing his ideas on renormalisation
group methods with us and for many interesting discussions. We thank
David Preiss for inspiring discussion on differentiability properties and for providing
notes on which Appendix D is based. We also thank S. Buchholz, S. Hilger,
G. Menz, F. Otto, E. Runa for many helpful suggestions and comments. The research
of S. Adams was supported by EPSRC grant number EP/I003746/1 and
by the Royal Society Exchange grant IE130438 \emph{The Challenge of Different Scales
in Nature}. S. Adams thanks the mathematics department at UBC for the warm
hospitality during his sabbatical stay 2013-2014. R. Koteck\'{y} was supported by the grants
GA\v{C}R P201/12/2613 \emph{Threshold phenomena in stochastic systems} 
and GA\v{C}R 16-15238S \emph{Collective behavior of large stochastic systems} 
and S. M{\"u}ller  by the
DFG Research group FOR 718 \emph{Analysis and stochastics in complex physical systems}  (2006{2013), by the Hausdorff Center for Mathematics (since 2008) and by
the CRC 1060 \emph{The mathematics of emergent effects}  (since 2013).

%%%

%%    Text AKM-acknowledgements- 200616
%We thank to D.B., D.P., ...
%
%The research of S.A. was supported by EPSRC grant  EP/I003746/1,
%R.K.  by the grants GA\v CR 201-09-1931, P201/12/2613, and MSM 0021620845, 
%and S.M.  by the DFG Research group 718 ``Analysis and Stochastics in Complex Physical Systems'' 
%and  the Hausdorff Center of Mathematics (DFG EXC 59).
%
%

%% Old text (before June 20, 2016)
%%
%The research of S. Adams was supported by EPSRC grant number EP/I003746/1 and by the Royal Society  Exchange grant  IE130438   'The Challenge of Different Scales in Nature'. S.~Adams thanks the mathematics department at UBC for the warm hospitality during his sabbatical stay 2013-2014. S. M\"uller was supported by the DFG Research group FOR 718 {\it Analysis and stochastics in complex physical systems} (2006--2013)
%and by the CRC 1060 {\it The mathematics of emergent effects} (since 2013).

%% file: AKM-ch1-introduction-30thJune-2016.tex
\chapter{Introduction}
This paper has two related goals.

First, we seek to identify uniform convexity properties for a class of
lattice gradient models with non-convex microscopic interactions.

Secondly, we extend the rigorous renormalisation group
techniques developed by Brydges and coworkes to models
without a discrete rotational symmetry of the interaction. 
 In the presence of symmetry,
the set of relevant terms is strongly restricted by the symmetry.

Regarding the first goal, we consider  gradient random fields $\{\varphi(x)\}_{x\in\Lscr}$
indexed by a lattice $\Lscr$ with values in $\R^m$, $\varphi(x)\in\R^m$. 
The term \emph{gradient} is referring to the assumption that the distribution depends only on  gradients 
$\nabla_{e}\varphi(x)=\varphi(x+e)-\varphi(x)$.

These type of fields are used as effective models of crystal deformation or phase separation.
In the former case, where $m=3$ and $\Lscr\subset\Z^3$, the value $\varphi(x)$ plays the role of a displacement
of an atom labelled by a site $x$ of a crystal under deformation. Even though the former case is our main motivation,  we will restrict our attention here, for simplicity,
to the latter case with $m=1$ and $\Lscr=\Z^d$.
This is a model describing a phase separation in $ \mathbb{R}^{d+1} $ with $\varphi(x)\in\R$ corresponding to the position of the (microscopic) phase separation surface. The model is a reasonably effective approximate description in spite of the fact that it ignores overhangs of separation surface as well as any correlations inside and between
the coexisting phases.

The distribution of the interface is given in terms of a Gibbs distribution with nearest neighbour interactions of gradient type, that is, the interaction between neighboring sites $x, x+{\rm e}_i $ depends only on the gradient $\nabla_i\varphi(x)=\varphi(x+{\rm e}_i)-\varphi(x), i=1,\ldots, d $.
More precisely, for any finite $\L\subset \Z^d$ we consider the Hamiltonian of the form
$$
H_\Lambda(\varphi)=\sum_{x\in\Lambda}\sum_{i=1}^dW(\nabla_i\varphi(x)),
$$
where $ W\colon\mathbb{R}\to\mathbb{R} $ is a perturbation of a quadratic functions, i.e.
$$
W(\eta)=\frac{1}{2}\eta^2+V(\eta) \quad\mbox{ with some perturbation } V\colon\mathbb{R}\to\mathbb{R}. 
$$
For a given boundary condition $ \psi\in\mathbb{R}^{\partial\Lambda}$, where 
%
%\lsm[Lhdg]{$\partial\Lambda$}{$=\{z\in\mathbb{Z}^d\setminus\Lambda\colon \abs{z-x}=1\mbox{ for some } x\in\Lambda\}$,  external boundary of $\Lambda$}
%
$\partial\Lambda=\{z\in\mathbb{Z}^d\setminus\Lambda\colon |z-x|=1\mbox{ for some } x\in\Lambda\}$, 
we consider the Gibbs distribution at inverse temperature $ \beta > 0 $  given by
%
%\lsm[ggLg]{$\g_{\Lambda,\beta}^{\psi}({\rm d}\varphi)$}{$=\frac{1}{Z_\Lambda(\beta,\psi)}\exp\big(-\beta H_\Lambda(\varphi)\big)\prod_{x\in\Lambda}{\rm d}\varphi(x)\prod_{x\in\partial\Lambda}\delta_{\psi(x)}({\rm d}\varphi(x))$,  \\ Gibbs distribution for a given boundary condition}
%
$$
\g_{\Lambda,\beta}^{\psi}({\rm d}\varphi)=\frac{1}{Z_\Lambda(\beta,\psi)}\exp\big(-\beta H_\Lambda(\varphi)\big)\prod_{x\in\Lambda}{\rm d}\varphi(x)\prod_{x\in\partial\Lambda}\delta_{\psi(x)}({\rm d}\varphi(x)),
$$
where the normalisation constant $ Z_\Lambda(\beta,\psi) $ is the integral of the density and is called the partition function. One is particularly interested in tilted boundary conditions
$$
\psi_{u}(x)=\langle x,u\rangle,\quad \mbox{ for some tilt } u\in\mathbb{R}^d.
$$
An object of basic relevance in this context is the \textit{surface tension\/} or \textit{free energy\/} defined by the limit
%
%\lsm[sgua]{$\sigma(u)$}{$=\lim_{\Lambda\uparrow\mathbb{Z}^d}-\frac{1}{\beta\abs{\Lambda}}\log Z_\Lambda(\beta,\psi_{u})$,  \\ surface tension}
%
\begin{equation}
\sigma_\beta(u)=-\lim_{\Lambda\uparrow\mathbb{Z}^d}\frac{1}{\beta|\Lambda|}\log Z_\Lambda(\beta,\psi_{u}).
\end{equation}
The surface tension $ \sigma_\beta(u) $ can also be seen as the price to pay for tilting a macroscopically flat interface.
The existence of the above limit follows from a standard sub-additivity argument. 

In the case of a \textit{strictly\/} convex potential, Funaki and Spohn show in \cite{FS97} that $ \sigma_\beta $ is convex as a function of the tilt. The simplest strictly convex potential is the quadratic one with $ V=0 $, which corresponds to a Gaussian model, also called the gradient free field. The convexity of the surface tension plays a crucial role in the derivation of the hydrodynamical limit of the Landau-Ginsburg model in \cite{FS97}. Strict convexity of the surface tension for strictly convex $W$ with $ 0<c_1 \le W{\dprime}\le c_2<\infty $, was proved in \cite{DGI00}. Under the assumption of the bounds of the second derivative of $W$, a large deviations principle for the rescaled profile with rate function given in terms of the integrated surface tension has been derived in \cite{DGI00}. Both papers \cite{FS97} and \cite{DGI00} use explicitly the conditions on the second derivative of $W$ in their proof. In particular they rely on the Brascamp-Lieb inequality and on the random walk representation of Helffer and Sj\"ostrand, which requires a strictly convex potential $W$.

In \cite{CDM09} Deuschel \textit{ et al} showed the strict convexity of the surface tension for non-convex potentials in the small $ \beta $ (high temperature) regime for potentials of the form
$$
W(t)=W_0(t)+g(t),
$$ where $ W_0 $ is strictly convex as above and where $ g \in\mathcal{C}^2(\mathbb{R}) $ has a negative bounded second derivative such that $ \sqrt{\beta}\| g^{\dprime}\|_{L^1(\mathbb{R})} $ is sufficiently small. These studies have been applied  in \cite{CD09} to large deviations principle for the profile.

In the present paper, we show the strict convexity of the surface tension for large enough $\beta$ (low temperatures) and sufficiently small tilt, using multi-scale techniques based on a finite range decomposition of the underlying background Gaussian measure in \cite{AKM09b}. 

Note also that, due to the gradient interaction, the Hamiltonian has a continuous symmetry. In particular this implies that no Gibbs measures on $ \mathbb{Z}^d $ exist for dimensions $ d=1,2 $ where the field 'delocalises', cf. \cite{FP}. If  one considers the corresponding random field of gradients (discrete gradient image of the height field $ \varphi $) it is clear that its distribution depends on the gradient of the boundary condition of the height field. One can also introduce gradient Gibbs measures  in terms of conditional distributions satisfying DLR equations, cf. \cite{FS97}. For strictly convex interaction $ W $ with  bounds on the second derivative, Funaki and Spohn in \cite{FS97} proved the existence and uniqueness of an extremal, i.e. ergodic, gradient Gibbs measure for each tilt $ u\in\mathbb{R}^d$. In the case of non-convex $W$, uniqueness of the ergodic gradient component can be violated, for tilt $ u=0 $ this has been proved in \cite{BK07}. However in this phase transition situation in \cite{BK07}, the surface tension is not strictly convex at tilt $ u=0 $.

The second goal of the present paper is to show in detail how  the rigorous renormalisation approach of 
Brydges and coworkers (see \cite{BY90} for early work, \cite{B07} for a survey and 
\cite{BS15a, BS15b, BBS15a, BS15c, BS15d, BBS15b}  for recent developments which go well beyond the gradient models
discussed in this paper)
can be extended to accommodate a class of models without a discrete rotational symmetry of the interaction.

In accordance with the general renormalization group strategy,
the resulting partition function $Z_\Lambda(\beta,\psi_{u})$ is obtained by a sequence of ``partial integrations'' (labelled by an index $k$).
The result of each of them is expressed in terms of two functions: the \emph{``irrelevant'' polymers} $K_k$ that are decreasing with each subsequent integration, and the \emph{``relevant'' ideal Hamiltonians} $H_k$---homogeneous quadratic functions of
gradients $\nabla\varphi$ parametrized by a fixed finite number of parameters. To fine-tune the procedure so that the final integration  yields a result with a straightforward bound we need to assure the smoothness of the procedure with respect  to the parameters of a suitably chosen ``seed Hamiltonian''. However, it turns out that the derivatives with respect to those parameters  lead to a loss of regularity of functions $K_k$ and $H_k$ considered as elements in a scale of Banach spaces.

A more detailed summary of the  strategy is presented  in Chapter~\ref{S:strategy} where the reader can get  an overview of our methods and techniques of the proof. 
First, however, we will summarize the main claims concerning the convexity of the surface tension $\sigma_\beta(u)$ in Chapter~\ref{S:results}.
The detailed formulations and proofs  are in 
Chapters~\ref{S:mainsteps}--\ref{S:Final}. Miscelaneous technical details %concerning discrete Sobolev estimates and bounds on boundary terms with help of discrete integration by part 
are deferred to Appendices.

Various extensions and generalisations of our work are possible.

First, Buchholz has very recently developed a new finite range decomposition 
for which no loss of regularity occurs in the problem we study  \cite{BU16}. However, in the present paper we decided to stick to the usual 
finite range decomposition and to explain how the loss for regularity can be overcome
by a suitable version of the chain rule and the implicit function theorem since we believe that
these tools might be useful in other contexts, too.

Secondly, we restrict ourselves to dimensions $d=2$ and $d=3$ because in that case there
are only two types of linear relevant terms: linear combinations of the first and second discrete derivatives of the field.
Our approach can be extended to higher dimensions by including linear terms in higher derivatives of the field. 
This only requires an extension of the  appropriate ``homogenisation projection operator''  $\Pi_2$ used in the definition of quadratic functions $H_k$ (see Chapter~\ref{S:renorm})
to  relevant polynomials and the corresponding 
discrete Poincar\'e type inequalities.
%, see
%Chapter~\ref{sec:contraction},  in particular
%Lemmas~\ref{BryLemma},   \ref{lemmatayest1},    and \ref{lemmatayest2}.
In fact,   Brydges and Slade  \cite{BS15b}
have recently developed a very general theory which allows one to define the
projection onto the relevant polynomials and to prove the necessary estimates.

Thirdly,  we focus on scalar valued field even though most our methods carry directly over to the 
vector valued case which is relevant in elasticity. The discussion of models relevant in elasticity
requires, however, also a number of other changes, e.g. the inclusion of non nearest neighbour interactions
and the consideration of symmetry under the left action of $\mathrm{SO}(m)$ (frame indifference). 
As a result  it is natural to replace our assumption that the microscopic interaction is convex close to its minimum
by a more complicated condition.
We will thus address the application of our ideas to vector valued fields and models relevant in elasticity in future work.

Fourthly, in this work we focus on the behaviour of the partition function in the large volume limit.
 As in the work of Bauerschmidt, Brydges and Slade \cite{BBS15b}  it should be possible
 to study finer properties, e.g., correlation functions. As a first step in that direction
Hilger has recently shown that the scaling limit of the random 
field becomes a free Gaussian field on the torus
(with the renormalised covariance) and that suitably averaged correlation functions converge
in the infinite volume limit \cite{Hi16}.
\smallskip

%% file: AKM-ch2-setting-30thJune-2016.tex
\chapter{Setting and Results}
\label{S:results}

\section{Setup} 
Let  $ L>0$  be a fixed integer.
For any integer $N$ we consider the space
\lsm[Lc]{$L$}{linear size of a renormalization block}
\lsm[Nc]{$N$}{the power yielding the size (of the torus) $L^N$}
\lsm[VyNc]{$\cbV_N$}{$=\{\varphi: \mathbb Z^d\to\mathbb R;\  \varphi(x+k)=\varphi(x)\  \forall k\in (L^N\mathbb Z)^d\}$,   set of fields taken as $\ell_2(\R^{L^{Nd}})$}
$$
 \cbV_N=\{\varphi: \mathbb Z^d\to\mathbb R;\  \varphi(x+k)=\varphi(x)\  \forall k\in (L^N\mathbb Z)^d\}
$$
that can be identified with the set of functions on the  \emph{torus} 
$\mathbb T_N=\bigl(\mathbb Z/L^N\mathbb Z\bigr)^d$.
\lsm[TxNc]{$\mathbb T_N$}{$=\bigl(\mathbb Z/L^N\mathbb Z\bigr)^d$, torus}
Using 
 \lsm[zzba]{$\abs{x}_\infty$}{$ \max_{i=1,\dots,d}\abs{x_i}$}
$\abs{x}_\infty= \max_{i=1,\dots,d}|x_i|$ for any $x\in \mathbb R^d$ (reserving the notation 
 \lsm[zzbb]{$\abs{x}$}{$=\sqrt{\sum x_i^2}$, the Euclidean norm}
$|x|$ for the Euclidean norm $\sqrt{\sum x_i^2}$),
the torus $\mathbb T_N$
may be represented by the lattice cube 
$\L_N =\{x\in\Z^d\colon \abs{x}_\infty\le \frac{1}{2}(L^N-1)\} $ of side $ L^N $, 
\lsm[LhNc]{$\L_N$}{$ =\{x\in\Z^d \colon \abs{x}_\infty \le \frac{1}{2} (L^N-1) \} $ (identified with torus  $\mathbb T_N$) }
once it is
equipped with the metric $\rho(x,y)=\inf\{\abs{x-y+k}_\infty: k\in (L^N\mathbb Z)^d\}$. 
\lsm[rgxaya]{$\rho(x,y)$}{$=\inf\{\abs{x-y+k}_\infty: k\in (L^N\mathbb Z)^d\}$}
We view  $\cbV_N$ as a Hilbert space with the scalar product
\lsm[zzab]{$(\cdot,\cdot)$}{the scalar product $(\varphi,\psi)=\sum_{x\in \mathbb T_N}\varphi(x)\psi(x)$}
$$
(\varphi,\psi)=\sum_{x\in \mathbb T_N}\varphi(x)\psi(x).
$$
By $ \cbX_N $ we denote the subspace
\begin{equation}
\label{E:XN}
\cbX_N=\{\varphi\in \cbV_N: \sum_{x\in \mathbb T_N }\varphi(x)=0\},
\end{equation}
of height fields whose sum over the torus is zero.
\lsm[XdNc]{$\cbX_N$}{$=\{\varphi\in \cbV_N: \sum_{x\in \mathbb T_N }\varphi(x)=0\}$}
We use  $ \l_N $ to denote  the $(L^{Nd}-1)$-dimensional Hausdorf measure on $ \cbX_N $. 
\lsm[lgNc]{$ \l_N $}{$(L^{Nd}-1)$-dimensional Hausdorf measure on $ \cbX_N $}
We equip the space $ \cbX_N $ with the $\sigma$-algebra $ \boldsymbol{\Bcal_{\cbX_N}} $  induced by the Borel $\sigma$-algebra with respect to the product topology
\lsm[BzXyNc]{$\boldsymbol{\Bcal_{\cbX_N}} $}{$\sigma$-algebra on $ \cbX_N $ induced by the Borel $\sigma$-algebra with respect to the product topology}
and use 
 \lsm[Mczz1aBxNb]{$ \Mcal_1(\cbX_N)$}{$=\Mcal_1(\cbX_N, \boldsymbol{\Bcal_{\cbX_N}}) $,  the set of probability measures on $ \cbX_N $}
$ \Mcal_1(\cbX_N)=\Mcal_1(\cbX_N, \boldsymbol{\Bcal_{\cbX_N}}) $ to denote the set of probability measures on $ \cbX_N $, referring to  elements in $ \Mcal_1(\cbX_N) $ as to \emph{random gradient fields}.

In this article we study a class of random gradient fields defined (as Gibbs measures) in terms of a non-convex  perturbation of a Gaussian gradient field. 
For a precise definition, we first introduce  the \emph{discrete derivatives}
\begin{equation}
\label{E:nablai}
\nabla_i\varphi(x)=\varphi(x+{\rm e}_i)-\varphi(x),\ \nabla_i^*\varphi(x)=\varphi(x-{\rm e}_i)-\varphi(x)
\end{equation}
on $\cbV_N$.  
\lsm[Nhiafga]{$\nabla_i\varphi(x)$}{$  =\varphi(x+e_i)-\varphi(x)$,  discrete derivative}
\lsm[Nhiafgb]{$\nabla_i^*\varphi(x)$}{$  =\varphi(x-e_i)-\varphi(x)$,  dual of discrete derivative $\nabla_i$}
Here, $e_i$, $ i=1,\dots,d$, are unit coordinate vectors in $\R^d$.
\lsm[eaia]{$e_i$}{unit coordinate vectors in $\R^d$}
Next, let $\Ecal_N(\varphi)$ be the  Dirichlet form 
\begin{equation}
\label{E:QN}
\Ecal_N(\varphi)=\frac{1}{2}\sum_{x\in \mathbb T_N}\sum_{i=1}^d \bigl(\nabla_i\varphi(x) \bigr)^2.
\end{equation}
Choosing a function $V\colon\R\to \R$
\lsm[EyNcpg]{$\Ecal_N(\varphi)$}{$=\frac{1}{2}\sum_{x\in \mathbb T_N}\sum_{i=1}^d \bigl(\nabla_i\varphi(x) \bigr)^2$}
\lsm[VcRxRx]{$V\colon\R\to \R$}{potential perturbation}
 (satisfying the conditions to be specified later), we consider the  Gibbs mesure on the torus corresponding to the Hamiltonian
\begin{equation}
H_{N}(\varphi)= \Ecal_N(\varphi)+ \sum_{x\in \mathbb T_N} \sum_{i=1}^d V(\nabla_i \varphi(x)).
\end{equation}
\lsm[HcNcpg]{$H_{N}(\varphi)$}{$= \Ecal_N(\varphi)+ \sum_{x\in \mathbb T_N} \sum_{i=1}^d V(\nabla_i \varphi(x))$, Hamiltonian on $\mathbb T_N$ (with no tilt)}

To be able to discuss random fields with a tilt  $u=(u_1\dots,u_d)\in\R^d$, 
\lsm[ub]{$u=(u_1\dots,u_d)\in\R^d$}{ a tilt}
we use the method  proposed by Funaki and Spohn \cite{FS97} who enforce the tilt on a measure 
defined on the torus space $\cbX_N$ by replacing the gradient 
$\nabla_i \varphi(x)$ in all definitions above by $\nabla_i \varphi(x)-u_i$,
$i=1,\dots,d$, $x\in \mathbb T_N$.

Namely, we define the Gibbs mesure on $ \mathbb T_N$  at inverse temperature $\beta$ as 
\lsm[bzz]{$\beta$}{inverse temperature}
\lsm[ggNcbgub]{$\g_{N,\beta}^{u}(\d\varphi)$}{$=\frac{1}{Z_{N,\beta}(u)} \exp\bigl(-\beta H_{N}^{u}(\varphi)\bigr)\l_N(\d\varphi)$,  random gradient field with Hamiltonian $H_N^{u}$ (with tilt $u$)}
\begin{equation}
\label{E:muNu}
\g_{N,\beta}^{u}(\d\varphi)=\frac{1}{Z_{N,\beta}(u)} \exp\bigl(-\beta H_{N}^{u}(\varphi)\bigr)\l_N(\d\varphi),
\end{equation}
where 
\lsm[HcNcpgub]{$H_{N}^{u}(\varphi)$}{$= \Ecal_N(\varphi)+\frac12 L^{Nd} \abs{u}^2+ \sum_{x\in \mathbb T_N} \sum_{i=1}^d V(\nabla_i \varphi(x)-u_i)$, Hamiltonian on $\mathbb T_N$ with tilt $u$}
\begin{equation}
\label{E:HNu}
H_{N}^{u}(\varphi)= \Ecal_N(\varphi)+\frac12 L^{Nd} \abs{u}^2+ \sum_{x\in \mathbb T_N} \sum_{i=1}^d V(\nabla_i \varphi(x)-u_i)
\end{equation}
(in the last equation we used the fact that substituting $\nabla_i \varphi(x)\mapsto \nabla_i \varphi(x)-u_i$ in $\Ecal_N$, the linear term $ \sum_{x\in \mathbb T_N} \sum_{i=1}^d u_i\nabla_i \varphi(x)$ vanishes  as $ \sum_{x\in \mathbb T_N} \nabla_i \varphi(x)=0$ for each $\varphi\in\cbV_N$ and each $i=1,\dots,d$).
Again, $Z_{N,\beta}(u)$ is the normalizing partition function
\begin{equation}
\label{E:ZNu}
Z_{N,\beta}(u)=\int_{\cbX_N} \exp\bigl(-\beta H_{N}^{u}(\varphi)\bigr)\l_N(\d\varphi).
\end{equation}
\lsm[ZcNcpgub]{$Z_{N,\beta}(u)$}{$=\int_{\cbX_N} \exp\bigl(-\beta H_{N}^{u}(\varphi)\bigr)\l_N(\d\varphi)$, partition function on $\mathbb T_N$ with tilt $u$}

Even though the ultimate goal, in general, is to characterize all limiting gradient Gibbs measures with a fixed mean tilt, and,  in  particular cases,  to prove their unicity, in this paper we will restrict our attention to the discussion of the strict convexity, in $u$, of the surface tension
\begin{equation}
\sigma_\beta(u):=-\lim_{N\to\infty}\frac{1}{\beta L^{dN}}\log Z_{N,\beta}(u).
\end{equation}
\lsm[sgbgub]{$\s_\beta(u)$}{$=-\lim_{N\to\infty}\frac{1}{\beta L^{dN}}\log Z_{N,\beta}(u)$, free energy (surface tension) with tilt $u$}
%

%%%%%%%%%%%%%%%%%%%%%%%%%%%%%%%%%%%%%%%%%%%%%%%%%%%%%%%%%%%%%%%%%%

\section{Main result}\label{sec:result}
To state our main result, we need a condition on smallness of the perturbation $V$. We will state it in terms of the function $\Kcal_{V,\beta,u}:\R^d\to\R$ associated
with the perturbation $V\colon\R\to \R$  determining the Hamiltonian $H_N^{u}$ in \eqref{E:HNu} (and with the (inverse) temperature $\beta\ge 0$ and the tilt $u\in\R^d$).
Namely, we take
\lsm[KcVcbgub]{$\Kcal_{V,\beta,u}(z)$}{$=\exp\bigl\{-\beta \sum_{i=1}^d U\bigl(\frac{z_i}{\sqrt\beta},u_i\bigr)\bigr\}-1$, the Mayer function for perturbation $V$}
\begin{equation}
\label{E:KVb}
\Kcal_{V,\beta,u}(z)=\exp\bigl\{-\beta \sum_{i=1}^d U\bigl(\frac{z_i}{\sqrt\beta},u_i\bigr)\bigr\}-1
\end{equation}
with
\lsm[Uc]{$U(s,t)$}{$=V(s-t)-V(-t)-V'(-t)s$}
\begin{equation}
\label{E:U(s,t)}
U(s,t)=V(s-t)-V(-t)-V'(-t)s.
\end{equation}
First, we rewrite the partition function in terms of the function $\Kcal_{V,\beta,u}$.
Consider the Gaussian measure   $\nu_{\beta}$ on $\cbX_N$  corresponding to the Dirichlet form $ \beta{\Ecal}_N(\varphi)$: 
\lsm[ngdapgbg]{$\nu_{\beta}(\d\varphi)$}{$=\frac{1}{Z_{N,\beta}^{(0)}} \exp\bigl(-\beta\Ecal_N(\varphi)\bigr)\l_N(\d\varphi)$, Gaussian measure  on $\cbX_N$}
\begin{equation}
\label{E:nubeta}
\nu_{\beta}(\d\varphi)=
\frac{1}{Z_{N,\beta}^{(0)}} \exp\bigl(-\beta\Ecal_N(\varphi)\bigr)\l_N(\d\varphi),
\end{equation}
with
\lsm[ZcNcbg0z]{$Z_{N,\beta}^{(0)}$}{$=\int_{\cbX_N} \exp\bigl(-\beta\Ecal_N(\varphi)\bigr)\l_N(\d\varphi)$}
\begin{equation}
\label{E:ZN0}
Z_{N,\beta}^{(0)}=\int_{\cbX_N} \exp\bigl(-\beta\Ecal_N(\varphi)\bigr)\l_N(\d\varphi).
\end{equation}
 To avoid overloading of the notation, here and  in future, we often skip the index referring to $N$ (as above in the case of measure $\nu_{\beta}$).
Now, the partition function \eqref{E:ZNu} is
\begin{multline}
\!\!\!\!\!\!\!\!Z_{N,\beta}(u)=Z_{N,\beta}^{(0)} \exp\bigl(-\tfrac\beta2 L^{Nd} \abs{u}^2\bigr)\int_{\cbX_N}\exp\bigl(-\beta\sum_{x\in \mathbb T_N} \sum_{i=1}^d V\bigl(\nabla_i \varphi(x)-u_i\bigr)\bigr)\nu_{\beta}(\d\varphi)=\\
=Z_{N}^{(0)}  \exp\bigl(-\beta L^{Nd}(\tfrac12  \abs{u}^2+ V(u))\bigr)\int_{\cbX_N}\exp\bigl(-\beta\sum_{x\in \mathbb T_N} \sum_{i=1}^d U\bigl(\tfrac{1}{\sqrt \beta}\nabla_i \varphi(x),u_i\bigr)\bigr)\nu(\d\varphi),
\end{multline}
where, denoting $\nu(\d\varphi)=\nu_{\beta=1}(\d\varphi)$ and $Z_{N}^{(0)}=Z_{N,\beta=1}^{(0)}$,
 the last equality was obtained by rescaling the field $\varphi$ by $\frac1{\sqrt\beta}$,
invoking the definition \eqref{E:U(s,t)} and using that $\sum_{x\in \mathbb T_N}\nabla_i \varphi(x)=0$.
\lsm[ngdapg]{$\nu(\d\varphi)$}{$=\nu_{\beta=1}(\d\varphi)$}
\lsm[ZcNc0z]{$Z_{N}^{(0)}$}{$=Z_{N,\beta=1}^{(0)}$}
Expanding the integrand 
\begin{equation}
\prod_{x\in \mathbb T_N}\Bigl(1+\exp\bigl\{-\beta \sum_{i=1}^d U\bigl(\tfrac{1}{\sqrt \beta}\nabla_i \varphi(x),u_i\bigr)\bigr\}-1\Bigr)
\end{equation} 
above and
introducing (with a slight abuse of notation),  the function
\lsm[KcVcXcpg]{$\Kcal_{V,\beta,u}(X,\varphi)$}{$=\prod_{x\in X} \Kcal_{V,\beta,u}(\nabla\varphi(x))$}

\begin{equation}
\label{E:KVX}\Kcal_{V,\beta,u}(X,\varphi)= \prod_{x\in X} \Kcal_{V,\beta,u}(\nabla\varphi(x))
\end{equation}
for any subset $X\subset \mathbb T_N$, we get
\begin{equation}
\label{E:ZN-K}
Z_{N,\beta}(u)=Z_{N,\beta}^{(0)} \exp\bigl(-\beta L^{Nd}(\tfrac12  \abs{u}^2+ V(u))\bigr) \int_{\cbX_N}  \sum_{X \subset \T_N}
\Kcal_{V,\beta,u}(X,\varphi)
\nu(\d\varphi).
\end{equation} 

It will be useful to generalize our formulation slightly and, instead of a particular $\Kcal_{V,\beta,u}$ above,
to consider for each $u$ a general function $\Kcal_{u}:\R^d\to\R$ and define 
\lsm[ZxNcua]{$\mathcal Z_{N}(u)$}{$=\int_{\cbX_N}  \sum_X  \Kcal_{u}(X,\varphi) \nu(\d\varphi)$}
\begin{equation}
\label{E:ZcalN-K(X)}  
\mathcal Z_{N}(u)= \int_{\cbX_N}  \sum_X  \Kcal_{u}(X,\varphi) \nu(\d\varphi)
\end{equation} 
with
\lsm[KcuaXcfg]{$\Kcal_{u}(X,\varphi)$}{$=\prod_{x\in X} \Kcal_{u}(\nabla\varphi(x))$ with a function $\Kcal_{u}:\R^d\to\R$}
\begin{equation}
\label{E:K}
\Kcal_{u}(X,\varphi)=\prod_{x\in X} \Kcal_{u}(\nabla\varphi(x)).
\end{equation}
Our main claim is that, under appropriate conditions on the function
$u\mapsto \Kcal_{u}$, the perturbative component of the surface tension,
\lsm[sguax]{$\varsigma(u)$}{$=-\lim_{N\to\infty}\frac{1}{ L^{dN}}\log {\mathcal Z}_{N}(u)$, the perturbative component of the surface tension}
\begin{equation}
\varsigma(u):=-\lim_{N\to\infty}\frac{1}{ L^{dN}}\log {\mathcal Z}_{N}(u)
\end{equation}
is sufficiently smooth for small $u$.

Before formulating it in detail, we observe that whenever the claim applies to 
the case $\Kcal_{u}=\Kcal_{V,\beta,u}$, the uniform smoothness of $\varsigma(u)$ implies that, for sufficiently large $\beta$ and small $\abs{u}$,  the surface 
tension $\sigma(u)$ is strictly convex, since, in view of \eqref{E:ZN-K}, we get
\begin{equation}\label{surface}
\sigma_\beta(u)=\tfrac12  \abs{u}^2+V(u)+ \frac{ \varsigma(u)}{\beta}-\lim_{N\to\infty}\frac{1}{\beta L^{dN}}\log Z_{N,\beta}^{(0)} 
\end{equation}
The last term is a constant that does not depend on $u$.

Given any $\zeta>0$,  
\lsm[zg]{$\zeta$}{a parameter  in the exponential weight of a norm (e.g. $\norm{\cdot}_{\zeta}$) }
consider the Banach space $\bE$ of functions 
\lsm[Ed]{$\bE$}{the Banach space with the norm $\norm{\cdot}_\zeta$}
$\Kcal:\R^d\to\R$ with the norm
\lsm[zzzg]{$\norm{\cdot}_\zeta$}{$\norm{\Kcal}_\zeta=\sup_{z\in\R^d}\sum_{\abs{\boldsymbol{\alpha}}\le r_0}\zeta^{\abs{\boldsymbol{\alpha}}}\bigl\vert\partial_{z}^{\boldsymbol{\alpha}} \Kcal(z)\bigr\vert \ex^{-\zeta^{-2}\abs{z}^2}$, norm in the Banach space $\bE$}
\begin{equation}
\label{E:||h}
\norm{\Kcal}_\zeta=\sup_{z\in\R^d}\sum_{\abs{\boldsymbol{\alpha}}\le r_0}\zeta^{\abs{\boldsymbol{\alpha}}}\bigl\vert\partial_{z}^{\boldsymbol{\alpha}} \Kcal(z)\bigr\vert \ex^{-\zeta^{-2}\abs{z}^2}.
\end{equation}
Here, the sum is over nonnegative integer multiindices $\boldsymbol{\alpha}=(\alpha_1,\dots,\alpha_d)$, $\alpha_i\in\mathbb{N}, i=1,\dots,d$ 
\lsm[agx]{$\boldsymbol{\alpha}$}{$=(\alpha_1,\dots,\alpha_d)$, $\alpha_i\in\mathbb{N}, i=1,\dots,d$, a multiindex}
\lsm[agxx]{$\abs{\boldsymbol{\alpha}}$}{$=\sum_{i=1}^d\alpha_i$ (for a multiindex $\boldsymbol{\alpha}$)}
with $\abs{\boldsymbol{\alpha}}=\sum_{i=1}^d\alpha_i\le r_0\in\N$, 
\lsm[ra]{$r_0$}{a bound on the order of derivatives used in the norm $\norm{\cdot}_\zeta$}
and  $\p^{\boldsymbol{\alpha}}=\prod_{i=1}^d\p_i^{\alpha_i}$.
\lsm[dzag]{$\p^{\boldsymbol{\alpha}}$}{$=\prod_{i=1}^d\p_i^{\alpha_i}$}
We also use $B_\delta(0)\subset \R^d$ to denote the ball  
$B_\delta(0)=\{u \mid \abs{u} < \delta\}$.
\lsm[Bcdg0z]{$B_\delta(0)$}{$=\{u \in \R^d\mid \abs{u} < \delta\}$}

\begin{theorem}[\textbf{Strict convexity of the surface tension}]
\label{T:conv}
Let $ r_0\ge 9 $. There exist constants $\delta_0>0$,  $\rho_0>0$, $M_0>0$, and $\zeta_0>0$
such that if the map $\R^d\supset B_\delta(0)\ni u\mapsto \Kcal_{u}\in \bE $ is $ C^3 $, satisfies the bounds
\begin{equation}
\norm{\Kcal_{u}}_\zeta \le \rho,
\end{equation}
and \begin{equation}
\sum_{i=1}^d\Bigl\Vert \frac{\partial}{\partial u_i}\Kcal_{u}\Bigr\Vert_\zeta+
\sum_{i,j=1}^d\Bigl\Vert  \frac{\partial^2}{\partial u_i\partial u_j}\Kcal_{u}\Bigr\Vert_\zeta
+\sum_{i,j,k=1}^d\Bigl\Vert  \frac{\partial^3}{\partial u_i\partial u_j\partial u_j}\Kcal_{u}\Bigr\Vert_\zeta \le M
\end{equation}
with $\zeta\ge \zeta_0$, $\rho\le \rho_0$, $\delta\le \delta_0$, $M<M_0$, and
 $u\in B_\delta(0)$, 
then  the surface tension
$\varsigma(u)$ exists with bounds on $\varsigma(u)$, $D \varsigma(u)$, $D^2 \varsigma(u)$,  and 
$D^3 \varsigma(u)$ depending only on $\rho$ and $M$ uniformly
in $ u\in B_{\delta}(0) $. 
\end{theorem}

The proof employs a multi-scale analysis based on ideas going back  to the work \cite{BY90}.  Even though we follow quite closely the approach outlined by Brydges
in  \cite{B07},  a fair amount of various deviations and generalisations is needed.
We believe that this fact and the demands on clarity
warrant an independent treatment and the presentation of the proof in full detail. 

The reader familiar with \cite{B07} may, however, find various shortcuts.
To facilitate a selective reading, we devote the next Chapter~\ref{S:strategy} to a  presentation
of the strategy of the proof, formulating then accurately all main steps of the proof and spelling out all needed extensions of  \cite{B07}   in Chapter~\ref{S:mainsteps}. 
The proof is  then executed in full detail in the remaining chapters.
%We also include the list of symbols at the end of the paper.

Before passing to the outline of the proof, we discuss two  particular classes of perturbative potentials for which the above theorem applies.  

First we verify the assumptions of Theorem~\ref{T:conv} for a class of perturbations 
of the form \eqref{E:KVb}. This yields a very simple example of a possibly  non-convex potential at low temperatures.
\begin{prop}%[\textbf{}]
\label{P:vanish}
Let  $r_0\in\N$,  $\zeta\in(0,\infty)$,  $M_0\ge 1$,  and  suppose that 
\begin{equation}
\label{E:V1}
V\in C^{r_0+5}(\R),
\end{equation}
\begin{equation}
\label{E:V2}
V(0)=V'(0)=V''(0)=0,
\end{equation}
\begin{equation}
\label{E:V3}
\norm{D^kV}_{\infty}\le M_0\ \text{ for }\ 2\le k\le r_0+5, 
\end{equation}
and
\begin{equation}
\label{E:V4}
V(s)\ge -\tfrac18 \zeta^{-2} s^2\   \text{ for each } s\in \R.
\end{equation}
Then, for any $\rho\in(0,1/2)$,  there exists $\beta_0=\beta_0(\zeta,  \rho, M_0, r_0)$, $\delta=\delta(\zeta,  \rho, M_0, r_0)$, and $M(\zeta, M_0, r_0)$
such that, for any $\beta \ge \beta_0$, the map $\R^d\supset B_{\delta}(0)\ni u\mapsto \Kcal_{V,\beta,u}\in \bE$ is $C^3$ 
 and, for any $u\in B_{\delta}(0)$, 
\begin{equation}
\norm{\Kcal_{V,\beta,u}}_\zeta \le \rho
\end{equation}
and \begin{equation}
\sum_{i=1}^d\Bigl\Vert \frac{\partial \Kcal_{V,\beta,u}}{\partial u_i}\Bigr\Vert_\zeta+
\sum_{i,j=1}^d\Bigl\Vert  \frac{\partial^2 \Kcal_{V,\beta,u}}{\partial u_i\partial u_j}\Bigr\Vert_\zeta
+\sum_{i,j,k=1}^d\Bigl\Vert  \frac{\partial^3\Kcal_{V,\beta,u}}{\partial u_i\partial u_j\partial u_j}\Bigr\Vert_\zeta \le M.
\end{equation}
Moreover, if $ r_0\ge 9 $,  there exists  $\bar\beta(M_0)$ and $\bar\delta(M_0)$  such that  for all  $\beta\ge \bar\beta_0$,
the function $\sigma_\beta: B_{\bar\delta}(0)\to\R$ given in \eqref{surface} is $C^3$ and uniformly strictly convex. 
\end{prop}

The proof will be given in Section~\ref{S:proofexamples}.
\begin{remark}
\noindent (i)  Notice that there is no loss of generality in the assumption \eqref{E:V2}. Indeed, the absolute term is just a shift by a constant, the linear term vanishes in view 
of the condition $\sum_{x\in \mathbb T_N}\nabla_i \varphi(x)=0$, and  the quadratic term may be absorbed into the \emph{a priori} quadratic part \eqref{E:QN}.

\noindent (ii)  The only smallness assumption on $V$ is \eqref{E:V4}. In terms of the full macroscopic potential 
$W(s)=\frac12 \abs{s}^2 +V(s)$ it reads
\begin{equation}
\label{E:W}
W(s)\ge \bigl(\tfrac12 W''(0) -\tfrac18 \zeta^{-2} \bigr)s^2.
\end{equation}
Of course, the factor $\frac18$ can be replaced by any $\theta<1$. If we could (almost) achieve the optimal value for $\zeta$,
$\zeta^{-2}=\frac12$, the condition \eqref{E:W} would simply say that $W$ is bounded from below by a nondegenerate quadratic function. Due to a number of technical points, however,  we need to choose $\zeta^{-2}$ rather small to assure the validity of Theorem~\ref{T:conv}. \hfill $\diamond $
\end{remark}

\medskip

Another example is the non-convex potential considered in \cite{BK07}.
The importance of this case lies in the fact that it is a non-convex potential for which the non-uniqueness  of a Gibbs state for a particular temperature and with a particular tilt is actually proven. For the sake of simplicity, the potential considered in \cite{BK07} was chosen in a particular form that corresponds to the replacement of 
$ \exp\bigl\{-\beta H_{N}(\varphi)\bigr\}$ by
\begin{equation}
\label{E:H-BK}
\prod_{x\in \mathbb T_N}\prod_{i=1}^d \Bigr[\ip\exp\Bigl\{-\frac12  \bigl(\nabla_i\varphi(x) \bigr)^2 \Bigr\}+(1-\ip)\exp\Bigl\{-\frac\kappa2  \bigl(\nabla_i\varphi(x) \bigr)^2 \Bigr\}\Bigr]
\end{equation}
(for parameters $\kappa_{\text{\rm O}}$ and $\kappa_{\text{\rm D}}$ from \cite{BK07}
we choose $\kappa_{\text{\rm O}}=1$ and  $\kappa_{\text{\rm D}}=\kappa$).
\lsm[kg]{$\kappa$}{parameter in $\Kcal_{\kappa,p,u}$}
This amounts to replacing $\Kcal_{V,\beta,u}(z)=\exp\bigl\{-\beta \sum_{i=1}^d V\bigl(\frac{z_i}{\sqrt\beta}-u_i\bigr)\bigr\}-1$ by
\begin{equation}
\label{E:V-BK}
\Kcal_{\kappa,\ip,u}(z)=\prod_{i=1}^d \Bigr[\ip+(1-\ip)\exp\Bigl\{\frac12(1-\kappa ) \bigl(z_i-u_i)^2 \Bigr\}\Bigr]-1.
\end{equation}
\lsm[pax]{$\ip$}{parameter in $\Kcal_{\kappa,\ip,u}$ (replacing $\beta$)}%
\lsm[Kckgpaub]{$\Kcal_{\kappa,p,u}(z)$}{$=\prod_{i=1}^d \bigr[p+(1-p)\exp\bigl\{\frac12(1-\kappa ) \bigl(z_i-u_i)^2 \bigr\}\bigr]-1$, Mayer function for the potential from \cite{BK07}}
Indeed, it is enough to observe that \eqref{E:H-BK} can be rewritten as 
\begin{equation}
\exp\bigl\{-\Ecal_N(\varphi)\bigr\}
\prod_{x\in \mathbb T_N}\prod_{i=1}^d \Bigr[\ip+(1-\ip)\exp\Bigl\{-\frac12(1-\kappa)  \bigl(\nabla_i\varphi(x) \bigr)^2 \Bigr\}\Bigr].
\end{equation}
Notice that temperature $\beta$ is in \eqref{E:H-BK} and \eqref{E:V-BK} is replaced by the parameter $\ip$. The phase transition (non-unicity of Gibbs state with the tilt $u=0$) mentioned above happens, for $\kappa$ sufficiently small, for a particular value $\ip=\ip_t(\kappa)$. 
\lsm[paxa]{$\ip_t$}{$=\ip_t(\kappa)$, corresponding phase transition value}%
However, this does not prevent the corresponding surface tension
to be convex in $u$ (at least for small $\abs{u}$) once $\ip$ is sufficiently close to $1$ (and thus bigger than $\ip_t$). This corresponds to the condition of sufficiently large $\beta$ in the previous Proposition.

Observing that the map $\R^d\ni u\mapsto \Kcal_{\kappa,\ip,u}\in \bE$ is clearly analytic for all $\ip$, what only needs to be proven to apply Theorem~\ref{T:conv} is the following claim.
\begin{prop}%[\textbf{}]
\label{P:BK}
Let    $\kappa\in (0,1)$ be given. 
There exist  $\delta>0$,  $\zeta=\zeta(\delta)$ and $M$ so that
so that  for any $\abs{u}\le \delta$  one has 
\begin{equation}
\label{E:Kpcc}
\norm{\Kcal_{\kappa,\ip,u}}_\zeta \le \rho
\end{equation}
and
\begin{equation}
\label{E:Kp}
\sum_{i=1}^d\Bigl\Vert \frac{\partial}{\partial u_i}\Kcal_{\kappa,\ip,u}\Bigr\Vert_\zeta+
\sum_{i,j=1}^d\Bigl\Vert  \frac{\partial^2}{\partial u_i\partial u_j}\Kcal_{\kappa,\ip,u}\Bigr\Vert_\zeta +
\sum_{i,j,k=1}^d \Bigl\Vert  \frac{\partial^3}{\partial u_i\partial u_j\partial u_k}\Kcal_{\kappa,\ip,u}\Bigr\Vert_\zeta 
\le M
\end{equation}
for any   $1-\ip$ sufficiently small (in dependence on $\rho$ and $\zeta$).
\end{prop}
The proof is given below in Section~\ref{S:proofexamples}

\section{Proofs of the given examples}\label{S:proofexamples}
We collect the outstanding proofs for our two examples above.

\begin{proofsect}{Proof of Proposition~\ref{P:vanish}}

\medskip

\noindent\textbf{Step 1.} Estimate for $\norm{\Kcal_{V,\beta,u}}_\zeta $.

This is the key estimate.
The main idea is that for $z_i$ small (and also $u_i$ small) we can use the Taylor expansion of $U(\frac{z_i}{\sqrt{\beta}},u_i)$ in 
$z_i$, while for large $z_i$ we rely on the weight $\ex^{-\zeta^{-2}\abs{z_i}^2}$ combined with the quadratic lower bound \eqref{E:V4}
on $V$. 

First, let us show that 
\begin{equation}
\label{E:Ubig}
-\beta U(\frac{z_i}{\sqrt{\beta}},u_i)\le \frac12 \zeta^{-2} z_i^2  \ \text{ for any  } \  z_i\in\R\ \text{ and any } \  \abs{u}<\delta,
\end{equation}
whenever $\delta\le \frac1{4 M_0} \zeta^{-2}$.

Indeed,  the Taylor expansion yields
\begin{equation}
\label{E:absU}
\beta\Babs{U(\frac{z_i}{\sqrt{\beta}},u_i)}\le \tfrac12 \babs{V''(s)} z_i^2
\end{equation}
with $\abs{s}\le \abs{u_i}+ \babs{\frac{z_i}{\sqrt{\beta}}}$.
Since $V''(0)=0$ implies that $\babs{V''(s)} \le M_0 \abs{s}$, the right hand side is bounded by $\frac12 M_0\bigl(\delta +\babs{\frac{z_i}{\sqrt{\beta}}}\bigr)z_i^2$ yielding the claim for  $\babs{\frac{z_i}{\sqrt{\beta}}}\le 3 \delta$.

On the other hand, for $\babs{\frac{z_i}{\sqrt{\beta}}}\ge 3 \delta$   we use \eqref{E:V4} and the observation that $\abs{a}\ge 3\abs{b}$ implies
that $(a-b)^2\le 2a^2$ to get
\begin{equation}
-\beta V(\frac{z_i}{\sqrt{\beta}}-u_i)\le \tfrac14 \zeta^{-2} z_i^2.
\end{equation}
Moreover, expanding $V'(-u_i)$ around $V'(0)=0$ up to the order $u_i^2$,   for $\abs{\frac{z_i}{\sqrt{\beta}}}\ge 3 \delta$ we get 
\begin{equation}
\beta \Babs{V'(-u_i)\frac{z_i}{\sqrt{\beta}}}\le \beta \frac{M_0}2 \delta^2 \Babs{\frac{z_i}{\sqrt{\beta}} }\le
\frac{M_0}{6} \delta z_i^2
\end{equation}
and, similarly,
\begin{equation}
\beta \abs{V(-u_i)}\le \beta \frac{M_0}{6} \delta^3  \le
\frac{M_0}{54} \delta z_i^2,
\end{equation}
yielding the claim since $M_0(\frac16+\frac1{54})\frac1{4M_0}< \frac14$.

As a result of \eqref{E:Ubig}, we are done once 
$ \abs{z}^2=\sum_{i=1}^d z_i^2\ge 2\zeta^2 \log\frac2\rho$.
Indeed, under this assumption,  we have
\begin{equation}
\label{E:eU-1}
\babs{{\rm e}^{-\beta U(\frac{z_i}{\sqrt{\beta}},u_i)}-1}{\rm e}^{- \zeta^{-2} \abs{z}^2}  \le
\max\bigl({\rm e}^{-\beta U(\frac{z_i}{\sqrt{\beta}},u_i)},1\bigr) {\rm e}^{- \zeta^{-2} \abs{z}^2}
\le {\rm e}^{- \frac12 \zeta^{-2} \abs{z}^2}\le\frac{\rho}2.
\end{equation}

Hence, we now focus on the case 
\begin{equation}
\label{E:zsmall}
\abs{z}^2\le 2\zeta^2 \log\frac2\rho.
\end{equation}
For sufficiently small $\rho$, set 
\begin{equation}
\label{E:delta1} 
\delta_1=\frac{\zeta^{-2}}{4 M_0}\min\bigl( 1 ,  \frac{\rho }{4 \log\frac2\rho }  \bigr)\le 1
\end{equation}
and
\begin{equation} 
\label{E:beta1} 
\beta_1=\frac{2\zeta^2\log\frac2\rho}{\delta_1^2}\ge 1.
\end{equation}
Then, for   $\beta\ge \beta_1$, the relation \eqref{E:zsmall} implies that $\abs{z}/\sqrt{\beta}\le \abs{z}/\sqrt{\beta_1}
\le \delta_1$ and \eqref{E:absU} thus for $\delta\le\delta_1$  yields 
\begin{equation}
\beta\sum_{i=1}^d\Babs{U(\frac{z_i}{\sqrt{\beta}},u_i)}\le  M_0 \delta_1 \abs{z}^2\le \frac{\rho}4.
\end{equation}
Since $\abs{{\rm e}^t-1}\le2\abs{t}$ for $t\le 1$, we get
\begin{equation}
\babs{{\rm e}^{-\beta\sum_{i=1}^d U(\frac{z_i}{\sqrt{\beta}},u_i)}-1}\le  \frac{\rho}2.
\end{equation}
Together with \eqref{E:eU-1} this shows that
\begin{equation}
\sup_{z\in\R^d}\abs{{\rm e}^{-\beta\sum_{i=1}^d U(\frac{z_i}{\sqrt{\beta}},u_i)}-1}{\rm e}^{-\zeta^{-2}\abs{z}^2}\le  \frac{\rho}2
\end{equation}
as long as $\abs{u}\le \delta\le \delta_1$ and $\beta\ge \beta_1$ with $\delta_1$ and $\beta_1$ given by
\eqref{E:delta1} and \eqref{E:beta1}, respectively. 

\medskip

\noindent \textbf{Step 2.} $z$-derivatives of  $\Kcal_{V,\beta,u} $.

We will employ Fa\`a di Bruno's chain rule for higher order derivatives \cite{H} of a function in the form ${\rm e}^f$,
\begin{equation}
\label{E:BdF}
{\rm e}^{-f}\p^{\boldsymbol{\a}}{\rm e}^f= \sum_{\begin{subarray}{c}   {\boldsymbol{\tau}}_1,{\boldsymbol{\tau}}_2,\dots, m_1, m_2, \dots \\  \sum_j m_j{\boldsymbol{\tau}}_j=\boldsymbol{\alpha} \end{subarray}} 
\frac{{\boldsymbol{\alpha}}!}{({\boldsymbol{\tau}}_1!)^{m_1} ({\boldsymbol{\tau}}_2!)^{m_2}\cdots m_1!m_2!\cdots }\prod_{j}(\p^{{\boldsymbol{\tau}}_j} f)^{m_j}.
\end{equation}
Here, the sum is over distinct partitions ${\boldsymbol{\tau}}_1, {\boldsymbol{\tau}}_2, \dots$  of the multiindex ${\boldsymbol{\alpha}}$ with 
  multiplicities $m_1, m_2, \dots$ (i.e., such that $\sum_j m_j{\boldsymbol{\tau}}_j=\boldsymbol{\alpha}$) and 
  ${\boldsymbol{\tau}}!= \tau_1!\dots\tau_d!$ for any multiindex ${\boldsymbol{\tau}}=(\tau_1,\dots,\tau_d)$.
  
  In our case, we have $f(z)=-\beta\sum_{i=1}^dU(\frac{z_i}{\sqrt{\beta}}, u_i)$ with
\begin{equation}
\label{E:pf}
\p_{z_j}f(z)=- \sqrt{\beta}\bigl(V'(\frac{z_j}{\sqrt{\beta}}-u_j)-V'(-u_j)\bigr).
\end{equation}
As for  the higher derivatives, only the ``diagonal'' ones, $\p_{z_j}^k f(z)$, are non-vanishing,
\begin{equation}
\label{E:pkf}
\p_{z_j}^k f(z)=-\frac1{\beta^{(k-2)/2} }V^{(k)}(\frac{z_j}{\sqrt{\beta}}-u_j).
\end{equation}
For $\abs{u_i}\le \delta$, we get 
\begin{equation}
\label{E:p2f}
\abs{\p_{z_j}^2 f(z)}= \babs{V''(\frac{z_j}{\sqrt{\beta}}-u_j)}\le  M_0\min\bigl(1,\delta+\babs{\frac{z_j}{\sqrt{\beta}}}\bigr)
\end{equation}
and thus, using that $\p_{z_j}f(0)=0$, also
\begin{equation}
\label{E:p1f}
\abs{\p_{z_j} f(z)}\le  M_0 \min\bigl(1,\delta+\babs{\frac{z_j}{\sqrt{\beta}}}\bigr)\abs{z_j}.
\end{equation}
Moreover,  in view of \eqref{E:pkf},    we have
\begin{equation}
\label{E:k}
\sup\abs{\p_{z_j}^k f(z)}\le \frac1{\beta^{(k-2)/2} } M_0
\end{equation}
for $k\ge 2$.
Combining \eqref{E:BdF}  with   \eqref{E:k} and with the particular implication of \eqref{E:p1f}, 
\begin{equation}
\label{E:1}
\abs{\p_{z_j} f(z)}\le  M_0 \abs{z},
\end{equation}
 observing that $\abs{z}^r\le 1+\abs{z}^{r_0}$ whenever $r\le r_0$,
and using that $M_0\ge 1$ and $\beta\ge 1$, we get
\begin{equation}
\label{E:BdFbound}
\Babs{\p^{\boldsymbol{\a}} {\rm e}^{-\beta\sum_{i=1}^dU(\frac{z_i}{\sqrt{\beta}}, u_i)}}\le C(r_0) {\rm e}^{-\beta\sum_{i=1}^dU(\frac{z_i}{\sqrt{\beta}}, u_i)} M_0^{r_0} (1+\abs{z}^{r_0})
\end{equation}
with a suitable constant $C(r_0)$.
 Using, further, \eqref{E:Ubig} and \eqref{E:BdFbound}, we get (note that $\zeta\ge 1$)
\begin{equation}
\label{E:BdFboundh}
\zeta^{\abs{\boldsymbol{\a}}}\Babs{\p^{\boldsymbol{\a}} {\rm e}^{-\beta\sum_{i=1}^dU(\frac{z_i}{\sqrt{\beta}}, u_i)}}{\rm e}^{-\zeta^{-2}\abs{z}^2}\le 
\xi\, {\rm e}^{-\frac14\zeta^{-2}\abs{z}^2}
\end{equation}
 with 
\begin{equation}
\label{E:xi(r_0,h,M_0)}
\xi=\xi(r_0,h,M_0)=2 C(r_0) \zeta^{2r_0} M_0^{r_0} \bigl(\tfrac{r_0}2\bigr)^{r_0/2} .
\end{equation}
Here, the factor $\xi$ is a bound on the term $C(r_0) \zeta^{r_0} M_0^{r_0} {\rm e}^{-\frac14\zeta^{-2}\abs{z}^2} (1+\abs{z}^{r_0})$
obtained with help of  the identity $\max_{t>0} {\rm e}^{-at^2} t^s= s^s {\rm e}^{-s} a^{-s}$ with  $t=\abs{z}^2$.
As a result,  the right hand side of \eqref{E:BdFboundh}  is bounded by $\rho/2$ whenever $\abs{z}^2\ge 4 \zeta^2 \log \frac{2\xi}{\rho}$.

For 
\begin{equation}
\label{E:zxi}
\abs{z}^2\le 4 \zeta^2 \log \frac{2\xi}{\rho}
\end{equation}
we  take
\begin{equation}
\label{E:delta2} 
\delta_2=\min\bigl(\delta_1 ,  \frac{\zeta^{-2} }{4 M_0 \log\frac{2\xi}{\rho} }  \bigr)
\end{equation}
and
\begin{equation}
\label{E:beta2} 
\beta_2=\max\bigl\{\beta_1,\frac{4\zeta^2\log\frac{2\xi}{\rho}}{\delta_2^2}\bigr\}.
\end{equation}
Then, for $\beta\ge \beta_2$ and $\abs{u}\le \delta\le\delta_2$,   the bound  \eqref{E:zxi} implies that $\frac{\abs{z_i}}{\sqrt{\beta}}\le \delta_2$, yielding,
in view of  \eqref{E:p1f}, the estimate
\begin{equation}
\label{E:1delta}
\abs{\p_{z_j} f(z)}\le 2 M_0 \delta_2\abs{z_j}
\end{equation}
and thus
\begin{equation}
\abs{f(z)}\le \sum_{j=1}^d 2 M_0 \delta_2\abs{z_j}^2\le 2,
\end{equation}
again in view of \eqref{E:zxi} and the definition of $\delta_2$. Hence, similarly as in \eqref{E:BdFboundh}, we get 
\begin{multline}
\Babs{\zeta^{\abs{\boldsymbol{\a}}}\p^{\boldsymbol{\a}} {\rm e}^{-f(z)} }{\rm e}^{-\zeta^{-2}\abs{z}^2}\le 
2C(r_0) \zeta^{r_0}{\rm e}^2 (2M_0)^{r_0} \abs{z}^{r_0}{\rm e}^{-\zeta^{-2}\abs{z}^2}\max\bigl(\delta_2, \tfrac1{\sqrt{\beta}}\bigr) \le\\  
\le \overline C(r_0,M_0,h)\max\bigl(\delta_2, \tfrac1{\sqrt{\beta}}\bigr)
\end{multline}
with $\overline C(r_0,M_0,h)=2C(r_0) \zeta^{2 r_0}{\rm e}^2 (2M_0)^{r_0}\bigl(\frac{r_0}2\bigr)^{\frac{r_0}2}$.
The factor $\max\bigl(\delta_2, \frac1{\sqrt{\beta}}\bigr)$ stems from the fact that each first and second derivative of $f$ contributes a factor bounded by $2M_0\delta_2$ (cf. \eqref{E:p1f} and \eqref{E:p2f}),
while each higher derivative the factor bounded by $\frac{M_0}{\sqrt{\beta}}$ (cf. \eqref{E:k}).
Taking now
\begin{equation}
\label{E:delta0} 
\delta_0=\min\bigl(\delta_2 ,  \tfrac{\rho}{\overline C(r_0,M_0,h)}   \bigr)
\end{equation}
and 
\begin{equation}
\label{E:beta0} 
\beta_0=\max\bigl(\beta_2,\bigl(\tfrac{\overline C(r_0,M_0,h)}{\rho}\bigr)^2\bigr),
\end{equation}
we get the sought claim
\begin{equation}
\norm{\Kcal_{V,\beta,u}}_\zeta \le \rho
\end{equation}
whenever $\abs{u}\le \delta\le \delta_0$ and $\beta\ge \beta_0$.

\medskip

\noindent\textbf{Step 3.} $u$-derivatives of $ \Kcal_{V,\beta,u} $.

The estimates for the $u$-derivatives of $ \Kcal_{V,\beta,u} $ are similar. Indeed,
\begin{equation}
\p_{u_i} \Kcal_{V,\beta,u} ={\rm e}^{f(z)} \p_{i} f(z),
\end{equation}
\begin{equation}
\p_{u_j}\p_{u_i} \Kcal_{V,\beta,u} ={\rm e}^{f(z)} \bigl(f_j(z) f_i(z)- f_{i,j}(z)\bigr),
\end{equation}
etc., where
\begin{equation}
f_i(z)=-\beta\sum_{i=1}^dU^{(1)}(\frac{z_i}{\sqrt{\beta}}, u_i),
\end{equation}
\begin{equation}
f_{i,i}(z)=-\beta\sum_{i=1}^dU^{(2)}(\frac{z_i}{\sqrt{\beta}}, u_i), \ \text{ and }\  f_{i,j}(z)=0 \text{ if } i \neq j.
\end{equation}
Here, the functions $U^{(\ell)}$ have the same structure as $U$, but with $V$ replaced by $(-1)^\ell \p^{\ell} V$, e.g.,
\begin{equation}
\label{E:U1(s,t)}
U^{(1)}(s,t)=V'(s-t)-V'(-t)-V''(-t)s.
\end{equation}
Thus, as in \eqref{E:k} and \eqref{E:1}, we get
\begin{equation}
\label{E:Uk}
\beta\sup\abs{\p_{z_j}^k U^{(\ell)}(\frac{z_i}{\sqrt{\beta}}, u_i)}\le \sup\abs{\p^{k+\ell} V}\le M_0
\end{equation}
and
\begin{equation}
\label{E:U1}
\beta\abs{\p_{z_j} U^{(\ell)}(\frac{z_i}{\sqrt{\beta}}, u_i)}\le  \sup\abs{\p^{2+\ell} V} \abs{z_i}\le M_0\abs{z_i}.
\end{equation}
In addition, we have a new estimate
\begin{equation}
\label{E:Unew}
\beta\abs{U^{(\ell)}(\frac{z_i}{\sqrt{\beta}}, u_i)}\le  \sup\abs{\p^{2+\ell} V} \abs{z_i}^2\le M_0\abs{z_i}^2.
\end{equation}
Thus, for $\abs{\boldsymbol{\beta}} \in\{1,2,3\}, \abs{\boldsymbol{\alpha}}\in\{0,\dots,r_0\}$,
\begin{equation}
\beta\babs{\p_{z}^{\boldsymbol{\alpha}}\p_u^{\boldsymbol{\beta}} {\rm e}^{-\beta\sum_{i=1}^dU(\frac{z_i}{\sqrt{\beta}}, u_i)}}
\le C(r_0) {\rm e}^{f(z)}  M_0^{\abs{\boldsymbol{\alpha}}+\abs{\boldsymbol{\beta}} }(1+\abs{z}^2)^{\abs{\boldsymbol{\alpha}}+\abs{\boldsymbol{\beta}} }.
\end{equation}
Estimate \eqref{E:Ubig} yields $\abs{f(z)} \le \frac12 \zeta^{-2} \abs{z}^2$ if $\abs{u}\le \frac1{4M_0}\zeta^{-2}$ (in particular if $\abs{u}<\delta_0$ defined in 
\eqref{E:delta0}). 
Then we easily conclude that 
\begin{equation}
\Bigl\Vert\partial^{\boldsymbol{\beta}} \Kcal_{V,\beta,u}\Bigr\Vert_\zeta\le M(r_0,h,M_0) \ \text{ for any }\ 
\abs{\boldsymbol{\beta}}\in\{1,2,3\}, \abs{u}\le \delta_0, \text{ and } \beta\ge \beta_0,
\end{equation}
with a suitable $M(r_0,h,M_0) $.

\medskip

\noindent \textbf{Step 4.} Uniform convexity of $\sigma(u)$.

To obtain uniform convexity of $\sigma(u)$, we first fix $\rho$ so small and $r_0$ and $\zeta$ so large that Theorem~\ref{T:conv} applies. 
Then for $\beta\ge \beta_0$ and  $\abs{u}<\delta_0$ we find that $\varsigma(u)$ is a $C^3$ function and its first three derivatives in $B_{\delta_0}(0)$ are controlled in terms of $\rho$ and $M=M(r_0,h,M_0)$. In particular, 
\begin{equation}
\abs{D^2 \varsigma(u)}\le M'(\zeta,M_0,\rho) \ \text{ if }\ u\in B_{\delta_0}(0).
\end{equation}
Note that for $\abs{s}\le \frac1{4M_0}$, we have $V''(s)\ge -\frac14$. Let
\begin{equation}
\bar\delta(M_0) =\min\bigl(\delta_0(\zeta,M_0,\rho, r_0),  \frac1{4M_0}\bigr)
\end{equation}
and
\begin{equation}
\bar\beta(M_0) =\max\bigl(\beta_0(\zeta,M_0,\rho, r_0),  \frac1{4M'(\zeta,M_0,r_0)}\bigr).
\end{equation}
Then
\begin{equation}
D^2 \sigma(u)\ge \Id - \tfrac14 \Id  - \tfrac14 \Id \ge   \tfrac12 \Id 
\end{equation}
for $u\in B_{\bar\delta}(0)$ and $\beta\ge \bar\beta$.

\qed
\end{proofsect}

%$\mathrm{p}\  \mathfrak{p}\ \mathit{p} p$

\bigskip

\begin{proofsect}{Proof of Proposition~\ref{P:BK}}
The proof is similar as the proof of  Proposition~\ref{P:vanish}. 

We will only indicate the main steps.
Again, skipping the indices in  $\Kcal_{\kappa,\ip,u}$  and rewriting
\begin{equation}
\Kcal(z)=\prod_{i=1}^d \Biggr[1+(1-\ip)\Bigr[\exp\Bigl\{\frac12(1-\kappa ) \bigl(z_i-u_i)^2 \Bigr\}-1\Bigr]\Biggr]-1,
\end{equation}
we have 
\begin{equation}
0\le \Kcal(z)\le 2^d (1-\ip)\exp\Bigl\{\frac12 \sum_{i=1}^d \bigl(z_i-u_i)^2 \Bigr\},
\end{equation}
and, with  suitable polynomials $P_{\boldsymbol{\alpha}}(z-u)$, also
\begin{equation}
\abs{\nabla^{\boldsymbol{\alpha}}\Kcal(z)}\le (1-p) P_{\boldsymbol{\alpha}}(z-u)  \exp\Bigl\{\frac12 \sum_{i=1}^d \bigl(z_i-u_i)^2 \Bigr\}.
\end{equation}

Taking now sufficiently small $u$ and, then, sufficiently large $\zeta$ we have  

$$
\norm{\Kcal}_\zeta\le
C (1-p)
$$ 
with the constant $C$ depending on $\zeta$. Similar bounds are valid for the remaining terms in \eqref{E:Kp}.
\qed
\end{proofsect}

%% file: AKM-ch3-strategy-30thJune-2016.tex
\chapter{The Strategy of the Proof}
\label{S:strategy}

Here we present, in rather broad brush, the main ideas of the proof. 
Accurate definitions of the needed notions then follow in the succeeding  chapter.

As mentioned above, to verify the claim of the theorem, we need to prove that the finite volume perturbative component of the surface tension  
\lsm[sguaxN]{$\varsigma_N(u)$}{$=-\frac{1}{ L^{dN}}\log {\mathcal Z}_{N}(u)$, the finite volume perturbative component of the surface tension}
\begin{equation}
\varsigma_N(u):=-\frac{1}{ L^{dN}}\log {\mathcal Z}_{N}(u)
\end{equation}
has bounded derivatives  uniformly in $ N \in\N $.

Here, the partition function ${\mathcal Z}_{N}(u)$ can be expressed, 
with a flavour of cluster expansions, in terms of the  functions  $\Kcal(X,\varphi)=\Kcal_u(X,\varphi)$ as shown in   \eqref{E:ZcalN-K(X)}.
However, here comes a difficulty: even though the function $\Kcal(X,\varphi)$ depends only on $\varphi(x)$
with $x$ in the set $X$ and its close neighbourhood and even if for  a disjoint union $X=X_1\cup X_2$
one has  $\Kcal(X,\varphi)=\Kcal(X_1,\varphi)\Kcal(X_2,\varphi)$, the Gaussian measure 
$\nu(\d\varphi)$ with its slowly decaying correlations does not allow to separate
the integral of $\Kcal(X,\varphi)$  into a product of integrals with 
the integrands $\Kcal(X_1,\varphi)$  and $\Kcal(X_2,\varphi)$. This is a non-locality that has to be overcome.

The strategy is to perform the integration in steps corresponding to increasing scales. 
Before showing what we mean by that, let us make one simple modification. Its importance will be in providing a parameter that will allow us to fine-tune the procedure in such a way that the final integration will eventually yield a result with a straightforward bound. 

The parameter in question will be chosen as a   \emph{symmetric $d\times d $-matrix} $ \bq\in\R^{d\times d}_{\rm sym} $. 
\lsm[qb]{$ \bq$}{a symmetric $d\times d $-matrix}
\lsm[Rxdazx]{$\R^{d\times d}_{\rm sym} $}{the set of symmetric $d\times d $-matrices}
Multiplying and dividing the integrand in \eqref{E:ZcalN-K(X)} by
\begin{equation}
\exp\Bigl\{-\tfrac12
\sum_{x\in \mathbb T_N}\sum_{i,j=1}^d q_{i,j}\nabla_i\varphi(x) \nabla_j\varphi(x)\Bigr\}=\exp\Bigl\{-\tfrac12\sum_{x\in \mathbb T_N}\langle \bq\nabla\varphi(x), \nabla\varphi(x)\rangle\Bigr\}
\end{equation}
and using the definition of the measure $\nu$ (by \eqref{E:nubeta} with $\beta=1$), we get
\begin{equation}
\label{E:ZcalN-mu}
\mathcal Z_{N}(u)=
\frac{Z_{N}^{(\bq)}}{Z_{N}^{(0)}}
\int_{\cbX_N}\exp\Bigl\{ - \tfrac12\sum_{x\in \mathbb T_N}\langle \bq\nabla\varphi(x), \nabla\varphi(x)\rangle \Bigr\}\sum_X  \Kcal(X,\varphi)\mu^{(\bq)}(\d\varphi).
\end{equation} 
Here, $\mu^{(\bq)}$ is the Gaussian measure on $ \cbX_N$ with the Green function ${\Cscr}^{(\bq)} $, 
\lsm[Ccqb]{${\Cscr}^{(\bq)} $}{the inverse of the operator ${\Ascr}^{(\bq)}$}
the inverse of the operator ${\Ascr}^{(\bq)}=\sum_{i,j=1}^d \bigl(\delta_{i,j}-q_{i,j}\bigr)\nabla_i^*\nabla_j$, 
\lsm[Acqb]{${\Ascr}^{(\bq)}$}{$=\sum_{i,j=1}^d \bigl(\delta_{i,j}+q_{i,j}\bigr)\nabla_i^*\nabla_j$}
\begin{equation}
\label{E:mu}
\mu^{(\bq)}(\d\varphi)=\frac{\exp\bigl\{-{\mathcal E}_{\bq}(\varphi)\bigr\}\lambda_N(\d\varphi)}
{Z_{N}^{(\bq)}},
\end{equation} 
\lsm[mgdapg]{$\mu^{(\bq)}(\d\varphi)$}{$=\frac1{Z_{N}^{(\bq)}}\exp\bigl\{-{\mathcal E}_{\bq}(\varphi)\bigr\}\lambda_N(\d\varphi)$}
with 
\begin{equation}
\label{E:Eq}
{\mathcal E}_{\bq}(\varphi)=\tfrac12 ({\Ascr}^{(\bq)}\varphi,\varphi)=\tfrac12
\sum_{x\in \mathbb T_N}\sum_{i,j=1}^d \bigl(\delta_{i,j}-q_{i,j}\bigr)\nabla_i\varphi(x) \nabla_j\varphi(x),
\end{equation}
\lsm[Eyqbpg]{${\mathcal E}_{\bq}(\varphi)$}{$=\tfrac12 ({\Ascr}^{(\bq)}\varphi,\varphi)=\tfrac12\sum_{x\in \mathbb T_N}\sum_{i,j=1}^d \bigl(\delta_{i,j}+q_{i,j}\bigr)\nabla_i\varphi(x) \nabla_j\varphi(x)$}
and
\begin{equation}
\label{E:ZNq}
Z_{N}^{(\bq)}=  \int_{\cbX_N}\exp\bigl\{-{\mathcal E}_{\bq}(\varphi)\bigr\}\lambda_N(\d\varphi).
\end{equation} 
\lsm[ZcNcqb]{$Z_{N}^{(\bq)}$}{$=\int_{\cbX_N}\exp\bigl\{-{\mathcal E}_{\bq}(\varphi)\bigr\}\lambda_N(\d\varphi)$}

Under a suitable assumption about the smallness of  $\bq$ (so that, in particular, the matrix 
$\boldsymbol 1 - \bq$
is positive definite), we will show that the Gaussian 
measure $\mu^{(\bq)}$ can be decomposed into a convolution $\mu^{(\bq)}(\d\varphi)= \mu^{(\bq)}_1\ast\dots\ast\mu^{(\bq)}_{N+1}(\d\varphi)$ where $\mu^{(\bq)}_1,\dots,\mu^{(\bq)}_{N+1}$ are
Gaussian measures with a particular finite range property.
Namely, the covariances ${\mathcal C}^{(\bq)}_{k}(x)$ of the measures $\mu^{(\bq)}_k$, $k=1,\dots,N+1$, 
\lsm[mgka]{$\mu^{(\bq)}_k(\d\varphi)$}{Gaussian measure with covariance $\Cscr^{(\bq)}_k$}
vanish for
$\abs{x}\ge \frac{1}{2} L^k$ with a fixed parameter $L$ with an additional bound on their derivatives with respect to $\bq$ of the order $L^{-(k-1)(d-1)}$.
(See next Chapter for careful definitions and exact formulations; here we concentrate just on the main ideas.)

Now, let us write  the integral in \eqref{E:ZcalN-mu} symbolically as 
\begin{equation}
\label{E:HqcircKq}
\int_{\cbX_N}(\ex^{-H^{(\bq)}}\circ \Kcal^{(\bq)})(\varphi)\mu^{(\bq)}(\d\varphi).
\end{equation}
 Here 
 % \lsm[Hcqb]{$H^{(\bq)}(X, \varphi)$}{$=-\tfrac12\sum_{x\in X}\sum_{i,j=1}^d q_{i,j}\nabla_i\varphi(x) \nabla_j\varphi(x)$}
\begin{equation}
\label{E:Hq}
H^{(\bq)}(X,\varphi)=\tfrac12
\sum_{x\in X}\sum_{i,j=1}^d q_{i,j}\nabla_i\varphi(x) \nabla_j\varphi(x)=
\tfrac12\sum_{x\in X}\langle\bq\nabla\varphi(x), \nabla\varphi(x)\rangle,
\end{equation}
the function  $\Kcal^{(\bq)}$ is defined as 
\lsm[KcqbXcpg]{$\Kcal^{(\bq)}(X,\varphi)$}{$=\exp\Bigl\{\tfrac12
\sum_{x\in X}\sum_{i,j=1}^d q_{i,j}\nabla_i\varphi(x) \nabla_j\varphi(x)\Bigr\}  \Kcal(X,\varphi)$}
\begin{equation} 
\label{E:Kq}
\Kcal^{(\bq)}(X,\varphi)=\exp\Bigl\{-\tfrac12\sum_{x \in X}\langle\bq\nabla\varphi(x), \nabla\varphi(x)\rangle\Bigr\}  \Kcal(X,\varphi),
\end{equation}
and $\circ$ is the  \emph{circle product} notation for the convolutive sum over subsets $X\subset\mathbb T_N$,
\begin{equation} 
(\ex^{-H^{(\bq)}}\circ \Kcal^{(\bq)})(\varphi)=
  \sum_{X \subset \mathbb T_N}  \ex^{-H^{(\bq)}(\T_N\setminus X, \varphi)} \, \,   \Kcal^{(\bq)}(X,\varphi), \end{equation}
where we set $H^{(\bq)}(\emptyset, \varphi) = \Kcal^{(\bq)}(\emptyset,\varphi) =1$.

Replacing $\mu^{(\bq)}$ in \eqref{E:HqcircKq} by the convolution $ \mu^{(\bq)}_1\ast\dots\ast\mu^{(\bq)}_{N+1}(\d\varphi)$,
we will proceed by integrating first over $\mu^{(\bq)}_1$. It turns out that the form of the integral is conserved.
Namely, starting from $H^{(\bq)}_0=H^{(\bq)}$ and $K^{(\bq)}_0=\Kcal^{(\bq)}$, we can define $H^{(\bq)}_1$ and $K^{(\bq)}_1$ so that 
\begin{equation}
\int_{\cbX_N}(\ex^{-H^{(\bq)}_0}\circ K^{(\bq)}_0)(\varphi+\xi)\mu^{(\bq)}_1(\d\xi)= (\ex^{-H^{(\bq)}_1}\circ K^{(\bq)}_1)(\varphi).
\end{equation}
Here, the function $K^{(\bq)}_1(X,\varphi)$ is defined (nonvanishing) only for sets $X$ consisting of $L^d$-blocks
and $H^{(\bq)}_1$ is again a quadratic form like $H^{(\bq)}_0$ but with modified coefficients $q_{i,j}$ and additional linear
and constant  terms. 
Recursively, one can define a sequence of pairs $(H^{(\bq)}_1,K^{(\bq)}_1), (H^{(\bq)}_2,K^{(\bq)}_2), \dots, (H^{(\bq)}_{N},K^{(\bq)}_{N})$ 
with each $H^{(\bq)}_k$ 
a quadratic form in $\nabla\varphi$  (plus linear and constant terms) and $K^{(\bq)}_k(X,\varphi)$ defined for sets
$X$ consisting of $L^{kd}$-blocks
so that
\begin{equation}
\label{E:ktok+1}
\int_{\cbX_N}(\ex^{-H^{(\bq)}_k}\circ K^{(\bq)}_k)(\varphi+\xi)\mu^{(\bq)}_{k+1}(\d\xi)= (\ex^{-H^{(\bq)}_{k+1}}\circ K^{(\bq)}_{k+1})(\varphi).
\end{equation}

Of course, the difficulty lies  in producing correct definitions of consecutive  pairs of functions $H^{(\bq)}_k,K^{(\bq)}_k$ so that not only \eqref{E:ktok+1} is valid, but also that the form of the quadratic function $H_k$ is conserved, the coarse-grained dependence of $K^{(\bq)}_k$ on blocks $L^{dk}$ is maintained, and, most importantly, the size of the perturbation $K^{(\bq)}_k$ in a conveniently chosen norm decreases (the variable $K^{(\bq)}_k$ is \emph{irrelevant} in the language of the renormalisation group theory). 
See Propositions~\ref{P:k->k+1}-\ref{P:Tnonlin} for an explicit form and properties of 
the renormalisation transformation $\bT^{(\bq)}_k\colon (H^{(\bq)}_k,K^{(\bq)}_k)\mapsto (H^{(\bq)}_{k+1},K^{(\bq)}_{k+1})$.

 Using now sequentially the formula \eqref{E:ktok+1}, we eventually get 
\begin{equation}
\int_{\cbX_N}(\ex^{-H^{(\bq)}_0}\circ K^{(\bq)}_0)(\varphi)\mu^{(\bq)}(\d\varphi)=
\int_{\cbX_N}(\ex^{-H^{(\bq)}_{N}}\circ K^{(\bq)}_{N})(\varphi)\mu^{(\bq)}_{N+1}(\d\varphi)
\end{equation}
and thus
\begin{equation}
\mathcal Z_{N}(u)=
\frac{Z_{N}^{(\bq)}}{Z_{N}^{(0)}}
\int_{\cbX_N}(\ex^{-H^{(\bq)}_{N}}\circ K^{(\bq)}_{N})(\varphi)\mu^{(\bq)}_{N+1}(\d\varphi).
\end{equation}

At this moment we will invoke an additional feature. Namely, the finite range decomposition can be constructed in such a way that the measures $ \mu_1^{\ssup{\bq}},\ldots,\mu_{N+1}^{\ssup{\bq}} $ depend smoothly on $ \bq $ (\cite{AKM09b}).  
As a result it turns out that, in dependence on  the original perturbation $\Kcal_u$ (or on
$V$, $\beta$, and $u$ in the explicit choice of $\Kcal_u$ as in  \eqref{E:KVb}),
one can choose the initial value $\bq=\bq(\Kcal_u)$ by an implicit function theorem in such a way that $H^{(\bq)}_N=0$.
\lsm[qbz0]{$\bq(\Kcal_u)$}{the value of $\bq$ yielding $H_N=0$}

However, here we encounter a  difficulty stemming from the fact that the action of $ T_k^{\ssup{\bq}} $, considered on a scale of function spaces, depends on $\bq$ with certain loss of regularity, see Chapter~\ref{S:smooth}. 
This leads to a need for  employing a  suitable version of implicit function theorem as well as a theorem about chain rule for composed  maps with loss of regularity (see Appendices~\ref{appChain} and \ref{appIFT} for the definitions and proofs).

Also,
the ``starting'' Hamiltonian $H^{(\bq)}_0$ will in general contain, 
 in addition to the quadratic term given by  \eqref{E:Hq},  also linear and 
constant terms, i.e., $H^{(\bq)}_0(X, \varphi) = \sum_{x \in X} \Hcal(x, \varphi)$ with 
\begin{equation} \label{E:3_starting_full}
\Hcal(x, \varphi) = \lambda + \sum_{i=1}^da_i\nabla\varphi(x)+\sum_{i,j=1}^d\bc_{i,j}\nabla_i\nabla_j\varphi(x) +
 \frac{1}{2}\sum_{i,j=1}^d\bq_{i,j}\nabla\varphi(x)\nabla_j\varphi(x),
\end{equation}
see \eqref{E:4_ideal_Hamiltonian} and \eqref{ideal}. Note, however,  that the constant and linear terms do not lead to a change of the measure $\mu^{(\bq)}$ since
by periodicity of $\varphi$ we have 
$ \sum_{x \in \T_N} \nabla_i \varphi(x) = 0$ and $ \sum_{x \in \T_N} \nabla_i  \nabla_j \varphi(x) = 0$.
For the purpose of this broad outline of the proof we will pretend that we can achieve $H^{(\bq)}_N=0$
with the choice
$$ \lambda = a = \bc = 0.$$
The general situation will be discussed in Chapter \ref{S:initial conditions} below. 

Finally, taking into account that the function $K^{(\bq)}_N(X,\cdot)$ is defined only for $X=\L_N$
or $X=\emptyset$, we  get
\begin{equation}
\label{E:ZNfinal}
\mathcal Z_{N}(u)= 
\frac{Z_{N}^{(\bq)}}{Z_{N}^{(0)}}
\int_{\cbX_N}\bigl(1+ K^{(\bq)}_{N}(\L_N,\varphi)\bigr)\mu^{(\bq)}_{N+1}(\d\varphi),
\end{equation}
with $\bq$ being implicitly dependent on $\Kcal = \Kcal_u$  by the condition that the iteration described above gives $H^{(\bq)}_N = 0$. 
%Notice that $\mu^{(\bq)}_{N+1}$  as well  as $K^{(\bq)}_{N}$   depend on $\bq$.
Note that this formula was derived under the assumption that the constant term $\lambda$ in the initial perturbation is zero.
In general, there is an additional term  depending on $\lambda$, see \eqref{E:4_formula_logZ} or \eqref{finalfreeenergy}.

Now, to get the sought smoothness with respect to $u$, we have to evaluate the derivatives with respect to $\bq$ and show the smooth dependence of implicitly defined $\bq$ as function of $u$. 
The smoothness with respect to $\bq$ is quite straightforward as the factor $Z_{N}^{(\bq)}$ can be explicitly computed by Gaussian integration and the derivatives of the integral term can easily be bounded  as a consequence of the iterative bounds on $K^{(\bq)}_N$. The smoothness of $\bq$ as function of $u$  follows by a careful examination of the corresponding implicit function yielding $\bq$ as function of the initial perturbation $\Kcal_u$ and by smoothness of 
$\Kcal_u$ as function of $u$ assumed in Theorem~\ref{T:conv} and proven for 
 the particular classes of potentials considered
in Propositions~\ref{P:vanish} and \ref{P:BK},
see Chapter \ref{sec:convexity}.

%% file: AKM-ch4-detailed-30thJune-2016.tex
\chapter{Detailed Setting of the Main Steps}\label{S:mainsteps}

\section{Finite range decomposition.}\label{SS:FRD}
First, we formulate the needed claim about the \emph{finite range decomposition} of the Green function ${\Cscr}^{(\bq)} $, the inverse of the operator ${\Ascr}^{(\bq)}=\sum_{i,j=1}^d \bigl(\delta_{i,j}-q_{i,j}\bigr)\nabla_i^*\nabla_j$ on $ \cbX_N$. We use 
$\norm{\bq}$ to denote the operator norm of $\bq$ viewed as operator on $\R^d$ equipped
with $\ell_2$ metric. Obviously, $\norm{\bq}\le \bigl(\sum_{i,j} q_{i,j}^2\bigr)^{1/2}$.
\lsm[qbz0z]{$\norm{\bq}$}{operator norm of $\bq$ viewed as operator on $\R^d$ equipped
with $\ell_2$ metric}

\begin{prop}
\label{P:FRD}
Let $ \bq\in\R^{d\times d}_{\rm sym} $ be a  symmetric $d\times d $-matrix 
such that $\norm{\bq}\le \tfrac12$.
% and thus $\boldsymbol 1+\bq$ is positive definite. 
There exist positive definite operators  ${\Cscr}^{(\bq)}_{k}$, $k=1,\dots, N+1$, on  $ \cbX_N$ such that
\lsm[Ccqbka]{${\Cscr}^{(\bq)}_{k}$}{finite range covariance operator}
\begin{equation}
\label{E:FRD}
{\Cscr}^{(\bq)}=\sum_{k=1}^{N+1} {\Cscr}^{(\bq)}_{k}.
\end{equation}
The operators ${\Cscr}^{(\bq)}_{k}$ commute with  translations on $\T_N$. In particular, there exists a function  ${\mathcal C}^{(\bq)}_{k}$ on $\T_N$ such that
$\bigl({\Cscr}^{(\bq)}_{k}\varphi\bigr)(x)=\sum_{y\in\T_N}{\mathcal C}^{(\bq)}_{k}(x-y)\varphi(y)$ for each $\varphi\in  \cbX_N$. Moreover,
\begin{equation}
\label{E:FRk}
{\mathcal C}^{(\bq)}_{k}(x) =0 \;\mbox{ if }\; \abs{x}_\infty\ge \frac{1}{2} L^k
\end{equation}
and,  
\lsm[Cxqbka]{${\mathcal C}^{(\bq)}_{k}$}{finite range covariance function}
for each multiindex $\boldsymbol{\alpha}$ 
with $\abs{\boldsymbol{\alpha}}\le 3$ and any 
$a\in\N_0$ there exists a constant
$c_{\boldsymbol{\alpha},a}$ such that
\lsm[caagaa]{$c_{\boldsymbol{\alpha},a}$}{coefficients in bounds of derivatives of finite range covariance function}%
\begin{equation}
\label{E:fluctk}
\sup_{\norm{\bq}\le \frac{1}{2}}\abs{\nabla^{\boldsymbol{\alpha}} D^{a}{\mathcal C}^{(\bq)}_{k}(x)(\dot{\bq},\dots,\dot{\bq})}\le c_{\boldsymbol{\alpha},a} L^{-(k-1)(d-2+\abs{\boldsymbol{\alpha}})}   L^{\upeta(|\boldsymbol{\alpha}|,d)} \norm{\dot{\bq}}^a
%\sup_{\norm{\dot{\bq}}\le \frac{1}{2}}\bigl\vert(\nabla^{\boldsymbol{\alpha}} D^{a}{\mathcal C}^{(\bq)}_{k})(x)(\dot{\bq},\dots,\dot{\bq})\bigr\vert\le c_{\boldsymbol{\alpha},a} L^{-(k-1)(d-2+\abs{\boldsymbol{\alpha}})}
%L^{\upeta(|\boldsymbol{\alpha}|,d)}
\end{equation}
\lsm[hgnadax]{$\upeta(n,d)$}{$=   \max(\tfrac14 (d+n-1)^2, d+n+6)+10$, the decay exponent in finite range decomposition}%
for all $x\in\T_N$ and all $ k=1,\ldots, N+1$, with 
\begin{equation}  \label{E:def_eta} 
 \upeta(n,d) =   \max(\tfrac14 (d+n-1)^2, d+n+6)+10 .
 \end{equation}
Here, $\nabla^{\boldsymbol{\alpha}}=\prod_{i=1}^d\nabla_i^{\alpha_i}$
and $D$ is the directional derivative in the direction $\dot{\bq}$.
\end{prop}

The proof can be found  in \cite{AKM09b} which is an extension of ideas in \cite{BT06} and \cite{BGM04}  applied to families of gradient Gaussian measures including vector valued functions. In fact there it is shown
that $\Cscr_k^{(\bq)}$ is (real) analytic in $\bq$ with the natural estimates for all derivatives with 
respect to $\bq$.

\begin{remark}  \label{R:FRD} Since the $\Cscr_k^{(\bq)}$ are translation invariant they are diagonal in the Fourier
basis given by $f_p(x) = L^{-dN/2} \ex^{i \langle p, x \rangle }$ with
\lsm[fapaxa]{$f_p(x)$}{$= L^{-dN/2} e^{i \langle p, x \rangle }$, Fourier basis functions}%
\begin{equation}
p \in \widehat{\T}_N=\Bigl\{p=(p_1,\dots, p_d)\colon p_i\in\bigl\{-\tfrac{(L^N-1)\pi}{L^N}, -\tfrac{(L^N-3)\pi}{L^N}
\ldots,0,\dots, \tfrac{(L^N-1)\pi}{L^N}  \bigr\} \Bigr\},
\end{equation}
\lsm[pa]{$p$}{$=(p_1,\dots, p_d) \in \widehat{\T}_N$, dual variables}%
\lsm[Txx]{$ \widehat{\T}_N$}{$=\bigl\{p=(p_1,\dots, p_d): p_i\in \{-\frac{(L^N-1)\pi}{L^N}, -\frac{(L^N-3)\pi}{L^N}
\ldots,0,\dots, \frac{(L^N-1)\pi}{L^N}\} \bigr\}$, dual torus}%
 i.e., 
\begin{equation}
\Cscr_k^{(\bq)} f_p = \widehat{\mathcal C}_k^{(\bq)}(p) f_p, 
\end{equation}
where the Fourier multiplier $ \widehat{\mathcal C}_k^{(\bq)}(p) $
\lsm[Cxqbkax]{$\widehat{\mathcal C}_k^{(\bq)}(p)$}{discrete Fourier transform of the kernel ${\mathcal C}_k^{(\bq)}$}
 is just the
discrete Fourier transform of the kernel ${\mathcal C}_k^{(\bq)}$.
Equation (4.62) and Lemma 4.3 in \cite{AKM09b} yield
\begin{equation}  \label{E:Fourier_estimate_FRD}
\frac{1}{L^{dN} }  \sum_{p\in\widehat{\T}_N\setminus\{0\}}      |p|^n  \, | D_q^a \widehat{\mathcal C}_k^{(\bq)}(p)(\dot{\bq}, \ldots, \dot{\bq})|
\leq  2^a a!   \, c(n,d)  L^{\upeta(n,d)}  L^{-(k-1) (d+n - 2)} . 
\end{equation}
This estimate implies   \eqref{E:fluctk} by the discrete Fourier inversion formula, but it will
also be of independent use later.  \hfill $\diamond$
\end{remark}

Now, if a \emph{random field} $ \varphi $ is distributed with respect to the Gaussian measure $\mu^{(\bq)}= \mu_{{\Cscr}^{(\bq)} } $ on $ \cbX_N$, where the \emph{covariance} $ {\Cscr}^{(\bq)} $  admits a finite range decomposition \eqref{E:FRD}, then there exist $N+1$  independent random fields $ \xi_k $, $k=1,\dots,N+1$, such that  each $ \xi_k $ is 
\lsm[xgka]{$ \xi_k $}{a random field distributed according to $ \mu_k=\mu_{{\Cscr}^{(\bq)}_{k}} $}
distributed according to the Gaussian measure $ \mu^{(\bq)}_k=\mu_{{\Cscr}^{(\bq)}_{k}} $ with the covariance $ {\Cscr}^{(\bq)}_{k} $ and, in distribution,
\begin{equation}
\varphi=\sum_{k=1}^{N+1}\xi_k ,
\end{equation}
or,
\begin{equation}
\int_{{\cbX_{N}}}F(\varphi)\mu^{(\bq)}(\d\varphi)=\E_{N+1}\cdots\E_1 F,
\end{equation}
where $ \E_k, k=1,\ldots, N+1 $, denote the expectations with respect to the Gaussian measures $ \mu^{(\bq)}_k $
and $F$ is taken as a function of  $\sum_{k=1}^{N+1}\xi_k $.
\lsm[Eyka]{$ \E_k$}{expectation with respect to  $ \mu_k=\mu_{{\Cscr}^{(\bq)}_{k}} $}

Taking into account that operators   ${\Cscr}^{(\bq)}_{k}$ are of full rank  on  $\mathcal X_N$, standard Gaussian calculus yields an expression  in terms of convolutions,
\begin{multline}
\int_{ \cbX_N} F(\varphi)\mu^{(\bq)}(\d\varphi)= \int_{ \cbX_N} F(\varphi) \mu^{(\bq)}_1\ast\dots\ast\mu^{(\bq)}_{N+1}(\d\varphi)= \\=\int_{ \cbX_N\times\dots \times  \cbX_N}F\Bigl(\sum_{k=1}^{N+1}\xi_k\Bigr)
\mu^{(\bq)}_1(\d\xi_1)\dots \mu^{(\bq)}_{N+1}(\d\xi_{N+1}).
\end{multline}

Our preferred formulation is to introduce renormalisation maps $\bR^{(\bq)}_k$ on functions on $\cbX_N$  by
\lsm[Rdka]{$\bR_k$}{renormalisation maps $(\bR_k F)(\varphi)= \int_{ \cbX_N}F(\varphi+\xi)\mu^{(\bq)}_k(\d\xi)$}%
\begin{equation}
\label{E:Rk}
(\bR^{(\bq)}_k F)(\varphi)= \int_{ \cbX_N}F(\varphi+\xi)\mu^{(\bq)}_k(\d\xi), k=1,\dots, N.
\end{equation}
Just to be on a firm ground, we can introduce the spaces $M(\cbX_N)$ of all 
\lsm[McXdNc]{$M(\cbX_N)$}{set of all functions on $\cbX_N$ measurable with respect to $\lambda_N$}
\lsm[ngka]{$\nu^{(\bq)}_k$}{the measure  on $\cbX_N$  with covariance ${\Cscr}^{(\bq)}_k+\dots+{\Cscr}^{(\bq)}_{N+1}$}
functions measurable with respect to $\lambda_N $  on $\cbX_N$
and view $\bR^{(\bq)}_k$ as a map
$\bR^{(\bq)}_k\colon \Ucal\subset M(\cbX_N)\to M(\cbX_N)$, where 
$$
\Ucal=\{F\colon\cbX_N\to\R\colon \mbox{r.h.s of \eqref{E:Rk} exists and is finite}\}. 
$$

The integration $\int_{ \cbX_N} F(\varphi)\mu^{(\bq)}(\d\varphi)$ can be viewed, for any
$F\in  M(\cbX_N)$,
as the consecutive application of maps $\bR^{(\bq)}_k$ with a final integration with respect to $ \mu^{(\bq)}_{N+1}$:
\begin{equation}
\int_{ \cbX_N} F(\varphi)\mu^{(\bq)}(\d\varphi)= \int_{ \cbX_N} (\bR^{(\bq)}_N\dots\bR^{(\bq)}_1 F)(\varphi) \mu^{(\bq)}_{N+1}(\d\varphi).
\end{equation}
Notice that for the operators  ${\Cscr}^{(\bq)}_{N}$ and ${\Cscr}^{(\bq)}_{N+1}$ (and the measures  $ \mu^{(\bq)}_{N}$  and  $ \mu^{(\bq)}_{N+1}$)
the condition \eqref{E:FRk} is void. However, the suppression condition  \eqref{E:fluctk}
still applies.

\section{Polymers, polymer functionals, ideal Hamiltonians and norms.}\label{S:norms}

There is a natural hierarchical paving corresponding to the correlation range 
\eqref{E:FRk} of random fields governed by Gaussian measures $ \mu_k $.

Namely, for $ k=0,1,2,\ldots,N $, we pave the torus $ \L_N $ by $ L^{(N-k)d} $ disjoint cubes of side length $ L^k $. These cubes are all translates
($L$ is odd) of
$ \{x\in\L_N\colon \abs{x}_\infty\le \frac{1}{2}(L^k-1)\} $ by vectors in $ L^k\Z^d $. We call such cubes  \emph{$k$-blocks} or blocks of $ k$-th generation, and use $\Bcal_k$ to denote  the set of all $ k$-blocks,
\lsm[Bxka]{$\Bcal_k=\Bcal_k(\L_N)$}{the set of all $ k$-blocks in $\L_N$}
$$
\Bcal_k=\Bcal_k(\L_N)=\{B\colon B \mbox{ is a } k\mbox{-block}\},\quad k=0,1,\ldots,N.
$$
Single vertices of the lattice are $0$-blocks, the starting generation for the renormalisation group transforms, ${\Bcal}_0=\Lambda_N$. The only $N$-block is the torus $\Lambda_N$ itself, ${\Bcal}_N=\{\L_N\}$.

 A union of $ k$-blocks is called a \emph{$ k$-polymer}.
 We use $ \Pcal_k= \Pcal_k(\L_N) $
\lsm[Pxka]{$\Pcal_k=\Pcal_k(\L_N)$}{the set of all $ k$-polymers in $\L_N$}
  to denote the set of all $k$-polymers in $ \L_N $ and we have $ \emptyset\in \Pcal_k$.  As $N$ is fixed through the major part of the paper, we often skip $ \L_N $ from the notation as indicated above.
    Notice that certain ambiguity stems from the fact that every $k$-polymer  is also $j$-polymer for any $j\le k$. Nevertheless, we abstain from  introducing $k$-polymer as a pair $(X,k)$ consisting of a set $X$ (union of $k$-blocks) and a label; the appropriate label will be always clear from the context.
    
    Any subset $X\subset\T_N $ is said to be \emph{connected} if for any $x,y\in X$
there exist a path $x_1=x, x_2, \dots, x_n=y$ such that $\abs{x_{i+1}-x_i}_{\infty}=1$, $i=1,\dots, n-1$. We use  $ \Ccal(X) $ to denote the \emph{set of connected components} of $ X$. 
\lsm[CxXc]{$\Ccal(X) $}{the set of all connected components of $ X$}
Two connected sets $X,Y\subset\L_N $ are said to be strictly disjoint if their union is not connected.
Notice that for any strictly disjoint $ X, Y\in {\Pcal}_k $, we have $\dist(X,Y)> L^k$.

We use $ \Pcal_k^{\rm c} $ to denote the set of all connected $k$-polymers and we define that $ \emptyset\notin \Pcal_k^{\rm c} $.
\lsm[Pxkaca]{$ \Pcal_k^{\rm c} $}{the set of all connected $k$-polymers}
For a polymer $ X\in {\Pcal}_k $, we use  $ {\Bcal}_k(X) $ 
to denote   the set of $k$-blocks in $ X $ and $ \abs{X}_k=\abs{{\Bcal}_k(X)} $ 
\lsm[BxkaXc]{${\Bcal}_k(X) $}{the set of $k$-blocks in $ X $}
to denote the number of $ k$-blocks in $ X $
 \lsm[Xczzka]{$\abs{X}_k$}{$ =\abs{{\Bcal}_k(X)}$}
and $ \Pcal_k(X) $ to denote the set of all polymers $Y$ consisting of subsets of blocks from $ \Bcal_k(X) $.
 \lsm[Pxkaa]{$ \Pcal_k(X) $}{the set of all polymers $Y$ consisting of subsets of blocks from $ \Bcal_k(X) $}
The set difference $ X\setminus Y\in \Pcal_k $ of two polymers $ X,Y\in \Pcal_k $ is again a polymer from
 $\Pcal_k $, $ X\setminus Y=\cup_{B\in X, B\notin Y} B $.
The \emph{closure} $ \overline{X} $ of a polymer $X\in \Pcal_k $ is the smallest polymer $ Y\in \Pcal_{k+1}$ of the next generation such that $X\subset Y $. 
 \lsm[Xczz]{$ \overline{X}$ }{the closure of $X$: the smallest polymer $ Y\in \Pcal_{k+1}$ of the next generation such that $X\subset Y $}

A polymer $ X\in \Pcal_k^{\rm c} $ is called \emph{small} if $ \abs{X}_k\le 2^d $ and we denote $ \Scal_k=\{X\in \Pcal_k^{\rm c}\colon \abs{X}_k\le 2^d\} $. 
 \lsm[Sxka]{$\Scal_k=\Scal_k(\L_N)$ }{$=\{X\in \Pcal_k^{\rm c}\colon \abs{X}_k\le 2^d\}$, the set of small polymers}
%
%{\color{red} Change of definition! I want $B^*$ to be a cube!} 
For any $B\in \Bcal_k$ we define  its \emph{small set neighbourhood} $B^*$ to be the cube of the side $(2^{d+1}-1)L^k$ centered at $B$.
 \lsm[Bczz*]{$B^*$}{$=$ the cube of the side $(2^{d+1}-1)L^k$ centered at $B$, the small set neighbourhood of $B$}
Notice that $B^*$ is the smallest cube for which $B\subset Y$ and $Y\in\Scal_k$ implies $Y\subset B^*$.
For any polymer $ X\in \Pcal_k $ we use $X^*$ to denote its \emph{small set neighbourhood},  $ X^*=\cup\{B^*\colon B\in  {\Bcal}_k(X) \} $.
 \lsm[Xczz*]{$X^*$}{$=\cup\{B^*\colon B\in  {\Bcal}_k(X) \}  $, the small set neighbourhood of $X$}
Notice that, strictly speaking, the operation of closure $ \overline{X} $ and 
small set neighbourhood  $ X^*$ should be amended by an index $k+1$ or $k$ indicating the scale from which the relevant blocks are taken. Again we will abstain 
from cumbersome indexing and avoid ambiguity by clearly stating to which $ \Pcal_k $ the considered set  $X$ is  taken to belong.

Having fixed the parameter $N$ and using a shorthand $\cbX $ for $\cbX _N$ in the following, we first introduce the space $M( \Pcal_k, \cbX)$ of all maps 
$F:  \Pcal_k \times \cbX \to \R$ such that for all $X\in \Pcal_k$ 
one has $F(X,\cdot)\in M(\cbX)$,
%M(\cbX,\nu_{k+1})$, 
the map  $F$ is $L^k$-periodic ($F(\tau_a(X),\tau_a(\varphi))=F(X,\varphi)$ for any $a\in (L^k \Z)^d$, where
 $\tau_a(B)=B+a$
 \lsm[tzzaa]{$\tau_a$}{a translation by a vector $a\in\Z^d$}
  and $\tau_a(\varphi)(x)=\varphi(x-a)$) and $F(X,\varphi)$ depends only on values of $\varphi$  on $X^*$
($\varphi,\psi\in \cbX, \   \varphi\bigr\vert_{X^*}=\psi\bigr\vert_{X^*} \implies  F(X,\varphi)=F(X,\psi)$ with $\varphi\bigr\vert_{X^*}$  denoting the restriction of $\varphi$ to $X^*$). 
 \lsm[McPxkaXz]{$M( \Pcal_k, \cbX)$}{the set of all $L^k$-periodic maps 
$F:  \Pcal_k \times \cbX \to \R$ such that $F(X,\cdot)\in M(\cbX,\nu_{k+1})$ for all $X\in \Pcal_k$}
 \lsm[zzzXczz]{$\varphi\bigr\vert_{X^*}$}{the restriction of $\varphi$ to $X^*$}

The sets $M(\Pcal^{\com}_k,\cbX), M(\Scal_k, \cbX)$, and $M(\Bcal_k, \cbX)$ are defined in an analogous way. We also consider the set $M^*(\Bcal_k, \cbX)\supset M(\Bcal_k, \cbX)$ of the maps $F\colon \Bcal_k \times \cbX \to \R$
with $F(B,\varphi)$ depending only  on values of $\varphi$  on the extended set $(B^*)^*$. 
 \lsm[McSxkaXz]{$M( \Scal_k, \cbX)$}{the set of all $L^k$-periodic maps 
$F:  \Scal_k \times \cbX \to \R$ such that  $F(X,\cdot)\in M(\cbX,\nu_{k+1})$ for all $X\in \Scal_k$}
 \lsm[McBxkaXz]{$M( \Bcal_k, \cbX)$}{the set of all $L^k$-periodic maps 
$F:  \Bcal_k \times \cbX \to \R$ such that  $F(B,\cdot)\in M(\cbX,\nu_{k+1})$ for all $B\in \Bcal_k$}
 \lsm[McBxkaXzx]{$M^*( \Bcal_k, \cbX)$}{the set of all $L^k$-periodic maps 
$F:  \Bcal_k \times \cbX \to \R$ such that $F(B,\cdot)\in M(\cbX,\nu_{k+1})$ for all $B\in \Bcal_k$ living on $(B^*)^*$}

For functions from $M( \Pcal_k, \cbX)$ we introduce the \emph{circle product},
 \lsm[zzaa]{  $\circ$}{$(F_1\circ F_2)(X,\varphi)=\sum_{Y\subset X} F_1(Y,\varphi) F_2(X\setminus Y,\varphi)$, the circle product of $F_1, F_2\in M( \Pcal_k, \cbX)$}
\begin{equation}
F_1, F_2\in M( \Pcal_k, \cbX), \  
(F_1\circ F_2)(X,\varphi)=\sum_{Y\subset X} F_1(Y,\varphi) F_2(X\setminus Y,\varphi),
\end{equation}
where we defined $ F(\emptyset,\varphi)=:1 $.
Notice, that the product is defined pointwise in the variable $\varphi$. We often skip it 
and write $(F_1\circ F_2)(X)=\sum_{Y\subset X} F_1(Y) F_2(X\setminus Y)$.
Observe that the circle product is commutative and distributive.

For  $F\in M(\Bcal_k, \cbX)$ and $X\in \Pcal_k$, we define
 \lsm[FcXcpg]{$F^X(\varphi)$}{$=\prod_{B\in  \Bcal_k(X)}F(B,\varphi)$}
\begin{equation}
F^X(\varphi)=\prod_{B\in  \Bcal_k(X)}F(B,\varphi).
\end{equation}
Extending any $F\in M(\Bcal_k, \cbX)$ to $M( \Pcal_k, \cbX)$ 
by taking 
 \lsm[FcXcpgx]{$F(X,\varphi)$}{$=F^X(\varphi)$ for $F\in M(\Bcal_k, \cbX)$}
\begin{equation}
\label{E:extF^X}
F(X,\varphi)=F^X(\varphi),
\end{equation}
we get 
\begin{equation}
(F_1+F_2)^X=\sum_{Y\subset X} F_1^Y F_2^{X\setminus Y}=(F_1\circ F_2)(X)
\end{equation}
directly from the definitions.

For each $ x\in\L_N $ we define the functions
\begin{equation}\label{ideal}
\begin{aligned}
\Hcal(x,\varphi)=\lambda + \sum_{i=1}^da_i\nabla\varphi(x)+\sum_{i,j=1}^d\bc_{i,j}\nabla_i\nabla_j\varphi(x) + \frac{1}{2}\sum_{i,j=1}^d\bd_{i,j}\nabla\varphi(x)\nabla_j\varphi(x)
\end{aligned}
\end{equation} 
with coefficients $ \l\in\R , a\in\R^d, \bc\in\R^{d\times d} $ and $ \bd\in\R^{d\times d}_{\rm sym} $. 

A special role will be played by a subspace 
$M_0(\Bcal_k, \cbX)\subset M(\Bcal_k, \cbX) $   of all quadratic functions  built from \eqref{ideal} of the form
 \lsm[MczaBxkaXd]{$M_0(\Bcal_k, \cbX)$}{the set of all ideal Hamiltonians: quadratic functions  of the form $H(B,\varphi)=\lambda\abs{B}+\ell(\varphi)+Q(\varphi)$ }
 \lsm[HcBxpg]{$H(B,\varphi)$}{ideal Hamiltonian  of the form $H(B,\varphi)=\lambda\abs{B}+\ell(\varphi)+Q(\varphi)$ }
\begin{equation}
\label{E:P}
H(B,\varphi)=\sum_{x\in B} \Hcal(x,\varphi)=\lambda\abs{B}+\ell(\varphi)+Q(\varphi), 
\end{equation}
where
\begin{equation}
\label{E:ell}
\ell(\varphi)=\sum_{x\in B} \bigl[\sum_{i=1}^d a_i\, \nabla_i\varphi(x) +
\sum_{i,j=1}^d \bc_{i,j}\, \nabla_i\nabla_j\varphi(x) \bigr]
\end{equation}
and
 \lsm[lfpg]{$\ell(\varphi)$}{$=\sum_{x\in B} \bigl[\sum_{i=1}^d a_i \,\nabla_i\varphi(x) +
\sum_{i,j=1}^d \bc_{i,j} \,\nabla_i\nabla_j\varphi(x) \bigr]$, linear term of ideal Hamiltonian}
\begin{equation}
\label{E:Q}
Q(\varphi,\varphi)=\frac12\sum_{x\in B}
\sum_{i,j=1}^d \bd_{i,j} \,\nabla_i\varphi(x)\, \nabla_j\varphi(x).
\end{equation}

 \lsm[Qcpg]{$Q(\varphi,\varphi)$}{$=\frac12\sum_{x\in B}
\sum_{i,j=1}^d \bd_{i,j} (\nabla_i\varphi)(x)(\nabla_j\varphi)(x)$, quadratic term of ideal Hamiltonian}
Sometimes we use the term  \emph{ideal Hamiltonians} for functions in $ M_0(\Bcal_k, \cbX) $.

%We define also the set $ M_0(\Scal_k, \cbX) $ of functions $H(X,\varphi)=H^X(\varphi)$.

% IDEOLOGY OF NORMS:
% There are 3 types of norms $\bnorm{\cdot}_{k,X}$ and $\bnorm{\cdot}^{k,X}$
% concern linear spaces $\cbX$ and its dual (in multilinear derivatives)
% Then there is weighted strong $\tnorm{\cdot}_{k,X}$ and weak $\norm{\cdot}_{k,X}$
% and finally $\norm{\cdot}_{k}$ (with a possible distinction 
% $\norm{ F}_{k}^{\rm(p)}$ and $\norm{ F}_{k}^{\rm(b)}$ and $\tnorm{\cdot}_{k}$

Our next aim is to introduce norms   $\norm{\cdot}_{k,r}$ and $\norm{\cdot}_{k+1,r}$ on $M( \Pcal_k, \cbX)$
%and $M(\Bcal_k, \cbX)$ 
(with $r=1,\dots,r_0$,  where $r_0$ is a fixed integer to be chosen later) and a norm $ \norm{\cdot}_{k,0} $ on $ M_0(\Bcal_k, \cbX)$.
We begin by introducing, for each $k\in\{0,1,\dots,N\}$ and $X\in \Pcal_k$,
two distinct (semi)norms  $\bnorm{\cdot}_{k,X}$ and $\bnorm{\cdot}_{k+1,X}$
on $\cbX $. For any $\varphi\in\cbX$ we define
\begin{equation}
\label{E:normphikX}
\bnorm{\varphi}_{k,X}= \max_{1\le s \le3}\sup_{x\in X^*}
\frac1h L^{k\bigl(\tfrac{d-2}2+s\bigr)}\bigl\vert\nabla^s\varphi(x)\bigr\vert
\end{equation}
and
 \lsm[zzkaXc]{$\bnorm{\cdot}_{k,X}$}{a norm on $\cbX $:  $\bnorm{\varphi}_{k,X}= \max_{1\le s \le3}\sup_{x\in X^*}
\frac1h L^{k\bigl(\tfrac{d-2}2+s\bigr)}\bigl\vert\nabla^s\varphi(x)\bigr\vert$}
 \lsm[zzkaXc1a]{$\bnorm{\cdot}_{k+1,X}$}{a norm on $\cbX $:  $\bnorm{\varphi}_{k+1,X}= \max_{1\le s\le3}\sup_{x\in X^*}
\frac1h L^{(k+1)\bigl(\tfrac{d-2}2+s\bigr)}\bigl\vert\nabla^s\varphi(x)\bigr\vert$}
\begin{equation}
\label{E:k+1}
\bnorm{\varphi}_{k+1,X}= \max_{1\le s \le3}\sup_{x\in X^*}
 \frac1h L^{(k+1)\bigl(\tfrac{d-2}2+s\bigr)}\bigl\vert\nabla^s \varphi(x)\bigr\vert, 
\end{equation}
where
\begin{equation}
\label{E:nablas}
\abs{\nabla^s\varphi(x)}^2= \sum_{\abs{\boldsymbol{\alpha}}=s}\abs{\nabla^{\boldsymbol{\alpha}}\varphi(x)}^2.
\end{equation}
\lsm[Nhiafgx]{$\abs{\nabla^s\varphi(x)}^2$}{$  =\sum_{\abs{\boldsymbol{\alpha}}=s}\abs{\nabla^{\boldsymbol{\alpha}}\varphi(x)}^2$}
%

%Let us emphasise that the supremum in \eqref{E:k+1} is taken over the $(k+1)$-polymer
%${\overline X}^*$.  {\color{red}  $k+1$ for $k=N$?} 

Next, for any $s$-linear function $S_k$ on $\cbX\times\dots\times\cbX$, we define
 \lsm[zzjaXc]{$\bnorm{\cdot}^{j,X}$}{$\bnorm{S}^{j,X}=\sup_{\bnorm{\dot{\varphi}}_{j,X}\le 1} \bigl\vert S_k(\dot{\varphi}, \dots, \dot{\varphi})\bigr\vert,\ j=k, k+1,$\\ for $s$-linear function $S_k$ on $\cbX\times\dots\times\cbX$}
\begin{equation}
\label{E:SjX}
\bnorm{S}^{j,X}=\sup_{\bnorm{\dot{\varphi}}_{j,X}\le 1}
  \bigl\vert S_k(\dot{\varphi}, \dots, \dot{\varphi})\bigr\vert,\ j=k, k+1,
\end{equation}
and, for any $F\in C^r( \cbX)$, also
\begin{equation}
\label{E:bnorm^jXr}
\bnorm{F(\varphi)}^{j,X,r}=\sum_{s=0}^r  \frac1{s !} \bnorm{D^s F(\varphi)}^{j,X}  .
\end{equation}
 \lsm[zzjaXcra]{$\bnorm{\cdot}^{j,X,r}$}{$\bnorm{F}^{j,X,r}=\sum_{s=0}^r  \frac1{s \symbol{33}} \bnorm{D^s F(\varphi)}^{j,X} ,\ j=k, k+1$, for $F\in C^r( \cbX)$}
Here, for $s=0$ we take 
\begin{equation}
\label{E:D0F}
\bnorm{ D^0 F(\varphi)}^{j,X}=  \bnorm{ F(\varphi)}.
\end{equation}

In particular, considering 
for any  $F\in M( \Pcal_k, \cbX)$ and any  $X\in \Pcal_k $  (and similarly also for any $F\in M(\Bcal_k, \cbX)$)
 the map $F(X): \cbX\to\R$ defined by $F(X)(\varphi)=F(X,\varphi)$ and its $s$th derivative $D^s F(X,\varphi)(\dot{\varphi}, \dots, \dot{\varphi})$, 
 we get
 \lsm[FcXcpg]{$F(X)(\varphi)$}{$=F(X,\varphi)$
  for $F\in M( \Pcal_k, \cbX)$}
\begin{equation}
\label{E:bnormup}
\bnorm{ F(X,\varphi)}^{j,X,r}=\sum_{s=0}^r  \frac1{s !} \sup_{\bnorm{\dot{\varphi}}_{j,X}\le 1}
  \bigl\vert D^s F(X,\varphi)(\dot{\varphi}, \dots, \dot{\varphi})\bigr\vert,\ j=k, k+1.
\end{equation}

Now, we are ready to introduce the weighted \emph{strong} norm $\tnorm{F(X)}_{k,X}$ as well as weighted \emph{weak} norm  $\norm{ F(X)}_{k,X,r}$, $ r=1, \dots, r_0$ depending on parameters $h$ and $\omega$ that will be used for tuning their properties.
Introducing the strong weight functions
\begin{equation}
\label{E:gX}
W_k^{X}(\varphi)=\exp\Bigl\{\sum_{x\in X} G_{k,x}(\varphi)\Bigr\}
\end{equation}
 \lsm[WcXcka]{$W_k^{X}(\varphi)$}{$=\exp\bigl\{\sum_{x\in X}  G_{k,x}(\varphi)\bigr\}$ the strong weight function}
with
\begin{equation}
\label{E:Gkxdef}
G_{k,x}(\varphi)=\frac1{h^2}\bigl( |\nabla\varphi(x)|^2+L^{2k}\abs{\nabla^2\varphi(x)}^2 +L^{4k}|\nabla^3\varphi(x)|^2\bigr), 
\end{equation}
 \lsm[Gc]{$G_{k,x}(\varphi)$}{$=\frac1{h^2} \bigl(  \abs{\nabla^2\varphi(x)}^2+ L^{2k}\abs{\nabla^2\varphi(x)}^2+
 L^{4k}\abs{\nabla^3\varphi(x)}^2\bigr)$}
we define the weighted strong norm
 \lsm[zzDcpaFcXc]{$\tnorm{\cdot}_{k,X}$}{the weighted strong norm, $\tnorm{F(X)}_{k,X}=\sup_{\varphi} \bnorm{F(X,\varphi)}^{k,X,r_0}
W_k^{-X}(\varphi)$}
\begin{equation}
\label{E:tnorm}
\tnorm{F(X)}_{k,X}=\sup_{\varphi} \bnorm{F(X,\varphi)}^{k,X,r_0}
W_k^{-X}(\varphi)
\end{equation}
with  $W_k^{-X}(\varphi)= \bigl(W_k^{X}(\varphi)\bigr)^{-1}$.
 \lsm[zzDcpaFcka]{$\tnorm{\cdot}_{k}$}{$\tnorm{F}_{k}=\tnorm{F(B)}_{k,B}$ for $F\in M( \Bcal_k, \cbX)$}
For $F\in M( \Bcal_k, \cbX)$, the norm $\tnorm{F(B)}_{k,B}$ actually does not depend on $B$ in view of periodicity of $F$, and we use the shorthand $\tnorm{F}_{k}$.

Further, let   $B_x\in  \Bcal_k$ be the $k$-block containing $x$
 \lsm[Bcxa]{$B_x$}{the $k$-block containing $x$}
and let $\partial X$ denote the bounda\-ry
 \lsm[Xcdg]{$\partial X$}{=$\{y\not\in X\vert\exists z\in X \,\text{such that}\, \abs{y-z}=1\}\cup 
\{y\in X\vert \exists  z\not\in X \,\text{such that}\, \abs{y-z}=1\}$, the boundary of $X$}
\begin{equation}
\label{E:pX}
\p X=\{y\not\in X\mid \exists  z\in X \text{ such that } \abs{y-z}=1\}\cup 
\{y\in X\mid \exists  z\not\in X \text{ such that } \abs{y-z}=1\}
\end{equation}
(recall that $\abs{\cdot}$ is the Euclidean norm).
Introducing  the weak weight functions
\begin{equation}
w_k^{X}(\varphi)=\exp\Bigl\{\sum_{x\in  X} \omega\bigl(2^d g_{k,x}(\varphi)+ G_{k,x}(\varphi)\bigr)+L^k\sum_{x\in \partial X} G_{k,x}(\varphi)\Bigr\}
\end{equation}
 \lsm[wakaXc]{$w_k^{X}(\varphi)$}{$= \exp\Bigl\{\sum_{x\in  X} \omega\bigl(2^d g_{k,x}(\varphi)+G_{k,x}(\varphi)\bigr)+L^k\sum_{x\in \partial X} G_{k,x}(\varphi)\Bigr\}$, 
the  weak weight function}
 \lsm[wgx]{$\omega$}{a parameter in the weightfunction $w_k^{X}(\varphi)$}
with $G_{k,x}(\varphi)$ as above and
\begin{equation}
g_{k,x}(\varphi)=\frac{1}{h^2}\sum_{s=2}^4L^{(2s-2)k}\sup_{y\in B^*_x}\abs{\nabla^s\varphi(y)}^2, 
\end{equation} 
 \lsm[gakaxa]{$g_{k,x}(\varphi)$}{$=\frac{1}{h^2}\sum_{s=2}^4L^{(2s-2)k}\sup_{y\in B^*_x}\abs{\nabla^s\varphi(y)}^2$}
we define  the weighted weak norm by
\begin{equation}
\norm{F(X)}_{k,X,r}=\sup_{\varphi}
\bnorm{ F(X,\varphi)}^{k,X,r} \, w_k^{-X}(\varphi) ,\   r=1,\dots,r_0.
\end{equation}
 \lsm[zzjaXcrax]{$\norm{\cdot}_{k,X,r}$}{$\norm{F(X)}_{k,X,r}=\sup_{\varphi}
\bnorm{ F(X,\varphi)}^{k,X,r} \, w_k^{-X}(\varphi) ,\   r=1,\dots,r_0$}
\lsm[haa]{$h$}{a parameter  in the norms $\tnorm{\cdot}_{k,X}$ or $\norm{\cdot}_{k,X,r}$ (via the weight functions $G_{k,x}$ and $g_{k,x}$)}

In addition we also introduce the norm  $\norm{ \mathbf{\cdot}}_{k:k+1,X,r}$ that can be viewed as being  ``halfway between'' $\norm{ \mathbf{\cdot}}_{k,X,r}$
and $\norm{\mathbf{\cdot}}_{k+1,U,r}$ with $U=\overline X \in \Pcal_{k+1} $. Namely,  we define  
\begin{equation}  \label{E:norm_k:k+1}
\norm{F(X)}_{k:k+1,X,r}=\sup_{\varphi}
\bnorm{ F(X,\varphi)}^{k+1,X,r} \, w_{k:k+1}^{-X}(\varphi), \  r=1,\dots,r_0.
\end{equation}
 \lsm[zzjaXcraxz]{$\norm{\cdot}_{k:k+1,X,r}$}{$\norm{F(X)}_{k:k+1,X,r}=\sup_{\varphi}
\bnorm{ F(X,\varphi)}^{k,X,r} \, w_{k:k+1}^{-X}(\varphi) ,\   r=1,\dots,r_0$}
with 
\begin{equation}
w_{k:k+1}^{X}(\varphi)=\exp\Bigl\{\sum_{x\in  X} \bigl((2^d\omega-1) g_{k:k+1,x}(\varphi)+ \omega G_{k,x}(\varphi)\bigr)+3 L^k\sum_{x\in \partial X} G_{k,x}(\varphi)\Bigr\},
\end{equation}
 \lsm[wakaXcka]{$w_{k:k+1}^{X}(\varphi)$}{$= \exp\Bigl\{\sum_{x\in  X} \bigl((2^d\omega-1)  g_{k:k+1,x}(\varphi)+\omega G_{k,x}(\varphi)\bigr)+3L^j\sum_{x\in \partial X} G_{k,x}(\varphi)\Bigr\}$, 
the  weak weight function}
where
\begin{equation}
g_{k:k+1,x}(\varphi)=\frac{1}{h^2}\sum_{s=2}^4L^{(2s-2)(k+1)}\sup_{y\in B^*_x}\abs{\nabla^s\varphi(y)}^2, 
\end{equation} 
 \lsm[gakaxaka]{$g_{k:k+1,x}(\varphi)$}{$=\frac{1}{h^2}\sum_{s=2}^4L^{(2s-2)(k+1)}\sup_{y\in B^*_x}\abs{\nabla^s\varphi(y)}^2$}
Notice that for the functions $g_{k:k+1,x}$ entering the norm   $\norm{\mathbf{\cdot}}_{k:k+1,X,r}$, we still take   $\sup_{y\in B^*_x}$ with $k$-block $B_x$. 
The prefactors $L^{(2s-2)(k+1)}$, however, involve the power $k+1$.
Also, the norm $\bnorm{ F(X,\varphi)}^{k+1,X,r}$ is used, involving ${\dot{\varphi}}_{k+1,X}$ in its definition.

For any $r\le r_0$, clearly, 
 \begin{equation}
\label{E:w<s}
\norm{F(X)}_{k,X,r}\le \tnorm{F(X)}_{k,X}.
\end{equation}
Inspecting the definitions, it is also easy to show that
\begin{equation}
\label{E:k:k+1<k+1X}
\norm{F(X)}_{k:k+1,X,r}  \le \norm{F(X)}_{k,X,r}
\end{equation}
once $\omega\ge 2^{d-1}$ (assuring that $2^d\o (L^2-1)\ge L^2$),
and, for any $U\in \Pcal_{k+1}\subset \Pcal_{k}$ and $F\in M( \Pcal_{k+1}, \cbX)\subset M( \Pcal_k, \cbX)$, also
\begin{equation}
\label{E:k:k+1<k+1}
\norm{F(U)}_{k+1,U,r}\le \norm{F(U)}_{k:k+1,U,r}\le \norm{F(U)}_{k,U,r}.
\end{equation}
 
Next, for any $F\in M( \Pcal_k^{\com}, \cbX)$ and a  parameter $\mathsf{A} \in \R_+$ we introduce
 \lsm[zzjaa]{$\norm{\cdot}_{k,r}^{(\mathsf{A})}$}{$\norm{ F}_{k,r}=\sup_{X\in \Pcal_k^{\rm c} }\norm{ F(X)}_{k,X,r} \Gamma_{k,A}(X)  ,\  r=1,\dots, r_0,$}
 \lsm[zzjab]{$\norm{\cdot}_{k:k+1,r}^{(\mathsf{A})}$}{$\norm{ F}_{k:k+1,r}=\sup_{X\in \Pcal_k^{\rm c} }
\norm{ F(X)}_{k:k+1,X,r} \Gamma_{k,\mathsf{A}}(X)  ,\  r=1,\dots, r_0$}
\begin{equation}
\label{E:weak_norm}
\norm{ F}^{(\mathsf{A})}_{k,r}=\sup_{X\in \Pcal_k^{\rm c} }
\norm{ F(X)}_{k,X,r} \Gamma_{k,\mathsf{A}}(X) ,\  r=1,\dots, r_0,
\end{equation}
where 
 \lsm[GhkaAc]{$\Gamma_{k,\mathsf{A}}(X)$}{=$\begin{cases} \mathsf{A}^{\bnorm{X}} &\text{ if }X\in\Pcal^{\rm c}_k\setminus\Scal_k\\ 1&\text{ if }X\in\Scal_k.\end{cases}$}
\begin{equation}
\label{E:GammaA}
\Gamma_{k,\mathsf{A}}(X)=\begin{cases} \mathsf{A}^{\bnorm{X}} &\text{if }X\in\Pcal^{\rm c}_k\setminus\Scal_k\\ 1&\text{if }X\in\Scal_k.\end{cases}
\end{equation}
Similarly we define also $\norm{ F}^{(\mathsf{A})}_{k:k+1,r}$. Note that this norm is only defined via functional on connected polymers. Whenever we estimate functionals on arbitrary polymers we simply consider the product over the connected components. Occasionally, when the parameter  $\mathsf{A}$ is clear from the context, we skip it and write just   $\norm{ F}_{k,r}$ and $\norm{ F}_{k:k+1,r}$.
For $F\in M(\Bcal_k, \cbX)$ we also define
 \lsm[zzjad]{$\norm{ \cdot}_{k,r}^{\rm(b)}$}{$\norm{ F}_{k,r}^{\rm(b)}=
\norm{ F(B)}_{k,B,r}$ for $F\in M(\Bcal_k, \cbX)$}
\begin{equation}
\label{E:FBnorm}
\norm{ F}_{k,r}^{\rm(b)}=
\norm{ F(B)}_{k,B,r} .
\end{equation}
Notice that the right hand side does not depend on $B$ in view of $L^k$-periodicity of $F$.
Any $F\in M( \Pcal_k, \cbX)$ can be restricted to $M(\Bcal_k, \cbX)$
with $\norm{ F}_{k,r}^{\rm(b)}\le \norm{ F}_{k,r}$.

Finally, on the subspace $ M_0(\Bcal_k, \cbX) $ we define an additional norm
$\norm{\cdot}_{k,0} $  by taking
 \lsm[zzja]{$\norm{ \cdot}_{k,0}$}{$\norm{ H}_{k,0}=L^{dk}\abs{\lambda}+L^{\frac{dk}{2}} h\sum_{i=1}^d\abs{a_i}+L^{\frac{(d-2)k}{2}}h\sum_{i,j=1}^d\abs{\bc_{i,j}}+\frac{h^2}2\sum_{i,j=1}^d\abs{\bd_{i,j}}$}
\begin{equation}\label{normideal}
\norm{ H}_{k,0}= L^{dk}\abs{\lambda}+L^{\frac{dk}{2}} h\sum_{i=1}^d\abs{a_i}+L^{\frac{(d-2)k}{2}}h\sum_{i,j=1}^d\abs{\bc_{i,j}}+\frac{h^2}2\sum_{i,j=1}^d\abs{\bd_{i,j}} 
\end{equation}
for any $H\in M_0(\Bcal_k, \cbX)$ of the form \eqref{E:P}.

Also, let us stress that the above norms depend on parameters like $L$, $h$, and $\mathsf{A}$
that are often skipped from the notation.
Finally we use the notation
\begin{equation}
\bM_{k,r} := \{ K \in M(\Pcal^{\com}_k, \cbX)  : \norm{K}_{k,r}^{(\mathsf{A})}  < \infty \}.
\end{equation}
Sometimes we write $\bM_r = \bM_{r,k}$ for brevity.
Note that  the norms  $\norm{K}_{k,r}^{(\mathsf{A})} < \infty$  for different $\mathsf{A} >0$ are equivalent (since there are only finitely many polymers).
Thus the definition of $\bM_{k,r}$ does not depend on $\mathsf{A}$.

\section[Renormalisation transformation $\bT_k: (H_k,K_k)\mapsto (H_{k+1},K_{k+1})$]{Definition of the renormalisation transformation $\bT_k: (H_k,K_k)\mapsto (H_{k+1},K_{k+1})$}\label{S:renorm}

Here, we introduce the renormalisation step at a scale  $k$, $ k=0,\ldots, N -1$. 
At each scale $k$, the interaction will be split between functions  $H_k$ and $K_k$.
(Here and in the following we suppress the notation indicating  the dependence on  $\bq$, reinstating it only when it will play a crucial role.)
The ``ideal local Hamiltonian'' part $H_k$ is collecting all \emph{relevant} (or marginal) directions under the renormalisation transformation, with all irrelevant ones delegated to the coordinate $K_k$.
There is only  limited number of parameters in the relevant coordinate $H_k$. 
Being given a pair $ (H_k,K_k) $,  $H_k\in M_0(\Bcal_k, \cbX)$
and $K_k\in M( \Pcal_k, \cbX)$, 
we define a pair $(H_{k+1},K_{k+1})$,
 $H_{k+1}\in M_0(\Bcal_{k+1}, \cbX)$
and $K_{k+1}\in M( \Pcal_{k+1}, \cbX)$,
so that 
\begin{equation}
\label{E:Rk+1}
\bR_{k+1}({\rm e}^{-H_k}\circ K_k)(\Lambda_N,\varphi)=({\rm e}^{-H_{k+1}}\circ K_{k+1})(\Lambda_N,\varphi)
\end{equation}
with
$(\bR_{k+1} F)(X,\varphi)= \int_{ \cbX}F(X,\varphi+\xi)\mu_{k+1}(\d\xi)$.

As the scale $k$ is fixed in the rest of this chapter, we will skip it and write 
$(H^\prime,K^\prime)$ for $(H_{k+1},K_{k+1})$, with  \eqref{E:Rk+1} becoming
\begin{equation}
\label{E:R'}
\bR ({\rm e}^{-H}\circ K)={\rm e}^{-H^\prime}\circ K^\prime.
\end{equation}

To define the Hamiltonian $H^{\prime}$ on the next scale,
we first introduce the projection
 \lsm[Tc2a]{$T_2$}{Taylor expansion up to the second order, $T_2 F(B,\dot{\varphi})= 
F(B,0) + DF(B,0)(\dot{\varphi})+\tfrac12 D^2F(B,0)(\dot{\varphi}, \dot{\varphi})$}
\begin{equation}
\label{E:Pi}
\Pi_2: M^*(\Bcal, \cbX) \to M_0(\Bcal, \cbX)
\end{equation}
as a ``homogenization''  of  the second order Taylor expansion $T_2$ around zero.
Namely, for any $F\in  M^*(\Bcal, \cbX)$ with
\begin{equation}
\label{E:T}
T_2 F(B,\dot{\varphi})= 
F(B,0) + DF(B,0)(\dot{\varphi})+\tfrac12 D^2F(B,0)(\dot{\varphi}, \dot{\varphi}),
\end{equation}
we define 
 \lsm[Ph]{$\Pi_2$}{the projection  from $M^*(\Bcal, \cbX)$ to $ M_0(\Bcal, \cbX)$:  $\Pi_2 F(B,\dot{\varphi})= 
F(B,0) + \ell(\dot{\varphi})+Q(\dot{\varphi},\dot{\varphi})$: $\ell$ agrees
with  $DF(B,0)$ on all quadratic functions $\dot{\varphi}$ on $(B^*)^*$ and
$Q$  agrees with  $\tfrac12 D^2F(B,0)$ on all affine functions $\dot{\varphi}$ on $(B^*)^*$}
\begin{equation}
\label{E:PiT}
\Pi_2 F(B,\dot{\varphi})= 
F(B,0) + \ell(\dot{\varphi})+Q(\dot{\varphi},\dot{\varphi})
\end{equation}
so that  $\ell$ is a (unique) linear function of the form \eqref{E:ell} that agrees
with  $DF(B,0)$ on all quadratic functions $\dot{\varphi}$ on $(B^*)^*$ and
$Q$ is a (unique) quadratic function of the form \eqref{E:Q} that agrees
with  $\tfrac12 D^2F(B,0)$ on all affine functions $\dot{\varphi}$ on $(B^*)^*$.
Strictly speaking,  we have in mind functions  $\dot{\varphi}\in \cbX$ such that 
they are quadratic or affine when restricted to $(B^*)^*$. 
 Since, for $B\in  \Bcal_k $, $k\le N-1$,  the set $(B^*)^*$ is not wrapped around the torus (as soon as $2^{d+2}\le L$), we do not need to be concerned with a possibility of a contradiction in the  assumption of  $\dot{\varphi}\in \cbX$ having a quadratic or   affine restriction to $(B^*)^*$. Clearly,   $\Pi_2 F  \in M_0(\Bcal, \cbX)\subset M(\Bcal, \cbX)$  whenever $F\in M^*(\Bcal, \cbX)$ and $\Pi_2 F =F$ for $F\in M_0(\Bcal, \cbX)$.
  In particular, we will consider the projection $\Pi_2$  on functions  $\overline F\in   M^*(\Bcal, \cbX)$ of the form
 \begin{equation}
\overline F(B,\varphi)=\sum_{\begin{subarray}{c}  X\in \cS \\  X\supset B  \end{subarray}} 
\frac1{\abs{X}} F(X,\varphi)
\end{equation}
for any $F\in M(\Scal, \cbX)$.

%Notice that projection $\Pi$ depends on the generation $k$ of blocks $B$. When this distinction is not obvious from the context, we will use the notation $\Pi_k$.
 
Now we are ready to define the iteration $H^{\prime}$. 
Recalling that $\bR=\bR_{k+1}$ is the mapping defined by convolution with $\mu_{k+1}$ and
starting from  
$H\in M_0(\Bcal, \cbX)$ and  $K\in M(\Pcal, \cbX)$, we define
\begin{equation}
\label{E:H'}
H^{\prime}(B^{\prime},\varphi) = 
\sum_{B\subset B^{\prime}}\Pi_2 \bigl((\bR H)(B, \varphi)-
\sum_{\begin{subarray}{c}  X\in \cS \\  X\supset B  \end{subarray}} 
\frac1{\abs{X}} (\bR K)(X,\varphi)\bigr).
\end{equation}

To define $K^{\prime}$, we first replace  the original variable  $H(B,\varphi)$ (or rather  $H(B,\varphi+\xi)$
in anticipation of the integration $\bR $)  by
$ \widetilde H(B,\varphi)$, the term in the right hand side sum above,
\begin{equation}
\label{E: tildeH}
 \widetilde H(B,\varphi) =
\Pi_2 \Bigl((\bR H)(B, \varphi)-
\sum_{\begin{subarray}{c}  X\in \cS \\  X\supset B  \end{subarray}} 
\frac1{\abs{X}} (\bR K)(X,\varphi)\Bigr).
\end{equation}
 \lsm[IcBcpgz~]{$\tilde I_k(B,\varphi) $}{$=\exp\bigl\{ - \widetilde H_k(B,\varphi)  \bigr\}$}
Writing  $\tilde I(B,\varphi) =\exp\bigl\{ - \widetilde H(B,\varphi)  \bigr\}$ instead of the original
 \lsm[IcBcpg]{$I_k(B,\varphi)$}{$=\exp\bigl\{ - H_k(B,\varphi)  \bigr\}$}
 $$I(B,\varphi+\xi)=\exp\bigl\{ - H(B,\varphi+\xi)  \bigr\},$$ and denoting $\tilde J= 1-\tilde I$, we introduce
 \lsm[JcBcpg]{$\tilde J_k(B, \varphi)$}{$= 1- \tilde I(B,\varphi)$}
\begin{equation}
\label{E:tildeKsimple}
\widetilde K=  \tilde J \circ (I-1)\circ K.
\end{equation}
Notice that we are considering here the extension of $\tilde I, \tilde J$, and $I$ to $M(\Pcal,\cbX)$,
resp. $M(\Pcal,\cbX\times \cbX)$,
according to \eqref{E:extF^X}.
Let us stress that the equation above (and in similar circumstances later)
is to be interpreted as an algebraic definition valid pointwise in the variables
$\varphi$ and $\xi$. It means that $\widetilde K$ is actually a function  
on $\Pcal\times \cbX\times\cbX$ defined explicitly by
 \lsm[KcXcpgz~]{$\widetilde K_k(X,\varphi,\xi)$}{$= \sum_{Y\in \Pcal_k(X)} (I_k(\varphi+\xi)-\tilde I_k(\varphi))^{X\setminus Y}(\varphi,\xi) K_k(Y,\varphi+\xi)$}
\begin{equation}
\label{E:tildeK}
\widetilde K(X,\varphi,\xi)= \sum_{\begin{subarray}{c}  Y,Z\in \Pcal_k(X) \\  Y\cap Z=\emptyset  \end{subarray}}  \tilde J^{X\setminus Y\cup Z}(\varphi)
\bigl(I(\varphi+\xi)-1\bigr)^Y K(Z,\varphi+\xi).
\end{equation}
Occasionally, we are skipping the polymer variable $X$ but wish to keep the field variables and write, slightly misusing the notation, say, $\widetilde K(\varphi,\xi)$
for the mapping $\widetilde K(\varphi,\xi): \Pcal \to \R $
defined by $\widetilde K(\varphi,\xi)(X) = \widetilde K(X,\varphi,\xi)$. Then the above algebraic equation reads
\begin{equation}
\widetilde K(\varphi,\xi)=\tilde J(\varphi)\circ \bigl( I(\varphi+\xi)-1\bigr) \circ K(\varphi+\xi).
\end{equation}
It is useful to observe that $I-\tilde I=(I-1)+\tilde J$ yields $I-\tilde I= \tilde J \circ(I-1)$ and thus $\widetilde K=(I-\tilde I)\circ K$
suggesting the interpretation of $\widetilde K(\varphi,\xi)$ as $K(\varphi+\xi)$ combined with the perturbation $I(\varphi+\xi)-\tilde I(\varphi)$.

Now, using $I(\varphi+\xi)=\tilde I(\varphi)+\tilde J(\varphi)+ \bigl(I(\varphi+\xi)-1\bigr) $,
we immediately infer that
\begin{equation}
I(\varphi+\xi)=  \tilde I(\varphi) \circ\tilde J(\varphi)\circ  \bigl(I(\varphi+\xi)-1\bigr)
\end{equation}
and thus
\begin{equation}
I(\varphi+\xi)\circ K(\varphi+\xi)=   \tilde I(\varphi) \circ\tilde J(\varphi)\circ\bigl(I-1\bigr)(\varphi+\xi) \circ  K(\varphi+\xi)= \tilde I(\varphi)\circ \widetilde K(\varphi,\xi).
\end{equation}
As a result,
\begin{equation}
\bR(I\circ K)(\L_N,\varphi)=  (\tilde I\circ (\bR \widetilde K))(\L_N,\varphi),
\end{equation}
or,  explicitly,
\begin{equation}
\label{E:inttildeK}
\bR(I\circ K)(\L_N,\varphi)=
\sum_{X\in\Pcal(\L_N)}\tilde I^{\L_N\setminus X}(\varphi)\int_{\cbX} \widetilde K(X,\varphi,\xi) \mu_{k+1}(\d\xi).
\end{equation}
Here we kept the index $k+1$ at $\mu_{k+1}$ to avoid a confusion with 
the measure $\mu=\mu_1\ast\dots\ast\mu_{N+1}$.

The function $K^\prime $ on the next scale satisfying \eqref{E:R'} will be defined by sorting the $X$-terms  according to  the next level closure $U$. While for any $X\in \Pcal(\L_N)\setminus\Scal(\L_N)$ we attribute the contribution to 
$K^\prime (U)$ with $U=\overline{X}\in\Pcal(\L_N)^{\prime}$, for $X\in \Scal(\L_N)$,
we (potentially) split the contribution%
\footnote{As will become clear later, the reason for doing so is a need to deal with 
relevant quadratic terms stemming from $K$'s with  $X\in \Scal$.  In anticipation, those terms  are already included as the second term in $\widetilde H^{\prime}$
(cf. \eqref{E:H'}) and the particular way of splitting them among $U$'s leads to the exact cancelations of the corresponding 
linearized terms. 
In particular, the linearization of the map $K\to K^\prime$ contains only terms starting with the third order in the Taylor expansion of $K(X,\varphi)$ for $X$ small (cf. \eqref{E:C}).
Using the fact that only the terms linear in $K(X)$ with $X\in\Scal$
are relevant in this context, it suffices to introduce a nontrivial $\chi$ only for such terms. Our definition is thus a slight simplification of the trick introduced by Brydges \cite{B07}. We thank Felix Otto and Georg Menz for discussions about this point.} between several $U$'s.
Namely,  introducing the factor
 $\chi(X,U)=\frac{\abs{\{B\in \Bcal(X)\colon \overline{{B}^*}=U\}}}{\abs{X}}$ 
 \lsm[hgXa]{$\chi(X,U)$}{$=\begin{cases}\frac{\abs{\{B\in \Bcal_k(X)\colon \overline{{B}^*}=U\}}}{\abs{X}} \text{ if } X\in \Scal_k(\L_N),\\
 \1_{U=\overline{X}} \text{ if } X\in \Pcal_k(\L_N)\setminus\Scal_k(\L_N),
 \end{cases}
$ for any connected $U\in \Pcal_{k+1}$}
for any $X\in \Scal(\L_N)$
 and $\chi(X,U)=\1_{U=\overline{X}}$ for $X\in \Pcal(\L_N)\setminus\Scal(\L_N)$ (including the case of $X$ consisting of several disjoint components from $\Scal(\L_N)$),
 we have 
 \begin{equation}
\label{E:IIKK}
(\tilde I\circ\widetilde K)(\Lambda_N,\varphi,\xi)=
\sum_{U\in {\Pcal}^\prime} {I^\prime}^{\L_N\setminus U}(\varphi)
\Bigl[\chi(X,U) \sum_{X\subset U} \tilde I^{U\setminus X}(\varphi) \widetilde K(X,\varphi,\xi)   \Bigr].
\end{equation}
Here we used the observation that, for any $X\in \Scal(\L_N)$ contributing to several $U$'s, we get $\sum_{ U\in \Pcal^{\prime}}\chi(X,U)=1$  and, also,
 that $X\subset {B}^*$ and thus $\overline{X}\subset \overline{{B}^*}$.
 
Defining now
\begin{equation}
\label{E:K'}
K^\prime (U,\varphi)=  \sum_{X\subset U } \chi(X,U) \tilde I^{U\setminus X}(\varphi) \int_{\cbX}\widetilde K(X,\varphi,\xi)  \mu_{k+1}(\d\xi)
\end{equation}
for any connected $U\in \Pcal^{\prime}$, and extending the definition by taking the corresponding product over connected components for a non-connected $U$,
we get
\begin{equation}
\bR(I\circ K)(\L_N,\varphi)=  ( I^\prime\circ  K^\prime )(\L_N,\varphi)
\end{equation}
in view of   \eqref{E:inttildeK} and \eqref{E:IIKK}.

Notice that if $K$ is $L^k$-periodic, then $K^{\prime}$ is obviously $L^{k+1}$-periodic. Also, the transform conserves the factorisation property of the coordinate $K$:
if $K$ factors on the scale $k$,
\begin{equation}
X, Y \in \Pcal, \text{ and } X\cap Y=\emptyset, \text { then }
K(X\cup Y, \varphi) = K(X, \varphi)  K(Y, \varphi), 
\end{equation}
then $K^{\prime}$ factors on the scale $k+1$.

Indeed, let $X_1,X_2\in \Pcal$ be such that their closures in $\Pcal^{\prime}$ are disjoint.
Then (assuming that $L> 2^{d+2}$) the range $\tfrac12 L^{k+1}$ of the covariance of
$\mu_{k+1}$ plus twice the possible reach of up to $2^d L^k$ of ${X_1}^*$ and ${X_2}^*$ out of the closures of $X_1$ and $X_2$, respectively, does not surpass the minimal distance $L^{k+1}$ of the closure  of $X_1$ from the closure of $X_2$, and thus
\begin{equation}
(\bR \widetilde K)(X_1\cup X_2, \varphi) = (\bR \widetilde K)(X_1, \varphi)  (\bR \widetilde K)(X_2, \varphi),
\end{equation}
inheriting the property from $K$, $ I$, and $\tilde I$.
Now it is easy to observe that this fact actually means that $K^{\prime}$ factors, as the pairs of sets contributing, according to \eqref{E:K'},  to  $K^{\prime}(U_1,\varphi)$ and $K^{\prime}(U_2,\varphi)$
with disjoint $U_1$ and $U_2$ are necessarily as discussed above.

Let us  summarise, reinstating the index $k$, what we have got.
\begin{prop}
\label{P:k->k+1}
Let  $k\in\{0,\ldots, N -1\}$,  $H_k\in M_0(\Bcal_k, \cbX)$,
and \\ $K_k\in M( \Pcal_k, \cbX)$ be such that it factors.
Let $H_{k+1}\in M_0(\Bcal_{k+1}, \cbX)$ be defined by 
 \lsm[Hcka1a]{$H_{k+1}(B^{\prime},\varphi)$}{$=\sum_{B\in \Bcal_k(B^{\prime})}
\Pi_2 \Bigl((\bR_{k+1} H_k)(B, \varphi)-
\sum_{\begin{subarray}{c}  X\in \cS^{(k)} \\  X\supset B  \end{subarray}} 
\tfrac1{\abs{X}_k} (\bR_{k+1} K_k)(X,\varphi)\Bigr)$}
\begin{equation}
\label{E:H_k+1}
H_{k+1}(B^{\prime},\varphi)=\sum_{B\in \Bcal_k(B^{\prime})}
\widetilde H_k(B,\varphi),
\end{equation}
where 
 \lsm[Hckaz~]{$\widetilde H_k(B,\varphi)$}{$=\Pi_2 \Bigl((\bR_{k+1} H_k)(B, \varphi)-
\sum_{\begin{subarray}{c}  X\in \cS^{(k)} \\  X\supset B  \end{subarray}} 
\frac1{\abs{X}_k} (\bR_{k+1} K_k)(X,\varphi)\Bigr)$}
\begin{equation}
\label{E:tildeH}
\widetilde H_k(B,\varphi)=\Pi_2 \Bigl((\bR_{k+1} H_k)(B, \varphi)-
\sum_{\begin{subarray}{c}  X\in \cS_k \\  X\supset B  \end{subarray}} 
\frac1{\abs{X}_k} (\bR_{k+1} K_k)(X,\varphi)\Bigr).
\end{equation}
Using  
$\widetilde K_k(\varphi,\xi)=\bigl(1-\ex^{-\tilde H_k(\varphi)}\bigr)\circ \bigl(\ex^{-H_k(\varphi+\xi)}-1\bigr)\circ K_k(\varphi+\xi)$, 
let $K_{k+1}\in M( \Pcal_{k+1}, \cbX)$  be defined by
 \lsm[Kcka1aUcpg]{$K_{k+1} (U,\varphi)$}{$=  \sum_{X\in\Pcal_k(U) } \chi(X,U)
 \exp\bigl\{-\sum_{B\in  \Bcal_k({U\setminus X})}
 \widetilde H_k(B,\varphi) \bigr\}\int_{\cbX}\widetilde K_k(X,\varphi,\xi) \mu_{k+1}(\d\xi)$} 
\begin{equation}
\label{E:K_k+1}
K_{k+1} (U,\varphi)=   \sum_{X\in\Pcal_k(U) } \chi(X,U)
 \exp\Bigl\{-\sum_{B\in  \Bcal_k({U\setminus X})}
 \widetilde H_k(B,\varphi) \Bigr\}\int_{\cbX}\widetilde K_k(X,\varphi,\xi) \mu_{k+1}(\d\xi)
\end{equation}
for any connected $U\in \Pcal^{\prime}$, with 
\begin{equation}
\label{E:chi}
\chi(X,U)=\begin{cases}\frac{\abs{\{B\in \Bcal_k(X)\colon \overline{{B}^*}=U\}}}{\abs{X}} \text{ if } X\in \Scal_k(\L_N),\\
 \1_{U=\overline{X}} \text{ if } X\in \Pcal_k(\L_N)\setminus\Scal_k(\L_N),
 \end{cases}
 \end{equation}
and by the corresponding product over connected components for any non-connected $U$.
Then $K_{k+1}\in M( \Pcal_{k+1}, \cbX)$, it factors, and
\begin{equation}
\label{E:Rk+1.}
\bR_{k+1}({\rm e}^{-H_k}\circ K_k)(\Lambda_N,\varphi)=({\rm e}^{-H_{k+1}}\circ K_{k+1})(\Lambda_N,\varphi).
\end{equation}
\end{prop}

As a result,   introducing 
\begin{equation}
\bT_k(H_k,K_k,\bq)=(H_{k+1},K_{k+1})
\end{equation}
with $H_{k+1}$ and $K_{k+1}$ defined by equations (\ref{E:H_k+1} -- \ref{E:K_k+1}), we get the renormalization map 
 \lsm[Tdka]{$\bT_k$}{map from $M_0(\Bcal_k, \cbX)\times  M( \Pcal_k, \cbX)\times \R^{d\times d}_{\rm sym}$ to $M_0(\Bcal_{k+1}, \cbX)\times M( \Pcal_{k+1}, \cbX)$,  $\bT_k((H_k,K_k))= (H_{k+1},K_{k+1})$}
\begin{equation}
\label{E:Tk}
\bT_k\colon  M_0(\Bcal_k, \cbX)\times  M( \Pcal_k, \cbX)
\times \R^{d\times d}_{\rm sym} \to M_0(\Bcal_{k+1}, \cbX)\times M( \Pcal_{k+1}, \cbX),
\end{equation}
$ k=0,1,\ldots, N-1 $.

\section{Key properties of the renormalisation transformation}\label{S:key}

Of course, defining the renormalisation map $\bT_k$ satisfying \eqref{E:Rk+1.} is only half of our task of the definition of the  renormalisation transform.
Another part lies in the verification that the choice of coordinates $H_k$ and $K_k$ 
together with the map  $(H_k,K_k)\mapsto (H_{k+1},K_{k+1})$ indeed isolates relevant and irrelevant variables with correct estimates. 
Notice that in the definition of $\bT_k$,  we explicitly included  the dependence on the matrix $\bq$.
It stems from the dependence
of  the starting Gaussian measure $\mu=\mu_{{\Cscr}^{(\bq)}}$ (and of the corresponding  generalised Laplacian  ${\Ascr}^{(\bq)}$)
on  $\bq$ and it transfers into such a dependence also for  the operators  ${\Cscr}^{(\bq)}_{\ssup{k}}$
obtained from the finite range decomposition,   for the corresponding Green functions  ${\mathcal C}^{(\bq)}_{k,0}$ and the measures $\mu_{k}$, and,
eventually, for the operators $\bT_k$. Even though this dependence often does not appear in our notation, in the following two Propositions, where we  state its key properties, we explicitly address this dependence and make it thus explicit also in the notation. 
For variables $H$ and $K$ we again skip the subscript $k$ and replace $k+1$ by a prime. 

It is easy to verify that, for any $\bq$, the origin $(H,K)=(0,0)$ 
is a fixed point of the transformation $\bT_k$.
Further,  the $H$-coordinate of the operator $ {\bT_k}$ has actually a linear dependence; we can write
\begin{equation}\label{map1}
{\bT_k}(H,K,\bq)=
(\bA_k^{(\bq)}H +  \bB_k^{(\bq)} K,  S_k(H,K,\bq))
\end{equation}
\lsm[Scka]{$S_k$}{the map $S\colon  M_0(\Bcal_k, \cbX)\times  M( \Pcal_k^{\com}, \cbX) \times \R^{d\times d}_{\rm sym} \to  M( (\Pcal_{k+1})^{\com}, \cbX)$ given by \\ $S (H_k,K_k,\bq)=K_{k+1}$}%
with appropriate linear operators $\bA_k^{(\bq)}$ and $\bB_k^{(\bq)}$.
While delegating the discussion of the explicit form and the properties of these operators 
 (as well as the linearization of the map $S_k$) to Proposition~\ref{P:Tlin},
we begin with the smoothness of the nonlinear part $S_k$.

The map $S_k$ is given as a composition of several maps and its smoothness will be a consequence of the smoothness of the composing maps.
To verify its smoothness  we find it useful to introduce a  notion differentiability that is rather easy to verify. 
\begin{definition} 
\label{4:C_*}
Let $\bX$ and $\bY$  be  normed linear spaces and $ \Ucal\subset \bX$ be open. 
We use  $C_*^m(\Ucal ,\bY)$ to denote the set of functions $G: \Ucal\to \bY$  such that for each $j\le m$ and $\dot x\in X$, the directional derivative 
\begin{equation}
\label{E:directionalp}
D^jf(x,\dot{x}^j) = \frac{\d^j}{\d t^j}G(x+t\dot{x})\Big|_{t=0}
\end{equation}
at any $x\in\Ucal$ exists and 
the map $(x,\dot{x})\in \Ucal\times \bX \to D^jG(x,\dot{x}^j)\in \bY$ is continuous. 
\end{definition}
The  technical reasons  for this definition will be apparent later and are explained in great detail in  Appendix~\ref{appChain}.
It turns out that this notion is weak only apparently. In particular,  for 
$m\ge 0$ the space $C^{m+1}_*(\Ucal,\bY)$ is contained in the usual space $C^m(\Ucal ,\bY)$ of Fr\'echet differentiable functions (with operator norms on multilinear forms from $L_m(\bX,\bY)$),  see Proposition~\ref{P:rel}.

Exploring the smoothness of the nonlinear part $S_k$ of the operator $ {\bT_k}$,  we run into  problems stemming from  a loss of regularity when deriving $S_k$ with respect to the parameter $\bq$. For example, it turns out that
 \begin{equation}\label{estsmoothnessnonlinp}
 \norm{D^{j'}_1 D^{j''}_2 D^{\ell}_3 S_k(H,K,\bq)
(\dot{H}^{j'},\dot{K}^{j''},\dot{\bq}^\ell)}^{(\mathsf{A})}_{k+1,r-2\ell}\le C\norm{\dot{H}}_0^{j'}(\norm{\dot{K}}_{k,r}^{(\mathsf{A})})^{j''}
\norm{\dot{\bq}}^{\ell},
\end{equation}
where the norm  $\norm{\cdot}^{(\mathsf{A})}_{k+1,r-2\ell} $ in the target space is weaker than the norm $\norm{\cdot}^{(\mathsf{A})}_{k,r}$ in the  domain space.
As a result we are compelled to consider the map $S_k$ with a suitable sequence of normed spaces $\bM=\bM_{r_0}\embed\bM_{r_0-2}\embed\dots\embed \bM_{r_0-2m}$, $r_0>2m$, defined as  the  spaces $M_r( \Pcal_k^{\com}, \cbX)$  endowed with the norms   $\norm{\cdot}_{k,r}^{(\mathsf{A})}$, $r=r_0,r_0-2,\dots,r_0-2m$, respectively,
and the space $\bM_0$ defined as  $M( \Bcal_k, \cbX)$ with the norm $\norm{\cdot}_{k,0}$.
Similarly, $\bM'=\bM'_{r_0}\embed\bM'_{r_0-2}\embed\dots\embed \bM'_{r_0-2m}$ are defined as  $M( \Pcal_{k+1}^{\com}, \cbX)$  with the norms $\norm{\cdot}_{r,k+1}^{(\mathsf{A})}$, $r=r_0,r_0-2,\dots,r_0-2m$.
Further, we will use  $\widetilde{\bM}_r$ to denote the closure of $\bM$ in $\bM_r$, and similarly for $\widetilde{\bM}'_{r}$.
\lsm[Mxcr]{$\bM_r$}{=$\bM_{k,r}=(M_r( \Pcal_k^{\com}, \cbX),\norm{\cdot}_{k,r}^{(\mathsf{A})})$}
\lsm[Mxc0]{$\bM_0$}{=$\bM_{k,0}=(M( \Bcal_k, \cbX),\norm{\cdot}_{k,0})$}

Considering now open subsets $\Ucal \subset \bM_0\times\bM$ and $\Vcal\subset \R^{d\times d}_{\rm sym}$, we will introduce 
the class of functions that can be described as those $G\colon\Ucal\times\Vcal\to \bM'$
\lsm[Cxcx]{$\widetilde C^m(\Ucal \times\Vcal)$}{the class of functions $G\colon\Ucal\times\Vcal\to \bM'$ for which the derivative $D^{j'}_1D^{j''}_2D^\ell_3G$ is a continuous map 
$\Ucal \times\Vcal\times{\bM}_{0}^{j''} \times \widetilde{\bM}_{r}^{j'} \times (\R^{d\times d}_{\rm sym})^\ell  \to \bM'_{r-2\ell}$}
for which the derivative $D^{j'}_1D^{j''}_2D^\ell_3G$ is a continuous map 
$\Ucal \times\Vcal\times{\bM}_{0}^{j''} \times \widetilde{\bM}_{r}^{j'} \times (\R^{d\times d}_{\rm sym})^\ell  \to \bM'_{r-2\ell}$. More formally, we introduce 
the set  $\widetilde C^m(\Ucal \times\Vcal,\bM^\prime)$ of maps $G:\Ucal\times\Vcal\to \bM'$ as follows
(see Definition~\ref{D:tildeC} in a more general setting):
\begin{definition}\label{D:4tildeC}
Let $r_0,m\in \N$, $r_0>2m$. We define $\widetilde C^m(\Ucal \times\Vcal,\bM^\prime)$ as the set of all maps $G:\Ucal\times\Vcal\to \bM'$ such that
\begin{enumerate}[\rm (a)]
\item\label{tildeCap}
$G\in C^m_*(\Ucal \times\Vcal,\bM'_{r_0-2m})$.
\item\label{E:tildeCb}
For each $0\le j'+j''+\ell\le m$, the function
\begin{multline*}
(H,K,\bq,\dot{H}_1,\dots ,\dot{H}_{j'},\dot{K}_1,\dots ,\dot{K}_{j''},\dot{\bq}_1,\dots, ,\dot{\bq}_\ell)\to \\
\to D^{j'}_1D^{j''}_2D_3^\ell G((H,K,\bq),\dot{\bq}_1,\dots, ,\dot{\bq}_\ell,\dot{K}_1,\dots ,\dot{K}_{j''},\dot{H}_1,\dots ,\dot{H}_{j'}),
\end{multline*}
(which is by an implication  of the claim (a)  (see Theorem~\ref{T:*=Ham}) defined as a map $\Ucal \times\Vcal\times  {\bM}_{0}^{j'} \times{\bM}^{j''} \times (\R^{d\times d}_{\rm sym})^\ell  \to \bM'_{r_0-2m}$) has an extension to a continuous mapping 
$\Ucal\times\Vcal \times{\bM}_{0}^{j'}\times \widetilde{\bM}_{r_0-2m+2\ell}^{j''}\times (\R^{d\times d}_{\rm sym})^\ell \to  \bM'_{r_0-2m}$. This extension is also denoted $D^{j'}_1 D^{j''}_2 D^{\ell}_3 G$.
\item\label{tildeCcp}
For each $0\le j'+j''+\ell\le m$ and $r=r_0, r_0-2, \dots, r_0-2m+2\ell$,   the restriction of  $D^{j'}_1 D^{j''}_2 D^{\ell}_3 G$  to $\Ucal\times\Vcal \times{\bM}_{0}^{j'}\times \widetilde{\bM}_{r}^{j''}\times (\R^{d\times d}_{\rm sym})^\ell$  (notice that it has  been already extended by (b)) has values in $\bM'_{r-2\ell}$ and is continuous as a mapping between these spaces.
\end{enumerate}
\end{definition}
Again, see Appendix~\ref{appChain}   for further context and properties of the notion of smoothness introduced in this way.
Contrary to Definition~\ref{D:tildeC} we abstain  from invoking the relevant sequences of normed spaces in the notation as here they are fixed from the context.

In the following we will consider the constants $d$, $\omega$, and $r_0$ to be fixed (assuming $d=2,3$,  $\omega\ge 2(d^2 2^{2d+1}+1)$ and we will not  mention possible dependence of various constants (like  $L_0$, $h_0$, and $\mathsf{A}_0$ below) on it.
For the proof of the results in Chapter~\ref{S:results} $ r_0=9 $ is sufficient, see comment in Remark~\ref{R:Tnonlin}).

For fixed values of  the parameters $L, h$, and $\mathsf{A}$ in the definition of the norms in Chapter~\ref{S:norms},
 let $\Ucal_{\rho}\subset \bM_0\times  \bM_{r_0}$ and $\Vcal\subset \R^{d\times d}_{\rm sym}$ be the neighbourhoods of the origin,
\lsm[UydgMc]{$\Ucal_{\rho}$}{=$\{(H,K) \in M_0(\Bcal_k, \cbX)\times  M(\Pcal_k, \cbX) \colon
\norm{ H}_{k,0} < {\rho}, \norm{ K}_{k,r_0}^{(\mathsf{A})} < {\rho} \}$}
\begin{equation}
\Ucal_{\rho}=\{(H,K) \in  \bM_0\times  \bM_{r_0} \colon
 \norm{ H}_{k,0} < {\rho}, \norm{ K}_{k,r_0}^{(\mathsf{A})} < {\rho} \}
\end{equation}
and 
\lsm[Vy]{$\Vcal$}{=$\{\bq \in  \R^{d\times d}_{\rm sym} \colon \norm{ \bq} < 1/2\}$}
\begin{equation}
\Vcal=\{\bq \in  \R^{d\times d}_{\rm sym} \colon \norm{ \bq} < 1/2\}.
\end{equation}
%and $ \Ocal_{{\rho}}=\{(H^\prime,K^\prime)\in M_0(\Bcal^\prime, \cbX)\times  M(\Pcal^\prime, \cbX)\colon \norm{ H^\prime}_0 <{\rho}, \norm{ K^\prime}_{r_0}<{\rho}\} $.
%%
%\lsm[OydgMc]{$\Ocal_{{\rho}}$}{=$\{(H,K)\in M_0(\Bcal_{k+1}, \cbX)\times  M(\Pcal_{k+1}, \cbX)\colon \norm{ H}_0 <{\rho}, \norm{ K}_{k+1,r_0}<{\rho}\}$}
%%
%
  
\begin{prop}[Smoothness of the nonlinear part $S_k$]
\label{P:Tnonlin}

There exists  a constant $L_0$ and, for any $L\ge L_0$, constants $h_0(L)$ and $\mathsf{A}_0(L)$, and for any $\mathsf{A}\ge \mathsf{A}_0$
a constant ${\rho}={\rho}(\mathsf{A})$ such that, for any $k=0,\dots, N-1$, any $L\ge L_0$, $h\ge   h_0$, and $\mathsf{A}\ge \mathsf{A}_0$   we have
\begin{equation}
\label{E:SinC}
 S_k\in \widetilde C^m(\Ucal_{\rho} \times\Vcal,\bM^\prime),
 \end{equation} 
 and there is a constant ${C}={C}(L,h,\mathsf{A})>0 $  such that 
 \begin{equation}\label{estsmoothnessnonlin}
\norm{D^{j'}_1 D^{j''}_2 D^{\ell}_3 S_k(H,K,\bq)
(\dot{H}^{j'},\dot{K}^{j''},\dot{\bq}^\ell)}^{(\mathsf{A})}_{k+1,r-2\ell}\le C\norm{\dot{H}}_0^{j'}(\norm{\dot{K}}_{k,r}^{(\mathsf{A})})^{j''}
\norm{\dot{\bq}}^{\ell},
\end{equation}
for any $(H,K)\in\Ucal_{{\rho}}$, $\bq\in\Vcal$, $0\le j'+j''+\ell\le m$, and $r=r_0, r_0-2, \dots, r_0-2m+2\ell$. 
%Moreover the derivatives are continuous in the following sense
%\begin{multline}
%\label{E:limder}
%\!\!\!\!\!\!\!\!
%\lim_{(\widetilde{H},\widetilde{K},\Qcal )\to(H,K,\bq)}\frac{1}{\norm{\dot{H}}_0^j\norm{\dot{K}}_{k,r}^i\norm{\dot{\bq}}^\ell}\Big\Vert D_{\bq}^{\ell}D_H^jD_K^i S_k(\widetilde{H},\widetilde{K},\widetilde{\bq})(\dot{H},\ldots,\dot{K},\ldots,\dot{K},\ldots,\dot{\bq})-\\
% -D_{\bq}^{\ell}D_H^jD_K^i S_k(H,K,\bq)(\dot{H},\ldots,\dot{H},\dot{K},\ldots,\dot{K},\dot{\bq},\ldots,\dot{\bq})\Big\Vert_{k+1,r-2\ell}^{(\mathsf{A})}=0,
%\end{multline}
%where $ \dot{H}\in M(\Bcal,\cbX)\setminus\{0\},\dot{K}\in M(\Pcal^{\com},\cbX)\setminus\{0\} $, and $ \dot{\bq}\in\R_{\rm sym}^{d\times d}\setminus\{0\} $. Here the limit is taken in the topology of $ \Ucal_\rho $, i.e., in particular $ \widetilde{K}\to K $ in $ \norm{\cdot}_{r_0} $.
\end{prop}

The proof will be deferred to  Chapter~\ref{S:smooth}, where we will split $S_k$ into a composition of several
partial maps and deal with their smoothness separately, isolating in detail the needed restrictions on various constants.
Here, instead, we offer  a heuristic explanation of the role of the principal constants.
The restrictions on $L$ are purely geometric (see Lemma~\ref{L:prop}, Lemma~\ref{L:contraction}, Lemma~\ref{L:large_set}, Lemma~\ref{L:small_set}, Lemma~\ref{L:ABcontraction}). In particular, by assuming that $L\ge L_0$ we have $L\ge 2^{d+1}$
implying, for example, that if $B\in\Bcal_k$, then the cube $B^*$ has the side at most $L^{k+1}$ and thus $\overline{B^*}\in \Scal_{k+1}$. The restrictions on the constant $h$ are more subtle (see  Lemma~\ref{L:prop}, Lemma~\ref{L:contraction}, Lemma~\ref{L:large_set}, Lemma~\ref{L:small_set}). Its role is to suppress large fields
in the norms $\norm{F(X)}_{k,X,r}$ and $\tnorm{F(X)}_{k,X}$
by employing the $h$-dependent  weight factors $W_k^X$ and $w_k^X$, respectively.
When evaluating  the norms of the maps $(H,K)\to \widetilde H$ (see \eqref{E:tildeH}) and
$K\to  \bR_{k+1} (K)$, a major part of the coarse grained increase is absorbed into the growth $L^k\to L^{k+1}$ of the corresponding factors in the functions $G_{k,x}$ and $g_{k,x}$
entering the weight factors. However, some surplus  remains stemming essentially
from the term $L^{\upeta(n,d)}$ in the fluctuation bound \eqref{E:fluctk} of the finite range decomposition.
A suppression of the relevant term is obtained by assuming that $h\ge h_0(L)=h_1 L^{\frac{d^2}2 +5d +16}$ with
$h_1$ depending only on $d$ and $\omega$. Finally, the constant $\mathsf{A}$ is responsible for combining the norms $\norm{\cdot}_{k,X,r}$ into a single norm $\norm{\cdot}_{k,r}^{(\mathsf{A})}$ (see Lemma~\ref{L:P_1} and Lemma~\ref{L:large_set}). However, it turns out that the map $K\to  \bR_{k+1} (K)$ leads to acquiring a factor
$2^{\abs{X}_k}$ in the norm $\norm{\cdot}_{k,X,r}$, yielding an inevitable loss in $\mathsf{A}$ in the norm $\norm{\cdot}_{k,r}^{(\mathsf{A})}$. Nevertheless, the loss can be recovered when combining the terms in \eqref{E:K_k+1} while passing to the next scale. Namely, using in the resulting sum stemming from evaluating the norm of \eqref{E:K_k+1} the geometric bound $\abs{X}_k\ge (1+\upalpha(d))\abs{\overline{X}}_{k+1} - (1+\upalpha(d)) 2^{d+1} \abs{\Ccal(X)} $ with a constant $\upalpha(d)>0$, we get the original $\mathsf{A}$ once we suppose that the map is restricted to sufficiently small domain,  e.g. assuming that  $\norm{\bR_{k+1} (K)}_{k:k+1,r}^{(\mathsf{A})}\le {\rho}(\mathsf{A})=(2\mathsf{A}^{2^{d+3}})^{-1}$ and taking $\mathsf{A}$ sufficiently large depending on $L$ (and $d$).

The next claim deals with the linearisation of the map $\bT_k$ at the fixed point $(H,K)=(0,0)$.
 For a linear operator $\bL$ between Banach spaces, we consider here the standard norm 
 \lsm[zzLd]{$\norm{ \bL}$}{$= \sup\{\norm{ \bL(f)}\colon\norm{ f}\le 1\}$,  norm of a linear operator $\bL$ between Banach spaces}
$\norm{ \bL}= \sup\{\norm{ \bL(f)}\colon\norm{ f}\le 1\}$,  with appropriate norms on the corresponding spaces.
Usually we indicate the corresponding norms in an appropriate way, e.g., $\norm{ \bL}_{k,r;k+1,0}$ and 
$\norm{ \bL}_{k,r;k+1,r}$, or simply $\norm{ \bL}_{r;0}$ and $\norm{ \bL}_{r}$,
 for a linear mapping  $\bL:\bM_r\to \bM'_0$ and $\bL:\bM_r\to \bM'_r$, respectively.

\begin{prop}[Linearisation of $\bT_k$]
\label{P:Tlin} 
\noindent The first derivative at $H=0$ and $K=0$ have a triangular form,
 \lsm[Bdkaqb]{$\bB_k^{(\bq)}$}{$(\bB_k^{(\bq)} \dot{K},0)= D\bT_k(0,0,\bq)(0,\dot{K})$, linearisation of 
 $\bT_k(\cdot,\cdot,\bq)$ at $(0,0)$,\\
 $(\bB_k^{(\bq)}\dot{K})(B^{\prime},\varphi)= -\sum_{B\in\Bcal_k(B^{\prime})} \Pi_2 
\sum_{\begin{subarray}{c}  X\in \cS_k \\  X\supset B  \end{subarray}} 
\frac1{\abs{X}_k} \Bigl( \int_{\cbX}
 \dot{K}(X,\varphi+\xi)
\mu^{(\bq)}_{k+1}(\d\xi)\Bigr)$\hfill}   
 \lsm[Adkaqb]{${\bA}^{(\bq)}_k$}{$({\bA}^{(\bq)}_k \dot{H},0)= D\bT_k(0,0,\bq)(\dot{H},0)$, linearisation of $\bT_k(\cdot,\cdot,\bq)$ at $(0,0)$, \\
$ ({\bA}^{(\bq)}_k\dot{H})(B^{\prime},\varphi)= \sum_{B\in\Bcal_k(B^{\prime})}
\bigl[\dot{H}(B,\varphi) + \sum_{x\in B}
\sum_{i,j=1}^d \dot{d}_{i,j} \nabla_i \nabla_j^*{\mathcal C}^{(\bq)}_{k+1}(0)\bigr]$}
 \lsm[Cdkaqb]{$\bC^{(\bq)}_k$}{$(0,\bC^{(\bq)}_k\dot{K})= D\bT_k(0,0,\bq)(0,\dot{K})$, linearisation of $\bT_k(\cdot,\cdot\bq)$ at $(0,0)$,\\
 $\bC^{(\bq)}_k(\dot{K})(U,\varphi)= \sum_{B: \overline{B^*}=U}\bigl(1 - \Pi_2\bigr)
\sum_{\begin{subarray}{c} Y\in  \Scal_k  \\  Y\supset B\end{subarray}} 
\frac1{\abs{Y}} \Bigl( \int_{\cbX} \dot{K}(Y,\varphi+\xi) \mu^{(\bq)}_{k+1}(\d\xi)\Bigr)+\\+ \sum_{\begin{subarray}{c}  X\in  \Pcal_k^{\rm c} \setminus \Scal_k \\  \overline X= U \end{subarray}} 
 \int_{\cbX}\dot{K}(X,\varphi+\xi) \mu^{(\bq)}_{k+1}(\d\xi)$}
\begin{equation}
\label{E:DbT}
D\bT_k(0,0,\bq)(\dot{H},\dot{K})=\left(\begin{matrix}\bA_k^{(\bq)} & \bB_k^{(\bq)}\\  \zero & \bC_k^{(\bq)}\end{matrix}\right)
\left(\begin{matrix}\dot{H}\\  \dot{K}\end{matrix}\right),
\end{equation}
with
\begin{equation}
\label{E:A}
(\bA_k^{(\bq)}\dot{H})(B^{\prime},\varphi)= \sum_{B\in\Bcal(B^{\prime})}
\bigl[\dot{H}(B,\varphi) + \sum_{x\in B}
\sum_{i,j=1}^d \dot{d}_{i,j} \nabla_i \nabla_j^*{\mathcal C}^{(\bq)}_{k+1}(0)\bigr],
\end{equation}
\begin{equation}
\label{E:B}
(\bB_k^{(\bq)}\dot{K})(B^{\prime},\varphi)= -\sum_{B\in\Bcal(B^{\prime})} \Pi_2 
\sum_{\begin{subarray}{c}  X\in \cS \\  X\supset B  \end{subarray}} 
\frac1{\abs{X}} \Bigl( \int_{\cbX}
 \dot{K}(X,\varphi+\xi)
\mu^{(\bq)}_{k+1}(\d\xi)\Bigr),
\end{equation}
and
\begin{multline}
\label{E:C}
(\bC_k^{(\bq)}\dot{K})(U,\varphi)= \sum_{B: \overline{B^*}=U}\bigl(1 - \Pi_2\bigr)
\sum_{\begin{subarray}{c} Y\in  \Scal  \\  Y\supset B\end{subarray}} 
\frac1{\abs{Y}} \Bigl( \int_{\cbX} \dot{K}(Y,\varphi+\xi) \mu^{(\bq)}_{k+1}(\d\xi)\Bigr)+\\+\mspace{-15mu} \sum_{\begin{subarray}{c}  X\in  \Pcal^{\rm c} \setminus \Scal \\  \overline X= U \end{subarray}} 
 \int_{\cbX}\dot{K}(X,\varphi+\xi) \mu^{(\bq)}_{k+1}(\d\xi).
\end{multline}

Further, let $\theta\in (1/4,3/4)$ and let $L_0$ and  $h_0=h_0(L)$ be as in Proposition~\ref{P:Tnonlin}. 
 \lsm[hgt]{$\theta$}{a contractivity constant for operator $\bC_k^{(\bq)}$}
There exists a constant $M=M(d)$ and, for any $L\ge L_0$,  a constant  $\mathsf{A}_0=\mathsf{A}_0(L)$,
such that for any $h\ge   h_0(L)$  and any $\mathsf{A}\ge \mathsf{A}_0(L)$,
the following bounds on the norms of operators $\bA_k^{(\bq)}$, $\bB_k^{(\bq)}$, and $\bC_k^{(\bq)}$ hold independently of $N$ and $k$ and for any $\norm{ \bq}\le  \frac{1}{2}$:
\begin{equation}
\label{E:boundsPropPT}
\norm{  \bC_k^{(\bq)} }_{r} \le \theta,\norm{{\bA_k^{(\bq)}}^{-1}}_{r;r} \le \frac1{\sqrt{\theta}}, \text{ and }  \norm{ \bB_k^{(\bq)}}_{r;0} \le M L^d,
\end{equation}
$r\ge 3$, and for all $\mathsf{A}\ge \mathsf{A}_0$ (note that for the contraction bound for $ \bC^{\ssup{\bq}} $ the choice $ h\ge h_0 $ is sufficient).
\end{prop}

\begin{remark}
 \label{R:Tnonlin}
\noindent (i) Notice that as a consequence of Proposition~\ref{P:Tnonlin},
the operators $\bA_k^{(\bq)}$, $\bB_k^{(\bq)}$, and $\bC_k^{(\bq)}$ are $m$-times differentiable with respect to $\bq$, $\norm{\bq}\le \frac{1}{2} $,  and there exists a finite constant $C=C(h,L)>0 $ such that
\begin{equation}\label{E:boundsOperatorswrt_q}
\norm{\partial^{\ell}_{\bq}  \bA_k^{(\bq)}\dot{H}}_{0}\le C \norm{\dot{H}}_0,  \
\norm{\partial^{\ell}_{\bq} \bB_k^{(\bq)}\dot{K}}_{0}\le C \norm{\dot{K}}_{2\ell+2}, \ 
\norm{\partial^{\ell}_{\bq} \bC_k^{(\bq)}\dot{K}}_{r-2\ell}\le C \norm{\dot{K}}_r, \ 
\end{equation}
for any $\ell=1,2,\ldots,m $ and any $r\ge 2\ell+3$ and $ \mathsf{A}\ge \mathsf{A}_0 $.

\noindent (ii) For the results in Chapter~\ref{S:results} we need $m=3 $. Thus $ r_0=9 $ is sufficient. \hfill $\diamond $

\end{remark}

\begin{proof}[Proof of Proposition \ref{P:Tlin}]
Here, we will only show the validity of the explicit formulas for the operators $\bA_k^{(\bq)}$, $\bB_k^{(\bq)}$, and $\bC_k^{(\bq)}$. The bounds needed for the remaining claims will be proven in Chapter~\ref{sec:contraction}.

Starting from \eqref{E:H_k+1} and \eqref{E:tildeH}, let us expand the linear and quadratic terms in $\dot{H}(B, \varphi+\xi)$ into the sum of the terms depending on $\varphi$, $\xi$, and the term proportinal to $\dot{Q}(\varphi,\xi)$. Observing that  the integral with respect to $\mu_{k+1}(\xi)$ of the terms linear in $\xi$ vanishes and that 
$\Pi_2(\dot{H}(B,\varphi))=\dot{H}(B,\varphi)$, we get the expression \eqref{E:A} for $\bA_k^{(\bq)}$
once we notice that
$\int_{\cbX}\dot{Q}(\xi,\xi) \mu_{k+1}(\d\xi)= \sum_{x\in B}
\sum_{i,j=1}^d \dot{d}_{i,j} \nabla_i \nabla_j^*{\mathcal C}^{(\bq)}_{k+1}(0)$.

The formula \eqref{E:B} follows directly from the second term on the right hand side of  \eqref{E:tildeH}.

When computing  $\bC_k^{(\bq)}$ we first observe that only linear terms in $\widetilde K$ can contribute.  Taking $\dot{H}=0$  and using thus \eqref{E:K_k+1} 
with
\begin{equation}
\widetilde H(B,\varphi)=- \Pi_2 
\sum_{\begin{subarray}{c}  X\in \cS \\  X\supset B  \end{subarray}} 
\frac1{\abs{X}}(\bR \dot{K})(X,\varphi)
\end{equation}
and
$\widetilde K(\varphi,\xi)=\bigl(1-\ex^{-\tilde H(\varphi)}\bigr)\circ  K(\varphi+\xi)$, 
we get
\begin{multline}
\bC_k^{(\bq)}(\dot{K})(U,\varphi)= 
\sum_{Y\in  \Scal} 
\chi(Y,U)  \int_{\cbX} D \widetilde K(0)(\dot{K})(Y,\varphi,\xi) \mu_{k+1}(\d\xi)+\\
+ \sum_{\begin{subarray}{c}  X\in  \Pcal^{\rm c} \setminus \Scal \\  \overline X= U \end{subarray}} 
 \int_{\cbX}D \widetilde K(0)(\dot{K})(X,\varphi,\xi) \mu_{k+1}(\d\xi).
\end{multline}
Writing $\chi(Y,U)=\sum_{\begin{subarray}{c}  B\in Y \\  \overline{B^*}=U  \end{subarray}} \frac1{\abs{Y}}$
and observing that 
\begin{equation}
D \widetilde K(0)(\dot{K})(B,\varphi,\xi)= \dot{K}(B,\varphi+\xi) -  D\ex^{-\tilde H(0)}(\dot{K})(B,\varphi)    \text { for } Y=B,
\end{equation}
\begin{equation}
D \widetilde K(0)(\dot{K})(Y,\varphi,\xi)= 
 \dot{K}(Y,\varphi+\xi) 
    \text { for } Y\not =B,
\end{equation}
and
\begin{equation}
D\ex^{-\tilde H(0)}(\dot{K})(B,\varphi) = 
\Pi_2\sum_{\begin{subarray}{c} Y\in  \Scal  \\  Y\supset B\end{subarray}} 
\frac1{\abs{Y}}   \bigl(\bR \dot{K}\bigr)(Y,\varphi) ,
\end{equation}
we get \eqref{E:C}.
\end{proof}

%Proposition~\ref{P:Tlin} yields the linearization of the maps $\bT_k$  at the fixed point $(H,K)= (0,0) $. 
% Let us now attend to the nonlinear behaviour of the full operator $\bT_k$.
% %
 %\lsm[gaka1s]{$S_{k+1}$}{$=S_{k+1}(H_k,K_k,\bq)$, the remnant of $K_{k+1}$ after subtracting the linearization of $\bT_{k}$}
%

\section{Fine tuning of the initial conditions}
\label{S:initial conditions}
Our next task is to implement in detail the idea of fine tuning outlined in Chapter \ref{S:strategy}.
More specifically we will choose an initial ideal Hamiltonian (as used in \eqref{E:3_starting_full} and defined in \eqref{ideal}),
\begin{equation}  \label{E:4_ideal_Hamiltonian}
\begin{aligned}
\Hcal(x,\varphi)=\lambda + \sum_{i=1}^da_i\nabla\varphi(x)+\sum_{i,j=1}^d\bc_{i,j}\nabla_i\nabla_j\varphi(x) +
 \frac{1}{2}\sum_{i,j=1}^d\bq_{i,j}\nabla\varphi(x)\nabla_j\varphi(x)
\end{aligned}
\end{equation} 

\lsm[HcNcpgubff]{$\Hcal(x,\varphi) $}{initial ideal Hamiltonian}
such that the final ideal Hamiltonian vanishes (note that in Chapter \ref{S:strategy} we considered only the simplified
case $\lambda = a = \bc = 0$).

Given an initial $\Kcal$ we 
want to evaluate the integral
$$ \mathcal Z_{N}(u) = 
\int_{\cbX_{N}} \prod_{x \in \Lambda} \big(1 + \Kcal(x, \varphi)\big) \, \mu(\d\varphi) 
= \int_{\cbX_{N}} (1 \circ \Kcal)(\Lambda, \varphi) \, \mu(\d\varphi).
$$
Analogously  to the calculation in Chapter \ref{S:strategy} cf. \eqref{E:ZNfinal} we can rewrite this integral as
\begin{equation}  \label{E:4_Z_with_Hcal}
\begin{aligned}
\mathcal Z_{N}(u) &= \int_{\cbX_{N}}  \ex^{\Hcal(\Lambda, \varphi)}  \big(  \ex^{-\Hcal} \circ   \ex^{-\Hcal} \Kcal\big)(\Lambda, \varphi)  \, \mu(\d\varphi) \\
&=  \frac{Z_N^{(\bq)}}{Z_N^{(0)}}  \, \ex^{L^{dN} \lambda} \, \,  \int_{\cbX_{N}}    \big(  \ex^{-\Hcal} \circ   \ex^{-\Hcal} \Kcal\big)(\Lambda, \varphi)  \,   \mu^{(\bq)}(\d\varphi) 
\end{aligned}
\end{equation}
where $Z_N^{(\bq)}$ and $Z_N^{(0)}$ are as in Chapter \ref{S:strategy}. Here we used that
$\sum_{x \in \Lambda} \nabla_i \varphi(x)  =0$ and $ \sum_{x \in \Lambda} \nabla_i \nabla_j \varphi(x)=  0$ because
$\varphi$ is periodic.

We will now show that for sufficiently small $\Kcal$ there exists an $\Hcal = \Hscr(\Kcal)$ such that
the second integral in \eqref{E:4_Z_with_Hcal} deviates from $1$ only by an exponential small term and such that
the derivatives of this term with respect to $\Kcal$ are also controlled.

To do so we proceed in two steps. We first show that given  sufficiently small  $\Kcal$ and $\Hcal$ there
exists an ideal Hamiltonian  $\Fcal_1(\Kcal, \Hcal) \in M_{0}$ 
\lsm[FcXcpgxa]{$\Fcal_{1}$}{ideal Hamiltonian map}
and a small 'irrelevant' term 
$\Fcal_{2N}(\Kcal, \Hcal) \in    \bM_{N,r}    $ 
\lsm[FcXcpgxab]{$\Fcal_{2N}$}{irrelevant term of the solution map}
such that
\begin{equation}  \label{E:4_last_KN}
 \int_{\cbX_{N}}    \big(  \ex^{-\Fcal_1(\Kcal,\Hcal)} \circ   \ex^{-\Hcal} \Kcal\big)(\Lambda, \varphi)  \,   \mu^{(\bq)}(\d\varphi)   =
  \int_{\cbX_{N}}  ( 1 + \Fcal_{2N}(\Kcal, \Hcal) )\, \mu^{(\bq)}_{N+1}(\d\varphi).
\end{equation}
As a byproduct of this construction we will see that for $\Kcal=0$ we have  $\Fcal_1(0, \Hcal)=0$ and
$\Fcal_{2N}(0, \Hcal) =0$ for all sufficiently small $\Hcal$.
Together with smoothness results  for $\Fcal_1$  this implies $D_{\Hcal} \Fcal_1(0,0) = 0$ and the implicit
function will guarantee that there exists a unique map $\Hscr$ mapping a neighbourhood of the origin in $ \bE $ to $ \bM_0 $  such  that
% Test for index
%\lsm[Hscr]{$\Hscr$}{ $\Hscr(\Kcal)$ is the self-consistent choice of the initial ideal Hamiltonian for a given perturbation $\Kcal$}
%\lsm[HHH]{AA}{test}
%\lsm[HH]{$\Hscr$}{test $\Hscr$}
%\lsm[Hscr]{text}
%
\begin{equation} \label{E:4_fix_Hscr}
\Fcal_1(\Kcal, \Hscr(\Kcal)) =  \Hscr(\Kcal). 
\end{equation}
Combining this with  \eqref{E:4_last_KN} and \eqref{E:4_Z_with_Hcal} we get
\begin{equation} \label{E:4_formula_logZ}
- \log \mathcal Z_{N}(u) = - \log \frac{Z_N^{(\bq)}}{Z_N^{(0)}}  - \lambda L^{dN} - \log  \int_{\cbX_{N}}  ( 1 + \Fcal_{2N}(\Kcal, \Hscr(\Kcal)) )\, \mu^{(\bq)}_{N+1}(\d\varphi),
\end{equation}
where   
\begin{equation}  \label{E:4_q_for_Hscr}
 \lambda = \pi_0 (\Hscr(\Kcal))\quad \hbox{and}   \quad \bq = \pi_2 (\Hscr(\Kcal))
 \end{equation}
denote the constant term in $\Hscr(\Kcal)$ and the coefficient matrix of the quadratic term, respectively.

%
%{We will first show that given sufficiently small $\Hcal$ and $\Kcal$ there exists a unique  $H_0 = \Fcal(\Kcal, \Hcal)$ such that
%\begin{equation} % \label{E:4_last_KN}
% \int_{\cbX_{N}}    \big(  e^{-\Fcal(\Kcal,\Hcal)} \circ   e^{-\Hcal} \Kcal\big)(\Lambda, \varphi)  \,   \mu^{(\bq)}(d\varphi)   =
%  \int_{\cbX_{N}}  ( 1 + K_N(\Kcal, \Hcal) )\, \mu^{(\bq)}_{N+1}(d\varphi)
%\end{equation}
%(with $K_N(\Kcal, \Hcal)$ exponentially small in $N$).
%
%Note that for $\Kcal = 0$ we can take $H_0=0$ and then  \eqref{E:4_last_KN} holds with $K_N =0$. In other words,
%$\Fcal(0, \Hcal) = 0$ for all sufficiently small $\Hcal$. In particular, assuming differentiability, we get
%$D_\Hcal H_0(0,0) = 0$. Thus the implicit function theorem implies that for sufficiently small $\Kcal$ there
%exists $\Hscr(\Kcal)$ such that
%\begin{equation}   % \label{E:4_fix_Hscr}
%H_0(\Kcal, \Hscr(\Kcal)) =  \Hscr(\Kcal). 
%\end{equation}
%Combining this with  \eqref{E:4_Z_with_Hcal}
%we get 
%\begin{equation}  %  \label{E:4_formula_logZ}
%- \log Z = - \log \frac{Z_N^{(\bq)}}{Z_N^{(0)}}  - \lambda L^{dN} - \log  \int_{\cbX_{N}}  ( 1 + K_N(\Kcal, \Hcal) )\, \mu^{(\bq)}_{N+1}(d\varphi).
%\end{equation}
%where $\lambda$ and $\bq$ are the constant term in $\Hscr(\Kcal)$ and the matrix in the quadratic term of 
%$\Hscr(\Kcal)$, respectively.}

%\oldfour{We now explain how to obtain the maps $H_0$ and $\Hscr$. To obtain $H_0$}

We now first explain how to construct the maps $\Fcal_1$ and $\Fcal_{2N}$.
% Old text (before May 26, 2016)
%Our next task is to choose the initial matrix $\bq$ so that the final ideal Hamiltonian $H_{N}$ vanishes. 
%To this end we  employ the implicit function theorem for a suitably chosen function. 
%It is useful to 
We rewrite the entire cascade of maps $\bT_{k}$ in terms of a single map on a suitably defined Banach space. First, we introduce the Banach spaces 
 \lsm[Yd]{$\bY_{\!\!r}$}{$=\big\{\by=(H_0,H_1,K_1,\dots,H_{N-1},K_{N-1},K_{N})\colon H_k\in \bM_{k,0}, K_k\in \bM_{k,r}\big\} $, the Banach space}
 \lsm[Yc]{$\by$}{an element $\by=(H_0,H_1,K_1,\dots,H_{N-1},K_{N-1},K_{N})$ of $\bY$}
\begin{equation} 
\label{E:Zr}
\bY_{\!\!r}=\big\{\by=(H_0,H_1,K_1,\dots,H_{N-1},K_{N-1},K_{N})\colon H_k\in \bM_{k,0}, K_k\in \bM_{k,r}\big\} 
\end{equation}
with the norms
 \lsm[zzYd]{$\norm{\cdot}_{\bY_r}$}{the norm on $\bY_r$, $\norm{\by}_{\bY_r}= \max_{k\in\{0,\dots,N-1\}} \frac{1}{\eta^{k}}\norm{H_k}_{k,0}\vee \max_{k\in\{1,\dots,N\}}  \frac{\alpha}{\eta^{k}}\norm{K_{k,r}}_k$}
 \lsm[hg]{$\eta$}{a parameter in the norm $\norm{\cdot}_{\bY_{r}}$}
 \lsm[ag]{$\alpha$}{a parameter in the norm $\norm{\cdot}_{\bY_{r}}$}
\begin{equation}
\label{E:normZr}
\norm{\by}_{\bY_{\!\!r}}= \max_{k\in\{0,\dots,N-1\}} \frac{1}{\eta^{k}}\norm{H_k}_{k,0}\vee \max_{k\in\{1,\dots,N\}}  \frac{\alpha}{\eta^{k}}\norm{K_k}_{k,r}
\end{equation}
for $r=1,\dots,r_0$ and with  parameters  $\eta\in (0,1)$ and $ \alpha\ge 1$ to be chosen later.
Here, to avoid ambiguity, we reinstated index $k$ also in the notation for normed spaces; we write $\bM_{k,0}$ and $\bM_{k,r} $ instead of $\bM_{0}$ and $\bM_{r} $
used previously.
Notice that the terms $K_0$ and $H_{N}$ are not present in $\by\in\bY_{\!\!r} $; while the latter is put  to be $0$, the former is singled out as an initial condition for a separate treatment.
Also, notice that $\norm{\by}_{\bY_{\!\!r}}\le \norm{\by}_{\bY_{\!\!r+1}}$ and thus $\bY_{\!\!r+1} \embed \bY_{\!\!r}$.

Taking into account the dependence of $\bT_{k}$ on $\bq$ (the matrix in the quadratic term of $\Hcal$)
 and  on the initial perturbation  $\Kcal\in \bE $ (see \eqref{E:||h})
we define the map
%
%%%\lsm[Tx]{$\cbT$}{the map from $\bY \times \bE \times \R^{d\times d}_{\rm sym}$ to $\bY$}
%
 \lsm[Tx]{$\cbT$}{the map from $\bY \times \bE \times M_0$ to $\bY$}
\begin{equation}
\cbT\colon  \bY_{\!\!r} \times \bE \times M_0 \to \bY_{\!\!r}
\end{equation}
%%%%
%%%\begin{equation}
%%%\cbT\colon  \bY_{\!\!r} \times \bE \times \R^{d\times d}_{\rm sym}\to \bY_{\!\!r}
%%%\end{equation}
by
 \lsm[yz]{$ \overline \by$}{ $=\cbT(\by,K,\Hcal)$ the $2(N+1)$-tuple defined by $\overline H_k$ and $\overline K_k$ }
\begin{equation}
\label{E:cbT}
\cbT(\by,\Kcal,\Hcal)=  \overline \by.
\end{equation}
Here, $\overline \by$ is given by recursive equations,
 \lsm[Hckaz]{$\overline H_k$}{$=\bA^{-1}_{k}\bigl(H_{k+1}-{\bB}_{k}K_k\bigr)$ }
 \lsm[Kcka1sz]{$\overline K_{k+1}$}{$=S_k(H_k,K_k,\bq)=\bC_{k}K_{k}+S_k(H_k,K_k,\bq)-D_2 S_k((0,0,\bq),K_k)$ }
\begin{equation}\label{defofT}
\begin{aligned}
&\overline H_k=\bA^{-1}_{k}\bigl(H_{k+1}-{\bB}_{k}K_k\bigr),       \\
& \overline K_{k+1}=S_k(H_k,K_k,\bq)=\bC_{k}K_{k}+S_k(H_k,K_k,\bq)-D_2 S_k((0,0,\bq),K_k)
\oldfour{.}
\end{aligned}
\end{equation}
for $k = 0, \ldots, N-1$.
Here $\bC_{k}K_{k}= D_2 S_k((0,0,\bq),K_k)$  and $S_k(H_k,K_k,\bq)-D_2 S_k((0,0,\bq),K_k)$ is the nonlinear part of  the map $ S_k$.
%\oldfour{Here $\bC_{k}K_{k}= D_2 S_k((0,0,\bq),K_k)$; and $S_k(H_k,K_k,\bq)-D_2 S_k((0,0,\bq),K_k)$ is the nonlinear part of  the map $ S_k$, and $k=0,\ldots, N-1$.} 
In addition, we set  $H_{N}=0$ and  define 
$K_0\in M(\Pcal_0,\cbX)$  by $K_0 = \ex^{-\Hcal} \Kcal$, i.e., 
%%%$K_0=\Kcal^{\ssup{\bq}}\in M(\Pcal_0,\cbX)$  
%%with $\Kcal^{\ssup{\bq}}$ given by equation \eqref{E:Kq}:
by
\begin{equation}
K_0(X,\varphi)   
:=\prod_{x\in X}\bigl(\exp(-\Hcal(x,\varphi)) \Kcal(\nabla\varphi(x))\bigr)
\end{equation}
with  $\Kcal\in\bE $ and $\Hcal\in M_0 $.
% and with $ H(x,\varphi) $ depending on $ \bq $ and given in \eqref{ideal}. 

Observe now that, for a given  $\Kcal$ and $\Hcal$, the $2N$-tuple  $\by$ is a fixed point of $\cbT$,
i.e., $\cbT(\by,\Kcal,\Hcal)=\by$
% \oldfour{     $\cbT(\by,\Kcal,\bq)=\by$} 
if and only if 
\begin{equation}
%\oldfour{{\bT}_{k} \quad}   
\bT_{k}(H_k,K_k,\bq)=(H_{k+1},K_{k+1}), k=0,\dots,N-1,
\end{equation}
with 
$K_0 = \ex^{-\Hcal} \Kcal$
%%%%%$K_0=\Kcal^{(\bq)}$ 
and $H_{N}=0$. 
Our task thus is to find a map  $\Fcal$
 from a neighbourhood of origin in  $\bE\times M_0$ to $ \bY_r$ so that
\begin{equation}
\label{E:TKqZ}
\cbT(\Fcal(\Kcal,\Hcal),\Kcal,\Hcal)=\Fcal(\Kcal,\Hcal).
\end{equation}

%
%\oldfour{
%\begin{equation*}
%%\label{E:TKqZ_old}
%\cbT(\Fcal(\Kcal,\Hcal),\Kcal,\bq)=\Fcal(\Kcal,\Hcal).
%\end{equation*}
%}

This can be done with help of the Implicit Function Theorem ~\ref{T:implicit}  using the bounds from
Propositions~\ref{P:Tlin} and \ref{P:Tnonlin} to verify its hypothesis. 
In Proposition~\ref{P:hatZ},  we will summarize the  smoothness properties of the obtained fixed point map $\Fcal$.
Note that for $\Kcal=0$ the vector $\by = 0$ is a fixed point for every $\Hcal$. Thus
\begin{equation}  \label{E:4_trivial_Kcal}
\Fcal(0, \Hcal) = 0.
\end{equation}

%We denote by 
%\begin{equation} \label{E:4_define_Fcal1}
%\hbox{ $\Fcal_1$ and   $\Fcal_{2N}$ the first and  last component of $\Fcal$, corresponding to $H_0$ and $K_N$.}
%\end{equation}

Taking now for $\Fcal_1$ and  $\Fcal_{2N}$ the first and  last component of $\Fcal$, corresponding to $H_0$ and $K_N$, the equality  \eqref{E:4_last_KN}  readily follows from the definition of $\Fcal$.

Now we can easily construct the map $\Hscr$. The condition 
\eqref{E:4_trivial_Kcal} and the differentiability of $\Fcal$ (see Proposition~\ref{P:hatZ})
imply that
\begin{equation}
D_{\Hcal} \Fcal_1(0, 0) = 0. 
\end{equation}
Thus we can apply the implicit function theorem in the space $C^m_*$ to get the following result.

\begin{theorem}
\label{T:tildeq}
Let $ 2m+3\le r_0 $. There exist  constants $ \rho_1 , \rho_2 >0$, and 
a parameter $\zeta>0$ in the definition of the norm on the space $\bE$ introduced in \eqref{E:||h} such that
there exists a
$C^m_*$-map $\Hscr  \colon B_{\bE}( \rho_1)\to B_{\bM_0}(\rho_2)$
\lsm[HcNcpgubf]{$\Hscr $}{the initial Hamiltonian map in Theorem~\ref{T:tildeq}}
satisfying the fixed point equations
\begin{equation}
\label{E:qfixed}
\Fcal_1 (\Kcal,\Hscr(\Kcal)))=\Hscr(\Kcal)
\end{equation}
and 
\begin{equation}
\cbT( \Fcal(\Kcal,\Hscr (\Kcal)),\Kcal,\Hscr (\Kcal))=\Fcal (\Kcal,\Hscr (\Kcal)) 
\end{equation}
 for all $\Kcal\in B_{\bE}(\rho_1)$.
Moreover, the $ C^m_*$- norm of  the map $ \Hcal $ 
is  bounded uniformly in $ N$. We may choose $\rho_2 < \frac14 h^{2}$. Then in view of \eqref{normideal} the matrix $\bq = \pi_2 \circ \Hscr(\Kcal)$  of the quadratic
part of $\Hscr(\Kcal)$  satisfies  $|\bq| < \frac12$.
\end{theorem}

%\oldfour{
%\begin{theorem}
%%\label{T:tildeq}
%Let $ 2m+3\le r_0 $. There exist  constants $ \rho_1 , \rho_2 >0$, and $\zeta>0$ such that
%there exists a $C^m_*$-map $\Hscr  \colon B_{\bE}( \rho_1)\to B_{M_0}(\rho_2)$
%$ \Qcal \colon B_{\bE}( \rho_1)\to B_{\R_{\rm sym}^{d\times d}}(\rho_2)$ satisfying the fixed point equations
%\begin{equation}
%%%\label{E:qfixed}
%\Fcal_1 (\Kcal,\Hscr(\Kcal)))=\Hscr(\Kcal)
%\end{equation}
%and 
%\begin{equation}
%\cbT( \Fcal(\Kcal,\Hscr (\Kcal)),\Kcal,\Hscr (\Kcal))=\Fcal (\Kcal,\Hscr (\Kcal))  
%\end{equation}
% for all $\Kcal\in B_{\bE}(\rho_1)$.
%%
%%The map $ \Lcal \colon B_{\bE}(\rho_1) \to\R $ defined by $ \Lcal (\Kcal)=\lambda(\Fcal (\Kcal,\Qcal (\Kcal))) $ is also $ C^m_* $. 
%%Moreover, the $ C^m_*$- norms of $ \Qcal $ and $ \Lcal $ are bounded uniformly in $ N$.
%\end{theorem}
%}

\section[Proof of strict convexity---Theorem~2.1]{Proof of strict convexity---Theorem~\ref{T:conv}}
\label{sec:convexity}
We are following the strategy outlined in Chapter~\ref{S:strategy}, but we now consider the full ideal Hamiltonian $\Hcal$
in \eqref{E:4_ideal_Hamiltonian} and
not just the quadratic part. To prove the strict convexity of
the surface tension $\sigma_\beta(u)$, we need to prove that its perturbative component $ \varsigma(u) $  is smooth in the tilt $u$. 
This amounts to obtaining a  uniform bound (in $ N \in\N $)   on the  approximation  
\begin{equation}
\varsigma_N(u):=-\frac{1}{L^{dN}}\log {\mathcal Z}_{N}(u)
\end{equation} 
with ${\mathcal Z}_{N}(u)$ defined in \eqref{E:ZcalN-K(X)}.
In view of  the  equality %%\eqref{E:ZNfinal} (old reference to Chapter 3)
%\oldfour{ \eqref{E:4_Z_with_Hcal}}
\eqref{E:4_formula_logZ},  applied with $\Kcal = \Kcal_u$, we have
\begin{equation}\label{finalfreeenergy}
\begin{aligned}
\varsigma_N(u)&=-\frac{1}{L^{d N}}\log \Big(\frac{Z_N^{\ssup{\bq}}}{Z_N^{\ssup{0}}}\Big)-
\lambda %  \Lcal (\Kcal_{u})
\\ & \quad +\frac{1}{L^{d N}}\log\Big(\int_{\cbX_N}\Big(1 + \Fcal_{2N}(\Kcal_u, \Hscr(\Kcal_u))(\L_N,\varphi)\Big)\mu_{N+1}^{\ssup{\bq}}(\d\varphi)\Big),
\end{aligned}
\end{equation}
where, as in \eqref{E:4_q_for_Hscr},
\begin{equation}  \label{E:4_q_for_Hscr_bis}
 \lambda = \pi_0 (\Hscr(\Kcal_u))\quad \hbox{and}   \quad \bq = \pi_2 (\Hscr(\Kcal_u))
 \end{equation}
denote the constant term in $\Hscr(\Kcal_u)$ and the coefficient matrix of the qudratic term, respectively.
%
%
%\oldfour{
%\begin{equation} 
%  %%%\label{finalfreeenergy}
%\varsigma_N(u)=-\frac{1}{L^{d N}}\log \Big(\frac{Z_N^{\ssup{\bq}}}{Z_N^{\ssup{0}}}\Big)-
%\Lcal (\Kcal_{u})+\frac{1}{L^{d N}}\log\Big(\int_{\cbX_N}\Big(1+K_N^{\ssup{\bq}}(\L_N,\varphi)\Big)\mu_{N+1}^{\ssup{\bq}}(\d\varphi)\Big).
%\end{equation}
%Here $\bq=\Qcal(\Kcal_u)$ is the initial condition tuned, in dependence on $u$ through the initial value $\Kcal_u$, in such a way that
%$ H_N^{\ssup{\bq}}=0 $.
%}
%
%\oldfour{
%Thus, controlling the smoothness of $\varsigma_N(u)$  consists  of three ingredients:
%smoothness of $\bT_k^{\ssup{\bq}}$ (and, in particular, $K_N^{\ssup{\bq}} $) with respect to the initial $\bq$, smoothness of  $\Qcal$, and finally, smoothness of initial $\Kcal=\Kcal_u$ in the tilt $u$ as formulated in the assumptions  of Theorem~\ref{T:conv}.
%}

 The proof of strict convexity thus consists of the following three steps. 
 
 \noindent \textbf{Step 1:}  Choose all needed constants according to Propositions~\ref{P:Tnonlin} and~\ref{P:Tlin}. 
In particular, we  choose (with a fixed  $ d $)  the constants  $L$, $h$,  $\mathsf{A}$, ${\bar\rho}={\bar\rho}(\mathsf{A})$,
and a constant $C$,  so
 that the claims from Propositions~\ref{P:Tnonlin} and ~\ref{P:Tlin} (i.e., differentiability and uniform smoothness of 
 the renormalization maps $\bT_k$ as well as the contractivity of the linearisation) are valid for  any  $(H,K,\bq)\in \Ucal_{\rho}$ (in particular,  $ \norm{\bq}\le \frac{1}{2} $).
 
\noindent\textbf{Step 2:}  Apply Theorem \ref{T:tildeq} to get the existence and smoothness properties of the map 
$\Hscr: B_{\bE}(\rho_1) \to B_{\bM_0}(\rho_2)$.

\noindent\textbf{Step 3:} Finally, address the dependence of $\Kcal_u$ on the tilt $u$:
 according to the assumptions  of Theorem~\ref{T:conv} we have a $C^3 $ tilt map 
 $ \tau$,   $u\mapsto \tau(u)=\Kcal_{u}$. Choosing $\delta$ sufficiently small, we have $ \tau(B_\delta(0))\subset B_{\bE}(\rho)\subset\bE $.

Having this in mind, we show that the right hand side of \eqref{finalfreeenergy} is  three times continuously differentiable
 in $ u $ with bounded derivatives,
by analysing each of the three terms separately. 

The first term on the right hand side of \eqref{finalfreeenergy} can easily be computed as
\begin{equation}
-\log \Big(\frac{Z_{N}^{\ssup{\bq}}}{Z_N^{\ssup{0}}}\Big)=\tfrac12\log\det \bigl(\Ascr^{\ssup{\bq}} \Cscr^{\ssup{0}}\bigr).
\end{equation}
Consider the dual torus
\begin{equation}
\widehat{\T}_N=\Bigl\{p=(p_1,\dots, p_d)\colon p_i\in\bigl\{-\tfrac{(L^N-1)\pi}{L^N},-\tfrac{(L^N-3)\pi}{L^N}
,\dots, \tfrac{(L^N-1)\pi}{L^N}\bigr\}, 
i=1,\dots, d \Bigr\},
\end{equation}
and the functions $f_p(x)={\rm e}^{i\langle p,x\rangle}$. The family $\left\{ |\Lambda_N|^{-1/2} f_p \right\}_{p \in \widehat{\T}_N\setminus\{0\}}$
is an orthonormal basis of $ \cbV_N$. The eigenvalues of $ \Acal^{\ssup{\bq}} $ are 
\begin{equation}
\sigma(p)=\langle q^{\ssup{p}},(\1+\bq)q^{\ssup{p}}\rangle=\sum_{l,j=1}^dq_l^{\ssup{p}}\big(\delta_{l,j}+\bq_{l,j}\big)q_j^{\ssup{p}},\  p\in\widehat{\T}_N
\end{equation}
with $q^{\ssup{p}}_j={\rm e}^{ip_j}-1$,  $j=1,\dots, d$.
Note that $ q_l^{\ssup{p}}q_j^{\ssup{p}}\approx p_lp_j $. The eigenvalues for $ \Acal^{\ssup{0}} $ and $ \Ccal^{\ssup{0}} $ are $ \langle q^{\ssup{p}},q^{\ssup{p}}\rangle\approx \norm{p}^2 $ and  $ \langle q^{\ssup{p}},q^{\ssup{p}}\rangle^{-1} \approx \norm{p}^{-2}, p\in\widehat{\T}_N $, respectively. We get
\begin{multline}\label{quotienteigenvalues}
\log\det \bigl(\Ascr^{\ssup{\bq}} \Cscr^{\ssup{0}}\bigr)=  \tr\log \bigl(\1 + \Ascr^{\ssup{\bq}} \Cscr^{\ssup{0}}\bigr)
=\sum_{p\in\widehat{\T}_N\setminus\{0\}}\log\Bigl(1+\frac{\langle q^{\ssup{p}},\bq q^{\ssup{p}}\rangle}{ \langle q^{\ssup{p}},q^{\ssup{p}}\rangle}\Bigr).
\end{multline}
Since the sum over the torus has $L^{dN}-1 $ terms it follows that 
$$
-\frac{1}{L^{d N}}\log \Big(\frac{Z_N^{\ssup{\bq}}}{Z_N^{\ssup{0}}}\Big)
$$  is a smooth function of $ \bq $ with derivatives bounded uniformly in $ N $. Thus
$$
u\mapsto -\frac{1}{L^{d N}}\log \Big(\frac{Z_N^{\pi_2(\Hscr(\Kcal_u))}}{Z_N^{\ssup{0}}}\Big)
$$ is a $C^3_*$ mapping  with uniformly bounded derivatives. Note that the chain rule initially states that this map is $ C^3_* $, but $ \R^d$ being a finite dimensional vector space it is actually a $C^3 $ mapping according to Proposition~\ref{P:rel}.  

As regards the second term  we know from Theorem~\ref{T:tildeq} and the chain rule that $u \mapsto \Hscr(\Kcal_u)$ is $ C^3_* $.
Thus the map $u \mapsto \lambda = \pi_0 (  \Hscr(\Kcal_u))$ is $C^3_*$ and hence $C^3$ because the map is defined a neighbourhood in the finite dimensional space $\R^d$.

% \oldfour{the second term  $ u\mapsto \Lcal(\Kcal_u) $  on the right hand side of \eqref{finalfreeenergy} is $ C^3_* $ and thus $C^3 $ as well for the same reasoning as before.}

Regarding the last term 
$$\log\Big(\int_{\cbX_N}\Big(1 + \Fcal_{2N}(\Kcal_u, \Hscr(\Kcal_u))(\L_N,\varphi)\Big)\mu_{N+1}^{\ssup{\bq}}(\d\varphi)\Big)$$
%\oldfour{
%$$  \log\Big(\int_{\cbX_N}\Big(1+K_N^{\ssup{\bq}}(\L_N,\varphi)\Big)\mu_{N+1}^{\ssup{\bq}}(\d\varphi)\Big) $$
%}
we first note that for a positive function $G$ the $k$-th derivative of $\log G$ is a polynomial 
in $\frac{1}{G}$ and the first $k$ derivatives of $G$. Since $\mu_{N+1}^{\ssup{\bq}}$ is a probability measure, it
 suffices to  show that
\begin{equation} \label{eq:six_bound_int_KN}
\left|  \int_{\cbX_N}            \Fcal_{2N}(\Kcal_u, \Hscr(\Kcal_u))(\L_N,\varphi)     \, \mu_{N+1}^{\ssup{\bq}}(\d\varphi)        \right| \le \frac12.
\end{equation}
and to estimate the derivatives of the integral.
We thus need to estimate
\begin{equation}  \label{E:4_integral_Tu}
 T(u) := \int_{\cbX_N}            \Fcal_{2N}(\Kcal_u, \Hscr(\Kcal_u))(\L_N,\varphi)     \, \mu_{N+1}^{\ssup{\bq}}(\d\varphi),
\quad \hbox{where  $\bq = \pi_2 (\Hscr(\Kcal_u))$,} 
\end{equation}
and its derivatives with respect to $u$. 
%%% Modified/ added by SM on June 23, 2016 at the request of RK.
The integral in  \eqref{E:4_integral_Tu} is exactly the application of the renormalisation map $R_1$, defined in \eqref{E:defR1},
evaluated at zero:
$$ T(u) =  (R_1^{(\bq)}P)(\Lambda_N, 0) \quad \hbox{where $P =  \Fcal_{2N}(\Kcal_u, \Hscr(\Kcal_u))$ and $\bq = \pi_2 (\Hscr(\Kcal_u))$.}
 $$
Thus we can apply the estimates for $R_1$ stated in Lemma \ref{L:R_1} and  in Lemma \ref{L:prop} (iv).
%In order to use  our previous results we write the integral  as the 
%application of the renormalisation map $R_1$  evaluated at zero. 
%% End modification
We introduce the notation
$$ \widetilde R_1(K,\Hcal) := (R_1^{(\bq)} K)(\Lambda_N, 0)  = R_1(K,q)(\Lambda_N, 0).$$
It will later be  convenient to view $\tilde R_1$ as a function of $K$ and $\Hcal$ even thus it 
depends on $\Hcal$ only through $\bq = \pi_2 (\Hcal)$. We get
$$ T(u) = \widetilde R_1\Big(  \Fcal_{2N}(\Kcal_u, \Hscr(\Kcal_u)), \Hscr(\Kcal_u) \Big).$$
Now by Lemma \ref{L:prop} (iv)  (note that there is only one $N$-block), Proposition \ref{P:hatZ},
the definition \eqref{E:normZr} of the norm on $\Fcal$, Theorem \ref{T:tildeq}
and the assumptions on $\Kcal_u$ in Theorem  \ref{T:conv} % \comment{Add label to the equation needed} 
we get
\begin{equation*}
|T(u)| \le \|   \Fcal_{2N}(\Kcal_u, \Hscr(\Kcal_u))\| \le 2 \frac{\eta^N}{\alpha} 
\| \Fcal(\Kcal_u, \Hscr(\Kcal_u))\|_{\bY_0} \le C \frac{\eta^N}{\alpha}.
\end{equation*}
Thus  \eqref{eq:six_bound_int_KN} holds if $N$ is large enough (note that $\alpha$ and $C$ are independent of $N$). 

To verify the differentiability of $T$ we recall  the notation 
$$
(F \diamond G)(x,\Hcal) = F(G(x,\Hcal), \Hcal)
$$ 
to rewrite $T(u)$ as
$$ T(u) = \big( \widetilde R_1 \diamond \Fcal_{2N} \big)(\Kcal_u, \Hscr(\Kcal_u))
$$

%\oldfour{
%Secondly, note that $ K_N^{\ssup{\bq}}(\L_N, \cdot)$ is just the last component of the vector $$\Fcal(\Kcal_u, \Qcal(\Kcal_u)).$$
%We denote this component by $\Fcal_{N,2}(\Kcal_u, \Qcal(\Kcal_u))$. Recalling the notation $\tau(u) = \Kcal_u$ we thus need to estimate
%$$ T(u) := \int_{\cbX_N}   \Fcal_{N,2}(\tau(u), (\Qcal \circ \tau)(u))  \,    \mu_{N+1}^{\ssup{(\Qcal \circ \tau)(u)}}(\d\varphi)$$
%and its derivatives with respect to $u$. In order to use  our previous results we write the integral  as the 
%application of the renormalisation map $R_1$  evaluated at zero. With the 
%notation
%$$ \widetilde R_1(K,q) := R_1(K,q)(0)$$
%we have
%$$ T(u) = \tilde R_1 ( \Fcal_{N,2}(\tau(u), (\Qcal \circ \tau)(u)), \Qcal \circ \tau)(u)).$$
%Now by Lemma \ref{L:prop} (iv)  (note that there is only one $N$-block), Proposition \ref{P:hatZ},
%the definition \eqref{E:normZr} of the norm on $\Fcal$, Theorem \ref{T:tildeq}
%and the assumptions on $\Kcal_u$ in Theorem  \ref{T:conv} % \comment{Add label to the equation needed} 
%we get
%\begin{equation*}
%|T(u)| \le  2  \| \Fcal_{N,2}(\tau(u), (\Qcal \circ \tau)(u))\|_{N,0}  \le  2 \frac{\eta^N}{\alpha}  \| \Fcal(\tau(u), (\Qcal \circ \tau)(u))\|_{\bY_0}
%\le  C \frac{\eta^N}{\alpha}.
%\end{equation*}
%Thus  \eqref{eq:six_bound_int_KN} holds if $N$ is large enough (note that $\alpha$ and $C$ are independent of $N$). 
%
%To verify the differentiability of $T$ we use the notation 
%$$
%(F \diamond G)(x,\bq) = F(G(x,\bq), \bq)
%$$ to rewrite $T(u)$ as
%$$ T(u) = (\tilde R_1 \diamond \Fcal_{N,2})( \tau(u), (\Qcal \circ \tau(u)).$$
%}

Now by Proposition \ref{P:hatZ} we have 
%\oldfour{$\Fcal_{N,2}  \in 
%\widetilde C^m(B_{\bX\times\R^{d\times d}_{\rm sym}}(\widehat\rho_1,\widehat\rho_2),\bY) \quad $}  
$\Fcal_{2N}  \in 
\widetilde C^m(B_{\bX\times\bM_0}(\widehat\rho_1,\widehat\rho_2),\bY)
$
with bounds on the derivatives which are independent of $N$. 
Here  $\bY = \bY_{\!\!r_0}\embed \bY_{\!\!r_0-2}\embed\dots \embed \bY_{\!\!r_0-2m}$
and in the domain we use the trivial scale  $\bX_m = \ldots = \bX_0 = \bE$.

By Lemma  \ref{L:R_1} we have 
$\widetilde R_1 \in \widetilde C^m(\bY \times B_{\widehat \rho_2}, \R)$ (as long as $\widehat \rho_2 < \frac14 h^2$),
%\oldfour{$\widetilde R_1 \in \widetilde C^m(\bY \times B_{\tfrac12}, \R)$},
again with bounds on the derivatives
which are independent of $N$.  Thus the chain rule
with loss of regularity, Theorem \ref{T:fullchain}, shows that
$\widetilde R_1 \diamond \Fcal_{2N} 
\in \widetilde C^m(B_{\bX\times M_0}(\widehat\rho_1,\widehat\rho_2), \R)$
%\oldfour{\quad 
% $\widetilde R_1 \diamond \Fcal_{N,2} 
%\in \widetilde C^m(B_{\bX\times\R^{d\times d}_{\rm sym}}(\widehat\rho_1,\widehat\rho_2), \R)$},
with uniformly bounded derivatives.
Since the scale $\bX_m = \ldots = \bX_0 = \bE$ is trivial (and since the target is just $\R$)
this implies that 
$\widetilde R_1 \diamond \Fcal_{2N} 
\in  C^m_*(B_{\bX\times M_0}(\widehat\rho_1,\widehat\rho_2), \R)$
%\oldfour{ \quad $\widetilde R_1 \diamond \Fcal_{N,2}  \in C^m_*(B_{E\times\R^{d\times d}_{\rm sym}}(\widehat\rho_1,\widehat\rho_2); \R)$.}
Together with the regularity of $\Hscr$ 
(see Theorem \ref{T:tildeq}) and the assumptions on $\Kcal_u$ in 
Theorem  \ref{T:conv} we get
$T \in C^3_*(B(\delta_0))$ with uniformly bounded derivatives. Since $B(\delta_0) \subset \R^d$ 
by Proposition \ref{P:rel} this is the same as $T \in C^3(B(\delta_0))$. \qed

%
%
%\oldfour{Together with the regularity of $\Qcal$ (see Theorem \ref{T:tildeq}) and the assumptions on $\Kcal_u$ in 
%Theorem  \ref{T:conv} we get
%$T \in C^3_*(B(\delta_0))$ with uniformly bounded derivatives. Since $B(\delta_0) \subset \R^d$ 
%by Proposition \ref{P:rel} this is the same as $T \in C^3(B(\delta_0))$.
%\qed}

%% file: AKM-ch5-norms-30thJune-2016.tex
\chapter{Properties of the Norms}\label{sec:properties}

As a preparation for the proof of Propositions~\ref{P:Tlin} and \ref{P:Tnonlin}, we first address the factorisation properties of the norms defined in Chapter~\ref{S:norms}
and prove a bound on the integration map $\bR_{k}$ defined in \eqref{E:Rk}.
Recalling that the norms  $\norm{\boldsymbol {\cdot}}_{k,X,r}$
depend on parameters $L, h$, and $\omega$, we summarise their properties 
in the following lemma. Using $\upeta(n,d)$ defined by \eqref{E:def_eta}, we introduce
$\upkappa(d) := \frac12\bigl(d+\upeta( 2\lfloor  \frac{d+2}{2}   \rfloor + 8, d)\bigr)$ with $\lfloor t \rfloor$ denoting
\lsm[kgx]{$\upkappa(d)$}{$= \frac12\bigl(d+\upeta( 2\lfloor  \frac{d+2}{2}   \rfloor + 8, d)\bigr)$}%
the integer value of $t$. Notice that $\upkappa(d)\le d^2/2+5d+16$.

\begin{lemma}\label{L:prop}
Let $\omega\ge  1+18\sqrt 2$, $N\in\N$, $N\ge 1$, and $L\in \N$ odd, $L\ge 3$.
Given $k\in\{0,\dots,N-1\}$, let $ K\in M(\Pcal_k,\cbX) $ factor (at the scale $k$), and let $ F\in M(\Bcal_k,\cbX) $. Then,
the norms $ \norm{\cdot}_{k,X,r} ,  \norm{\cdot}_{k:k+1,X,r} , r\in\{1,\dots,r_0\}$, and $ \tnorm{\cdot}_{k,X}, X\in\Pcal_k $, satisfy the following conditions:
\begin{enumerate}
 \item[(i)] 
$\norm{K(X)}_{k,X,r}\le\prod_{Y\in\Ccal(X)}\norm{K(Y)}_{k,Y,r}$ and 

\noindent
$\norm{K(X)}_{k:k+1,X,r}\le\prod_{Y\in\Ccal(X)}\norm{K(Y)}_{k:k+1,Y,r}$,
\item[(iia)] 
$\norm{F^XK(Y)}_{k,X\cup Y,r}\le\norm{K(Y)}_{k,Y,r}  \tnorm{F}_{k}^{\abs{X}_k}$ as well as
\item[(iib)] 
$\norm{F^XK(Y)}_{k:k+1,X\cup Y,r}\le\norm{K(Y)}_{k:k+1,Y,r}   \tnorm{F}_{k}^{\abs{X}_k}$
for $X,Y\in\Pcal_k$ disjoint,
\item[(iii)]
$\tnorm{\1(B)}_{k,B}=1$ for $B\in\Bcal_k$,
\item[(iv)]  %Let $\upeta(n,d)$ be given by \eqref{E:def_eta} and let
%$\upeta(d) := \upeta( 2\lfloor  \frac{d+2}{2}   \rfloor + 8, d)$ 
% where $\lfloor t \rfloor$ denotes the integer value of $t$.
There exists a constant $h_1=h_1(d,\omega)$ depending only on the dimension $d$ and value of the parameter $\omega$, such that for any $h\ge L^{\upkappa(d)} h_1$ and  $X\in\Pcal_k$, we have $\norm{(\bR_{k+1}K)(X)}_{k:k+1,X,r}\le 2^{|X|_k}\norm{K(X)}_{k,X,r}$.
\end{enumerate}
\end{lemma}

\begin{proofsect}{Proof}

\noindent (i) 
%RK seminorms: thus I prefer to skip mentioning Banach spaces etc.
Notice first that for any $F_1, F_2\in M(\Pcal_k,\cbX) $ and any (not necessarily disjoint) $X_1, X_2\in \Pcal_k$,
we have 
\begin{equation}
\label{E:F_1F_2}
\bnorm{F_1(X_1)(\varphi) F_2(X_2)(\varphi)}^{k,X_1\cup X_2,r}\le\bnorm{F_1(X_1)(\varphi)}^{k,X_1,r} \bnorm{F_2(X_2)(\varphi)}^{k,X_2,r}.
\end{equation}
Indeed, using the definition of the norm $\bnorm{\cdot}^{k,X,r}$ and fact that a Taylor expansion of a product is the product of Taylor expansions,
we have 
\begin{equation}
%\label{E:F_1F_2}
\bnorm{F_1(X_1)(\varphi) F_2(X_2)(\varphi)}^{k,X_1\cup X_2,r}\le\bnorm{F_1(X_1)(\varphi)}^{k,X_1\cup X_2,r} \bnorm{F_2(X_2)(\varphi)}^{k,X_1\cup X_2,r}.
\end{equation}
Observing now that for any $\dot{\varphi}\in \cbX_N$ we have $\bnorm{\dot{\varphi}}_{k,X_1}\le\bnorm{\dot{\varphi}}_{k,X_1\cup X_2}$,
we get
\begin{equation}
\label{}
\sup_{\bnorm{\dot{\varphi}}_{k,X_1\cup X_2}\le 1}|D^sF_1(X_1)(\varphi)(\dot{\varphi},\ldots,\dot{\varphi})|\le \sup_{\bnorm{\dot{\varphi}}_{k,X_1}\le 1}|D^sF_1(X_1)(\varphi)(\dot{\varphi},\ldots,\dot{\varphi})|,
\end{equation}
implying
\begin{equation}
\label{E:embeddingnorms}
\bnorm{F_1(X_1)(\varphi)}^{k,X_1\cup X_2,r}\le \bnorm{F_1(X_1)(\varphi)}^{k,X_1,r}
\end{equation}
and similarly for $F_2$, yielding thus \eqref{E:F_1F_2}.

Iterating \eqref{E:F_1F_2} we can use it for  $K(X,\varphi)= \prod_{Y\in\Ccal(X)} K(Y)(\varphi)$, yielding
\begin{equation}
\label{E:KleprodK}
\bnorm{K(X,\varphi)}^{k,X,r}\le\prod_{Y\in\Ccal(X)}\bnorm{K(Y)(\varphi)}^{k,Y,r}
\end{equation}
and, similarly, 
\begin{equation}
\bnorm{K(X,\varphi)}^{k+1,X,r}\le\prod_{Y\in\Ccal(X)}\bnorm{K(Y)(\varphi)}^{k+1,Y,r}
\end{equation}

To conclude, it then suffices to observe that 
\begin{equation}
\begin{aligned}
w_k^X(\varphi)= \prod_{Y\in\Ccal(X)} w_k^Y(\varphi)\  \text{ and }\  w_{k:k+1}^X(\varphi)= \prod_{Y\in\Ccal(X)} w_{k:k+1}^Y(\varphi).
\end{aligned}
\end{equation}
Here, in both cases,  we use the fact that the partition    $X=\cup_{Y\in \Ccal(X)}  Y$ splits both $X$ and its boundary $\partial X$
into  disjoint components:
$Y_1, Y_2\in \Ccal(X)$, $Y_1\neq Y_2$ implies that $\dist(Y_1, Y_2)>L^k$ and thus $Y_1\cap Y_2=\emptyset$, $\partial Y_1\cap\partial Y_2=\emptyset$, and
$\partial X=\cup_{Y\in \Ccal(X)} \partial Y$.

\noindent (iia) 
Using (iterated) \eqref{E:F_1F_2} for   $\prod_{B\in\Bcal_k(X)}F(B)(\varphi) K(Y)(\varphi)$ , we have 
\begin{equation}
\label{}
\bnorm{\bigl(F^XK(Y)\bigr)(\varphi)}^{k,X\cup Y,r}\le\prod_{B\in\Bcal_k(X)}\bnorm{F(B)(\varphi)}^{k,B,r}\bnorm{K(Y)(\varphi)}^{k,Y,r}.
\end{equation}
Bounding the right hand side by
\begin{equation}
\label{}
\prod_{B\in\Bcal_k(X)}\tnorm{F(B)}_{k,B}\norm{K(Y)}_{k,Y,r} 
\prod_{B\in\Bcal_k(X)}W_k^B(\varphi)w^Y_k(\varphi), 
\end{equation}
we get (ii) once we verify that
\begin{equation}
\label{E:w<Ww}
 \prod_{B\in\Bcal_k(X)}W_k^B(\varphi) w^{Y}_k(\varphi)\le w^{X\cup Y}_k(\varphi).
\end{equation}
Inserting the definitions of the strong and weak weight functions, \eqref{E:w<Ww} is satisfied once
\begin{equation}
\label{E:kgX}
L^k \sum_{x\in \p Y} G_{k,x}(\varphi)\le \sum_{x\in X} \bigl(2^d \omega  g_{k,x}(\varphi) +(\omega-1) G_{k,x}(\varphi)   \bigr)+ L^k \sum_{x\in \p(X\cup Y)} G_{k,x}(\varphi).
\end{equation}
To verify this, it suffices to notice that each  $y\in  \p  Y\setminus  \p(X\cup Y)$ is necessarily contained in $\partial B$ for some
$B\in  {\Bcal}_k(X)$ (a block on the boundary of $X$ touching $Y$).  Thus, it suffices to show that for each such $B$ one has
\begin{equation}
\label{E:GinpB}
L^k \sum_{x\in \p B} G_{k,x}(\varphi)\le \sum_{x\in B} \bigl(2^d \omega  g_{k,x}(\varphi) +(\omega-1) G_{k,x}(\varphi)   \bigr).
\end{equation}
Indeed, applying Proposition~\ref{boundaryest} (a), we have
\begin{multline}
h^2 L^k \sum_{x\in \p B} G_{k,x}(\varphi)\le \\ \le 2 \mathfrak{c}\bigl(\sum_{x\in B}|\nabla \varphi(x)|^2+L^{2k}\sum_{x\in U_1(B)}|\nabla^2 \varphi(x)|^2\bigr)+
L^k \sum_{x\in \p B}\sum_{s=2}^3L^{(2s-2)k}|\nabla^s \varphi(x)|^2\le\\
\le h^2 2\mathfrak{c} \sum_{x\in B}  G_{k,x}(\varphi)  +h^2 2 \mathfrak{c} L^k\sum_{z\in \p B} g_{k,z}(\varphi),
\end{multline}
where $z$ is any point $z\in B$. 
Observing that   the size of the set $\p  B$ 
is at most $(L^k+2)^d- (L^k-2)^d\le 2^d L^{(d-1)k}$  once $2\le L$, we get
the seeked bound once 
\begin{equation}
\label{E:assume1a}
2\mathfrak{c}\le \omega-1.
\end{equation}
Observing that $\mathfrak{c}< 3\sqrt{2}$, this condition is satisfied with our choice of $\omega$.

\noindent(iib) The proof is similar, with \eqref{E:kgX} replaced by
\begin{equation}\label{E:k:k+1gX}
\begin{aligned}
3L^k \sum_{x\in \p Y} G_{k,x}(\varphi)\le &  \sum_{x\in X} \bigl((2^d \omega-1)  g_{k:k+1,x}(\varphi) +(\omega-1) G_{k,x}(\varphi)   \bigr)\\
&  +3  L^k \sum_{x\in \p(X\cup Y)} G_{k,x}(\varphi)
\end{aligned}
\end{equation}
that, in its turn, needs \eqref{E:GinpB} in a slightly stronger version,
\begin{equation}
\label{E:k:k+1GinpB}
3 L^k \sum_{x\in \p B} G_{k,x}(\varphi)\le \sum_{x\in B} \bigl((2^d \omega-1)  g_{k:k+1,x}(\varphi) +(\omega-1) G_{k,x}(\varphi)   \bigr).
\end{equation}
This is satisfied once
\begin{equation}
\label{E:assume1}
6\mathfrak{c}\le \omega-1.
\end{equation}

\noindent (iii) follows immediately from the definition.
 
 \noindent (iv)  Since convolution commutes with differentiation we have
% The integration and the differentiation implicitly contained in the norm 
%\noindent
%$\norm{(\bR_{k+1}K)(X)}_{k:k+1,X,r}$
%can be interchanged. Namely, recalling the definition
%%% \eqref{E:SjX}, we have
\begin{equation}
D^s \int K(\varphi+\xi)\mu_{k+1}(\d\xi)  =   \int D^s K(\varphi+\xi)  \mu_{k+1}(\d\xi).
\end{equation}
For  a vector $(A_0, A_1, \ldots, A_r)$ 
consisting of $A_0 \in \R$ and multilinear symmetric maps $A_s\colon \cbX^{\otimes s} \to \R, s\in\N$,
we consider the norm 
\begin{equation}
|(A_0, \ldots, A_r)| := \sum_{s=0}^r  \frac{1}{s!}  |A_s|^{k+1,X}
\end{equation}
with $|A_s|^{k+1,X}$ defined by  \eqref{E:SjX}.
Then 
$$
|K(\varphi), DK(\varphi), \ldots, D^r K(\varphi))| = |K(\varphi)|^{k+1,X,r}.
$$
Now fix $\varphi$ and apply Jensen's inequality to 
map $\xi \mapsto (K(\varphi + \xi), \ldots, D^r K(\varphi + \xi))$. This yields
\begin{equation}
\Big| \int K(\varphi+\xi)\mu_{k+1}(\d\xi)\Big| ^{k+1,X,r}  =   \int | K(\varphi+\xi)|^{k+1,X,r}  \mu_{k+1}(\d\xi).
\end{equation}
Since 
\begin{equation}
\label{E:dotphi}
\bnorm{\dot{\varphi}}_{k,X}\le L^{-\frac{d}{2}}\bnorm{\dot{\varphi}}_{k+1,X},
\end{equation}
we also have
\begin{equation}
\bnorm{K(X,\varphi+\xi)}^{k+1,X,r}\le \bnorm{K(X,\varphi+\xi)}^{k,X}.
\end{equation}
As a result,
\begin{equation}
\norm{(\bR_{k+1}K)(X)}_{k:k+1,X,r}\le  \sup_{\varphi} \int   \bnorm{K(X,\varphi+\xi)}^{k,X,r} \mu_{k+1}(\d\xi)w_{k:k+1}^{-X}(\varphi).
\end{equation}
Estimating   the integrand $\bnorm{K(X,\varphi+\xi)}^{k,X,r}$  from  above by 
%%% end new %%%%old
%\noindent (iv) The integration and the differentiation implicitly contained in the norm 
%
%\noindent
%$\norm{(\bR_{k+1}K)(X)}_{k:k+1,X,r}$
%can be interchanged. Namely, recalling the definition \eqref{E:SjX}, we have
%\begin{equation}
%\bigl|D^s \int K(\varphi+\xi)\mu_{k+1}(d\xi)\bigr|^{k+1,X} \le \int \bnorm{D^s K(\varphi+\xi)}^{k+1,X} \mu_{k+1}(d\xi).
%\end{equation}
%Observing, directly from the definition, that 
%\begin{equation}
%\label{E:dotphi}
%\bnorm{\dot{\varphi}}_{k,X}\le L^{-\frac{d}{2}}\bnorm{\dot{\varphi}}_{k+1,X},
%\end{equation}
%we get
%\begin{equation}
%\bnorm{D^sK(X,\varphi+\xi)}^{k+1,X}\le \bnorm{D^sK(X,\varphi+\xi)}^{k,X}.
%\end{equation}
%(For $s=0$ this is trivial since $\bnorm{K(X,\varphi+\xi)}^{k,X}=\abs{K(X,\varphi+\xi)}$ actually does not depend on $k$.)
%As a result,
%\begin{equation}
%\norm{(\bR_{k+1}K)(X)}_{k:k+1,X,r}\le  \sup_{\varphi} \int   \bnorm{K(X,\varphi+\xi)}^{k,X,r} \mu_{k+1}(d\xi)w_{k:k+1}^{-X}(\varphi).
%\end{equation}
%
%
%Evaluating now the integrand $\bnorm{K(X,\varphi+\xi)}^{k,X,r}$  above by 
$$\norm{K(X)}_{k,X,r}  w_k^X(\varphi+\xi),$$ 
the proof of the needed bound amounts to showing that
\begin{equation}
\label{toshowfluct}
\int_{\cbX_N} w_k^X(\varphi+\xi) \mu_{k+1}(\d \xi)\le 2^{|X|} w_{k:k+1}^X(\varphi) .
\end{equation}
\end{proofsect}

As this result  will be used also later in different circumstances, we state it as a separate Lemma.

\begin{lemma}
\label{L:strongerproperties}
Let $\omega\ge 1+6\sqrt 2$. 
%Let $\upeta(n,d)$ be given by \eqref{E:def_eta} and let
%$\upeta(d) := \upeta( 2\lfloor  \frac{d+2}{2}   \rfloor + 8, d)$ 
% where $\lfloor t \rfloor$ denotes the integer value of $t$.
There exists a constant $h_1=h_1(d,\omega)$ such that for any   $N\ge 1$,  $L$ odd, $L\ge 5$, $h\ge L^{\upkappa(d)} h_1$,
$k\in\{0,\dots,N-1\}$, $ K\in M(\Pcal_k,\cbX) $, and any  $X\in\Pcal_k$, we have 
\begin{equation}
\label{toshowfluct:str}
\int_{\cbX_N} w_k^X(\varphi+\xi) \mu_{k+1}(\d \xi)\le 2^{|X|_k} w_{k:k+1}^X(\varphi).
\end{equation}
\end{lemma}

\begin{proofsect}{Proof}
We will prove the bound \eqref{toshowfluct:str} in three  steps: 

\medskip

\noindent 
{\bf Step 1.}  Expanding the terms $(\nabla\varphi(x)+\nabla\xi(x))^2 $ in $\sum_{x\in X}G_{k,x}(\varphi+\xi)$ and using the
 Cauchy's inequality $ (a+b)^2\le 2a^2+2b^2 $ for the remaining terms (those that are preceded by a power in $L$ that allows to absorb the resulting prefactors  while passing to the next scale), we have
\begin{multline}
\label{est1}
h^2 \sum_{x\in X}G_{k,x}(\varphi+\xi)\le \sum_{x\in X}\bigl(|\nabla\varphi(x)|^2+|\nabla\xi(x)|^2 \bigr) +2 \bigl|\sum_{x\in X} \nabla\varphi(x)\nabla\xi(x)\bigr|+\\
+ 2\sum_{x\in X}\Bigl(L^{2k}\abs{\nabla^2\varphi(x)}^2 +L^{2k}\abs{\nabla^2\xi(x)}^2 +L^{4k}|\nabla^3\varphi(x)|^2+L^{4k}|\nabla^3\xi(x)|^2\Bigr) .
\end{multline}
For the remaining terms occurring in $w_k^X(\varphi+\xi)$, we simply write (again by Cauchy's inequality)
\begin{equation}
g_{k,x}(\varphi+\xi)\le 2g_{k,x}(\varphi)+2g_{k,x}(\xi)
\end{equation}
and 
\begin{equation}
L^kG_{k,x}(\varphi+\xi)\le 2L^kG_{k,x}(\varphi)+2L^kG_{k,x}(\xi).
\end{equation}

\medskip

\noindent 
{\bf Step 2.} In view of  Proposition~\ref{estmixedterm}, we bound
the mixed term $2 \bigl|\sum_{x\in X}\nabla\varphi(x)\nabla\xi(x)\bigr|$ by
\begin{multline}
\label{mixedterm}
L^{2k}\sum_{x\in X\cup \partial^- X}\abs{\nabla^2 \varphi(x)}^2+
L^k\sum_{x\in\partial^- X}|\nabla \varphi(x)|^2   +
\frac{1+\mathfrak{c} d}{L^{2k}}\!\!\!\sum_{x\in X\cup\partial^- X}\!\!\! \xi(x)^2+
\mathfrak{c} \sum_{x\in X}|\nabla \xi(x)|^2.
\end{multline}
The sum over $X$ in the first term  above will be estimated by the regulator $ g_{k:k+1,x}(\varphi) $ of the next generation. Namely, 
combining, for any $x\in X$, its terms with the corresponding $\varphi$-terms on the second line in \eqref{est1},  we have
\begin{multline}
\label{E:blockmixed}
3 L^{2k}|\nabla^2\varphi(x)|^2 +2L^{4k} |\nabla^3\varphi(x)|^2\le\\ \le
3 L^{-2} L^{2(k+1)}|\nabla^2\varphi(x)|^2 +2L^{-4}L^{4(k+1)} |\nabla^3\varphi(x)|^2\le 
3 L^{-2} h^2g_{k:k+1,x}(\varphi), 
\end{multline}
where we are assuming that
\begin{equation}
\label{E:assume2}
2L^{-2}\le 3.
\end{equation}
The remaining sum over $\p^- X\setminus X$, together with the second term in \eqref{mixedterm},
will be absorbed into the sum $\sum_{x\in \p X} G_{k,x}(\varphi)$. 
Collecting now all the $\varphi$-terms in $\log w_k(\varphi+\xi)$ with expanded mixed term, we get the bound
\begin{equation}
\label{est2}
\sum_{x\in X}2^{d+1} \omega  g_{k,x}(\varphi) +\sum_{x\in X}\omega  G_{k,x}(\varphi)+
3\omega  L^{-2} \sum_{x\in X} g_{k:k+1,x}(\varphi) +3L^k \sum_{x\in \p X} G_{k,x}(\varphi).
\end{equation}
This is bounded by
\begin{equation}
\log w_{k:k+1}^X(\varphi)= \sum_{x\in  X} \bigl((2^{d}\omega-1) g_{k:k+1,x}(\varphi)+\omega  G_{k,x}(\varphi)\bigr)+3L^{k}\sum_{x\in \partial X} G_{k,x}(\varphi)
\end{equation}
once
\begin{equation}
\label{E:assume3:str}
(3+2^{d+1})\omega  \le (2^{d}\omega-1) L^2.
\end{equation}
This condition, including also \eqref{E:assume2}, are satisfied once $L\ge 5$.

Turning now to the $\xi$-terms in $h^2 \log w_k(\varphi+\xi)$ with expanded mixed term,
we get the bound
\begin{multline}
%\label{}
\sum_{x\in X} h^2 2^{d+1}\omega g_{k,x}(\xi) +\sum_{x\in X}\omega \bigl( (1+\mathfrak{c}) \abs{\nabla\xi(x)}^2 +2L^{2k}\abs{\nabla^2\xi(x)}^2
+2L^{4k}\abs{\nabla^3\xi(x)}^2\bigr)+\\+
\omega(1+\mathfrak{c}d) L^{-2k} \sum_{x\in X\cup\p^- X} \xi(x)^2+ 2 L^k \sum_{x\in\p X} h^2G_{k,x}(\xi).
\end{multline}
Bounding the last term with the help of Proposition~\ref{boundaryest}, we get 
\begin{multline}
%\label{}
\sum_{x\in X} h^2 2^{d+1}\omega g_{k,x}(\xi) +\sum_{x\in U_1(X)}\bigl(\omega(1+\mathfrak{c}d) L^{-2k}\xi(x)^2+(\omega  (1+\mathfrak{c})+4\mathfrak{c}) \abs{\nabla\xi(x)}^2 +\\+(2\omega+8\mathfrak{c})L^{2k}\abs{\nabla^2\xi(x)}^2
+(2\omega+8\mathfrak{c})L^{4k}\abs{\nabla^3\xi(x)}^2+4\mathfrak{c} L^{6k}\abs{\nabla^4\xi(x)}^2\bigr).
\end{multline}

Finally, the  term  $g_{k,x}(\xi)$  containing $l_\infty$-norm of $ \nabla^s\xi$, $s=2, 3, 4$,  is bounded with the help of the Sobolev inequality from Proposition~\ref{Sobolovpropi-iv}. Taking $B^*$ for the $B_n$ with $n=(2^{d+1}-1)L^k$, we get
\begin{equation}
\norm{\nabla^s\xi}^2_{l_\infty(B^*)}\le \mathfrak{C}^2(2^{d+1}-1)^2\frac{1}{L^{kd}}\sum_{l=0}^{\widetilde{M}}L^{2lk}\sum_{x\in B^*}|\nabla^l\nabla^s\xi|^2(x),
\end{equation}
where $\widetilde{M}=\lfloor\frac{d+2}2\rfloor$ is the integer value of $\frac{d+2}2$ and in computing the pre-factor we took into account that 
$2\lfloor\frac{d+2}2\rfloor-d\le 2$. Notice that the constant $\mathfrak{C}$ depends (also through $\widetilde{M}$) only on the dimension $d$.
As a result, we are getting
\begin{multline}
%\label{}
\sum_{x\in X} h^2 2^{d+1}\omega g_{k,x}(\xi) \le  \\
\le 2^{d+1}\omega\sum_{x\in X} \sum_{s=2}^4 L^{(2s-2)k} \mathfrak{C}^2(2^{d+1}-1)^2
\frac{1}{L^{kd}}\sum_{l=0}^{M}L^{2lk}\sum_{y\in B_x^*}|\nabla^l\nabla^s\xi|^2(x)\le \\
\le   2^{d+1}\omega 2^{d+1} \mathfrak{C}^2(2^{d+1}-1)^{d+2} 3 L^{-2k}
\sum_{l=2}^{M+4}L^{2lk}\sum_{y\in X^*}|\nabla^l\xi|^2(x),
\end{multline}
where in the last inequality we took into account that each point $y\in X^*$ may accur in $B_x^*$ for at most $(2^{d+1}-1)^{d}L^{dk}$ points $x\in X$.

Summarising, under the conditions  \eqref{E:assume2},  \eqref{E:assume3:str},   we have 
\begin{equation}\label{est2step}
w^{X}_k(\varphi+\xi)
\le w_{k:k+1}^X(\varphi)\exp\Big(h^{-2}\frac{\overline C}{L^{2k}}\sum_{x\in X^*}\sum_{l=0}^{M+4}L^{2lk}|\nabla^l\xi(x)|^2\Big)
\end{equation}
with the constant
\begin{equation}
\label{E:barC}
\overline C=\max\{\omega(1+\mathfrak{c}d), \omega(1+\mathfrak{c})+4\mathfrak{c}, 2(\omega+8\mathfrak{c})+3 2^{d+1} \o \mathfrak{C}^2(2^{d+1}-1)^{d+2} \}
\end{equation}
that depends, afters $\omega$ is chosen,  only on the dimension $d$.

\medskip

\noindent 
{\bf Step 3.} We first bound the term in $ \xi $ in \eqref{est2step} by a smooth Gaussian and then bound the remaining integral.  Let $ \eta_{X^*} $ be a smooth cut-off function such that $ \supp\, \eta_{X^*}\subset (X^*)^*, \eta_{X^*}=1 $ on $ X^* $, and
\begin{equation}
\label{E:eta*bound}
\bigl|\nabla^l\eta_{X^*}\bigr|\le \varTheta L^{-lk}.
\end{equation}
Then the bound in \eqref{est2step} implies taht
\begin{equation}
\label{bkbound1}
w_k^X(\varphi+\xi)\le w^X_{k:k+1}(\varphi)\exp\big(\frac{1}{2}\varkappa(\Bscr_k\xi,\xi)\big),
%\int_{\cbX_N}\exp\big(\frac{1}{2} \varkappa( {\Bscr}_{k}\xi,\xi)\big)\,\mu_{k+1}(\d\xi),
\end{equation}
where $\varkappa=2\overline Ch^{-2}$ and   
\begin{equation}
\label{bkbound2}
( {\Bscr}_{k}\xi,\xi)=\frac{1}{L^{2k}}\sum_{x\in\L_N}\sum_{l=0}^{M+4}L^{2lk}\big|\eta_{X^*}(x)(\nabla^l\xi)(x)\big|^2.
\end{equation}
Explicitly,
\begin{equation}
\label{bkbound3}
{\Bscr}_{k}={\Bscr}^{\ssup{0}}_{k}+\sum_{l=1}^{M+4}{\Bscr}^{\ssup{l}}_{k}
\end{equation}
with
\begin{equation}
{\Bscr}_{k}^{\ssup{l}}\xi=\frac{1}{L^{2k}}(\nabla^l)^*\eta_{X^*}^2\nabla^l\xi, \  l=1,\dots, \widetilde{M}+4,\  \text{  and }\
{\Bscr}_{k}^{\ssup{0}}\xi=\frac{1}{L^{2k}}\varPi(\eta_{X^*}^2\xi), 
\end{equation}
where $\varPi\colon\cbV_N\to\cbX_N$ is the projection $(\varPi \varphi)(x)= \varphi(x)-\frac1{\abs{\Lambda_N}}\sum_{y\in\Lambda_N} \varphi(y)$
(for $l\ge 1$ the projection is not needed since $(1,\nabla_i^* \varphi)=(\nabla_i 1, \varphi)=0$).

It remains only to show that 
$$
\int_{\cbX_N}\,\exp\big(\frac{1}{2} \varkappa(\Bscr_k\xi,\xi)\big)\mu_{k+1}(\d\xi)\le 2^{|X|}.
$$  A formal Gaussian calculation with respect to the measure   $\mu_{k+1}$ with the covariance operator ${\Cscr}_{k+1}$ yields
\begin{equation}
\begin{aligned}
\int_{\cbX_N}\exp\big(\frac{1}{2}\varkappa( {\Bscr}_{k}\xi,\xi) \big)\mu_{k+1}(\d\xi)&=\Big(\frac{\det({\Cscr}_{k+1}^{-1}-\varkappa {\Bscr}_{k})}{\det ({\Cscr}_{k+1}^{-1})}\Big)^{-\frac{1}{2}}\\ & =\det\Big(\id -\varkappa {\Cscr}_{k+1}^{\frac{1}{2}}{\Bscr}_{k} {\Cscr}_{k+1}^{\frac{1}{2}}\Big)^{-\frac{1}{2}}.
\end{aligned}
\end{equation}

To justify this calculation we will derive a bound on the spectrum $
\sigma({\Cscr}_{k+1}^{\frac{1}{2}}{\Bscr}_{k} {\Cscr}_{k+1}^{\frac{1}{2}})$ in the following lemma.

\begin{lemma} Using the shorthand 
$\upeta(d) := \upeta( 2\lfloor  \frac{d+2}{2}   \rfloor + 8, d)=2\upkappa(d)-d$, we have:
\lsm[hgnac]{$\upeta(d)$}{$=  \upeta( 2\lfloor  \frac{d+2}{2}   \rfloor + 8, d)$}%
%
% where $\lfloor t \rfloor$ denotes the integer value of $t$.
\label{L:spectrum}
\begin{enumerate}
 \item[(i)] The operators $
{\Cscr}_{k+1}^{\frac{1}{2}}{\Bscr}_{k} {\Cscr}_{k+1}^{\frac{1}{2}}$  are symmetric and positive definite.
\end{enumerate}
There exist constants $M_0$ and $M_1$ that depend only on  the dimension $d$  such that for any $N$ and any $k=1,\dots,N,$
\begin{enumerate}
\item[(ii)] 
$\sup\sigma({\Cscr}_{k+1}^{\frac{1}{2}}{\Bscr}_{k} {\Cscr}_{k+1}^{\frac{1}{2}})\le M_0 L^{d+\upeta(d)}$   and
\item[(iii)]
$\tr \Big({\Cscr}_{k+1}^{\frac{1}{2}}{\Bscr}_{k} {\Cscr}_{k+1}^{\frac{1}{2}}\Big)\le M_1 |X|_k L^{\upeta(d)}$.
\end{enumerate}
\end{lemma}
Postponing momentarily the proof of the Lemma, we first observe that    $ \varkappa <\frac{1}{2M_0L^{d+\upeta(d)}} $ with $h\ge L^{\upkappa(d)}4\overline{C} M_0$, and thus the eigenvalues $\lambda_j$, $j=1,\dots, L^{Nd}-1$ of $ \varkappa {\Cscr}_{k+1}^{\frac{1}{2}}{\Bscr}_{k} {\Cscr}_{k+1}^{\frac{1}{2}} $ lie between $0$ and $ \frac{1}{2} $. The formal Gaussian calculation is then justified and
\begin{multline}
\log\det \Big(\id -\varkappa {\Cscr}_{k+1}^{\frac{1}{2}}{\Bscr}_{k} {\Cscr}_{k+1}^{\frac{1}{2}}\Big)\ge \sum_{i}\log(1-\l_i)\ge \sum_{i}-2\l_i=-2\tr\Big(\varkappa {\Cscr}_{k+1}^{\frac{1}{2}}{\Bscr}_{k} {\Cscr}_{k+1}^{\frac{1}{2}} \Big)\\  \ge -2M_1 L^{\upeta(d)}\varkappa|X|_k= -4\overline C M_1 L^{\upeta(d)} h^{-2}|X|_k.
\end{multline}
Hence
\begin{equation}
\det\Big(\id -\varkappa {\Cscr}_{k+1}^{\frac{1}{2}}{\Bscr}_{k} {\Cscr}_{k+1}^{\frac{1}{2}}\Big)^{-\frac{1}{2}}\le {\ex}^{\frac{2\overline  C M_1|X|_k}{h^2}L^{\upeta(d)}} \le {\ex}^{\frac{2\overline  C M_1|X|_k}{h_1^2}L^{-d}} 
\end{equation}
and the Lemma~\ref{L:strongerproperties} follows with 
\begin{equation}  \label{E:defh_1}
 h_1(d,\omega)^2\ge 4\overline C   \max\bigl(M_0, \tfrac{M_1}{5^d 2\log 2}\bigr) .
 \end{equation}
\qed
\end{proofsect}

\begin{proofsect}{Proof of Lemma~\ref{L:spectrum}}

\noindent The claim (i) follows from definitions.

\noindent The estimate (ii) follows from the estimate
\begin{equation}
\norm{{\Bscr}_{k}{\Cscr}_{k+1}\xi}_2 \le M_0 L^{d+\upeta(d)}\norm{\xi}_2  \mbox{ for all } \xi\in\cbX_N.
\end{equation}
For $ {\Bscr}_{k}^{\ssup{0}} $, we first observe that
\begin{equation}
\label{E:Bk0}
L^{2k}\norm{{\Bscr}_{k}^{\ssup{0}}\xi}_2 =\norm{\varPi(\eta_{X^*})^2\xi}_2 \le\norm{(\eta_{X^*})^2\xi}_2 \le \norm{\xi}_2 .
\end{equation}
In view of Proposition~\ref{P:FRD}, the operator $ {\Cscr}_{k+1} $ acts by convolution with respect to the function ${\Ccal}_{k+1}$. 
With the bounds \eqref{E:Bk0}, \eqref{E:FRk}, \eqref{E:fluctk}, and $c_{\max}=\max_{\abs{\boldsymbol{\alpha}}\le 2(M+4)}c_{\boldsymbol{\alpha},0}$, 
we have
(recall that $\upeta(0,d) \le  \upeta( 2\lfloor  \frac{d+2}{2}   \rfloor + 8, d) = \upeta(d))$)
\begin{equation}
\norm{{\Bscr}_{k}^{\ssup{0}}{\Cscr}_{k+1}\xi}_2 \le L^{-2k}\norm{{\Cscr}_{k+1}\xi}_2 \le L^{-2k}\sum_{z\in\L_N}|{\Ccal}_{k+1}(z)|\norm{\xi}_2 \le c_{\max} L^{d+\upeta(d)}\norm{\xi}_2 .
\end{equation}
For $ {\Bscr}_{k}^{\ssup{l}} $ we use the discrete product rule
\begin{equation}
\label{}
\nabla_i(fg)=\nabla_i f\text{\sf S}_i g+\text{\sf S}_if\nabla_i g,
\end{equation}
where
\begin{equation}
\label{}
(\text{\sf S}_i f)(x):=\tfrac{1}{2}f(x)+\tfrac{1}{2}f(x+\ex_i).
\end{equation}
The operations $ \text{\sf S}_i $ commute with all discrete derivatives. Using multiindex notation 
\begin{equation}
\label{}
\nabla^{\boldsymbol{\alpha}}:=\prod_{i=1}^d\nabla_i^{\alpha_i} \  \text{  and  }\  
\text{\sf S}^{\boldsymbol{\alpha}}:=\prod_{i=1}^d \text{\sf S}_i^{\alpha_i}, 
\end{equation}
we get the Leibniz rule
\begin{equation}
\label{E:Leibniz} 
\nabla^{\boldsymbol{\gamma}} (fg)=\sum_{\boldsymbol{\alpha}+\boldsymbol{\beta}=\boldsymbol{\gamma}} C_{\boldsymbol{\alpha},\boldsymbol{\beta}}\bigl(\text{\sf S}^{\boldsymbol{\alpha}}\nabla^{\boldsymbol{\beta}} f\bigr)\bigl(\text{\sf S}^{\boldsymbol{\beta}}\nabla^{\boldsymbol{\alpha}} g\bigr),
\end{equation}
with suitable constants $C_{\boldsymbol{\alpha},\boldsymbol{\beta}}$. Thus
\begin{equation}
\label{formulaforA_k^l}
{\Bscr}_{k}^{\ssup{l}} {\Cscr}_{k+1}\xi =L^{(2l-2)k}\sum_{|\boldsymbol{\gamma}|=l}\sum_{\boldsymbol{\alpha}+\boldsymbol{\beta}=\boldsymbol{\gamma}}C_{\boldsymbol{\alpha},\boldsymbol{\beta}} \text{\sf S}^{\boldsymbol{\alpha}}
(\nabla^{\boldsymbol{\beta}})^*(\eta_{X^*})^2
 \text{\sf S}^{\boldsymbol{\beta}}(\nabla^{\boldsymbol{\alpha}})^*\nabla^{\boldsymbol{\gamma}}{\Cscr}_{k+1}\xi.
\end{equation}
Notice that $ \norm{\text{\sf S}^\beta}=1 $ (with the operator norm induced by $l^2$ norms on $\Vcal_N$). Further,  using \eqref{E:eta*bound}, \eqref{E:nablas}, and again \eqref{E:Leibniz}, we have
\begin{equation}
\label{}
\bigl|(\nabla^{\boldsymbol{\beta}})^*(\eta_{X^*})^2\bigr|\le \varTheta^2 C_{\max}L^{-k|\boldsymbol{\beta}|}
\end{equation}
with
\begin{equation}
\label{E:Cmax}
C_{\max}=\sum_{\substack{\boldsymbol{\alpha}, \boldsymbol{\beta}\\
\abs{\boldsymbol{\alpha}+\boldsymbol{\beta}}\le M+4}} C_{\boldsymbol{\alpha},\boldsymbol{\beta}}.
\end{equation}
As a result we get, 
 recalling  that $l \le \widetilde{M}+4$, where $\widetilde{M} = \lfloor  \frac{d+2}{2} \rfloor$, and that  $$\upeta(2(\widetilde{M}+4), d) = \upeta(d),$$
\begin{multline}
\norm{{\Bscr}_{k}^{\ssup{l}} {\Cscr}_{k+1}}\le\\
\le  L^{(2l-2)k}\sum_{\abs{\boldsymbol{\gamma}}=l}
\sum_{\boldsymbol{\alpha}+\boldsymbol{\beta}=\boldsymbol{\gamma}}
C_{\boldsymbol{\alpha},\boldsymbol{\beta}} \varTheta^2 C_{\max} L^{-k\abs{\boldsymbol{\beta}}}
L^{(k+1)d}  c_{\max} L^{-k(d-2+\abs{\boldsymbol{\alpha}}+l  )}L^{\upeta(d)}\le \\
\le
\varTheta^2 C_{\max}^2 c_{\max}
L^{d+\upeta(d)}.
\end{multline}
This completes the proof of (ii) with $M_0= \varTheta^2 C_{\max}^2 c_{\max} $.

\medskip

\noindent  To prove the estimate (iii), we first observe that $ {\Cscr}_k \1_{\L_N}=0 $. Hence $ {\Bscr}_{k} {\Cscr}_k $ can be viewed as an operator from $\Vcal_N$ (instead of $ \cbX_N $) to $\Vcal_N$ with the same trace. 
 To compute the trace of $ {\Bscr}_{k} {\Cscr}_{k+1} $ we now use the  orthonormal basis given by the unit coordinate vectors
\begin{equation}
{\ex}_x(z)=\left\{\begin{array}{r@{\,,\,}l} 1 & z=x,\\ 0 & z\not= x.\end{array}\right.
\end{equation}
According to \eqref{formulaforA_k^l}, for $ l\ge 1 $  we get 
 \begin{equation}
\big|( {\ex}_x,{\Bscr}_{k}^{\ssup{l}} {\Cscr}_{k+1}{\ex}_x)\big|=0 \ \text{ whenever  } \ x\notin (X^*)^*. 
 \end{equation}
For $ x\in (X^*)^*$ we use \eqref{formulaforA_k^l} and the bound
\begin{equation}
\sup_{z}\big|(\nabla^{\boldsymbol{\alpha}})^*\nabla^{\boldsymbol{\gamma}}{\Ccal}_{k+1}(z)\big|\le c_{\max} 
L^{-k(d-2+\abs{\boldsymbol{\alpha}}+\abs{\boldsymbol{\gamma}}  )}L^{\upeta(d)}
\end{equation}
to conclude that
\begin{equation}
\big|( {\ex}_x,{\Bscr}_{k}^{\ssup{l}}{\Cscr}_{k+1}{\ex}_x)\big|\le \varTheta^2 C_{\max}^2 c_{\max}  L^{-kd+\upeta(d)}
\end{equation}
and
\begin{equation}
\tr {\Bscr}^{\ssup{l}}_{k} {\Cscr}_{k+1}=\sum_{x\in\L_N}({\ex}_x,{\Bscr}^{\ssup{l}}_{k}{\Cscr}_{k+1}{\ex}_x) 
\le\varTheta^2 C_{\max}^2 c_{\max}   2^{d+2}L^{\upeta(d)}\abs{X}_k.
\end{equation}

For $ {\Bscr}_{k}^{\ssup{0}} $, we explicitly express the projection, $\varPi{\ex}_x={\ex}_x-\1_{\L_N}\frac{1}{|\L_N|} $, yielding
\begin{multline}
L^{2k}({\ex}_x,{\Bscr}_{k}^{\ssup{0}} {\Cscr}_{k+1}{\ex}_x)=(\varPi{\ex}_x,\eta^2_{X^*}{\Cscr}_{k+1}{\ex}_x)=\\=({\ex}_x,\eta^2_{X^*}{\Cscr}_{k+1}{\ex}_x)-(\1_{\L_N}\frac{1}{|\L_N|},\eta^2_{X^*}{\Cscr}_{k+1}{\ex}_x)\\=\eta^2_{X^*}(x){\Ccal}_{k+1}(0)-\frac{1}{|\L_N|}(\1_{\L_N},\eta^2_{X^*}{\Cscr}_{k+1}{\ex}_x).
\end{multline}
Therefore
\begin{multline}
\tr {\Bscr}^{\ssup{0}}_{k} {\Cscr}_{k+1}=\sum_{x\in\L_N}({\ex}_x,{\Bscr}^{\ssup{0}}_{k}{\Cscr}_{k+1}{\ex}_x)=\\=L^{-2k}\Big(\sum_{x\in\L_N}\eta^2_{X^*}(x)\Big){\Ccal}_{k+1}(0)-\frac{1}{|\L_N|}( \1_{\L_N},\eta^2_{X^*}{\Cscr}_{k+1}\1_{\L_N})
\le\\
\le L^{-2k} c_{\max}  L^{-k(d-2)}L^{\upeta(d)}
\sum_{x\in (X^*)^*}1  \le  c_{\max} 2^{d+2}L^{\upeta(d)}|X|_k.
\end{multline}
Thus
$$
\tr\Big({\Bscr}_{k}{\Cscr}_{k+1}\Big)\le C\big|\big(X^*\big)^*\big|_k\le (M+5) \varTheta^2 C_{\max}^2 c_{\max}   2^{d+2} L^{\upeta(d)}|X|_k.
$$
We get the claim  (iii)  with $M_1= (\widetilde{M}+5) \varTheta^2 C_{\max}^2 c_{\max}   2^{d+2}$.
\qed
\begin{remark}
Notice that, with the particular values of $M_0$ and $M_1$ given above, we can choose $h_1$ fulfilling 
\eqref{E:defh_1}
by taking 
\begin{equation}
\label{E:defh_1=}
h_1^2= \overline{C} (\widetilde{M}+5) M_0.
\end{equation} \hfill $ \diamond $
\end{remark}
\end{proofsect}

%% file: AKM-ch6-smoothness-30thJune-2016.tex
\chapter{Smoothness}\label{S:smooth}

We  prove Proposition \ref{P:Tnonlin} asserting  the smoothness of the renormalisation map
\begin{equation}
\label{E:S}
S\colon  \Ucal \times B_{\frac12}  \subset \left( M_0(\Bcal, \cbX)\times  M( \Pcal^{\com}, \cbX) \right)
\times \R^{d\times d}_{\rm sym} \to  M( (\Pcal^\prime)^{\com}, \cbX)
\end{equation}
on a suitable scale of functions spaces.
Here, $\Bcal=\Bcal_k$, $\Pcal=\Pcal_k$, and $\Pcal^\prime=\Pcal_{k+1}$ with $k$ fixed. (Later, when the dependence of the map $S$ on $k$ will be crucial, 
we will use the notation $S_k$  instead of $S$.)
Let us recall the explicit formula  \eqref{E:K_k+1} for $K_{k+1}=K^{\prime}=S(H,K,\bq)$,
\begin{equation}\label{E:K'p}
K^{\prime}(U,\varphi)=\sum_{X\in\Pcal(U)}\chi(X,U)\widetilde{I}^{U\setminus X}(\varphi)\int_{\cbX} \big(\widetilde J(\varphi)\circ P(\varphi+\xi)\big)(X)\mu_{k+1}(\d \xi)
\end{equation}
with $\widetilde{I}={\rm e}^{-\widetilde{H}}$,  $\widetilde{J}=1-\widetilde{I}$,
$P=(I-1)\circ K$, and ${I}={\rm e}^{-{H}}$.

It will be useful to split the map $S$ into a composition of a series of maps and to deal with them one by one.
To this end, we first recall the notation for relevant normed spaces. 
In Section~\ref{S:key} we have already introduced the 
sequence of normed spaces $\bM=\bM_{r_0}\embed\bM_{r_0-2}\embed\dots\embed \bM_{r_0-2m}$, 
defined as 
$\bM_{r} = \{ K \in M(\Pcal^{\com}, \cbX)  : \norm{K}_{k,r}^{(\mathsf{A})}  < \infty \} $
 and equipped with the norm $\norm{\cdot}_{k,r}^{(\mathsf{A})}$, 
 $r=r_0,r_0-2,\dots,r_0-2m$, the space $\bM_0=(M( \Bcal_k, \cbX),\norm{\cdot}_{k,0})$, 
 and the sequence of spaces
 $\bM'=\bM'_{r_0}\embed\bM'_{r_0-2}\embed\dots\embed \bM'_{r_0-2m}$ 
 with  $\bM'_{r}=\{ K \in M( \Pcal^{\com}_{k+1}, \cbX),\norm{K}_{r,k+1}^{(\mathsf{A})} < \infty \}$,
 equipped with the norm $\norm{\cdot}_{k+1,r}^{(\mathsf{A})}$,
  $r=r_0,r_0-2,\dots,r_0-2m$.
We also  introduce  the space 
$\bM_{\ttt}=   \{ F \in M(\Bcal, \cbX), \tnorm{F}_{k} < \infty \}$.
    
 One difficulty is that convolution with the measure $\mu_{k+1}$ does not preserve the factorization 
 in connected $k$-polymers.\footnote{We are grateful to S. Buchholz for pointing this out and for suggesting the use
 of the norm $\| \cdot \|_{k,r}^{(\mathsf{A},\mathsf{B})}$.} More precisely, if 
 $$
 K(X, \varphi) = \prod_{Y \in  \Ccal(X)} K(Y, \varphi)$$
 and if 
 $$ RK(X, \varphi) := \int_{\cbX} K(X, \varphi + \xi) \mu_{k+1}(\d \xi),$$ 
 then in general 
 $$ RK(X, \varphi) \ne \prod_{Y \in  \Ccal(X)} RK(Y, \varphi)$$
 because the support of the covariance $\Cscr_{k+1}$  has range bounded by $L^{k+1}/2$ but not by $L^k/2$. 
 Thus we cannot only consider functionals defined for connected $k$-polymers but we need to consider functionals which 
 involve all $k$-polymers and we define
 \begin{equation}
 \widehat \bM_r := \{ K \in M(\Pcal_k, \cbX), \norm{K}_{k,r}^{(\mathsf{A},\mathsf{B})} < \infty \}, 
 \end{equation}
 \begin{equation}
  \widehat \bM_{:,r} := \{ K \in M(\Pcal_k, \cbX), \norm{K}_{k:k+1,r}^{(\mathsf{A},\mathsf{B})} < \infty \},
  \lsm[McPhat]{$\widehat \bM_r$ }{   $   \{ K \in M(\Pcal_k, \cbX), \norm{K}_{k,r}^{(\mathsf{A},\mathsf{B})} < \infty \}$}
  \lsm[McPhat:]{$\widehat \bM_{:,r}$ }{  $ \{ K \in M(\Pcal_k, \cbX), \norm{K}_{k:k+1,r}^{(\mathsf{A},\mathsf{B})} < \infty \}$}
 \end{equation}
where 
\begin{equation}    \label{E:weight_B}
\norm{K}_{k,r}^{(\mathsf{A},\mathsf{B})} := \sup_{X \in \Pcal_k \setminus \emptyset}  \Gamma_{\mathsf{A}}(X)  \mathsf{B}^{|\Ccal(X)|}  \,  \|K(X)\|_{k,X,r}
\end{equation}
with
\begin{equation}  \label{E:6_Gamma_general}
\Gamma_{\mathsf{A}}(X) := \prod_{Y \in \Ccal(X)}  \Gamma_{\mathsf{A}}(Y)    \quad \text{for $X \in \Pcal \setminus \emptyset$}
\end{equation}
and where $\norm{\cdot}_{k:k+1,r}^{(\mathsf{A},\mathsf{B})}$ is defined in the same way using $\|K(X)\|_{k:k+1,X, r}$.
Note that the definition of the spaces does not depend on the weights $\mathsf{A} > 0$ and $\mathsf{B} > 0$ since there are only finitely many polymers.

The map $S$ will be rewritten as a composition of several partial maps:

\noindent
The exponential map,
\begin{align}
&E \colon  \bM_0 \to  \bM_{\ttt} \text{ defined by } \\  \label{E:defE}
&E(\widetilde H)= \exp\{-\widetilde H\}=\widetilde I, 
 \lsm[Ec]{$E$}{the map  $E \colon (\bM_0,\norm{\cdot}_{k,0}) \to (\bM_{\ttt},\tnorm{\cdot}_{k,\cdot}) $
  defined by $E(H)(B,\varphi)=\exp\{-H(B,\varphi)\}$
}
\intertext{   three polynomial maps,}
&P_1\colon  \bM_{\ttt}\times  \bM_{\ttt} \times    \widehat \bM_{:,r_0}      % \bM_{:, r_0}
 %(M( \Pcal^{\com}, \cbX),\norm{\cdot}^{(A/4)}_{k:k+1,r}) 
\to \bM'_{r_0}  \text{ defined by } \\ \label{E:defP1}
&P_1(\widetilde{I}, \widetilde{J}, \widetilde{P})(U,\varphi)=\!\!\!\!\sum_{\heap{X_1,X_2\in\Pcal(U)}{X_1\cap X_2=\emptyset}}\!\!\!\chi(X_1\cup X_2,U)\widetilde I^{U\setminus (X_1\cup X_2)}(\varphi)\widetilde J^{X_1}(\varphi) \widetilde{P}(X_2,\varphi),\\
\lsm[Pcka1]{$P_1$}{$P_1(\widetilde{I}, \widetilde{J}, \widetilde{P})(U,\varphi)=\sum_{\heap{X_1,X_2\in\Pcal(U)}{X_1\cap X_2=\emptyset}}\chi(X_1\cup X_2,U)\widetilde I^{U\setminus (X_1\cup X_2)}(\varphi)\widetilde J^{X_1}(\varphi) \widetilde{P}(X_2,\varphi)$ mapping $(M(\Bcal_k, \cbX),\tnorm{\cdot}_{k})\times  (M( \Bcal_k, \cbX),\tnorm{\cdot}_{k})
\times (M( \Pcal_k^{\com}, \cbX),\norm{\cdot}^{(A/2)}_{k:k+1,r})$ into $ (M( (\Pcal_{k+1})^{\com}, \cbX),\norm{\cdot}^{(\mathsf{A})}_{k+1,r})$
}
&P_2 \colon  \bM_{\ttt} \times \bM_{r} \to  \bM_r  
% % old text said   \bM_{:, r}, but I think the : is a typo
% %(M( \Pcal^{\com}, \cbX),\norm{\cdot}^{(A/2)}_{k,r})
\text{ defined by } \\  \label{E:defP2}
&P_2(I,K)=(I-1)\circ K,\\
\lsm[Pcka2]{$P_2$}{$P_2(I,K)=(I-1)\circ K$ 
mapping }
%%$(M( \Bcal_k, \cbX),\tnorm{\cdot}_{k})\times (M( \Pcal_k^{\com}, \cbX),\norm{\cdot}^{(\mathsf{A})}_{k,r})$ 
%into $(M( \Pcal_k^{\com}, \cbX),\norm{\cdot}^{(\mathsf{A})}_{k,r})$
%
%}
&  P_3\colon \bM_r \to \widehat \bM_r,   \\  \label{E:defP3}
&  (P_3 K)(X, \varphi) = \prod_{Y \in \Ccal(X)} K(Y, \varphi)  
\lsm[Pcka3]{$P_3$}{$(P_3 K)(X, \varphi) = \prod_{Y \in \Ccal(X)} K(Y, \varphi)$    }
%%%%{  \newsix{$(P_3 K)(X, \varphi) = \prod_{Y \in \Ccal(X)} K(Y, \varphi)$ mapping}
% $M(\Pcal_k^{\com}, \norm{\cdot}_{k,r}^{A/2})$ to $M(\Pcal_k, \norm{\cdot}_{k,r}^{A/2,B})$}
%}
%
\intertext{and, finally, two linear renormalisation maps that are the source of loss of regularity,}
&R_1\colon     \widehat \bM_{r_0}   %% old version  \bM_{r_0}
\times B_{\frac{1}{2}} \to  \widehat \bM_{:, r_0}    % old version  \bM_{:, r_0}
\text{ defined by } \\  \label{E:defR1}%(M( \Pcal^{\com}, \cbX),\norm{\cdot}^{(A/2)}_{k:k+1,r}),\\ 
&R_1(P,\bq)(X,\varphi)= (\bR^{(\bq)} P)(X,\varphi)=\int_{\cbX}P(X,\varphi+\xi)\mu_{k+1}^{(\bq)}(\d\xi), 
\quad X\in\Pcal,      \\            % old version  X\in\Pcal^{\rm c},                       
&R_2\colon  \bM_0\times \bM_{r_0}  \times B_{\frac{1}{2}} 
  \to \bM_0  \text{ defined by } \\  \label{E:defR2}
&R_2(H,K,\bq)(B,\varphi)=
\Pi_2 \Bigl((\bR^{(\bq)} H)(B,\varphi)-
\sum_{\begin{subarray}{c}  X\in \cS \\  X\supset B  \end{subarray}} 
\tfrac1{\abs{X}} (\bR^{(\bq)} K)(X,\varphi)\Bigr),
\end{align}
where we write $ B_{\frac{1}{2}}=  \left\{   \bq \in   \R^{d\times d}_{\rm sym}  : \|q \| < \tfrac12 \right\} $.
\lsm[Rcka1]{$R_1$}{$R_1(P,\bq)(X,\varphi)= (\bR^{(\bq)} P)(X,\varphi)= \int_{\cbX}P(X,\varphi+\xi)\mu_{k+1}^{(\bq)}(\d\xi)$ 
mapping $(M( \Pcal_k^{\com}, \cbX),\norm{\cdot}^{(\mathsf{A})}_{k,r})
\times (\R^{d\times d}_{\rm sym},\norm{\cdot})$ into $(M(\Pcal_k^{\com}, \cbX),\norm{\cdot}^{(\mathsf{A})}_{k:k+1,r})$
}
\lsm[Rcka2]{$R_2$}{$R_2(H,K,\bq)(B,\varphi)=\Pi_2 \Bigl((\bR^{(\bq)} H)(B,\varphi)-\sum_{\begin{subarray}{c}  X\in \cS \\  X\supset B  \end{subarray}} 
\tfrac1{\abs{X}} (\bR^{(\bq)} K)(X,\varphi)\Bigr)$ 
mapping $(M_0( \Bcal_k, \cbX),\norm{\cdot}_{k,0})\times (M( \Pcal_k^{\com}, \cbX),\norm{\cdot}^{(\mathsf{A})}_{k,r})
\times (\R^{d\times d}_{\rm sym},\norm{\cdot})$ into $(M_0( \Bcal_k, \cbX),\norm{\cdot}_{k,0})$
}

In terms of these maps we have
\begin{equation}
\label{E:S-P1}
S(H,K,\bq)=P_1\bigl(E(R_2(H,K,\bq)), \, 1-E(R_2(H,K,\bq)),  \, R_1(P_3(P_2(E(H),K)),\bq)\bigr).
  \end{equation}
\lsm[Scka00]{$S$}{the map $S$ is composed as \\
$S(H,K,\bq)=P_1\bigl(E(R_2(H,K,\bq)), 1-E(R_2(H,K,\bq)), R_1(P_2(E(H),K),\bq)\bigr)$
}

Notice that the norms on the corresponding spaces are chosen in a natural way, 
  with the exception of the space  $ M( \Pcal, \cbX)$  in the role of the domain space of
  the map $P_1$ as well as the target space of the map $R_1$, that comes equipped with the norm 
   $\norm{\cdot}^{(\mathsf{A},\mathsf{B})}_{k:k+1,r_0}$. This is driven by the bound (iv) from Lemma~\ref{L:prop} that 
   makes  the norm $\norm{K(X, \cdot)}_{k:k+1, r}$ natural for the map $R_1$.
   The additional weight $\mathsf{B}^{|\Ccal(X)|}$ in the norms of $\widehat \bM_r$ and $\widehat \bM_{:,r}$
    plays an important role in the estimates for the map $P_1$
   and is a substitute for the fact that we no longer deal with maps which factor in connected $k$-polymers. More precisely
  % if  the number of connected components is large $|\Ccal(X)|$ is large and
   if $K$ factors we can use the bound
  (i)  from Lemma~\ref{L:prop}  to conclude that 
  $$\norm{K(X)}_{k,X,r}  \le  \prod_{Y \in \Ccal(X)} \norm{K(Y)}_{k,Y,r} 
  \le \Gamma_{\mathsf{A}}(X)^{-1} \left[ \|K\|_{k,r}^{(\mathsf{A})} \right]^{|\Ccal(X)|}.
  $$ 
  This provides additional smallness  if $\norm{K}_{k,r}^{(\mathsf{A})}$ is small
  and the number of connected components $|\Ccal(X)|$ is large. If $K$ does not factor
  we can use the bound 
  $$\norm{K(X)}_{k,X,r} \le \Gamma_{\mathsf{A}}(X)^{-1} \mathsf{B}^{-|\Ccal(X)|}  \norm{K}_{k,r}^{(\mathsf{A},\mathsf{B})}
  $$ instead to get a good decay  for a large number of components.

  The dependence on the parameters $\mathsf{A}$ and $\mathsf{B}$  in the definition of the weak norms \eqref{E:weak_norm} 
   and in the norm  \eqref{E:weight_B} plays an important role here, we thus  incorporate it explicitly into the notation and write, e.g.,  $\norm{\cdot}^{(\mathsf{A})}_{k,r}$. Note that for a fixed $N$ (where $L^N$ is the system size) the norms
   $\norm{\cdot}^{(\mathsf{A})}_{k,r}$ and  $\norm{\cdot}^{(\mathsf{A},\mathsf{B})}_{k,r}$ are equivalent for all $\mathsf{A} > 0$ and $\mathsf{B}> 0$ (because there are only
   finitely many polymers), but the constant in the equivalences depend strongly  on $N$. Since we are interested in bounds
   on the derivatives which are independent of $N$ a careful choice of the parameters $\mathsf{A}$ and $\mathsf{B}$ is crucial.
   
 In the following sections we will show that all maps introduced 
  above belong to the class  $\widetilde C^m(\bX \times B_{\frac12}, \bY)$,
  introduced in Appendix \ref{appChain}, 
   for suitable scales of spaces $\bX = \bX_m  \embed \ldots \embed \bX_0$ and
   $\bY = \bY_m  \embed \ldots \embed \bY_0$. Finally we will use the 
   chain rule in the $\widetilde C^m$ spaces to show that the  same regularity for the composed map $S$,
   see  Section~\ref{S:finalsmooth}.
   In fact  the maps above actually possess arbitrarily many Fr\'{e}chet derivatives (or are even real-analytic) 
but the setting of the $\widetilde C^m$ spaces is setting which naturally goes with the estimates that are independent of 
$N$ (where $L^N$ is the system size).

Let us first discuss the partial maps one by one, starting from the most interior one in the composition \eqref{E:S-P1}.
  
\section{Immersion $E \colon \bM_0 \to \bM_{\ttt}$}
While the norm $\norm{H}_{k,0}$ is expressed directly in terms of the co-ordinates $ \lambda, a,\bc,\bd $ of the ideal Hamiltonian $H\in \bM_0$, the terms involving $E(H)(B,\varphi)={\rm e}^{-H(B,\varphi)}$ will be  evaluated with the help of the norm $\tnorm{\cdot}_{k}$.
Considering thus the map $E\colon \bM_0 \to \bM_{\ttt} $, we have:
\begin{lemma}\label{immersion}\
We have $\tnorm{H}_{k}\le  5 \norm{H}_{k,0}$ for any $H\in \bM_0$. Moreover,  there exist  constants $\delta=\delta(r_0)$ and $C=C(r_0)$ so that 
$E$ is  smooth  on $B_\delta=\{H\in \bM_0\colon \norm{H}_{k,0} <\delta\} $ with uniformly bounded  derivatives,
\begin{equation}\label{E:DImm}
\tnorm{D^jE(H)(\dot{H},\ldots,\dot{H})}_k\le C\norm{\dot{H}}_{k,0}^j,\qquad j\le m.
\end{equation}
In particular we have
\begin{equation}\label{immersion2}
\tnorm{E(H)-1}_k\le C\norm{H}_{k,0}.
\end{equation}
\end{lemma}

\begin{remark}
The definition of norm $ \tnorm{\cdot}_k $ involves the parameter  $r_0 $ (see \eqref{E:tnorm}) but the statement does not depend on $ r_0 $.
 \hfill $\diamond$
\end{remark}

\begin{proofsect}{Proof}
Let $H\in \bM_0$ and $ B\in\Bcal $.
First, we estimate $\tnorm{H(B,\cdot)}_{k,B}$ by $\norm{H}_{k,0} $.
In view of the definitions \eqref{E:tnorm} and \eqref{E:bnorm^jXr}, we need to compute the norms 
$$\bnorm{D^p H(B,\varphi)}^{k,B}, \quad p=0,1,2,$$
(the higher derivatives vanish as $H$ is a quadratic function).

Starting with $p=0$ and recalling the definitions \eqref{E:P}--\eqref{E:Q}, we get
\begin{multline}
\bnorm{D^0 H(B,\varphi)}^{k,B}=\abs{H(B,\varphi)}\le |\lambda|L^{dk} +L^{\frac{dk}{2}}\sum_{i=1}^d|a_i|\big(\sum_{x\in B}|\nabla\varphi(x)|^2\big)^{1/2}+\\+L^{\frac{dk}{2}}\sum_{i,j=1}^d|\bc_{i,j}| \big(\sum_{x\in B}|\nabla^2\varphi(x)|^2\big)^{1/2}
+\sum_{x\in B}|\nabla\varphi(x)|^2\tfrac12\sum_{i,j=1}^d|\bd_{i,j}|.
\end{multline}
Here, when evaluating the term $\sum_{x\in B} \sum_{i=1}^d\abs{a_i} \abs{\nabla_i\varphi(x)}$,
we first apply the Cauchy-Schwarz inequality in $\R^d$ and
using the bound $ |a| =\bigl(\sum_{i=1}^d |a_i|^2\bigr)^{1/2}\le  \sum_{i=1}^d |a_i| =|a|_1$,
we then employ  the Cauchy-Schwarz inequality for the second time on the sum
$\sum_{x\in B} 1\cdot  \abs{\nabla\varphi(x)}$ with
$\abs{\nabla\varphi(x)}^2= \sum_{i=1}^d\abs{\nabla_i\varphi(x)}^2$. Similarly  we treat  the next term with
$\abs{\nabla^2\varphi(x)}^2= \sum_{i,j=1}^d\abs{\nabla_i\nabla_j\varphi(x)}^2$. 
In the last term we just use the bound 
	$\bigl\vert\frac12\sum_{i,j=1}^d\bd_{i,j}\nabla_i\varphi(x) \nabla_j\varphi(x)\bigr\vert\le
	\tfrac12\norm{\bd} \abs{\nabla\varphi(x) }^2$ and then evaluate the operator norm, 
$\norm{\bd}\le (\sum_{i,j+1}^d\bd_{i,j}^2)^{1/2} \le \sum_{i,j=1}^d\abs{\bd_{i,j}}$.

Hence,
\begin{multline}
\label{E:D^0H}
\big|H(B,\varphi)\big|\le\\
\le \norm{H}_{k,0}\Big(1+\frac{1}{h}\big(\sum_{x\in B}|\nabla\varphi(x)|^2\big)^{1/2} +\frac{1}{h}L^k\big(\sum_{x\in B}|\nabla^2\varphi(x)|^2\big)^{1/2}
+\frac{1}{h^2}\sum_{x\in B}|\nabla\varphi(x)|^2\Big)\\
\le 2\norm{H}_{k,0}\Big(1+\frac{1}{h^2}\sum_{x\in B}\big(|\nabla\varphi(x)|^2+L^{2k}|\nabla^2\varphi(x)|^2\big)\Big)\le 2\norm{H}_{k,0}\Big(1+\log W^B(\varphi)\Big),
\end{multline}
where we took into account the definition \eqref{E:gX} of the weight function $W^B(\varphi)=W_k^B(\varphi)$.

Similarly, taking into account that $DH(B,\varphi)(\dot{\varphi})=\ell(\dot{\varphi}) +2Q(\varphi,\dot{\varphi})$,
we get
\begin{multline}
\label{E:DH}
\bnorm{DH(B,\varphi)}^{k,B}=\sup_{\abs{\dot{\varphi}}_{k,B}\le 1}\abs{\ell(\dot{\varphi})+2Q(\varphi,\dot{\varphi})}\le\\
\le 
\sup\{\abs{\ell(\dot{\varphi})}+\abs{2Q(\varphi,\dot{\varphi})}:  
 \sup_{x\in B^*}\abs{\nabla\dot{\varphi}(x)}\le h L^{-\frac{kd}2}\ \text{ and }\ \sup_{x\in B^*} \abs{\nabla^2\dot{\varphi}(x)} \le  h L^{-\frac{kd}2-k}\}
\\
\le 
 h L^{-\tfrac{kd}2} \bigl\{L^{kd} \sum_{i=1}\abs{a_i}+L^{kd}L^{-k}\sum_{i,j=1}^d 
\abs{\bc_{i,j}}
+\sum_{i,j=1}^d|\bd_{i,j}| \sum_{x\in B}\abs{(\nabla\varphi)(x)\bigr)}\bigr\}
 \le \\ \le\norm{H}_{k,0}\bigl(1+\frac{1}{h}\bigl(\sum_{x\in B}|\nabla\varphi(x)|^2\bigr)^{1/2}\bigr)\le 2\norm{H}_{k,0}\bigl(1+\log W^B(\varphi)\bigr)
\end{multline}
and
\begin{equation}
\label{E:D^2H}
\bnorm{D^2H(B,\varphi)}^{k,B}\le 2 h^2   L^{-dk} L^{dk} \sum_{i,j=1}^d|\bd_{i,j}| \le 2\norm{H}_{k,0}.
\end{equation}
Recalling that $D^3H(B,\varphi)(\dot{\psi},\dot{\psi},\dot{\psi} )=0$, we finally get
\begin{equation}
\tnorm{H}_{k}=\tnorm{H(B,\cdot)}_{k,B}\le 5\sup_{\varphi}   W_k^{-B}(\varphi)
\norm{H}_{k,0} (1+\log W^B(\varphi))\le 5 \norm{H}_{k,0} .
\end{equation}

To get $\tnorm{E(H)}_{k}$,  we need to compute the norms $\bnorm{D^p E(H)(B,\varphi)}^{k,B}$, $p=0,\dots,r_0$.
Using again Fa\`a di Bruno's chain rule for higher order derivatives  and the bounds 
\eqref{E:D^0H}, \eqref{E:DH}, and \eqref{E:D^2H}, we get
\begin{equation}
\label{E:DE^kB}
\bnorm{D^p E(H)(B,\varphi)}^{k,B}\le B_{r_0} \ex^{-H(B,\varphi)} 
\Bigl(1+ 2 \norm{H}_{k,0} (1+\log W^B(\varphi)) \Bigr)^p
\end{equation}
with the constant $B_{r_0}\le r_0^{r_0}$ bounding the number of partitions of the set $\{1,\dots,p\}$.
\lsm[Bcraxx]{$B_{r_0}$}{$\le r_0^{r_0}$, bound on the number of partitions (for Bruno di Fa\`a formula)}%
%
%
%\comment{Possible alternative to Fa di Bruno: write $E(H) = \sum_{j=0}^\infty   (-H)^j$ and use
%\eqref{E:F1_F2} with $X_1 = X_2 = B$ to get 
%$|(-H)^j(B,\varphi)|^{k,B,r} \le \left( |H(B, \varphi)|^{k,B,r} \right)^j$ and thus \\
%$ |E(H)|^{k,B,r} \le \ldots $\\
%does not work because we get exponential growth in $|\nabla \varphi|^2$ in the exponent}
%
Hence,
\begin{multline}
\tnorm{E(H)}_{k}\le B_{r_0}\sup_{\varphi}  \ex^{-H(B,\varphi)} W^{-B}(\varphi)
\sum_{p=0}^{r_0}\bigl(1+ 2 \norm{H}_{k,0} (1+\log W^B(\varphi))\bigr)^p\le \\  \le 
B_{r_0}\sum_{p=0}^{r_0}\sup_{\varphi}  \ex^{2 \norm{H}_{k,0} (1+\log W^{B}(\varphi))} W^{-B}(\varphi)\ex^{2p\norm{H}_{k,0}(1+\log W^B(\varphi)) }\le\\ \le
(r_0+1) B_{r_0}\ex^{2 (1+r_0)\norm{H}_{k,0}}\sup_{\varphi}  \ex^{2 \norm{H}_{k,0} (1+r_0)\log W^{B}(\varphi)} W^{-B}(\varphi)<e(r_0+1) B_{r_0}
\end{multline}
once $ \norm{H}_{k,0} $ is  sufficiently small to assure that  $2 \norm{H}_{k,0} (1+r_0)\le 1$
(we took into account that $W^{B}(\varphi)\ge 1$).

Computing the derivative of the exponent $E(H)$ as a composed function, we get $DE(H)(\dot{H})(B,\varphi)=E(H)(B,\varphi)\dot{H}(B,\varphi)$.  Using, similarly as when proving \eqref{E:F_1F_2}, the fact that a Taylor expansion of a product is the product of Taylor expansions, we get
\begin{equation}
\bnorm{DE(H)(\dot{H})(B,\varphi)}^{k,B, r_0}\le \bnorm{E(H)(B,\varphi)}^{k,B, r_0} 
 \bnorm{\dot{H}(B,\varphi)}^{k,B, r_0}.
\end{equation}
Applying now \eqref{E:DE^kB}  and \eqref{E:D^0H}--\eqref{E:D^2H}, we get
\begin{multline}
\bnorm{DE(H)(\dot{H})(B,\varphi)}^{k,B,r_0}\le\\
\le \ex^{-H(B,\varphi)}  (r_0 +1)
\bigl(1+ 2 \norm{H}_{k,0} (1+\log W^B(\varphi)) \bigr)^{r_0}
\, 5 \, \norm{\dot{H}}_{k,0} (1+\log W^B(\varphi))
\end{multline}
yielding
\begin{alignat}1
& \tnorm{D E(H)(\dot{H})}_{k} %\le
\\
\nonumber 
\le &  \sup_{\varphi} \ex^{-H(B,\varphi)}  (r_0+1) W^{-B}(\varphi)
\ex^{2 r_0 \norm{H}_{k,0} (1+\log W^B(\varphi))} \,  10 \,   \norm{\dot{H}}_{k,0} 
\ex^{\frac12 \log W^B(\varphi)}   %\le
\\
\le  &  10 \,  \ex (r_0+1) \norm{\dot{H}}_{k,0} 
\nonumber 
\end{alignat}
if $4 r_0 \norm{H}_{k,0} \le 1$.
Similarly, we get the bounds for higher derivatives. Formally the estimate \eqref{immersion2} follows from \eqref{E:DImm} and the identity
$$
E(H)-1=\int_0^1DE(tH)(H)\,\d t.
$$
\qed 
\end{proofsect}

\section{The map $P_2$}

\begin{lemma}\label{L:S3smooth}  
Consider the map $P_2\colon  \bM_{\ttt}\times \bM_{r} \to \bM_r$ defined in \eqref{E:defP2},
restricted to  $B_{\rho_1}(1)\times B_{\rho_2}\subset \bM_{\ttt}\times \bM_{r}$ with the balls $B_{\rho_1}(1)=\{I\colon \tnorm{I-1}_k<\rho_1\} $ and $ B_{\rho_2}=\{K\colon \norm{K}^{(\mathsf{A})}_{k,r}<\rho_2\}$
and   the target space  $\bM_r$ equipped with the norm   $\norm{\cdot}^{(\mathsf{A}/2)}_{k,r}$.
For any $ \mathsf{A} \ge 2 $ and $\rho_1,\rho_2$ such that
\begin{equation}
\label{E:rhoassum}
\rho_1< (2\mathsf{A})^{-1},\text{ and } \rho_2< (2\mathsf{A}^{2^d})^{-1},
\end{equation}
the map $P_2$ restricted to $B_{\rho_1}\times B_{\rho_2}$ is smooth and satisfies the bound
\begin{equation} \label{E:S3smooth}
\begin{aligned}
\frac{1}{j_1!j_2!}&\bigl\Vert\big(D_1^{j_1} D_2^{j_2} P_2)(I,K)(\dot{I},\dots,\dot{I},\dot{K},\ldots,\dot{K})\bigr\Vert^{(\mathsf{A}/2)}_{k,r}
\\& \quad \le (2\mathsf{A})^{j_1}  \bigl(2\mathsf{A}^{2^d}\bigr)^{j_2}\tnorm{\dot{I}}_k^{j_1} \bigl(\norm{\dot{K}}^{(\mathsf{A})}_{k,r}\bigr)^{j_2}
\end{aligned}
\end{equation}
for any $j_1,j_2\in\N$. In particular,
\begin{equation}\label{E:S3bis}
\norm{P_2(I,K)}_{k,r}^{(\mathsf{A}/2)}\le 2\mathsf{A}\tnorm{I-1}_k+2\mathsf{A}^{2^d}\norm{K}_{k,r}^{(\mathsf{A})}.
\end{equation}
\end{lemma}

\begin{proofsect}{Proof}
Recall that
\begin{equation}
\bigl((I-1)\circ K\bigr)(X)=\sum_{Y\in\Pcal(X)}(I-1)^{X\setminus Y}K( Y), \  X\in\Pcal^{\rm c},
\end{equation}
with $ (I-1)^{X\setminus Y}=\prod_{B\in\Bcal(X\setminus Y)}\bigl(I(B)-1\bigr) $ and $K(Y)=\prod_{Z\in \Ccal(Y) }K(Z)$, where $\Ccal(Y)$ denotes the set of components of $Y\in\Pcal$.

Hence, 
\begin{multline}
\frac{1}{j_1!j_2!}\big(D_1^{j_1} D_2^{j_2}\bigl((I-1)\circ K\bigr)(X)(\dot{I},\dots,\dot{I}, \dot{K},\dots,\dot{K})=\\=\sum_{\substack{Y\in\Pcal(X),  Y_1\in\Pcal(X\setminus Y),\abs{Y_1}=j_1\\  \Jcal\subset \Ccal( Y), \abs{\Jcal}=j_2 }}(I-1)^{(X\setminus Y)\setminus Y_1}\dot{I}^{Y_1}\prod_{Z\in \Ccal(Y)\setminus \Jcal }K(Z)\prod_{Z\in  \Jcal }\dot{K}(Z).
\end{multline}
Further, recall that, by definition of the norm $\norm{K}^{(\mathsf{A})}_{k,r} $,  we have
$$ 
 \norm{K(Z)}_{k,Z,r}\le \G_A(Z)^{-1}\norm{K}^{(\mathsf{A})}_{k,r} \mbox{ for any } Z\in\Pcal_k^{\com}.
 $$
 Notice also that 
\begin{equation}
\label{E:AGA}
\mathsf{A}^{\abs{Z}- 2^d}\le\max(1, \mathsf{A}^{\abs{Z}-2^d})\le \G_{\mathsf{A}}(Z)\le \mathsf{A}^{\abs{Z}}
\end{equation}
for any $\mathsf{A}\ge 1$ and any $Z\in\Pcal^{\com}$.
Using the bounds (iia)  and (i) from Lemma~\ref{L:prop}, assumptions \eqref{E:rhoassum}, as well as the  lower bound on ${\G}_{\mathsf{A}}(Z)$ above and the fact that the number of terms in the sum is bounded by $2^{|X|}$,
we get
\begin{multline}\label{estP_2}
\norm{P_2(I,K)(X)}_{k,X,r}\le\\
\le \sum_{Y\in\Pcal(X)}\tnorm{I-1}_k^{|X\setminus Y|}\big(\norm{K}_{k,r}^{\ssup{\mathsf{A}}}\big)^{|\Ccal(Y)|}\mathsf{A}^{2^d|\Ccal(Y)|}\mathsf{A}^{-|Y|}
\le \mathsf{A}^{-|X|}2^{|X|}=  \bigl(\tfrac{\mathsf{A}}{2}\bigr)^{-|X|},
\end{multline}
cf. \cite[Lemma~6.3]{B07}. 
Similarly, using that $\binom{n}{j}\le 2^n$, we get the claim
\begin{equation}
\begin{aligned}
\frac{1}{j_1!j_2!}&\norm{(D_1^{j_1} D_2^{j_2} P_2)(I,K)(X)(\dot{I},\dots,\dot{I},\dot{K},\ldots,\dot{K})}_{k,X,r}\\
& \le \sum_{Y\in\Pcal(X)}\tbinom{\abs{X\setminus Y}}{j_1}\tnorm{I-1}_k^{\abs{X\setminus Y}-j_1}\tnorm{\dot{I}}_k^{j_1}  \tbinom{\abs{\Ccal(Y)}}{j_2} 
\\ & \quad \times \bigl(\norm{K}^{(\mathsf{A})}_{k,r}\bigr)^{\abs{\Ccal(Y)}-j_2}\bigl(\norm{\dot{K}}^{(\mathsf{A})}_{k,r}\bigr)^{j_2}
\mathsf{A}^{2^d\Ccal(Y)}\mathsf{A}^{-\abs{Y}}\\
& \le \sum_{Y\in\Pcal(X)}
2^{\abs{X\setminus Y}}(2\mathsf{A})^{-(\abs{X\setminus Y}-j_1)}\tnorm{\dot{I}}_k^{j_1}  
2^{\abs{\Ccal(Y)}}(2\mathsf{A}^{2^d})^{-(\abs{\Ccal(Y)}-j_2)}\bigl(\norm{\dot{K}}^{(\mathsf{A})}_{k,r}\bigr)^{j_2}
\\ & \quad \times \mathsf{A}^{2^d\Ccal(Y)}\mathsf{A}^{-\abs{Y}}=\\
&=\sum_{Y\in\Pcal(X)}
2^{j_1}\mathsf{A}^{-(\abs{X\setminus Y}-j_1)}\tnorm{\dot{I}}_k^{j_1}  
2^{j_2}\mathsf{A}^{j_2 2^d}\bigl(\norm{\dot{K}}^{(\mathsf{A})}_{k,r}\bigr)^{j_2}
A^{-\abs{Y}} \\
& \le 2^{\abs{X}} (2\mathsf{A})^{j_1}\tnorm{\dot{I}}_k^{j_1} (2\mathsf{A}^{2^d})^{j_2}  \bigl(\norm{\dot{K}}^{(\mathsf{A})}_{k,r}\bigr)^{j_2}
\mathsf{A}^{-\abs{X}} .
\end{aligned}
\end{equation}
Finally, \eqref{E:S3bis} follows from the fact that $ P_2(1,0)=0 $ and 
\begin{equation}
\begin{aligned}
\frac{\d}{\d t}P_2(1+t(I-1),tK)&=D_1P_2(1+t(I-1),tK)(I-1)\\ &\quad +D_2P_2(1+t(I-1),tK)K.
\end{aligned}
\end{equation}
\qed 
\end{proofsect}

\section{The map $P_3$}
  \begin{lemma} \label{L:P_3}  Let $\mathsf{A} \ge 1$, $\mathsf{B} \ge 1$.
 Consider the map $P_3\colon  \bM_{r} \to  \widehat \bM_{r}   $ 
 defined by
 \begin{equation}
 (P_3 K)(X) = \prod_{Y \in \Ccal(X)} K(Y).
 \end{equation}
restricted to  $B_{\rho} = \{ K \in  \bM_{r} :  \norm{K}^{(\mathsf{A})}_{k,r}<\rho\}  $ 
%with the balls $B_{\rho_1}(1)=\{I\colon \tnorm{I-1}_k<\rho_1\} $ and $ B_{\rho_2}=\{K\colon \norm{K}^{(\mathsf{A})}_{k,r}<\rho_2\}$
and   the target space  $\widehat \bM_{r}$ equipped with the norm   $\norm{\cdot}^{(\mathsf{A},\mathsf{B})}_{k,r}$.
For any 
\begin{equation}
\rho \le (2\mathsf{B})^{-1}
\label{E:rhoassum_P3}
\end{equation}
the map $P_3$ restricted to $B_\rho$  is smooth and satisfies the bound
\begin{equation} \label{E:P3smooth}
\frac{1}{j!} \bigl\Vert\big(D_1^{j}  P_3)(K)(\dot{K},\ldots,\dot{K})\bigr\Vert^{(\mathsf{A},\mathsf{B})}_{k,r}
\le   \bigl( 2 \mathsf{B}\norm{\dot{K}}^{(\mathsf{A})}_{k,r}\bigr)^{j}
\end{equation}
for any $j_1,j_2\in\N$.
%, j_1+j_2\le r_0$. 
 \end{lemma}
 
 \begin{proofsect}{Proof} The proof is similar to, but simpler  than, the proof of Lemma \ref{L:S3smooth}.
 We have
 \begin{equation}
 \frac{1}{j!} D^j P_3(K)(X) (\dot{K}, \ldots, \dot{K}) = \sum_{\Jcal \subset \Ccal(X), |\Jcal| = j}
 \prod_{Z \in \Ccal(X) \setminus \Jcal} K(Z)  \prod_{Z \in \Jcal} \dot{K}(Z).
 \end{equation}
 Thus using the estimate $\tbinom{|\Ccal(X)}{j} \le 2^{|\Ccal(X)}$ and the identity
 $\Gamma_{\mathsf{A}}(X) = \prod_{Z \in \Ccal(X)} \Gamma_{\mathsf{A}}(Z)$  and arguing as in the proof of Lemma \ref{L:S3smooth}
 we get
 \begin{equation}  
 \begin{aligned}
\mathsf{B}^{|\Ccal(X)|}\Gamma_{\mathsf{A}}(X)  \frac{1}{j!} \norm{ D^j P_3(K)(X) (\dot{K}, \ldots, \dot{K}) }_{k,X,r}  \le &
 (2\mathsf{B})^{|\Ccal(X)|} \left(  \|K\|_{k,r}^{(\mathsf{A})}\right)^{|\Ccal(X)|-j } \times \\ &\times \left(  \|\dot{K}\|_{k,r}^{(\mathsf{A})}\right)^{j}.
\end{aligned}
 \end{equation}
Since $2 \mathsf{B}  \|K\|_{k,r}^{(\mathsf{A})} \le 2 \mathsf{B} \rho \le 1$ it follows that 
\begin{equation}
\frac{1}{j!} \norm{ D^j P_3(K)(X) (\dot{K}, \ldots, \dot{K}) }_{k,r}^{(\mathsf{A},\mathsf{B})}  \le  \left( 2\mathsf{B}  \|\dot{K}\|_{k,r}^{(\mathsf{A})}\right)^{j}
\end{equation}
and this finishes the proof.
\qed
 \end{proofsect}

\section{The map $R_1$}

\begin{lemma}\label{L:R_1} 
Let $m\in\N$, $ 2m\le r_0$, and for any $ n=0,1,\dots,m$, let
 $\bX_n $  denote the space $\widehat  \bM_{r_0 - 2m + 2n} $ equipped with the 
 norm $ \norm{\cdot}_{\bX_n}= \norm{\cdot}_{k,r_0-2m+2n}^{(\mathsf{A},\mathsf{B})} $ 
 and  $ \bY_{n} $  the space $ \widehat  M_{:,r_0 - 2m + 2n} $ equipped with the 
 norm $ \norm{\cdot}_{ \bY_{ n   }} = \norm{\cdot}_{k:k+1,r_0-2m+2n}^{(\mathsf{A}/2,\mathsf{B}/ 2^{2^d})} $. 
 Further, let $ B_{\frac{1}{2}}=\{\bq\in\R_{\rm sym}^{d\times d}\colon\norm{\bq}<\frac{1}{2}\} $.
Consider the map $R_1\colon  \bX\times B_{\frac{1}{2}}\to\bY$ defined in \eqref{E:defR1} 
with $\bX=\bX_{\!m} = \widehat \bM_{r_0}$ and $\bY=\bY_{\!\!m} = \widehat \bM_{:,r_0}$.
There exists a constant $C = C(r_0, d)$ such that for any 
$h\ge L^{\upkappa(d)} h_1$ with $h_1=h_1(d, \omega)$  and $\upkappa(d)$ as in Lemma~\ref{L:prop} (iv)
(see \eqref{E:defh_1=}),   $\mathsf{A}\ge 2$,     and any  $r=1,\ldots, r_0$, 
we have 
\begin{equation}
\label{E:S1inC}
R_1\in \widetilde C^m(\bX\times B_{\frac{1}{2}},\bY).
\end{equation}
Moreover the constants in the estimates of the relevant derivatives are independent of $k$ and $N$. More precisely for $ 0\le \ell\le m, 0\le n\le m-\ell $, there are $ C(n,d)> 0 $ such that
\begin{align}
\norm{D_2^\ell R_1(P,\bq,\dot{\bq}^\ell)}_{\bY_n}&\le C(n,d)\norm{P}_{\bX_{n+\ell}}\norm{\dot{\bq}}^\ell,\label{L6.3-A}\\\label{L6.3-B}
\norm{D_1D_2^\ell R_1(P,\bq,\dot{P},\dot{\bq}^\ell)}_{\bY_n} &\le C(n,d)\norm{\dot{P}}_{\bX_{n+\ell}}\norm{\dot{\bq}}^\ell,\\\label{L6.3-C}
D_1^2D_2^\ell R_1(P,\bq,\dot{P}^2,\dot{\bq}^\ell)&=0.
\end{align}
\end{lemma}

\begin{remark}  \label{R:R_1} 
(i)  Note that \eqref{L6.3-B} follows from \eqref{L6.3-A} since $R_1 $ is linear in the first argument, whereas \eqref{L6.3-C} is trivial.

%\noindent (ii)  The restriction $ m\le 3 $ is only included since three derivatives are sufficient for our purposes and since we wanted to avoid the introduction of another parameter $ m_0 $ for the maximum number of derivatives. The proof works also for an arbitrary number of derivatives (see Step 7 in the proof below). 

\noindent (ii) The proof below actually shows that
\begin{equation} \label{E:6_R1_X}
\norm{D_2^\ell R_1(P,\bq,\dot{\bq}^\ell)(X)}_{k:k+1,X,n} \le C(n,d) 2^{|X|} \norm{P(X)}_{k,X,{n+\ell}}\norm{\dot{\bq}}^\ell
\end{equation}
The estimate \eqref{L6.3-A} then follows by the choice of weights $\mathsf{A}/2$ and $\mathsf{B}/2^{2^d}$ on the target space, see Step 2
of the proof.

\noindent (iv) It follows from Step 1 in the proof, the bound 
$$
\norm{\bR^{\ssup{\bq}}P(X)}_{k:k+1,X,r}\le 2^{|X|_k}\norm{P(X)}_{k,X,r}$$ 
in Step 2 of the proof and the linearity 
of $R_1$ in the first argument that $R_1$ is actually a real-analytic map from $\bX_r \times B_{\frac12}$ to $\bY_r$
without any loss of regularity. The bounds on the corresponding derivatives depend, however, on the system size $N$
and the level $k$, while the bounds stated in Lemma \ref{L:R_1} do not.
  \hfill $\diamond$
\end{remark}

\begin{proofsect}{Proof}
Recall from \eqref{E:defR1} that
$$
\begin{aligned}
R_1(P,\bq)(X,\varphi)&=(\bR^{\ssup{\bq}}P)(X,\varphi)=\int_{\cbX}P(X,\varphi+\xi)\,\mu_{k+1}^{\ssup{\bq}}(\d\xi).
\end{aligned}
$$
The fact that $ R_1 $ maps $ \bM_m\times B_{\frac{1}{2}} $ to $ \bY_{\!\!m}$ follows from Lemma~\ref{L:prop}(iv).
Note that $R_1 $ is linear in $P$. Thus by Lemma~\ref{L:linear} it suffices to show that
\begin{enumerate}
\item[(i)] For each $P\in\bX_m $ and $ 0\le \ell\le m $ the map $ \bq\mapsto \bR^{\ssup{\bq}}P $ is in $ C_*^\ell(B_{\frac{1}{2}};\bY_{m-\ell}) $.
\item[(ii)] For each $ \bq_0 \in B_{\frac{1}{2}}$ there exist $ \delta, C>0 $ such that
$$
\norm{D^\ell_{\bq}\bR^{\ssup{\bq}}(P,\bq),\dot{\bq}^\ell)}_{\bY_n}\le C\norm{P}_{\bX_{n+\ell}}\norm{\dot{\bq}}^\ell
$$ for any  $ 0\le \ell\le m, 0\le n\le m-\ell $, and for all $ (P,\bq,\dot{\bq})\in\bX_m\times B_\delta(\bq_0)\times\R_{\rm sym}^{d\times d} $.
\end{enumerate}
We split the proof of (i) and (ii) into seven steps below. Note that the required constant $ C$ will be given as the maximum of all constants in \eqref{L6.3-A} and \eqref{L6.3-B}. We first show (i) in step 1 below. Indeed we even show that $ \bq\mapsto \bR^{\ssup{\bq}} P $ is real-analytic with values in $ \bY_m\subset \bY_{m-\ell} $.

\bigskip

\noindent \textbf{Step 1:}  Assume that $ P\in \bM_r=(M_r(\Pcal^{\rm c}_k,\cbX),\norm{\cdot}^{\ssup{\mathsf{A}}}_{k,r}) \,$ for some $ r\in\{r_0,\ldots,r_0-2m\} $. Then the map
$$
\bq\mapsto R_1(P,\bq) 
$$ is real-analytic from $ B_{\frac{1}{2}} $ to $ \bM_{:,r}=(M_r(\Pcal^{\rm c}_k,\cbX),\norm{\cdot}^{\ssup{\mathsf{A}/2}}_{k:k+1,r}) $. First it suffices to show the result for $ r=0 $, since differentiation with respect to $ \varphi $ commutes with $ \bR^{\ssup{\bq}} $. Secondly it suffices to consider a fixed polymer $X$, since there are only finitely many polymers. Thus we need to show the following: If
$$
\norm{P(X)}_{k,X,0}=\sup_{\xi}\,\frac{|P(X,\xi)|}{w_k^X(\xi)}<\infty,
$$ then the map
$$
B_{\frac{1}{2}}\ni \bq \mapsto\int_{\cbX}\, P(X,\cdot+\xi)\,\mu_{\Cscr_{k+1}^{\ssup{\bq}}}(\d\xi)
$$ is real-analytic with values in the space of continuous functions $ F $ of the field with the weighted norm
$$
\norm{F}_{k:k+1,X,0}=\sup_{\varphi}\,\frac{|F(\varphi)|}{w^X_{k:k+1}(\varphi)}.
$$
This follows from Gaussian calculus (see Lemma~\ref{Clemma1}), Lemma~\ref{L:spectrum} and the properties of the finite range decomposition, see Proposition~\ref{P:FRD}. To see this recall \eqref{bkbound1}, i.e.,
$$
w^X_k(\varphi+\xi)\le w^X_{k:k+1}(\varphi)\ex^{\frac{1}{2}\varkappa(\Bscr_k\xi,\xi)},
$$
where $ \varkappa=2\overline{C}h^{-2} $ and $ \Bscr_k $ is given by \eqref{bkbound2}. If $ h_1 $ and $ \upkappa(d) $ are chosen as in Lemma~\ref{L:prop} and $ h\ge L^{\upkappa(d)}h_1 $ then it follows from Lemma~\ref{L:spectrum} that for $ \bq\in B_{\frac{1}{2}} $ and $ \Cscr_{k+1}=\Cscr_{k+1}^{\ssup{\bq}} $ we have
\begin{equation}\label{bkest}
0\le \Cscr_{k+1}^{1/2}\varkappa\Bscr_k\Cscr_{k+1}^{1/2}\le \frac{1}{2}\Id \mbox{ and hence } \Cscr_{k+1}^{-1}>\varkappa\Bscr_k,
\end{equation}
i.e.,
$$
B_{\frac{1}{2}}\ni\bq\mapsto \Ucal_k,
$$
where we define
$$
\Ucal_k:=\{\Cscr\in\sym^{\ssup{+}}(\cbX)\colon \Cscr^{-1}>\varkappa\Bscr_k\}.
$$
By Lemma~\ref{Clemma1} the map
$$
\Cscr\mapsto \int_{\cbX}\,P(\cdot+\xi)\,\mu_{\Cscr}(\d\xi)
$$ is real-analytic from $ \Ucal_k $ to the desired space. Finally,  by Proposition~\ref{P:FRD} 
 and \eqref{bkest} the 
map $ \bq\mapsto\Cscr_{k+1}^{\ssup{\bq}} $ is real-analytic from $ B_{\frac{1}{2}} $ to $ \Ucal_k $.

   Hence $ \bq\mapsto R_1(P,\bq) $ is real-analytic from $ B_{\frac{1}{2}} $ to the space $ \bM_{:,r} $, and thus (i) is proven.

\bigskip

In the remaining steps we are going to prove (ii). In step 2 we show the bounds for $ \ell=0 $ followed by the bound for $ \ell =1 $ in step 3 to step 6. The bounds for higher derivatives are then finally settled in step 7.

\medskip

\noindent\textbf{Step 2: Bounds on $\boldsymbol{R^{\ssup{\bq}}}$. } 
By Lemma~\ref{L:prop}(iv) we have for all $ \bq\in B_{\frac{1}{2}} $ the following estimate
$$
\norm{\bR^{\ssup{\bq}}P(X)}_{k:k+1,X,r}\le 2^{|X|_k}\norm{P(X)}_{k,X,r}.
$$ 

For connected polymers $Y$ we have
\begin{equation}
2^{|Y|} \Gamma_{\mathsf{A}/2}(Y) \le 2^{2^d} \Gamma_{\mathsf{A}}(Y). 
\end{equation}
Thus for general polymer $X$ we get
\begin{equation}
2^{|X|} \Gamma_{\mathsf{A}/2}(X) \le  2^{2^d |\Ccal(X)|}  \Gamma_{\mathsf{A}}(X).
\end{equation}
and thus
\begin{equation}
2^{|X|} \left(\mathsf{B}/ 2^{2^d}\right)^{|\Ccal(X)|}   \Gamma_{\mathsf{A}/2}(X) \le  \mathsf{B}^{|\Ccal(X)|}  \Gamma_{\mathsf{A}}(X).
\end{equation}
Therefore
$$
\norm{\bR^{\ssup{\bq}}P}_{k:k+1,r}^{\ssup{\mathsf{A}/2, \mathsf{B}/2^{2^d}}}\le \norm{P}_{k,r}^{\ssup{\mathsf{A},\mathsf{B}}},
$$
and hence with $ r=r_0-2m+2n $ we obtain
\begin{equation}\label{L:6.3est}
\norm{R_1(P,\bq)}_{\bY_n}=\norm{\bR^{\ssup{\bq}}P}_{\bY_n}\le \norm{P}_{\bX_n},\quad \mbox{ for all }\bq\in B_{\frac{1}{2}}.
\end{equation}

\bigskip

\noindent\textbf{Step 3: Bounds for $ \boldsymbol{D_2 R_1(P,\bq,\dot{\bq})}$.} Let $ \bq\in B_{\frac{1}{2}} $ and $ \norm{\dot{\bq}}=1 $ and write $ \bgamma(t)=\bq+t\dot{\bq} $ in the following. By Lemma~\ref{L:der_cov1} and \eqref{1stder} we have 
$$
\begin{aligned}
D_2R_1&(P,\bq,\dot{\bq})(X,\varphi)=\frac{\d}{\d t}\Big|_{t=0}\int_{\cbX}\,P(X,\varphi+\xi)\,\mu_{\Cscr_{k+1}^{\ssup{\bgamma(t)}}}(\d\xi)\\&=\int_{\cbX}\,A_{\dot{\Cscr}_{k+1}}P(X,\varphi+\xi)\,\mu_{\Cscr_{k+1}^{\ssup{\bq}}}(\d\xi)=(\bR^{\ssup{\bq}} A_{\dot{\Cscr}_{k+1}}P)(X,\varphi)
\end{aligned}
$$ 
with
$$ 
\dot{\Cscr}_{k+1} =\frac{\d}{\d t}\Big|_{t=0}\Cscr_{k+1}^{\ssup{\bgamma(t)}}
$$
and where the functional $ A_{\dot{\Cscr}_{k+1}} $ is defined as
$$
 A_{\dot{\Cscr}_{k+1}}P(X,\xi)=\sum_{i,j=1}^{L^{dN}-1} D^2P(X,\xi,e_i,e_j)(\dot{\Ccal}_{k+1})_{i,j},
$$ 
where $ \{e_j\}_{j=1}^{L^{dN}-1} $ is any orthonormal basis of $ \cbX $ and $ (\dot{\Ccal}_{k+1})_{i,j}=(\dot{\Cscr}_{k+1} e_i,e_j) $. By Step 2 we obtain the following bound for the derivative with respect to $ \bq $, for $0 \le n \le m-1$,
\begin{equation}
\norm{D_2R_1(P,\bq,\dot{\bq})}_{\bY_{n}}\le %2^{2^d}
\norm{A_{\dot{\Cscr}_{k+1}}P}_{\bX_{n}}.
\end{equation}

\bigskip

\noindent \textbf{Step 4: Estimate for $ \boldsymbol{\norm{A_{\dot{\Cscr}_{k+1}}P}}$. } We now express and estimate the functional $ A_{\dot{\Cscr}_{k+1}}P $ using 
the orthonormal Fourier basis
$ \{f_p\}_{p\in\widehat{\T}_N} $  of the (complexified space) $\cbX$ given by 
\begin{equation}
f_p(x)=\frac{{\rm e}^{i\langle p,x\rangle}}{L^{dN/2}},\quad p\in\widehat{\T}_N, x\in\L_N.
\end{equation}
%We denote by $\hat \Ccal_{k+1}^{\ssup{\bq}}$ the Fourier multiplier of $\Cscr_{k+1}^{\ssup{\bq}}$.
We denote by 
 $ \dot{\widehat{\Ccal_{k+1}}}(p) $ the Fourier 
 multiplier of $ \dot{\Cscr}_{k+1} $. Now $ \Cscr_{k+1}^{\ssup{\bq}} $ and hence $ \dot{\Cscr}_{k+1} $ are diagonal in the Fourier basis and 
$$
\dot{\Cscr}_{k+1} f_p=\dot{\widehat{\Ccal_{k+1}}}(p)f_p \quad \mbox{ with } \dot{\widehat{\Ccal_{k+1}}}(p)\in\R. 
$$
Thus by \eqref{hess}
$$ 
\begin{aligned}
A_{\dot{\Cscr}_{k+1}}P(X,\xi)&= \sum_{p\in\widehat{\T}_N}D^2P(X,\xi,\dot{\Ccal}_{k+1}f_p,\overline{f}_p)\\
&\sum_{p\in\widehat{\T}_N}D^2P(X,\xi,f_p,\overline{f}_p)\dot{\widehat{\Ccal_{k+1}}}(p).
\end{aligned}
$$

\noindent  We claim that
 \begin{equation}
| A_{\dot{\Cscr}_{k+1}} D^2P(X, \xi) \dot{\Ccal}_{k+1})| ^{k, X, r-2}  
\le 
r (r-1)  |P(X, \xi)|^{k,X,r}  
 \sum_{p\in\widehat{\T}_N\setminus\{0\}}  |f_p|_{k,X} ^2 | \dot{\widehat{\Ccal_{k+1}}}|(p) \, .
\end{equation}
whenever $ \dot{\Cscr}_{k+1} $ is diagonal in the Fourier basis. 
In particular we now show that there exists a $ C(n,d)>0 $ such that   for $0 \le n \le m-1$ the following estimate holds,
\begin{equation}\label{estbis}
\norm{A_{\dot{\Cscr}_{k+1}}P}_{\bX_{n}}\le C(n)\norm{P}_{\bX_{n+1}}\sum_{p\in\widehat{\T}_N}|f_p|^2\dot{\widehat{\Ccal_{k+1}}}(p).
\end{equation}

Indeed, using the fact that $\dot{\widehat{\Ccal_{k+1}}}(p) $ is real and the definition of the trace we have 
\begin{equation}\label{E:tracecomp}
\begin{aligned}
G(X,\xi)
:=& \tr\big(D^2 P(X,\xi)\dot{\Cscr}_{k+1})\big)=A_{\dot{\Cscr}_{k+1}}P(X,\xi)\\
= &\sum_{p\in\widehat{\T}_N\setminus\{0\}}\big(f_p,D^2 P(X,\xi)f_p\big)\dot{\widehat{\Ccal_{k+1}}}(p)\\
=&\sum_{p\in\widehat{\T}_N\setminus\{0\}}\big(\mbox{Re}\,(f_p),D^2 P(X,\xi) \mbox{Re}\,(f_p)\big)\dot{\widehat{\Ccal_{k+1}}}(p)\\
&\quad +\sum_{p\in\widehat{\T}_N\setminus\{0\}}\big(\mbox{Im}\,(f_p),D^2 P(X,\xi) \mbox{Im}\,(f_p)\big)\dot{\widehat{\Ccal_{k+1}}}(p).
\end{aligned}
\end{equation}
By a standard symmetrisation argument we have
\begin{equation}
|D^\sigma P(X,\xi)(\dot{\varphi}_1,\ldots,\dot{\varphi}_\sigma)|\le\frac{\sigma^\sigma}{\sigma!}|D^\sigma P(X,\xi)|^{k,X}
\prod_{i=1}^\sigma |\dot{\varphi}_i|_{k,X}.
\end{equation} 
Set
\begin{equation}\label{Mconst}
M :=   \sum_{p\in\widehat{\T}_N\setminus\{0\}}  |f_p|_{k,X} ^2 | \dot{\widehat{\Ccal_{k+1}}}|(p)\,.
\end{equation}
Then for all $\dot{\varphi}$ with $|\dot{\varphi}|_{k,X} \leq 1$
we have
\begin{equation}
\begin{aligned}
|D^s &G(X,\xi) (\dot{\varphi}, \ldots, \dot{\varphi})|  \le \\
&\le \sum_{p\in\widehat{\T}_N\setminus\{0\}} |D^{s+2} P(X,\xi) (\dot{\varphi}, \ldots, \dot{\varphi}; \mbox{Re}\,( f_p), \mbox{Re}\,(f_p))|  
| \dot{\widehat{\Ccal_{k+1}}}|(p)  \\
 &+ \sum_{p\in\widehat{\T}_N\setminus\{0\}}  |D^{s+2} P(X,\xi) (\dot{\varphi}, \ldots, \dot{\varphi}; \mbox{Im}\,( f_p), \mbox{Im}\,(f_p))|  
| \dot{\widehat{\Ccal_{k+1}}}|(p) \\
 & \le 2\frac{(s+2)^{s+2}}{(s+2)!}|D^{s+2}P(X,\xi)|^{k,X,r}  |\dot{\varphi}|^s_{k,X} M.
\end{aligned}
\end{equation}
Hence $|D^s G(X, \xi)|^{k,X} \leq C(r_0)M |D^{s+2}P(X, \xi, X)|^{k,X} $, 
for all $ s\le r_0-2 $ and $ C(r_0)=2\frac{r_0^{r_0}}{r_0!} $. This yields
\begin{multline}
|G(X,  \xi)|^{k,X,r-2} \le \\
\le M \sum_{s=0}^{r-2} \frac{1}{s!} |D^{s+2} P(X, \xi)|^{k,X}
\le r (r-1) C(r_0)M \sum_{s=0}^{r-2} \frac{1}{(s+2)!} |D^{s+2} P(X, \xi) |^{k,X} \le \\   
\le r (r-1) C(r_0)
M |P(X, \xi)|^{k,X,r} 
\end{multline}
and hence the assertion \eqref{estbis}. 
Note that  in the proof we only used the fact that $\Cscr^{(\bgamma(t))}_{k+1}$ is diagonal in the Fourier basis. Hence the
same computation yields the corresponding result for the higher derivatives
 \begin{multline}  
 \label{E:est_trace_S1_h}
\abs{\tr (D^2P(X, \xi) \frac{\d^j}{\d t^j} {\Ccal}_{k+1}^{\bgamma(t)})}^{k, X, r-2}  
\le \\
\le
r (r-1) C(r_0) \abs{P(X, \xi)}^{k,X,r}  
 \sum_{p\in\widehat{\T}_N\setminus\{0\}}  \abs{f_p}_{k,X} ^2 \abs*{\frac{\d^j}{\d t^j}\widehat{\Ccal_{k+1}^{(\bgamma(t))}}(p)} \, . 
\end{multline}

\bigskip

\noindent\textbf{Step 5: Estimate for the term \eqref{Mconst} involving the Fourier multiplier. } Let
$$ 
\bgamma(t)=\bq+t\dot{\bq} \quad \mbox{ with } \bq\in B_{\frac{1}{2}} \mbox{ and } \norm{\dot{\bq}}=1.
$$
We claim that, with our choice of $h$, there exists $ C=C(n,d)>0 $ such that
\begin{equation}  \label{E:est_Fmult_S1}
 \sum_{p\in\widehat{\T}_N\setminus\{0\}}  |f_p|_{k,X} ^2 \abs*{\frac{\d^j}{\d t^j}\widehat{\Ccal_{k+1}^{(\bgamma(t))}}(p) }
 \le C  j!.
\end{equation}
%with $\tilde \eta(d) = \max(\frac14 (d+5)^2, d+12) + d +14$.
To see this note first that by the definition of the $| \cdot |_{k, X}$ norm
\begin{equation}
\abs{f_p}_{k,X} \le  \frac{1}{h}  \frac{1}{L^{Nd/2}} L^{kd/2 } \max(|p|, L^{k} |p|^2, L^{2k} |p|^3 )  \, .
\end{equation}
%By estimate (5.62) in  \cite{AKM09b} and Lemma 4.3 in  \cite{AKM09b} we have
The estimate \eqref{E:Fourier_estimate_FRD} in Remark \ref{R:FRD} can be rewritten as
\begin{equation}   \label{E:weighted_Fourier_estimate}
 \sum_{p\in\widehat{\T}_N\setminus\{0\}}  |p|^n \,  \abs*{\frac{\d^j}{\d t^j}\widehat{\Ccal_{k+1}^{(\bgamma(t))}}(p)}  
 \leq C 2^j j!  \,   L^{\upeta(n,d)  + n + d - 2}  L^{-k(n+d-2)} L^{dN},
\end{equation}
where $\upeta(n,d) = \max(\frac14 (d+n-1)^2, d+n+6) +10$. Applying this estimate with $n=2,4$ and $6$ and using the monotonicity of $\eta(n,d)$ in $n$, we need a  bound on $\upeta(6,d) + 4 +d$. It turns out that $\upeta(6,d) + 4 +d\le 2\upkappa(d)$ whenever $d\ge 2$. Indeed, this amounts to showing that $\upeta(6,d)+4\le \upeta(12,d)$ (with $2 \lfloor\frac{d+2}2\rfloor +8=12$ for $d=2$). Using this and  assuming that $h_1\ge 1$,  we can conclude that 
\begin{equation}
\label{E:hL<1}
h^{-2}L^{\upeta(n,d)  + n + d - 2}\le1
\end{equation} 
for $n=2,4,6$, implying  thus \eqref{E:est_Fmult_S1}.

\bigskip

\noindent\textbf{Step 6: Estimate for $\boldsymbol{D_2R_1(P,\bq,\dot{\bq})} $. } It follows from Step 3, \eqref{estbis} with $ \dot{\Cscr}_{k+1}=\frac{\d}{\d t}\big|_{t=0}\Cscr_{k+1}^{\ssup{\bgamma(t)}} $, and Step 5 with $j=1 $ for any $ 0\le n\le m-1 $ that there exists $ C(n,d)> 0 $ such that 
\begin{equation}
\norm{D_2R_1(P,\bq,\dot{\bq})}_{\bY_{n }}\le %2^{2^d}
\norm{A_{\dot{\Cscr}_{k+1}} P}_{\bX_{n}}\le C(n,d)\norm{P}_{\bX_{n+1}}.
\end{equation}

\bigskip

\noindent \textbf{Step 7: Bounds for the higher derivatives $ \boldsymbol{D_2^\ell R_1(P,\bq,\dot{\bq}^\ell)}$. } 
These bounds follow from Gaussian calculus in Lemma~\ref{Clemfinal}, the chain rule and the estimates for $ \frac{\d^j}{\d t^j}\Cscr_{k+1}^{\ssup{\bgamma(t)}} $ (see step 5). We consider first the case $ \ell=2 $. As in \eqref{Hdef} in appendix~\ref{app5} we set
$$
H(\Cscr)(\cdot)=\int_{\cbX}\,P(X,\cdot+\xi)\,\mu_{\Cscr}(\d\xi),
$$
respectively, 
$$
\widetilde{h}(t)(\cdot)=\int_{\cbX}\,P(X,\cdot+\xi)\,\mu_{\Cscr_{k+1}^{\bgamma(t)}}(\d\xi).
$$
By Lemma~\ref{Clemfinal} and \eqref{2ndder} we obtain
$$
\begin{aligned}
D_2^2& R_1(P,\bq,\dot{\bq},\dot{\bq})(X,\varphi)=\frac{\d^2}{\d t^2}\Big|_{t=0}R_1(P,\bgamma(t))(X,\varphi)=D^2H(\Cscr_{k+1},\dot{\Cscr}_{k+1},\dot{\Cscr}_{k+1})\\ &+DH(\Cscr,\ddot{\Cscr}_{k+1})
=R_1(A_{\dot{\Cscr}_{k+1}}^2 P,\bq)(X,\varphi)+R_1(A_{\ddot{\Cscr}_{k+1}}P,\bq)(X,\varphi)
\end{aligned}
$$ where we use that
$$
\dot{\Cscr}_{k+1}=\frac{\d}{\d t}\Big|_{t=0}\Cscr^{\ssup{\bgamma(t)}}_{k+1} \quad \mbox{ and } \ddot{\Cscr}_{k+1}=\frac{\d^2}{\d t^2}\Big|_{t=0}\Cscr^{\ssup{\bgamma(t)}}_{k+1}.
$$
By step 2 we have the estimate
$$
\norm{D_2^2R_1(P,\bq,\dot{\bq},\dot{\bq})}_{\bX_{n}}\le %2^{2^d}
\big(\norm{A^2_{\dot{\Cscr}_{k+1}} P}_{\bX_{n}}+\norm{A_{\ddot{\Cscr}_{k+1}} P}_{\bX_{n}}\big).
$$

Now step 4 and step 5 yield the following bound, for $0 \le n \le m-2$, 
$$
\norm{A_{\ddot{\Cscr}_{k+1}} P}_{\bX_{n}}\le C(n)\norm{P}_{\bX_{n+1}}\le C(n)\norm{P}_{\bX_{n+2}}.
$$
Applying now the steps 4 and 5 twice we get that
$$
\norm{A^2_{\dot{\Cscr}_{k+1}}P}_{\bX_{n}}\le C(n)\norm{A_{\dot{\Cscr}_{k+1}}P}_{\bX_{n+1}}
\le C(n)\norm{P}_{\bX_{n+2}},
$$
and thus the required estimate for the second derivative $ D_2^2 R_1 $. For general $ \ell\ge 2 $ it follows from Lemma~\ref{Clemfinal} and the chain rule that
$$ D^\ell_2R_1(P,\bq,\dot{\bq}^\ell) $$ is a linear combination of terms of the form 
$$
R_1(A_{\dot{\Cscr}_1}\cdots A_{\dot{\Cscr}_\kappa} P,\bq)
$$ where
$$
\dot{\Cscr}_i:=\frac{\d^{j_i}}{\d t^{j_i}}\Big|_{t=0}\Cscr_{k+1}^{\ssup{\bgamma(t)}} \quad \mbox{ with } \sum_{i=1}^\kappa j_i=\ell.
$$
Thus the desired estimate follows from step 2 and a $\kappa$-fold application of \eqref{estbis} and step 5.
\qed
\end{proofsect}

\section{The map $R_2$}

\begin{lemma}  \label{L:R_2}  Let $m\in \N$, $2m +2 \le r_0$. For $n =0, , \ldots,  m$ let $\bZ_n$ denote the space
$ \bM_{r_0 - 2m + 2n}$ equipped with the norm $ \norm{\cdot}_{ \bZ_n}= \norm{\cdot}_{k,r_0-2m+2n}^{(\mathsf{A})} $.
Let $\bX_n = \bM_0 \times \bZ_n$, 
$\bY_n = \bM_0$ (for all $n$) and 
$ B_{\frac{1}{2}}=\{\bq\in\R_{\rm sym}^{d\times d}\colon\norm{\bq}<\frac{1}{2}\}$. 
Consider the map $R_2\colon  \bM_0\times\bX  \times B_{\frac{1}{2}}\to \bY $, defined in \eqref{E:defR2} with $\bX = \bX_m = \bM_0 \times \bM_{r_0}$ and  $\bY = \bY_m = \bM_0$. There exists a constant $C = C(d)$ such that for any 
$h\ge L^{\upkappa(d)} h_1$ with $h_1=h_1(d, \omega)$ and $\upkappa(d)$  as in Lemma~\ref{L:prop} (iv),
  $\mathsf{A}\ge 1$, we have
  \begin{equation}
\label{E:S2inC}
R_2\in \widetilde C^m(\bX\times B_{\frac{1}{2}},\bY).
\end{equation}

Moreover for   any $\bq$ and $\dot{\bq}$ with $|\bq| < \frac12$ and $|\dot{\bq}| \le 1$, and any $\ell \leq m$, 
we have
\begin{equation} \label{E:S2_linear}
\begin{aligned}
\norm{D^j_1 D^n_2 D^{\ell}_3 R_2(H,K,\bq)&(\dot{H}, \dot{K}, \dot{\bq}, \ldots, \dot{\bq})}_{k,0}\\
&\quad\leq 
C\left\{
   \begin{array}{ll}    \norm{H}_{k,0} +\norm{K}_{\bZ_\ell} & \mbox{if  } j=0, n=0,\\
    \norm{\dot{H}}_{k,0} & \mbox{if  } j=1, n=0,\\
     \norm{\dot{K}}_{\bZ_\ell} & \mbox{if  } j=0, n=1;
  \end{array}
  \right.
  \end{aligned}
\end{equation}
and 
\begin{equation} \label{E:S2_higher}
D^j_1 D^m_2 D^{\ell}_3 R_2(H,K,\bq) = 0 \quad  \mbox{if  } j+m \geq 2.
\end{equation}
\end{lemma}

\begin{remark} \label{R:R_2}
It follows from Remark \ref{R:R_1} and Lemma \ref{L:Taylor} below  that the map $R_2$ is actually  a real 
analytic map from $\bM_0 \times  \bM_2 $  to $\bM_0$.
\hfill $\diamond$
\end{remark}

First, we estimate the main component of $R_2$, namely the map $\Pi_2$.
\begin{lemma}\label{L:Taylor}\
Let $B \in \Bcal_k$,  $X \in \Scal_k$ with $X \supset B$, and let $K \in 
M(\Pcal_k, \cbX)$. Then 
\begin{equation} 
\norm{\Pi_2 K(X, \cdot)}_{k,0} \leq [2^d (d^{\frac32} + d)  + d^{\frac12}]  \,  \bnorm{K(X,0)}^{k,X,2} .
\end{equation}
\end{lemma}
 
Note that since $X \in \Scal_k$ we have $X \subset B^*$ and thus the
maps $\varphi \mapsto K(X,\varphi)$ can be viewed as an element
of $M^*(\Bcal, \cbX)$ on which the projection $\Pi_2$ was defined.

%%% old version 
%\begin{lemma}\label{L:Taylor}\
%The linear map $ F(B,\cdot)\mapsto \Pi_2F(B,\cdot) $ is bounded, i.e., 
%there exists  $ C>0 $ so that for any $F\in M^*(\Bcal, \cbX)$ one has
%$$
%\norm{\Pi_2 F(B)}_{k,0}\le C\norm{F(B)}_{k,B,r}.
%$$
%\end{lemma}
%
\begin{proofsect}{Proof}
%
%{\bf Comment SM}: go trough the proof and use right norms for $\nabla^s \varphi$
%($l_2$ based) consistently and track $d$ dependence).
%
Let $H = \Pi_2 K(X, \cdot)$. By definition we have
%\begin{equation}
$H(B, \dot{\varphi}) = L^{dk} \lambda + \ell(\dot{\varphi}) + Q(\dot{\varphi}, \dot{\varphi})$,
%\end{equation}
where
\begin{align}
\ell(\dot{\varphi}) &= \sum_{x \in B} \sum_{i=1}^d a_i \nabla_i \dot{\varphi} 
      +  \bc_{i,j}  \nabla_i \nabla_j \dot{\varphi}(x) \\
      Q(\dot{\varphi}, \dot{\varphi}) &= \frac12  \sum_{x \in B} \sum_{i,j =1}^d
         \bd_{i,j} \nabla_i \dot{\varphi}(x) \nabla_j \dot{\varphi}(x)
\end{align}
and
\begin{align}
L^{dk} \lambda &= K(X,0)   \label{E:S2_lambda} \\
\ell(\dot{\varphi})  &= DK(X,0)(\dot{\varphi})    \quad \forall \dot{\varphi} \mbox{   quadratic + affine in  } (B^*)^* \\
Q(\dot{\varphi},\dot{\varphi}) &= \frac12 D^2 K(X, 0)(\dot{\varphi},\dot{\varphi})   \quad
\forall  \dot{\varphi}  \mbox{   affine in  }  (B^*)^*
\end{align}

% Let $F\in   M^*(\Bcal, \cbX)$ and let   $H(B,\varphi)= \Pi_2F(B,\varphi) $.
%By definition
%\begin{equation}
%\bnorm{F(B,0)}^{k,B^*}+\bnorm{DF(B,0)}^{k,B^*}+\frac{1}{2}\bnorm{D^2F(B,0)}^{k,B^*}\le 
%\bnorm{F(B,0)}^{k,B^*,2}.
%\end{equation}
%As before, we write $ H(B,\dot{\varphi})=H(B,0)+l(\dot{\varphi})+Q(\dot{\varphi}) $. 
To estimate $\bd_{i,j}$ and $a_i$ we consider functions
$ \dot{\varphi} $ which are linear  on  $ ((B^*)^*)^* $  
%(actually on a slightly
%larger set so that the second and third derivaties still vanish on $(B^*)^*$ 
\begin{equation}
\dot{\varphi} =\sum_{i=1}^d \eta_i \pi_i,
\end{equation}
where $ \eta=(\eta_i)_{i=1,\dots,d}\in\R^d $, 
 \lsm[hg]{$\eta=(\eta_i)$}{$\in \R^d$}
and  $\pi_i$ is the co-ordinate projection $\pi_i(x)=x_i $ for $ x\in\Z^d $. 
 \lsm[pgia]{$\pi_i$}{the co-ordinate projection $\pi_i(x)=x_i $ for $ x\in\Z^d $}
Then for $x \in (B^*)^*$ we have  $ \nabla_i\dot{\varphi}(x)=\eta_i $ and $ \partial^{\alpha}\dot{\varphi}(x)=0 $  if $ |\alpha| =2 $
or $|\alpha| = 3$. Hence,
\begin{equation}
\begin{aligned}
L^{dk}\abs{ \frac12 \sum_{i,j=1}^d\bd_{i,j} \eta_i\eta_j}&=\abs{Q(\dot{\varphi})}=\Bigl|\frac{1}{2}D^2K(X,0)(\dot{\varphi},\dot{\varphi})\Bigr|\le \frac{1}{2}\bnorm{D^2K(X,0)}^{k,X}\bnorm{\dot{\varphi}}_{k,X}^2\\&= \frac{1}{2}\bnorm{D^2K(X,0)}^{k,X}h^{-2}  \sum_{i=1}^d \abs{\eta_i}^2L^{dk}.
\end{aligned}
\end{equation}
This yields
%\begin{equation}
$\max_{|\eta|_2 =1} \abs{ \frac12 \sum_{i,j=1}^d\bd_{i,j} \eta_i\eta_j} \le 
\frac{1}{2}  h^{-2} \bnorm{D^2K(X,0)}^{k,X}$
%\end{equation}
and thus
\begin{equation}  \label{E:S2_d}
\sum_{i,j =1}^d |\bd_{i,j}| \le  d \left( \sum_{i,j =1}^d |\bd_{i,j}|^2 \right)^{\frac12} 
\le d^{\frac32} \left( \lambda_{\rm max}(\bd^2)  \right)^{1/2} 
\le  \frac12 d^{\frac32}   h^{-2} \bnorm{D^2K(X,0)}^{k,X} . 
\end{equation}

Similarly, we have 
\begin{equation}
L^{dk}\sum_{i=1}^da_i\eta_i=\ell(\dot{\varphi})=DK(X,0)(\dot{\varphi})\le\bnorm{DK(X,0)}^{k,X}h^{-1}
\left(  \sum_{i=1}^d  |\eta_i|^2  \right)^{\frac12}   L^{\frac{dk}{2}}.
\end{equation}
The choice $\eta_i = a_i$ yields
\begin{equation}  \label{E:S2_a}
\sum_{i=1}^d|a_i|   \le d^{\frac12}   \left(  \sum_{i=1}^d |a_i|^2   \right)^{\frac12} 
\le d^{\frac12} h^{-1} L^{-\frac{dk}{2} }  \bnorm{DK(X,0)}^{k,X}.
\end{equation}

For the evaluation of the second derivative we use a test function which satisfies
 \lsm[hgiaja]{$\eta_{i,j}$}{coefficients of a quadratic test function $\dot{\varphi}(x)=\frac{1}{2}\sum_{i,j=1}^d\eta_{i,j}(x-\overline{x})_i(x-\overline{x})_j$}
\begin{equation}
\dot{\varphi}(x)=\frac{1}{2}\sum_{i,j=1}^d\eta_{i,j}(x-\overline{x})_i(x-\overline{x})_j\quad\forall x\in ((B^*)^*)^*,
\end{equation}
where $\overline{x}=\frac{1}{|B|}\sum_{x\in B} x $ and $ \eta_{i,j}=\eta_{j,i} $. Then,
for any $ x\in (B^*)^* $, 
\begin{equation}
\begin{aligned}
\nabla_j\dot{\varphi}(x)
        &=\sum_{i=1}^d\eta_{i,j}(x-\overline{x})_i,\\
\nabla_i\nabla_j\dot{\varphi}(x)&=\eta_{i,j}, \text{ and }\\
\nabla^\alpha\dot{\varphi}(x)&=0\quad\mbox{ for } |\alpha| = 3.
\end{aligned}
\end{equation}
Now  $ \abs{(x-\overline{x})_i}\le \frac{2^{d+1} -1}{2} L^k \le 2^d L^k $ for any $ x\in (B^*)^* $
and thus $|\nabla_j \dot{\varphi}(x)| \le  d^{\frac12} (\sum_{i=1}^d |\eta_{i,j}|^2)^{\frac12} 2^d L^k$ which yields
\begin{multline}
\bnorm{\dot{\varphi}}_{k,B}\le
\frac{1}{h}\big( 2^d   d^{\frac12}  L^k   L^{\frac{kd}{2}}  ( \sum_{i,j=1}^d|\eta_{i,j}|^2 )^{\frac12}
+ |L^{k(\frac{d}{2} + 1)}\big)    ( \sum_{i,j=1}^d|\eta_{i,j}|^2 )^{\frac12}\le\\
\le  (2^d d^{\frac12} + 1)  h^{-1}     L^{k(\frac{d}{2}+1)}    ( \sum_{i,j=1}^d|\eta_{i,j}|^2 )^{\frac12}      .
\end{multline}
Note that $\sum_{x\in B}\eta_{i,j}(x-\overline{x})_i a_i$ vanishes in view of the definition of $\overline{x}$.
Hence
\begin{multline}
\sum_{i,j=1}^dL^{dk}\eta_{i,j}\bc_{i,j} = \\=
\ell(\dot{\varphi})
 \le \bnorm{DK(X,0)}^{k,X}\bnorm{\dot{\varphi}}_{k,X}    
 \le  (2^d d^{\frac12} + 1)  h^{-1}     L^{k(\frac{d}{2}+1)}   
 ( \sum_{i,j=1}^d|\eta_{i,j}|^2 )^{\frac12}     \bnorm{DK(X,0)}^{k,X} 
\end{multline}
Taking $\eta_{i,j} = \bc_{i,j}$ we get
\begin{equation}  \label{E:S2_c}
\sum_{i,j=1}^d\abs{\bc_{i,j}}\le  d \left( \sum_{i,j=1}^d\abs{\bc_{i,j}}^2 \right)^{\frac12}
\le (2^d d^{\frac32} + d) h^{-1} L^{- (\frac{d}{2} -1) k}   \bnorm{DK(X,0)}^{k,X} .
\end{equation}
This yields the assertion with
\begin{equation}
C(d) = \max(1, d^{\frac12} + 2^d (d^{\frac32} + d), d^{\frac32})
=  d^{\frac12} + 2^d (d^{\frac32} + d).
\end{equation}
 \qed
\end{proofsect}

\begin{proofsect}{Proof of Lemma~\ref{L:R_2}}
We first note that  $R_2(H,K,\bq) = R_{2,a}^{(\bq)} H + R_{2,b}^{(\bq)} K$ where
$R_{2,a}^{(\bq)}$ and $R_{2,b}^{(\bq)}$ are linear maps. Thus \eqref{E:S2_higher} is obvious. To prove the remaining statements we can consider 
the maps $ H\mapsto R_{2,a}^{(\bq)} H $ and $ K\mapsto R_{2,b}^{(\bq)} K $ separately. We will establish the relevant estimates for the directional derivatives $ t\mapsto R_{2,a}^{\ssup{\bq+t\dot{\bq}}} $ and $ t\mapsto R_{2,b}^{\ssup{\bq+t\dot{\bq}}} $. The assertion on the existence and continuity of the total derivatives then follows as in the proof of Lemma~\ref{L:R_1}, using in particular the continuity of the map $ \bq\mapsto \bR^{\ssup{\bq}} $.  We first consider the map
\begin{equation}
R_{2,a}^{(\bq)} H :=  \Pi_2 \bR^{(\bq)} H
\end{equation}
which acts on ideal Hamiltonians. 
The integral of an odd functions against $\mu_{k+1}^{(\bq)}$ is zero
and 
\begin{equation}
\int_{\cbX} Q(\xi, \xi) \, \mu_{k+1}^{(\bq)} =  L^{dk}  \frac12 \sum_{i,j} \bd_{i,j} \nabla_i \nabla_j^* \Ccal^{(\bq)}_{k+1}(0)
\end{equation}
(cf. \eqref{E:A}). Thus 
$\bR^{(\bq)}  H$ is again an ideal Hamiltonian and 
the action of $R_2^{(a)}$ in the coordinates
$(\lambda, a, \bc, \bd)$ for $H$ is simply
\begin{equation}
(\lambda, a, \bc, \bd) \mapsto (\lambda + \sum_{i,j} \bd_{i,j} \nabla_i \nabla_j^* \Ccal^{(\bq)}_{k+1} (0) , a, \bc, \bd)
\end{equation}
By \eqref{E:fluctk} we have  $|\nabla_i \nabla_j^* \Ccal^{(\bq)}_{k+1} (0)| \le   C(d) L^{\eta(2,d)} L^{-dk}$ and thus
\begin{equation}
\norm{R_{2,a}^{(\bq)} H } \leq (1 + C(d) h^{-2} L^{\upeta(2,d)} ) \norm{H}_{k,0}
\le C(d) \norm{H}_{k,0},
\end{equation}
where we used the  lower bound on $h$ in the assumption of the lemma. 
The estimates for $D^\ell_q R_{2,a}^{(\bq)} H$ follow in the same way from 
\eqref{E:fluctk} since $h^2 \ge  L^{\upkappa(d)}\ge L^{\upeta(8,d)}$. 

Now let $X \in \Scal:k$ with $X \supset B$ and let $K \in M(\Pcal_k, \cbX)$.  We will estimate
\begin{equation}   \label{E:S2_single}
\Pi_2 \bR^{(\bq)} K(X, \cdot)
\end{equation}
and its derivatives with respect to $\bq$. 
The operator $R_{2,b}^{(\bq)}$ is obtained by taking a sum over
all such $X$ (for a fixed block $B$) with weight $\frac{1}{|X|_k}$.  Since there are at most
$(3^d - 1)^{2^d}$ such polymers $X$ is suffices to estimate 
\eqref{E:S2_single}.

By Lemma \ref{L:Taylor} and Lemma \ref{L:prop} (iv) we have
\begin{multline}
\frac{1}{C(d)} \norm{\Pi_2 \bR^{(\bq)} K(X, \cdot)}_{k,0}
\le
\bnorm{ \bR^{(\bq)} K(X, 0)}^{k,X,2}  \\
\le \int_{\cbX}  \bnorm{K(X, \xi)}^{k,X,2} \, \mu_{k+1}^{(\bq)}(d\xi)   
\le 2^{|X|_k} \norm{K(X)}_{k,X,2} \le 2^{2^d}  \norm{K(X)}_{k,X,2} \\ \le  2^{2^d}   \norm{K}_{k,2}^{(\mathsf{A})} .
\end{multline}

The derivatives with respect to $\bq$ are estimated using Gaussian calculus
and the estimates used in the proof of Lemma \ref{L:R_1}. Let $\norm{\bq} < \frac12$ and   $\norm{\dot{\bq}} = 1$,
and consider the curve $\bgamma(t) = \bq + t \dot{\bq}$ on a sufficiently small interval $(-a,a)$. 
Let 
\begin{equation}
G(X, \varphi) := \tr\big[D^2 K(X,\varphi)
\dot{\Cscr}^{\ssup{\bq}}_{k+1} \big] .
\end{equation}
Then (see Appendic~\ref{app5})
\begin{equation}
\frac{\d}{\d t}\Bigr|_{t=0}   (\bR^{(\bgamma(t))}  K)(X, \varphi)   = (\bR^{(\bq)}  G)(X, \varphi)
%\tfrac{1}{2}\int_{\cbX}\tr\big[D^2 K(X,\varphi+\xi)
%\dot{\Cscr}^{\ssup{\bq}}_{k+1}   \big]\mu_{k+1}^{\ssup{\bq}}(\d \xi).
\end{equation}
Now by \eqref{E:est_trace_S1_h}  and \eqref{E:est_Fmult_S1}  as well
as the assumption on $h$ we have
$$
|G(X, \varphi)|^{k,X,2} \leq C |K(X, \varphi)|^{k,X,4}.
$$
Using again Lemma \ref{L:Taylor} and Lemma \ref{L:prop} (iv) 
we get
\begin{multline}
\frac{1}{C(d)} \norm{D_q \Pi_2 \bR^{(\bq)} K(X, \cdot)(\dot{\bq})}_{k,0} =  \frac{1}{C(d)}
\norm*{ \frac{\d}{\d t}\Bigr|_{t=0} \Pi_2 \bR^{(\bgamma(t))} K(X, \cdot)}_{k,0} \\
\le
\bnorm{ (\bR^{(\bq)} G)(X, 0)}^{k,X,2} 
 \le 2^{2^d}  \norm{G(X)}_{k,X,2}  \le   C 2^{2^d}  \norm{K(X)}_{k,X,4} \le
 C 2^{2^d} \norm{K}_{k,4}^{(\mathsf{A})}  .
\end{multline}

The higher derivatives with respect to $t$ are estimated in a similar way 
using the functions 
\begin{align}
G_2(X, \varphi) &:= \tr\big[D^2 K(X,\varphi)
\ddot{\Cscr}^{\ssup{\bq}}_{k+1} \big],  \quad
G_3(X, \varphi) :=  \tr\big[D^2 G(X,\varphi)
\dot{\Cscr}^{\ssup{\bq}}_{k+1} \big], \\
 G_4(X, \varphi) &:= \tr\big[D^2 K(X,\varphi)
\dddot{\Cscr}^{\ssup{\bq}}_{k+1} \big], \quad 
 G_5(X, \varphi) := \tr\big[D^2 G(X,\varphi)
\ddot{\Cscr}^{\ssup{\bq}}_{k+1} \big], \\
 G_6(X, \varphi) &:=  \tr\big[D^2 G_3(X,\varphi)
\dot{\Cscr}^{\ssup{\bq}}_{k+1} \big]  .
\end{align}
and the
estimates (see  \eqref{E:est_trace_S1_h}  and \eqref{E:est_Fmult_S1})
\begin{align}
|G_2(X,  \xi)|^{k, X,2}   +   |G_4(X,  \xi)|^{k, X,2}  &\le C   |K(X,  \xi)|^{k, X,4}, 
 \\
|G_3(X,  \xi)|^{k, X,2} +  |G_5(X,  \xi)|^{k, X,2} & \le C |G(X, \xi)|^{k,X,4} \le  C   |K(X,  \xi)|^{k, X,6}, \\
\quad |G_6(X,  \xi)|^{k, X,2}  & \le  C |G_3(X, \xi)|^{k,X, 4} \le C   |K(X,  \xi)|^{k, X,8} .
\end{align}
\qed
\end{proofsect}

\section{The map $P_1$}

\begin{lemma}\label{L:P_1}
Consider the map 
$$
P_1\colon  \bM_{\ttt}\times  \bM_{\ttt}
\times \widehat \bM_{:,r} \to  \bM'_{r}
$$ defined in \eqref{E:defP1},
restricted to  
$B_{\rho_1}(1)\times B_{\rho_2}\times   \widehat \bM_{:,r}  \subset \bM_{\ttt}\times \bM_{\ttt} \times \widehat \bM_{:,r}$ 
with the balls $B_{\rho_1}(1)$ and $ B_{\rho_2}$  
 defined in terms of respective norms $\tnorm{\cdot}_k$, 
 i.e., $ B_{\rho_1}(1)=\{\widetilde{I}\in \bM_{\ttt}\colon\tnorm{\widetilde{I}-1}_k<\rho_1\} $ and 
 $ B_{\rho_2}=\{\widetilde{J}\in \bM_{\ttt}\colon \tnorm{\widetilde{J}}_k<\rho_2\} $, 
and   the target space $\bM'_{r}$  equipped with the norm   $\norm{\cdot}^{(\mathsf{A})}_{k+1,r}$.
There exists $\mathsf{A}_0=\mathsf{A}_0(L,d)$ such that for any $ \mathsf{A} \ge \mathsf{A}_0 $ and $\rho_1,\rho_2$, and  $\tilde B$ such that
\begin{equation}
 \rho_1\le 1/2,\ \rho_2< (2\mathsf{A}^{1+2^{d+2}})^{-1}  \text{ and } \   \tilde{\mathsf{B}} \ge \mathsf{A}^{2^{d+3}}
\end{equation}
the map $P_1$ is smooth  and, for any $j_1,j_2 \in\N$,   satisfies the bounds

\begin{multline}\label{estP_1deriv}
\tfrac{1}{j_1!}   \tfrac{1}{j_2!}      \norm{D^{j_1}_{1}D_{2}^{j_2} P_1
(\widetilde{I}, \widetilde{J}, \widetilde{P})(\dot{\widetilde{I}},\dots,\dot{\widetilde{I}},\dot{\widetilde{J}},\dots,\dot{\widetilde{J}})   }_{k+1,r}^{(\mathsf{A})}\le\\ 
\le 
\tnorm{\dot{\widetilde{I}}}_k^{j_1} \bigl(\mathsf{A}^{1+2^{d+2}} \tnorm{\dot{\widetilde{J}}}_k\bigr)^{j_2}
\, \max\left(  \norm{  \widetilde{P}}_{k:k+1,r}^{(\mathsf{A}/4, \tilde{\mathsf{B}})}, 1 \right),
\end{multline}
\begin{multline}\label{estP_1deriv_bis}
\tfrac{1}{j_1!}   \tfrac{1}{j_2!}      \norm{D^{j_1}_{1}D_{2}^{j_2}  D_3 P_1   
(\widetilde{I}, \widetilde{J}, \widetilde{P})(\dot{\widetilde{I}},\dots,\dot{\widetilde{I}},\dot{\widetilde{J}},\dots,\dot{\widetilde{J}},   \dot{\widetilde{P}})   }_{k+1,r}^{(\mathsf{A})}\le\\ 
\le 
\tnorm{\dot{\widetilde{I}}}_k^{j_1} \bigl(\mathsf{A}^{1+2^{d+2}} \tnorm{\dot{\widetilde{J}}}_k\bigr)^{j_2} 
\, \norm{  \dot{\widetilde{P}}}_{k:k+1,r}^{(\mathsf{A}/4, \tilde{\mathsf{B}})},
\end{multline}
\begin{equation} \label{estP_1deriv_ter}
D^{j_1}_{1}D_{2}^{j_2}  D_3^{j_3}  P_1 = 0  \quad \text{for $j_3 \ge 2$.}
\end{equation}

\end{lemma}

\begin{proofsect}{Proof}  Since $P_1$ is affine in the last argument, \eqref{estP_1deriv_ter} is obvious and \eqref{estP_1deriv_bis} follows
from \eqref{estP_1deriv}. Indeed since $\widetilde P(\emptyset) \equiv 1$ the map  $P_1$ can be written as
\begin{equation} P_1(\widetilde I, \widetilde J, \widetilde P) = P_1^0(\widetilde I, \widetilde J) + 
P_1^1(\widetilde I, \widetilde J, \widetilde P)
\end{equation}
with
\begin{equation}
P_1^0(\widetilde I, \widetilde J)(U)  = \sum_{X_1 \in \Pcal(U)}  \chi(X_1, U) \widetilde{I}^{U\setminus X_1} \widetilde{J}^{X_1},
\end{equation}
\begin{equation}
P_1^1(\widetilde I, \widetilde J, \widetilde P) = \sum_{\heap{X_1,X_2\in\Pcal(U)}{X_1\cap X_2=\emptyset, X_2 \neq \emptyset}}
\chi(X_1 \cup X_2, U)
\widetilde{I}^{U\setminus (X_1 \cup X_2)} \widetilde{J}^{X_1} \widetilde{P}(X_2)
\end{equation}
Since $P_1^1$ is linear in $P$ we have
\begin{equation}
D_3 P_1(\widetilde I, \widetilde J, \widetilde P)(\dot{\widetilde{P}}) =  P_1^1(\widetilde I, \widetilde J, \dot{\widetilde{P}}) =
\lim_{\lambda \to \infty} \tfrac{1}{\lambda}  P_1(\widetilde I, \widetilde J,  \lambda \widetilde P)
\end{equation}
and an  analogous identity holds for 
$\tfrac{1}{j_1!}   \tfrac{1}{j_2!}     D^{j_1}_{1}D_{2}^{j_2}  D_3 P_1$.
Thus \eqref{estP_1deriv_bis} follows
from \eqref{estP_1deriv}.

To prove \eqref{estP_1deriv}  we first consider the case $j_1 = j_2 = 0$. Pick $ U\in\Pcal_{k+1}^{\rm c} $. 
Taking into account that 
$$\norm{F(U)}_{k+1,U,r}\le \norm{F(U)}_{k:k+1,U,r},$$ and applying Lemma~\ref{L:prop} (iib)   we get
\begin{multline}
\label{E:P_1initial}  
\!\!\!\! \norm{   P_1^1(\widetilde I,\widetilde J,\widetilde P)(U)}_{k+1,U,r}\le\\
\le\!\!\!\!\! \sum_{\heap{X_1,X_2\in\Pcal(U)}{X_1\cap X_2=\emptyset, X_2 \neq \emptyset}}
\chi(X_1\cup X_2, U)
\tnorm{\widetilde I}_k^{|U\setminus (X_1\cup X_2)|}\tnorm{\widetilde J}_k^{|X_1|} 
   \,   \norm{\widetilde P(X_2)}_{k:k+1,X_2,r} 
\\ 
\le \sum_{\heap{X_1,X_2\in\Pcal(U)}{X_1\cap X_2=\emptyset}}\chi(X_1\cup X_2, U)  \, 
2^{|U\setminus (X_1\cup X_2)|}     \mathsf{A}^{-(1+2^{d+2}) |X_1|}
\norm{\widetilde P}_{k:k+1,r}^{(\mathsf{A}/4, \tilde{\mathsf{B}})}      \Gamma_{\mathsf{A}/4}(X_2)^{-1} \tilde{\mathsf{B}}^{-|\Ccal(X_2)|}   %^{|\Ccal(X_2)}
\end{multline} 

Now 
\begin{equation}
\Gamma_{\mathsf{A}/4}(X_2) \ge  \left(\tfrac{\mathsf{A}}{4}\right)^{|X_2| - 2^d |\Ccal(X_2)|}
\end{equation}
and using that $\tilde{\mathsf{B}} \ge \mathsf{A}^{2^{d+3}}$  and $2^{d+3} - 2^d \ge 2^{d+2}$ we get
\begin{multline}
\label{E:P_1initial_bis}
\!\!\!\! \norm{P_1^1(\widetilde I,\widetilde J,\widetilde P)(U)}_{k+1,U,r}\le\\
\le 
4^{|U|}\sum_{\heap{X_1,X_2\in\Pcal(U)}{X_1\cap X_2=\emptyset, X_2 \neq \emptyset}}\chi(X_1\cup X_2, U)
\mathsf{A}^{-(1+2^{d+2})\abs{X_1}-\abs{X_2}-2^{d+2} \abs{\Ccal(X_2)})}   \, \,  \norm{\widetilde P}_{k:k+1,r}^{(\mathsf{A}/4, \tilde{\mathsf{B}})} .
\end{multline}

Now, we will rely on the combinatorial Lemma 6.16 from \cite{B07}
stated in \eqref{E:Xk+1all} in Lemma~\ref{L:Xk+1largeandall}, 
\begin{equation}
\abs{X}_k\ge (1+\upalpha(d))\abs{\overline{X}}_{k+1} - (1+\upalpha(d)) 2^{d+1} \abs{\Ccal(X)} \ \text{ with }\ \upalpha(d)=\tfrac1{(1+2^d)(1+6^d)}.
\end{equation}
\lsm[agx]{$\upalpha(d)$}{$=\tfrac1{(1+2^d)(1+6^d)}$ from the bound $\abs{X}_k\ge (1+\upalpha(d))\abs{\overline{X}}_{k+1} - (1+\upalpha(d)) 2^{d+1} \abs{\Ccal(X)}$ for any $X\in \Pcal_k$}% 
Applying this inequality with $X = X_1 \cup X_2$ and using the trivial estimate
$\Ccal(X_1 \cup X_2) \le |X_1| + \Ccal(X_2)$, we get

\begin{equation}   \label{E:6_reblock}
(1+2^{d+2})\abs{X_1}_k+\abs{X_2}_k+2^{d+2} \abs{\Ccal(X_2)} \ge (1+\upalpha(d))\abs{\overline{X_1\cup X_2}}_{k+1}
\end{equation}
and thus
\begin{equation}  \label{E:6_sum_P1}
\begin{aligned}
\!\!\!\! \norm{P_1^1 &(\widetilde I,\widetilde J,\widetilde P)(U)}_{k+1,U,r}
\\ & \le 4^{|U|_k}\!\!\!\!  \sum_{\heap{X_1,X_2\in\Pcal(U)}{X_1\cap X_2=\emptyset, X_2 \neq \emptyset}}\chi(X_1\cup X_2, U) 
\mathsf{A}^{-(1+\upalpha(d))\abs{\overline{X_1\cup X_2}}_{k+1}} \, \, 
 \norm{\widetilde P}_{k:k+1,r}^{(\mathsf{A}/4, \tilde{\mathsf{B}})}.
 \end{aligned}
\end{equation}

Similarly we obtain for $P_1^0$
\begin{alignat}1
\label{E:P_1_0_initial}
 \norm{   P_1^0(\widetilde I,\widetilde J)(U)}_{k+1,U,r}  &  \le
\sum_{X_1 \in \Pcal(U)}
\chi(X_1, U)
\tnorm{\widetilde I}_k^{|U\setminus X_1|}\tnorm{\widetilde J}_k^{|X_1|} 
\\   \nonumber
&  \le    2^{|U|} \sum_{X_1 \in \Pcal(U)} \chi(X_1, U)     \mathsf{A}^{-(1+2^{d+2}) |X_1|}    
\end{alignat} 
Since $\upalpha(d) \le 1 \le 2^{d+2}$ and since $|X_1|_k \ge  |\overline{X_1}|_{k+1}$ it is easy to combine the estimates
for $P_1^1$ and $P^0_1$. To prove \eqref{estP_1deriv} for $j_1 = j_2 = 0$ it thus suffices to show that
\begin{equation}   \label{E:6_P1_key_est}
\Gamma_{\mathsf{A}}(U)  \, \,   4^{|U|_k}\!\!\!\!  \sum_{\heap{X_1,X_2\in\Pcal(U)}{X_1\cap X_2=\emptyset}}\chi(X_1\cup X_2, U) 
\mathsf{A}^{-(1+\upalpha(d))\abs{\overline{X_1\cup X_2}}_{k+1}}  \le 1.
\end{equation}
for any $ U\in\Pcal_{k+1}^{\rm c} $ once 
\begin{equation}
\mathsf{A}\ge \mathsf{A}_0(L,d) =(12)^{L^d (1+2^d)(1+6^d)}.
\end{equation}

If  $\abs{U}_{k+1}\le 2^{d}$ then $\Gamma_{\mathsf{A}}(U) = 1$ and  we use $|U|_{k}=L^d|U|_{k+1}$ 
as well as the fact that the 
sum in \eqref{E:6_P1_key_est} has at most $3^{|U|_k} \le 3^{L^d 2^d}$
terms, each contributing at most $\mathsf{A}^{-1} \le \mathsf{A}^{- 2^d  \upalpha(d)} $ 
to bound the left hand side of \eqref{E:6_P1_key_est} by
  \begin{equation}
4^{(2L)^d} 3^{(2 L)^d} \mathsf{A}^{-1} 
%A^{-(1+2^{d+2})}
%\le \bigl( \tfrac{ \newsix{(}3\cdot 4\newsix{)}^{L^d}}{A^{ \upalpha(d)     }}   \bigr)^{2^d}  
\le \left( (12)^{L^d}  \mathsf{A}^{- \upalpha(d)  }  \right)^{2^d} \le 1.
\end{equation}

For $\abs{U}_{k+1}> 2^{d}$, there is no $B\in \Pcal_k$ such that $U=\overline{B^*}$ and as a result $X_1\cup X_2$ is not small and $U=\overline{X_1\cup X_2}$
(cf. definition \eqref{E:chi} of $\chi(X_1\cup X_2, U)$). Hence, using again that the number of terms in the sum is bounded by $ 3^{|U|_{k}}$, we  can bound the left hand side of  \eqref{E:6_P1_key_est} by
\begin{alignat}1
& \, \mathsf{A}^{|U|_{k+1}} \,  4^{L^d|U|_{k+1}}  \mathsf{A}^{-(1+\upalpha(d))|U|_{k+1}}
\sum_{\heap{X_1,X_2\in\Pcal(U)}{X_1\cap X_2=\emptyset}}\chi(X_1\cup X_2, U) 
\\
\le &   \, (12)^{L^d|U|_{k+1}}  \mathsf{A}^{-\upalpha(d)|U|_{k+1}} \le 1  \nonumber 
\end{alignat}
once    $(12)^{L^d}  \mathsf{A}^{-\upalpha(d)}\le 1$.

For the derivatives

\begin{multline}
\tfrac{1}{j_1!}\tfrac{1}{j_2!} D^{j_1}_{1}D_{2}^{j_2} P_1^1(\widetilde{I}, \widetilde{J}, \widetilde{P})(U)(\dot{\widetilde{I}},\dots,\dot{\widetilde{I}},\dot{\widetilde{J}},\dots,\dot{\widetilde{J}})
%=
\\=\!\!\!
\sum_{\heap{X_1,X_2\in\Pcal(U)}{X_1\cap X_2=\emptyset, X_2 \neq \emptyset}}\!\!\! \chi(X_1\cup X_2,U)\!\!\!\!\!\!\!\!\!\!\!\!
\sum_{\substack{Y_1\in\Pcal(U\setminus (X_1\cup X_2)),\abs{Y_1}=j_1\\  Y_2\in\Pcal(X_1),\abs{Y_2}=j_2 }}\!\!\!\!\!\!\!\!\!
\widetilde I^{(U\setminus (X_1\cup X_2))\setminus Y_1}(\dot{\widetilde{I}})^{Y_1}\widetilde J^{X_1\setminus Y_2}(\dot{\widetilde{J}})^{Y_2}
%\!\!\!\!\!\! \;\;\;\;\times\\ 
\widetilde P(X_2)
\end{multline}

we proceed as above   in \eqref{E:P_1initial} and \eqref{E:P_1initial_bis}   to get 

\begin{multline}
\tfrac{1}{j_1!}   \tfrac{1}{j_2!}    \norm{D^{j_1}_{1}D_{2}^{j_2} P_1^1(\widetilde{I}, \widetilde{J}, \widetilde{P})(U)(\dot{\widetilde{I}},\dots,\dot{\widetilde{I}},\dot{\widetilde{J}},\dots,
\dot{\widetilde{J}})}_{k+1,U,r}\le\\
\shoveleft{\le \sum_{\heap{X_1,X_2\in\Pcal(U)}{X_1\cap X_2=\emptyset, X_2 \neq \emptyset}}\chi(X_1\cup X_2,U)
\tbinom{\abs{U\setminus (X_1\cup X_2)}}{j_1}\tnorm{\widetilde{I}}_k^{\abs{U\setminus (X_1\cup X_2)}-j_1}
\tbinom{\abs{X_1}}{j_2}
}   \times  \\
\shoveright{
\times   \tnorm{\widetilde{J}}_k^{\abs{X_1}-j_2}  \| P(X_2)\|_{k:k+1,X_2,r}\
\tnorm{\dot{\widetilde{I}}}_{k}^{j_1}  \tnorm{\dot{\widetilde{J}}}_k^{j_2}
\le}\\
\shoveleft{\le
\sum_{\heap{X_1,X_2\in\Pcal(U)}{X_1\cap X_2=\emptyset, X_2 \neq \emptyset}}\chi(X_1\cup X_2,U)
2^{\abs{U\setminus (X_1\cup X_2)}}  
2^{\abs{U\setminus (X_1\cup X_2)} - j_1}  %sm 1.4. 2016 changed sign here
2^{\abs{X_1}}
\times}\\
\shoveright{      (2\mathsf{A}^{1+2^{d+2}})^{-\abs{X_1}+j_2} \,    
(\tfrac{\mathsf{A}}{4})^{-|X_2| + 2^d \abs{\Ccal(X_2)}}   \,  \mathsf{A}^{- 2^{d+3} \abs{\Ccal(X_2)}} \,      
\norm{\widetilde P}_{k:k+1,r}^{(\mathsf{A}/4, \tilde{\mathsf{B}})}  \, 
\tnorm{\dot{\widetilde{I}}}_{k}^{j_1}  \,  \tnorm{\dot{\widetilde{J}}}^{j_2}_k   \le 
   }\\ 
\shoveleft{\le   \norm{\widetilde P}_{k:k+1,r}^{(\mathsf{A}/4, \tilde{\mathsf{B}})}  \, 
  \tnorm{\dot{\widetilde{I}}}_k^{j_1} \bigl( \mathsf{A}^{1+2^{d+2}} \tnorm{\dot{\widetilde{J}}}_k\bigr)^{j_2} 
\times}\\ 
\times
 4^{|U|}\sum_{\heap{X_1,X_2\in\Pcal(U)}{X_1\cap X_2=\emptyset, X_2 \neq \emptyset}}\chi(X_1\cup X_2, U)
\mathsf{A}^{-(1+2^{d+2})\abs{X_1}-\abs{X_2}-2^{d+2} \abs{\Ccal(X_2)})}.
\end{multline}
Similarly we get
\begin{multline}
\tfrac{1}{j_1!}   \tfrac{1}{j_2!}    \norm{D^{j_1}_{1}D_{2}^{j_2} P_1^0(\widetilde{I}, \widetilde{J})(U)(\dot{\widetilde{I}},\dots,\dot{\widetilde{I}},\dot{\widetilde{J}},\dots,
\dot{\widetilde{J}})}_{k+1,U,r}\le\\
\shoveleft{\le  \sum_{X_1 \in \Pcal(U)} \chi(X_1, U)
2^{\abs{U\setminus X_1}}  
2^{\abs{U\setminus X_1)} - j_1}  
2^{\abs{X_1}}
\, 
 (2\mathsf{A}^{1+2^{d+2}})^{-\abs{X_1}+j_2} \,   
\tnorm{\dot{\widetilde{I}}}_{k}^{j_1}  \,  \tnorm{\dot{\widetilde{J}}}^{j_2}_k     }  \\
\shoveleft{    \le  
  \tnorm{\dot{\widetilde{I}}}_k^{j_1} \bigl( \mathsf{A}^{1+2^{d+2}} \tnorm{\dot{\widetilde{J}}}_k\bigr)^{j_2} 
\, 
 4^{|U|}  \sum_{X_1 \in \Pcal(U)} \chi(X_1, U)   }  \mathsf{A}^{- (1+2^{d+2})\abs{X_1} }  \\
\end{multline}
Now      \eqref{estP_1deriv}       follows as in the case $j_1=j_2 = 0$ by 
using \eqref{E:6_reblock} and \eqref{E:6_P1_key_est}  as well as the obvious estimates
$\upalpha(d) \le 1 \le 2^{d+2}$ and $|X_1|_k \ge | \overline{X_1}|_{k+1}$.

\qed
 \end{proofsect}

\section[Proof of Proposition~4.6]{Proof of Proposition~\ref{P:Tnonlin}}\label{S:finalsmooth}

Proposition~\ref{P:Tnonlin} now follows from the estimates on the maps $ E,P_1,R_1,R_2 $,  $P_2$ 
and $P_3$ and the chain rule, Theorem~\ref{T:fullchain}, in connection with Remark~\ref{D:remDerivative} which provides uniform control of the relevant derivatives. For the convenience of the reader we spell out the details. 
We first write $ S $ as a composition of five maps $ \bF_1,\ldots,\bF_{5} $ 
 and describe the scales of Banach spaces $\bX^{\ssup{i}},i=1,\ldots, 5$,  
 on which these maps are defined. 
 Then we recursively identify neighbourhoods $ \bU^{\ssup{i}}\subset\bX^{\ssup{i}} $ such that 
$$
\bF_i\in \widetilde{C}^m(\bU^{\ssup{i}}\times B_{\frac{1}{2}}),\quad i=1,\ldots, 5, % was 4
$$ 
and verify that $ \bF_i(\bU^{\ssup{i}}\times B_{\frac{1}{2}})\subset \bU^{\ssup{i-1}} $ for $ i\ge 2 $ and that each map $ \bF_i $ satisfies the assumptions of the chain rule Theorem~\ref{T:fullchain}. Recall the definitions in Appendix~\ref{appChain} and denote by $ \diamond $ the composition defined by
\begin{equation}
\big(\bF\diamond\bG\big)(\bx,\bp):=\bF(\bG(\bx,\bp),\bp).
\end{equation}

Define
\begin{equation}
\tilde{\mathsf{B}} = \mathsf{A}^{2^{d+3}}, \quad \mathsf{B} = 2^{2^d} \tilde{\mathsf{B}}.
\end{equation}
In the following we will always assume
\begin{equation}
r_0 \ge 2m + 2. 
\end{equation}
We also assume that 
\begin{equation}
\mathsf{A} \ge \mathsf{A}_0(L,d)
\end{equation}
where $\mathsf{A}_0(L,d)$ is the quantity in Lemma \ref{L:P_1} and
\begin{equation}
h \ge L^{\upkappa(d)} h_1 \quad \text{with}  \quad h_1 = h_1(d, \omega)
\end{equation}
and $\upkappa(d)$ as in  Lemma~\ref{L:prop} (iv)
(see \eqref{E:defh_1=}).

Note that
\begin{equation}
S=\bF_1\diamond\bF_2\diamond\bF_3\diamond\bF_4\diamond\bF_5,
\end{equation}

where the maps $ \bF_i,i=1,\ldots, 5 $,
 and the scales of Banach spaces are given by
\begin{equation}
\begin{aligned}
\bF_1\colon\bX^{\ssup{1}}\times B_{\frac{1}{2}}\to\bX^{\ssup{0}},\quad  \bF_1(K_1,K_2,K_3,
\bq)=P_1(K_1,K_2,K_3),
\end{aligned}
\end{equation}
with
\begin{equation}
\begin{aligned}
\bX^{\ssup{1}}_{n} &=\bM_{\ttt}^2\times     (\widehat \bM_{:,r_0-2m + 2n},
\norm{\cdot}_{k:k+1, r_0-2m + 2n }^{(\mathsf{A}/4, \tilde{\mathsf{B}})})     \\
\bX^{\ssup{0}}_{n}&=(\bM_{r_0-2m +2n}^{\prime},\norm{\cdot}_{k+1,r_0 - 2m + 2n}^{(\mathsf{A})}),\\
B_{\frac{1}{2}}&=\{\bq\in\R_{\rm sym}^{d\times d}\colon\norm{\bq}<\frac{1}{2}\};
\end{aligned}
\end{equation}
and
\begin{equation}
\begin{aligned}
\bF_2&\colon\bX^{\ssup{2}}\times B_{\frac{1}{2}}\to\bX^{\ssup{1}}, \quad \bF_2(H,K,\bq):=(E(H),1-E(H),R_1(K,\bq)),
\end{aligned}
\end{equation}
with
\begin{equation}
\begin{aligned}
\bX^{\ssup{2}}_{n}&=(\bM_0, \norm{\cdot}_{k,0})\times(\widehat \bM_{r_0 - 2m + 2n}, 
 \norm{\cdot}_{k, r_0 - 2m +2n}^{(\mathsf{A}/2,\mathsf{B})});
\end{aligned}
\end{equation}
and

\begin{equation}
\begin{aligned}
\bF_3\colon\bX^{\ssup{3}} \to \bX^{\ssup{2}}, \quad\bF_3(H,K):= (H, P_3(K)),
\end{aligned}
\end{equation}
with
\begin{equation}
\bX^{\ssup{3}}_n=(\bM_0,\norm{\cdot}_{k,0}) \times ( \bM_{r_0 - 2m + 2n},\norm{\cdot}_{k,r_0 -2m + 2n}^{(\mathsf{A}/2)})
\end{equation}

\begin{equation}
\begin{aligned}
\bF_4\colon\bX^{\ssup{4}}\times B_{\frac{1}{2}}\to\bX^{\ssup{3}},\quad \bF_4(H,\widetilde{K},K,\bq):=(R_2(H,K,\bq),P_2(\widetilde{K},K)),
\end{aligned}
\end{equation}
with
\begin{equation}
\bX^{\ssup{4}}_n=(\bM_0,\norm{\cdot}_{k,0})\times \bM_{\ttt}\times (\bM_{r_0 -2m + 2n},\norm{\cdot}_{k,r}^{(\mathsf{A})})
\end{equation}
and
\begin{equation}
\bF_5\colon\bX^{\ssup{5}}\times B_{\frac{1}{2}}\to\bX^{\ssup{4}},\quad \bF_5(H,K):=(H,E(H),K),
\end{equation}
with
\begin{equation}
\bX^{\ssup{5}}_n=(\bM_0,\norm{\cdot}_{k,0})\times(\bM_{r_0 - 2m + 2n},\norm{\cdot}_{k,r_0 - 2m + 2n}^{(\mathsf{A})}).
\end{equation}

\noindent Let
\begin{equation}
\begin{aligned}
\bU^{\ssup{1}}&=B_{\rho_1}(1)\times B_{\rho_2}\times \widehat \bM_{:,r_0} \subset\bX^{\ssup{1}}_m\;\mbox{ with }\\
\rho_1&\le \frac{1}{2},\;\rho_2<\big(2\mathsf{A}^{1+2^{d+2}}\big)^{-1}.   %,\; \rho_3<\big(2A^{2^{d+3}}\big)^{-1}.
\end{aligned}
\end{equation}
Then by Lemma~\ref{L:P_1} we have 
\begin{equation}
\bF_1\in \widetilde{C}^m(\bU^{\ssup{1}}\times B_{\frac{1}{2}}, \bX^{\ssup{0}}),
\end{equation}
and the derivatives of $ \bF_1 $ satisfy the assumptions of the chain rule, Theorem~\ref{T:fullchain}. 
Let $ C_{\ref{immersion}} $ denote the constant in \eqref{immersion2} in Lemma~\ref{immersion} (we may assume that 
$C_{\ref{immersion}}\ge 1$) and let
\begin{equation}
\rho_3= \frac{1}{C_{\ref{immersion}}  }  
 \min\{\rho_1,\rho_2\}=       \frac{\rho_2}{C_{\ref{immersion}}}.
\end{equation}
Then $ H\in B_{\rho_3} $ implies that $ E(H)-1\in B_{\rho_1}\cap B_{\rho_2}  \subset \bM_{\ttt}^2$. 
Thus the choice
$$
\bU^{\ssup{2}}:=B_{\rho_3} \times \widehat \bM_{r_0}
$$
yields
\begin{equation}
\bF_2(\bU^{\ssup{2}}\times B_{\frac{1}{2}})\subset\bU^{\ssup{1}}.
\end{equation}
Moreover by Lemma~\ref{immersion} and Lemma~\ref{L:R_1} the map $ \bF_2\colon\bU^{\ssup{2}}\times B_{\frac{1}{2}}\to\bX^{\ssup{1}}_{m} $ satisfies the assumptions of the chain rule, Theorem~\ref{T:fullchain}.

\bigskip

Let 
$$ \rho_4 := (2\mathsf{B})^{-1},   \quad \bU^{\ssup{3}}  = B_{\rho_3} \times B_{\rho_4}  $$
Then 
\begin{equation}
\bF_3( \bU^{\ssup{3}}  \times B_{\frac12}) \subset \bU^{\ssup{2}} 
\end{equation}
and by Lemma  \ref{L:P_3} the map $F_3$ is a smooth map on  $\bU^{\ssup{3}}$ and
on  $\bU^{\ssup{3}}$  satisfies the assumptions of 
the chain rule Theorem~\ref{T:fullchain}.  Note that we are applying Lemma~\ref{L:chainnoq} for those maps which do not depend on $ \bq $ like $ F_1,F_2 $ and  $ F_5$.

We have $ \rho_4\le 1 $. Let $ C_{\ref{L:R_2}} $ be the constant in
Lemma~\ref{L:R_2}   and let 
\begin{equation}
\rho_5=\frac{\rho_3}{2C_{\ref{L:R_2}}},\;  \rho_6=\frac{\rho_4}{4\mathsf{A}}, \;
\rho_7=\min\Big\{\frac{\rho_3}{2C_{\ref{L:R_2} }},\frac{\rho_4}{4\mathsf{A}^{2^d}}\Big\}.
\end{equation}
Then it follows from \eqref{E:S3bis} in Lemma~\ref{L:S3smooth} and 
Lemma~\ref{L:R_2} (with $r_1=r_0$)   that
\begin{equation}
\bF_4(B_{\rho_5}\times B_{\rho_6}(1)\times B_{\rho_7}\times B_{\frac{1}{2}})\subset B_{\rho_3}\times B_{\rho_4}=\bU^{\ssup{3}}.
\end{equation}
 Set $ \bU^{\ssup{4}}:=B_{\rho_5}\times B_{\rho_6}(1)\times B_{\rho_7} $. 
 Then $ \bF^{\ssup{4}}\colon\bU^{\ssup{4}}\times B_{\frac{1}{2}}\to\bX^{\ssup{3}}_{m} $ satisfies the assumptions of the chain rule. 
 
 Finally set
 \begin{equation}
 \rho_8=\frac{\rho_6}{C_{\ref{immersion}}},\;\rho_{9}=\rho_7, \;\mbox{ and }\bU^{\ssup{5}}=B_{\rho_8}\times B_{\rho_{9}}.
 \end{equation}
Then $\bF_5(\bU^{\ssup{5}}\times B_{\frac{1}{2}}) \subset \bU^{\ssup{4}}$ and  $ \bF_5\colon\bU^{\ssup{5}}\times B_{\frac{1}{2}}\to\bX^{\ssup{4}}_m $ satisfies the assumptions of the chain rule. Now an application of the chain rule, Theorem~\ref{T:fullchain}, shows that the conclusions of Proposition~\ref{P:Tnonlin} hold with $ \rho=\min\{\rho_8,\rho_{9}\} $.

\qed

%% file: AKM-ch7-contraction-30thJune-2016.tex
\chapter{Linearization of the Renormalization Map}\label{sec:contraction}
Here we prove Proposition~\ref{P:Tlin} summarizing the properties of the linearization \eqref{E:DbT} of the maps $\bT_k$
at the fixed point $(H_k,K_k)=(0,0)$ guaranteeing that $H_k$ and $K_k$ are the relevant and irrelevant variables, respectively.
First, we prove the contraction property of  the operator $\bC^{(\bq)}$  in Section \ref{S:proof of Tlin}. 
We finish the proof of Proposition~\ref{P:Tlin} in Section~\ref{S:proof of Tlin} with the bounds on the operators ${\bA^{(\bq)}}^{-1}$ and $\bB^{(\bq)}$.

\section{Contractivity of operator $\bC^{(\bq)}$}
\begin{lemma}
\label{L:contraction}
Let $\theta\in (\frac14,\frac34)$ and $\o\ge  2 (d^2 2^{2d+1} +1) $. Consider the  constant $h_1=h_1(d,\o)$ and $\upkappa(d)$ chosen from Lemma~\ref{L:prop} and let $L\ge 2^d+1$, $h\ge L^{\upkappa(d)} h_1(d,\o)$. There exists  $\mathsf{A}_0=\mathsf{A}_0(d,L)$ such that
\begin{equation}
 \norm{\bC^{(\bq)}}_{r}^{(\mathsf{A})}=\sup_{\norm{K}_{k,r}^{(\mathsf{A})}\le 1}\norm{\bC^{(\bq)} K}_{k+1,r}^{(\mathsf{A})}\le \theta.
\end{equation}
for any $\norm{\bq}\le\tfrac12$, any $k=1,\dots, N$, $r=1,\dots,r_0$, and any   $\mathsf{A}\ge \mathsf{A}_0$.
\end{lemma}

\begin{proofsect}{Proof}
Let us begin by evaluating the large set term: the last term on the right hand side of \eqref{E:C}. 
\end{proofsect}
\begin{lemma}
\label{L:large_set}
Let $L\ge 2^d+1$ and $\o\ge 18\sqrt2+1$. Whenever $h\ge L^{\upkappa(d)} h_1$, and
$\mathsf{A}$ such that $2 \mathsf{A}^{-\frac{2\upalpha}{1+2\upalpha}}\le \frac18\updelta(d,L)$ with $\upalpha$ from Lemma~\ref{L:Xk+1largeandall} and  $\updelta(d,L)$ from Lemma~\ref{L:Peierls}, then
\begin{equation}
\label{E:large_set}
\norm{F}_{k+1,r}^{(\mathsf{A})}\le \tfrac\theta2 \norm{ K }_{k,r}^{(\mathsf{A})}
\end{equation}
for any $K\in M( \Pcal_k, \cbX)$.
Here, the function $F\in M( \Pcal_{k+1}, \cbX)$ is defined by
\begin{equation}
\label{E:large_setF}
F(U,\varphi)=\sum_{\begin{subarray}{c}  X\in  \Pcal_k^{\rm c} \setminus \Scal_k \\  \overline X= U \end{subarray}} 
 \int_{\cbX}K(X,\varphi+\xi) \mu_{k+1}(\d\xi).
\end{equation}
\end{lemma}
\begin{proofsect}{Proof}
Considering, for any $X\subset U$, the function $(\bR_{k+1} K)(X,\varphi)$ and its norm $\bnorm{(\bR_{k+1} K)(X,\varphi)}^{k+1, U,r}$ as defined by \eqref{E:bnorm^jXr}, we have
\begin{equation}
\label{E:FURX}
\sup_{\varphi}\bnorm{(\bR_{k+1} K)(X,\varphi)}^{k+1, U,r}   w_{k+1}^{-U}\le 
\sup_{\varphi}
\bnorm{(\bR_{k+1} K)(X,\varphi)}^{k+1,X,r}  w_{k:k+1}^{-X}.
\end{equation}
To see it, we just notice that, as in \eqref{E:embeddingnorms} in the proof of  Lemma~\ref{L:prop},
one has  
\begin{equation}
\label{E:U<X}
\abs{(\bR_{k+1} K)(X,\varphi)}^{k+1, U,r}\le \abs{(\bR_{k+1} K)(X,\varphi)}^{k+1, X,r}
\end{equation}
 and that 
\begin{equation}
\label{E:wk+1U}
w_{k+1}^{-U}(\varphi) \le  w_{k:k+1}^{-X}.
\end{equation}
 The last inequality amounts to
\begin{multline}
\sum_{x\in X}\bigl((2^d \omega-1)  g_{k:k+1,x}(\varphi)+ \omega G_{k,x}(\varphi)\bigr) + 3L^k  \sum_{x\in\p X} G_{k,x}(\varphi)\le\\ \le
\sum_{x\in U}\omega \bigl(2^{d} g_{k+1,x}(\varphi)+ G_{k+1,x}(\varphi)\bigr) + L^{k+1} \sum_{x\in\p U} G_{k+1,x}(\varphi).
\end{multline}
This is clearly valid since $g_{k:k+1,x}(\varphi)\le  g_{k+1,x}(\varphi)$, $G_{k,x}(\varphi)\le G_{k+1,x}(\varphi) $, and any $x\in \p X\setminus \p U$ is necessarily
contained in $\p B$ for some $B\in \Bcal_k(U\setminus X)$ and, in view of  \eqref{E:k:k+1GinpB}, for each such $B$ one has
\begin{equation}
3 L^{k} \sum_{x\in \p B} G_{k,x}(\varphi)\le \sum_{x\in B} \omega\bigl(2^{d} g_{k+1,x}(\varphi) + G_{k+1,x}(\varphi)   \bigr)
\end{equation}
once $\omega\ge 6c+1$.

Combining  now \eqref{E:FURX} with the bound from Lemma~\ref{L:prop} (iv), we get
\begin{multline}
\Gamma_{k+1,A}(U) \norm{F(U)}_{k+1,U,r}\le  
\mathsf{A}^{|U|_{k+1}}\sum_{\begin{subarray}{c}  X\in  \Pcal_k^{\rm c} \setminus \Scal_k \\  \overline X= U \end{subarray}} 2^{|X|_k}\norm{K(X)}_{k,X,r} \le \\ \le
\norm{K}_{k,r}^{(\mathsf{A})}\; \mathsf{A}^{|U|_{k+1}}   \sum_{\begin{subarray}{c}  X\in  \Pcal_k^{\rm c} \setminus \Scal_k \\  \overline X= U \end{subarray}}  (\tfrac{\mathsf{A}}{2})^{-|X|_k}
\le \norm{K}_{k,r}^{(\mathsf{A})}\; \sum_{\begin{subarray}{c}  X\in  \Pcal_k^{\rm c} \setminus \Scal_k \\  \overline X= U \end{subarray}}  (2 \mathsf{A}^{-\frac{2\upalpha}{1+2\upalpha}})^{|X|_k} \le \tfrac\theta2\norm{K}_{k,r}^{(\mathsf{A})}.
\end{multline}
Here, in the last two inequalities, we first used $\abs{X}_k\ge (1+2\upalpha(d))\abs{\overline{X}}_{k+1}$
for any $X$ contributing to the sum (see \cite[Lemma~6.15]{B07}; \eqref{E:Xk+1large} in Lemma \ref{L:Xk+1largeandall})
and then applied Lemma~\ref{L:Peierls} assuming that $2 \mathsf{A}^{-\frac{2\upalpha}{1+2\upalpha}}\le \frac\theta2\updelta(d,L)$.
%Note that the independence on the generation index follows with \cite[Lemma~6.12]{B07} and \cite[Lemma~6.18]{B07}. 
\qed
\end{proofsect}

\medskip

Turning to the first term on the right hand side of \eqref{E:C}, we have:
\begin{lemma}
\label{L:small_set}
Let $L\ge 7$, $\omega\ge 2 (d^2 2^{2d+1} +1)$, $h\ge L^{\upkappa(d)}h_1$, and $K\in M( \Pcal_k, \cbX)$ with $G\in M( \Pcal_{k+1}, \cbX)$ defined by
\begin{equation}
\label{E:small_setG}
G(U,\varphi)=\sum_{\heap{B\in\Bcal_k(U) }  {\overline{B^*}=U}}\bigl(1-\Pi_2\bigr)\sum_{\heap{X\in\Scal_k,}{X\supset B}}\frac{1}{|X|_k}(\bR_{k+1} K)(X,\varphi).
\end{equation}
Then
\begin{equation}
\label{E:small_set}
\norm{G}_{k+1,r}^{(\mathsf{A})}\le  2^{d+2^d} (3^d-1)^{2^d}\bigl( 5 L^{-\frac{d}{2}} + 2^{d+3} L^{\frac{d}2-2} +9L^{-1}  \bigr) \norm{ K }_{k,r}^{(\mathsf{A})}
\end{equation}
for any $\mathsf{A}>1$.
\end{lemma}
\begin{remark}
Notice that \eqref{E:small_set} is used later only for $ d\le 3 $. Our method can be extended also to include higher dimension when employing additional higher order terms to estimate the projection of the second Taylor polynomial. \hfill $ \diamond$
\end{remark}

\begin{proofsect}{Proof}
Notice  first that the sum vanishes unless $ U\in\Scal_{k+1} $ and, necessarily, for any contributing $X$, one has  $X\subset U$ and $X^*\subset U^*$.
As a result, the norms in \eqref{E:small_set} contain only the contributions of small sets and do not depend on $\mathsf{A}$ according to the definition of the factor $\G_{j,\mathsf{A}}(X)$,  $j=k, k+1$.
Considering $R\in  M^*( \Bcal_k, \cbX)$ defined by 
$R(B,\varphi)=\sum_{\heap{X\in\Scal_k}{X\supset B}}\frac{1}{|X|_k}(\bR_{k+1} K)(X,\varphi)$ and
replacing the operator $1-\Pi_2$ by $(1-T_2)+(T_2-\Pi_2)$, we split $G(U,\varphi)$ into two terms,
\begin{equation}
\label{E:G_1}
G_1(U,\varphi)=\sum_{\heap{B\in\Bcal_k(U) }  {\overline{B^*}=U}}  (1-T_2) R(B,\varphi)
\end{equation}
and
\begin{equation}
\label{E:G_2}
G_2(U,\varphi)=\sum_{\heap{B\in\Bcal_k(U) }  {\overline{B^*}=U}}  (T_2-\Pi_2) R(B,\varphi),
\end{equation}
and evaluate them separately in  Lemma~\ref{lemmatayest1} and Lemma~\ref{lemmatayest2}.

First, however, considering the norm $\bnorm{ F(X,\varphi)}^{j,X,r}$, $j=k, k+1$, as defined in \eqref{E:bnormup} for any 
$ F\in\Mcal(\Pcal_k,\cbX)$ with $ X\in\Pcal_k$ and $ \varphi\in\cbX $, we prove the following.
 \end{proofsect}
 
 \medskip

\begin{lemma}\label{BryLemma}
Let $ F\in\Mcal(\Pcal_k,\cbX)$, $ X\in\Pcal_k $, $r=1,\dots,r_0$, and $j=k, k+1$. Then

\begin{equation}
\label{taylorest}
\bnorm{F(X,\varphi)-T_2F(X,\varphi)}^{j,X,r}\le (1+\bnorm{\varphi}_{j,X})^3\sup_{t\in (0,1)} \sum_{s=3}^r \frac1{s!} \bnorm{D^{s}F(X,t\varphi)}^{j,X}.
\end{equation}
\end{lemma}

\begin{proofsect}{Proof}
Cf. \cite[Lemma 6.8]{B07}. Introducing the shorthands $$f(\varphi)=(1-T_2)F(X,\varphi)$$ and $$f_s(\varphi)=D^sF(X,\varphi)(\dot{\varphi},\dots, \dot{\varphi})$$
for any $s\ge 1$, we express the terms contributing to the left hand side of \eqref{taylorest} with the help of the integral form of the Taylor polynomial remainder, 
\begin{equation}
f(\varphi)=  \int_0^1  \frac{(1-t)^2}2 D^3F(X,t\varphi)(\varphi,\varphi,\varphi)\,  \d t,
\end{equation}
\begin{multline}
Df(\varphi)(\dot{\varphi})=f_1(\varphi) -f_1(0) -Df_1(0)(\varphi) = \int_0^1  (1-t) D^2f_1(t\varphi)(\varphi,\varphi)\,  \d t=\\=
\int_0^1  (1-t) D^3 F(X,t\varphi)(\dot{\varphi},\varphi,\varphi) \, \d t,
\end{multline}
\begin{multline}
\frac12D^2f(\varphi)(\dot{\varphi},\dot{\varphi})=\frac12\bigl(f_2(\varphi) -f_2(0)\bigr) =\\= \frac12\int_0^1  Df_2(t\varphi)(\varphi) \, \d t=\int_0^1  D^3F(X,t\varphi)(\dot{\varphi},\dot{\varphi},\varphi) \, \d t,
\end{multline}
and, for $s\ge 3$,
\begin{equation}
\frac1{s!}D^s f(\varphi)(\dot{\varphi},\dots, \dot{\varphi})= \frac1{s!}D^s F(X,\varphi)(\dot{\varphi},\dots, \dot{\varphi}).
\end{equation}
Summing all the right hand sides above and using the bound 
\begin{equation}
\abs{D^{s+m}F(X,t\varphi)(\dot{\varphi},\dots, \dot{\varphi},\varphi,\dots,\varphi)}\le \abs{D^{s+m}F(X,t\varphi)}^{j,X} \bnorm{\dot{\varphi}}_{j,X}^s\bnorm{\varphi}_{j,X}^m,
\end{equation}
as well as the fact that 
\begin{equation}
\bnorm{\varphi}_{j,X}^3  \int_0^1  \frac{(1-t)^2}2 \,\d t+\bnorm{\varphi}_{j,X}^2 \int_0^1  (1-t) \,\d t  + \frac12\bnorm{\varphi}_{j,X} +\frac1{3!} = \frac1{3!}  (1+\bnorm{\varphi}_{j,X})^3,
\end{equation}
we get the seeked result.
\qed
\end{proofsect}

\medskip

\begin{lemma}\label{lemmatayest1}
Let $ K\in\Mcal(\Scal_k,\cbX)$,  $ X\in\Scal_k $, $B\in \Bcal_k(X) $, and $ U=\overline{B^*}$, and assume that $L\ge 7$, $\omega\ge 2 (d^2 2^{2d+1} +1)$, and $h\ge L^{\upkappa(d)}h_1$. Then
\begin{equation}
\label{E:1-T2X}
\sup_{\varphi}\bnorm{(\bR_{k+1} K)(X,\varphi)-T_2 (\bR_{k+1} K)(X,\varphi)}^{k+1,X,r} w_{k+1}^{-U}(\varphi)\le 5  L^{-\frac{3d}{2}}2^{|X|_k}\norm{K(X)}_{k,X,r}.
\end{equation}
For $G_1$ defined in \eqref{E:G_1} we have
\begin{equation}
\label{E:G1bound}
\norm{G_1(U)}_{k+1,U,r} \le 5\, 2^{d+2^d} (3^d-1)^{2^d} L^{-\frac{d}{2}}  \norm{K}_{k,r}^{(\mathsf{A})}.
\end{equation}
\end{lemma}

\begin{proofsect}{Proof}
Lemma~\ref{BryLemma} yields
\begin{multline}
\bnorm{(\bR_{k+1} K)(X,\varphi)-T_2 (\bR_{k+1} K)(X,\varphi)}^{k+1,X,r} \le\\
\le  (1+\bnorm{\varphi}_{k+1,X})^3\sup_{t\in (0,1)}
\sum_{s=3}^r \frac1{s!}
\bnorm{D^s (\bR_{k+1} K)(X,t\varphi)}^{k+1,X}
\end{multline}
for any $ \varphi\in\cbX $.
Interchanging  differentiation and  integration,  we get 
\begin{multline}
\sum_{s=3}^r \frac1{s!}\bnorm{D^s(\bR_{k+1} K)(X,t\varphi)}^{k+1,X} \le\\
\le \sum_{s=3}^r \frac1{s!}\sup_{\dot{\varphi} \not= 0}\,\int_{\cbX}\mu_{k+1}(\d \xi)\,\Bigl| \frac{D^{s} K(X,t\varphi+\xi)(\dot{\varphi},\dots,\dot{\varphi})}{\bnorm{\dot{\varphi}}^{s}_{k+1,X}}\Bigr|=\\
=\sum_{s=3}^r \frac1{s!}\sup_{\dot{\varphi} \not= 0}\,\int_{\cbX}\mu_{k+1}(\d \xi)\,\Bigl| \frac{D^{s} K(X,t\varphi+\xi)(\dot{\varphi},\dots,\dot{\varphi})}{\bnorm{\dot{\varphi}}^{s}_{k,X}}\frac{\bnorm{\dot{\varphi}}^{s}_{k,X}}{\bnorm{\dot{\varphi}}^{s}_{k+1,X}}\Bigr|\le\\
\le L^{-\frac{3d}{2}}\int_{\cbX}\mu_{k+1}(\d \xi) \,\bnorm{K(X,t\varphi+\xi)}^{k,X,r}.
\end{multline}
In the last inequality we used the bound \eqref{E:dotphi}.
Next, we apply $$\bnorm{K(X,t\varphi+\xi)}^{k,X,r}\le\norm{K(X)}_{k,X,r}  w_{k}^X(t\varphi+\xi) $$ and \eqref{toshowfluct:str}, to get
\begin{equation}
\label{E:sums=13}
\sum_{s=3}^r \frac1{s!}
\bnorm{D^s (\bR_{k+1} K)(X,t\varphi)}^{k+1,X} 
\le 2^{|X|_k} L^{-\frac{3d}{2}}\norm{K(X)}_{k,X,r}\, \frac{w_{k:k+1}^X(\varphi) }{w_{k+1}^U(\varphi) }
w_{k+1}^U(\varphi) .
\end{equation}
Here we also used the fact that $w_{k:k+1}^X(t\varphi)$ is monotone in $t$.

Bounding $ (1+\bnorm{\varphi}_{k+1,X})^3 $ with the help of 
\begin{equation}
\label{E:q}
(1+u)^3\le 5 {\rm e}^{u^2}
\end{equation}
(proven by showing that $\min_{u\ge 0} \frac{{\rm e}^{ u^2}}{(1+u)^3}\ge \frac15$),
we would like to show that
\begin{equation}
\label{E:boundon|phi|}
 \bnorm{\varphi}_{k+1,X}^2\le \log\frac{w_{k+1}^U(\varphi) }{w_{k:k+1}^X(\varphi) }.
\end{equation}

Notice, first, that
\begin{multline}
\label{E:logwlogw}
\log\frac{w_{k+1}^U(\varphi) }{w_{k:k+1}^X(\varphi) }\ge  \sum_{x\in U\setminus X}\bigl( (2^{d} \o-1) g_{k+1,x}(\varphi)+ \o G_{k+1,x}(\varphi)\bigr)+
\sum_{x\in  U} g_{k:k+1,x}(\varphi)+\\+
L^{k}(L-3)\sum_{x\in\p U} G_{k+1,x}(\varphi)-3L^{k}\sum_{x\in\p X\setminus \p U} G_{k,x}(\varphi)\ge\\ \ge
\sum_{x\in U\setminus X}(2^{d} \omega-1) g_{k+1,x}(\varphi)+L^{k}(L-3)\sum_{x\in\p U} G_{k+1,x}(\varphi).
\end{multline}
To verify the last inequality, we show that
\begin{equation}
3L^{k}\sum_{x\in\p X\setminus \p U} G_{k,x}(\varphi)\le \sum_{x\in  U} g_{k:k+1,x}(\varphi)+  \sum_{x\in U\setminus X}\omega G_{k+1,x}(\varphi)
\end{equation}
in analogy with \eqref{E:k:k+1gX}.
Indeed, arguing that any $x\in\p X\setminus \p U$ is contained in $\p B$ for $B\in \Bcal_k( U\setminus X)$,
and applying again Proposition~\ref{boundaryest} (a), we have
\begin{multline}
h^2  L^k \sum_{x\in \p B} G_{k,x}(\varphi)\le\\
\le  2 c\bigl(\sum_{x\in B}|\nabla \varphi(x)|^2+L^{2k}\sum_{x\in U_1(B)}|\nabla^2 \varphi(x)|^2\bigr)+
L^k \sum_{x\in \p B}\sum_{s=2}^3L^{(2s-2)k}|\nabla^s \varphi(x)|^2\le\\
\le h^2 2\mathfrak{c} \sum_{x\in B}  G_{k,x}(\varphi)  +h^2 2 \mathfrak{c}  L^k\sum_{x\in \p B} L^{-2} g_{k:k+1,z}(\varphi),
\end{multline}
where $z$ is any point $z\in B$. 
Using $\abs{\p  B}\le 2^d L^{(d-1)k}$, we get
the seeked bound once  $\o\ge 18\sqrt 2$ and $L\ge 5$ (when $6\mathfrak{c} \le \o$ and $ 6\mathfrak{c}  L^{-2}\le 1$).

In view of \eqref{E:logwlogw} and using that $\bnorm{\varphi}_{k+1,X}^2\le \bnorm{\varphi}_{k+1,U}^2$,
it suffices to show that 
\begin{equation}
\label{E:needk+1U}
 \bnorm{\varphi}_{k+1,U}^2\le \sum_{x\in U\setminus X}(2^{d}\omega-1) g_{k+1,x}(\varphi)+ L^{k}(L-3)\sum_{x\in \p U} G_{k+1,x}(\varphi).
\end{equation}

Clearly,
\begin{equation}
h^2 \bnorm{\varphi}_{k+1,U}^2\le \sum_{1\le s\le 3}L^{(k+1)(d-2+2s)}\max_{x\in U^*} \abs{\nabla^s\varphi(x)}^2
\end{equation}
Applying Lemma~\ref{L:6.20}, we get
\begin{equation}
L^{(k+1)d}\max_{x\in U^*} \abs{\nabla\varphi(x)}^2\le  
 \frac{2 L^{(k+1)d}}{\abs{\p U}} \sum_{x\in \p U} \abs{\nabla\varphi(x)}^2 + 2L^{(k+1)d} (\diam U^*)^2 \max_{x\in U^*} \abs{\nabla^2 \varphi(x)}^2.
\end{equation}
Using  that $\abs{\p U}\ge 2 d L^{(k+1)(d-1)}$, the first term above is covered by the second term on the right hand side of \eqref{E:needk+1U}
once $L\ge 7$,
 \begin{equation}
\frac{2 L^{(k+1)d}}{\abs{\p U}} \le \frac{2 L^{(k+1)d}}{2d L^{(k+1)(d-1)}}=\frac1d L^{k+1}\le L^{k}(L-3).
\end{equation}
Taking into account that  $\diam U^*\le d 2^{d} L^{k+1}$ (here we use the fact that $U$ is necessarily contained in a block of the side $2 L^{k+1}$), the second term is bounded by
$d^2 2^{2d+1} L^{(k+1)(d+2)} \max_{x\in U^*} \abs{\nabla^2 \varphi(x)}^2$ and will be treated together with the remaining terms $\max_{x\in U^*} \abs{\nabla^s\varphi(x)}^2$, $s=2,3$, 
contained in $\bnorm{\varphi}_{k+1,U}^2$.

Using the fact that the number of $(k+1)$-blocks in $U$ is at most $2^d$, we get
\begin{equation}
\max_{x\in U^*} \abs{\nabla^s\varphi(x)}^2\le 2^d \sum_{B\in \Bcal_{k+1}(U)} \max_{x\in B^*} \abs{\nabla^s\varphi(x)}^2.
\end{equation}
This yields
\begin{multline}
(d^2 2^{2d+1} L^{(k+1)(d+2)}+L^{(k+1)(d+2)}) \max_{x\in U^*} \abs{\nabla^2\varphi(x)}^2  \le \\
\le 2^d (d^2 2^{2d+1} +1)L^{(k+1)(d+2)}\sum_{B\in \Bcal_{k+1}(U)} \max_{x\in B^*} \abs{\nabla^2\varphi(x)}^2.
\end{multline}
and
\begin{equation}
L^{(k+1)(d+4)} \max_{x\in U^*} \abs{\nabla^3\varphi(x)}^2  \le 2^d L^{(k+1)(d+4)}\sum_{B\in \Bcal_{k+1}(U)} \max_{x\in B^*} \abs{\nabla^3\varphi(x)}^2.
\end{equation}

Each of the terms on the right hand sides will be bounded by the corresponding term in 
\begin{equation}
h^2 \sum_{x\in B\setminus X} (2^{d}\omega-1) g_{k+1,x}(\varphi)=(2^{d}\omega-1)\sum_{x\in B\setminus X} \sum_{s=2}^4L^{(2s-2)(k+1)}\sup_{y\in B^*_x}\abs{\nabla^s\varphi(y)}^2, 
\end{equation} 
Indeed, observing that $g_{k+1,x}(\varphi)$ is constant over each $(k+1)$-block $B\subset U$, and the volume of $B\setminus X$ is at least $L^{kd}(L^d-2^d)=L^{(k+1)d}(1-(\tfrac2{L})^d)$
since the number of  $k$-blocks in $X$ is at most $2^d$, while $B$ consists of $L^d$ of them, we need
\begin{equation}
2^d (d^2 2^{2d+1} +1)L^{(k+1)(d+2)} \le (2^{d}\omega-1) L^{(k+1)d} (1-(\tfrac2{L})^d) L^{2(k+1)}
\end{equation}
and
\begin{equation}
2^d L^{(k+1)(d+4)} \le (2^{d}\omega-1) L^{(k+1)d} (1-(\tfrac2{L})^d) L^{4(k+1)}.
\end{equation}
These conditions are satisfied once $\omega\ge 2 (d^2 2^{2d+1} +1)$.

In summary, combining \eqref{E:sums=13}, \eqref{E:q}, and \eqref{E:boundon|phi|}, %\eqref{E:sums=13-E:boundon|phi|}
we have
\begin{multline}
(1+\bnorm{\varphi}_{k+1,X})^3\sum_{s=3}^r \frac1{s!}
\bnorm{D^s (\bR_{k+1} K)(X,t\varphi)}^{k+1,X} \le\\
\le 5  L^{-\frac{3d}{2}}2^{|X|_k}\norm{K(X)}_{k,X,r}\,  w_{k+1}^U(\varphi) .
\end{multline}
for any $ \varphi\in\cbX $ and any $t\in(0,1)$, finishing  thus the proof of the inequality \eqref{E:1-T2X}.  

To prove the bound \eqref{E:G1bound}, we use that $\abs{\Bcal_k(U)}\le (2L)^d$ and the obvious bound $\abs{\{X\in\Scal_k\mid X\supset B\}}  \le (3^d-1)^{2^d}$, to get
\begin{multline}
\label{E:G1Ub}
\norm{G_1(U)}_{k+1,U,r}\le5\,  L^{-\frac{3d}{2}}  \sum_{\heap{B\in\Bcal_k(U) }  {\overline{B^*}=U}}\sum_{\heap{X\in\Scal_k}{X\supset B}}\frac{1}{|X|_k} 2^{|X|_k}\norm{K(X)}_{k,X,r}\le\\
\le  5\,  L^{-\frac{3d}{2}} (2L)^d (3^d-1)^{2^d} \norm{K}_{k,r}^{(\mathsf{A})}  2^{2^d}
\le   5\,  2^{d+2^d} (3^d-1)^{2^d} L^{-\frac{d}{2}}  \norm{K}_{k,r}^{(\mathsf{A})}.
\end{multline}
\qed
\end{proofsect}

\medskip

\begin{lemma}\label{lemmatayest2}
Let $ K\in\Mcal(\Scal_k,\cbX)$,   $ U=\overline{B^*}$, and assume that $L\ge 7$ and $\omega\ge 2 (d^2 2^{2d+1} +1)$. For $G_2$ defined in \eqref{E:G_2} we have
\begin{multline}
\label{E:G2bound}
\norm{G_2(U)}_{k+1,U,r} \le\\
\le  2^{2^d+d+1} (3^d-1)^{2^d} \bigl((2^{d+2}-1) L^{\frac{d}{2}-2}+(8 L^{-1}+2L^{-2}) \bigr)\norm{K}_{k,r}.
\end{multline}

\end{lemma}

%\begin{proofsect}{Proof}
Recall that $G_2(U,\varphi)=\sum_{\heap{B\in\Bcal_k(U) }  {\overline{B^*}=U}}  (T_2-\Pi_2) R(B,\varphi)$
with $R\in  M^*( \Bcal_k, \cbX)$ defined by 
$R(B,\varphi)=\sum_{\heap{X\in\Scal_k}{X\supset B}}\frac{1}{|X|_k}(\bR_{k+1} K)(X,\varphi)$.
The polynomial $\Pi_2 R(B,\varphi)= \lambda \abs{B}+\ell(\varphi) +Q(\varphi,\varphi)$
is characterised by taking a unique linear function $\ell(\varphi)$ of the form \eqref{E:ell},
$\ell(\varphi)=\sum_{x\in (B^*)^*} \bigl[\sum_{i=1}^d a_i\, \nabla_i\varphi(x) +
\sum_{i,j=1}^d \bc_{i,j}\, \nabla_i\nabla_j\varphi(x) \bigr]$,
  that agrees with $ D  R(B,0)(\varphi)$
 on all quadratic functions $\varphi$ on $(B^*)^*$ and
 a unique quadratic function $Q(\varphi, \varphi)$  of the form \eqref{E:Q},
 $Q(\varphi,\varphi)=\sum_{x\in (B^*)^*}
\sum_{i,j=1}^d \bd_{i,j} \,\nabla_i\varphi(x)\, \nabla_j\varphi(x)$, that agrees
with  $\tfrac12 D^2 R(B,0)(\varphi, \varphi)$ on all affine functions $\varphi$ on $(B^*)^*$.

In view of the definition of the map $\bR_{k+1}$ we can write 
$$R(B,\varphi)=\int_{\cbX}\mu_{k+1}(\d \xi)\,   R_\xi(B,\varphi)
$$
with 
$$
R_\xi(B,\varphi)=\sum_{\heap{X\in\Scal_k}{X\supset B}}\frac{1}{|X|_k} K(X,\xi+\varphi).
$$
Observing that  
$$
\begin{aligned}
D (\bR_{k+1} K)(X,0)(\varphi)&=\int_{\cbX}\mu_{k+1}(\d \xi)\,  D K(X,\xi)(\varphi),\\
D^2 (\bR_{k+1} K)(X,0)(\varphi,\varphi)&=\int_{\cbX}\mu_{k+1}(\d \xi)\,  D^2 K(X,\xi)(\varphi,\varphi),
\end{aligned}
$$ and introducing, similarly as above,
$\Pi_2 R_\xi(B,\varphi)= \lambda_\xi \abs{B}+\ell_\xi(\varphi) +Q_\xi(\varphi,\varphi)$, the unicity implies that
$\ell(\varphi)=\int_{\cbX}\mu_{k+1}(\d \xi)\, \ell_\xi(\varphi)$ and  $Q(\varphi,\varphi)=\int_{\cbX}\mu_{k+1}(\d \xi)\, Q_\xi(\varphi,\varphi)$.

Given that $G_2(B,\varphi)= (T_2-\Pi_2) R(B,\varphi)$ is a polynomial of second order, we have $\abs{G_2(B,\varphi)}^{k+1,U,r}=\abs{G_2(B,\varphi)}^{k+1,U,2}$. 
In a preparation for the evaluation of this norm, we first evaluate separately  the absolute value of the linear and quadratic terms $P_1(\varphi)$ and $P_2(\varphi)$ in 
$G_2(B,\varphi)$.

Observing that for any affine function $\varphi_1$  and any quadratic function $\varphi_2$ on $(B^*)^*$ we have $P_1(\varphi-\varphi_1-\varphi_2)=P_1(\varphi)$,
we get
\begin{multline}
\label{T2est1}
\bigl|  P_1(\varphi)\bigr| 
=\bigl|  \int_{\cbX}\mu_{k+1}(\d \xi)\, \bigl(D R_\xi(B,0)(\varphi-\varphi_1-\varphi_2)-\ell_\xi(\varphi- \varphi_1-\varphi_2)\bigr)\bigr|\le \\
\le (2^{d+2}-1) \sum_{\heap{X\in\Scal_k}{X\supset B}} \frac{1}{|X|_k}  \norm{K(X)}_{k,X,r}\bnorm{\varphi- \varphi_1-\varphi_2}_{k,B^*} \int_{\cbX}\mu_{k+1}(\d \xi)w_k^{X}(\xi)\le
\\ \le  2^{2^d} (3^d-1)^{2^d} (2^{d+2}-1)  \norm{K}_{k,r} \bnorm{\varphi- \varphi_1-\varphi_2}_{k,B^*} .
\end{multline}
Here,  we first used the inequalities
\begin{equation}
\label{E:ellxi}
\abs{\ell_\xi(\varphi)} \le (2^{d+2}-2)  \sum_{\heap{X\in\Scal_k}{X\supset B}}\frac{1}{|X|_k} \abs{K(X,\xi)}^{k,X,r}   \bnorm{\varphi}_{k,B^*}
\end{equation}
and
\begin{equation}
\abs{D R_\xi(B,0)(\varphi)}\le \sum_{\heap{X\in\Scal_k}{X\supset B}}\frac{1}{|X|_k} \abs{K(X,\xi)}^{k,X,r}   \bnorm{\varphi}_{k,X}
\end{equation}
combined with the bounds $\abs{K(X,\xi)}^{k,X,r}\le  \norm{K(X)}_{k,X,r}w_k^{X}(\xi)$ and  $\bnorm{\varphi}_{k,X}\le \bnorm{\varphi}_{k,B^*}$, and then the bounds $ \int_{\cbX}\mu_{k+1}(\d\xi)w_k^{X}(\xi)   \le 2^{|X|_k} $,
and, as in \eqref{E:G1Ub},  $\abs{\{X\in\Scal_k\mid X\supset B\}}  \le (3^d-1)^{2^d}$.
To verify \eqref{E:ellxi}, we first observe that $\ell_\xi(\varphi)= \sum_{i=1}^d a_i(\xi)\, s_i +
\sum_{i,j=1}^d \bc_{i,j}(\xi)\,  t_{i,j}$
where $s_i=s_i(\varphi)=\sum_{x\in (B^*)^*} \nabla_i\varphi(x)$ and $t_{i,j}=t_{i,j}(\varphi)=  \sum_{x\in (B^*)^*} \nabla_i\nabla_j\varphi(x)$. The same values of ``average slopes'' $\mathbf{s} =\{s_i\}$ and $
\mathbf{t}=\{t_{i,j}\}$ are obtained with the quadratic function 
\begin{equation}
\varphi_{\bold s, \bold t}(x)=L^{-dk} (2^{d+2}-3)^{-d} \sum_i  \bigl(s_i - \sum_j (t_{i,j} + t_{j,i}) \overline{x_j}\bigr) x_i+ L^{-dk}(2^{d+2}-3)^{-d}\sum_{i,j}  t_{i,j} x_i x_j,
\end{equation}
where $\overline{x_j}=L^{-dk}(2^{d+2}-3)^{-d}\sum_{y\in B} y_j$ (notice  that $(B^*)^*$ contains  $(2^{d+2}-3)^{d}$ $k$-blocks).
Further, observe that
\begin{multline}
\label{E:|phist|}
h \bnorm{\varphi_{\bold s, \bold t}}_{k,X}  = \max \Bigl(L^{\frac{dk}2}\max_{x\in X^*} \abs{\nabla\varphi_{\bold s, \bold t}(x)}, 
L^{\frac{dk}2+k}\max_{x\in X^*}\abs{\nabla^2\varphi_{\bold s, \bold t}(x)}\bigr)\le\\
\le L^{-\frac{dk}2}(2^{d+2}-3)^{-d}\max \Bigl(\abs{\bold s} + 2\abs{\bold t } \tfrac12 L^k (2^{d+2}-3), 
L^k\abs{\bold t}\Bigr)=\\
 =L^{-\frac{dk}2}(2^{d+2}-3)^{-d}\abs{\bold s}+ L^{-\frac{dk}2+k} (2^{d+2}-3)^{-d+1}\abs{\bold t}\le   (1+2^{d+2}-3) h \bnorm{\varphi}_{k,B^*}.
\end{multline}
Here,  the last inequality, valid  for any $\varphi$ such that $s_i(\varphi)=s_i$ and $t_{i,j}(\varphi)=t_{i,j}$, is implied by obvious bounds
$\max_{x\in (B^*)^*}\abs{\nabla_i \varphi(x)}\ge L^{-dk}(2^{d+2}-3)^{-d} \abs{s_i}$ and $\max_{x\in (B^*)^*}\abs{\nabla_i \nabla_j\varphi(x)}\ge L^{-dk}(2^{d+2}-3)^{-d} \abs{t_{i,j}}$.

Now, for the quadratic  function $\varphi_{\bold s, \bold t}$ we have  $\ell_\xi(\varphi_{\bold s, \bold t})=D R_\xi(B,0)(\varphi_{\bold s, \bold t})$. As a result,
\begin{multline}
\label{E:elxix}
\abs{\ell_\xi(\varphi)}=\abs{\ell_\xi(\varphi_{\bold s, \bold t})}\le \\
\le
\sum_{\heap{X\in\Scal_k}{X\supset B}}\frac{1}{|X|_k} \abs{D K(X,\xi)(\varphi_{\bold s, \bold t})}\le \sum_{\heap{X\in\Scal_k}{X\supset B}}\frac{1}{|X|_k} \abs{K(X,\xi)}^{k,X,r} 
\bnorm{\varphi_{\bold s, \bold t}}_{k,X}\le
\\  \le  (2^{d+2}-2)  
\sum_{\heap{X\in\Scal_k}{X\supset B}}\frac{1}{|X|_k} \abs{K(X,\xi)}^{k,X,r} 
\bnorm{\varphi}_{k,B^*}.
\end{multline}
Here,  the last inequality, valid  for any $\varphi$ such that $s_i(\varphi)=s_i$ and $t_{i,j}(\varphi)=t_{i,j}$, is implied by obvious bounds
$\max_{x\in (B^*)^*}\abs{\nabla_i \varphi(x)}\ge L^{-dk}(2^{d+2}-3)^{-d} \abs{s_i}$ and $\max_{x\in (B^*)^*}\abs{\nabla_i \nabla_j\varphi(x)}\ge L^{-dk}(2^{d+2}-3)^{-d} \abs{t_{i,j}}$.

Choosing now,  for any fixed $\varphi$,  the functions $\varphi_1$ and $\varphi_2$ as an optimal approximation in accordance with  the Poincar\'{e} inequalities, 
\begin{equation}
\label{T2est2}
\inf_{\varphi_1\,\mbox{{\tiny affine} }}\bnorm{\varphi-  \varphi_1}_{k,B^*}\le \frac1h L^{k(\frac{d}{2}+1)}\sup_{x\in (B^*)^*}|\nabla^2 \varphi(x)|\le  L^{-(\frac{d}{2}+1)} \bnorm{\varphi}_{k+1,B^*}
\end{equation}
and 
\begin{equation}\label{T2est3}
\inf_{\heap{ \varphi_1\,\mbox{{\tiny affine,} }}{  \varphi_2\,\mbox{{\tiny quadratic } }}}\bnorm{  \varphi- \varphi_1-\varphi_2}_{k,B^*}\le \frac1h L^{k(\frac{d}{2}+2)}\sup_{x\in (B^*)^*}|\nabla^3 \varphi(x)|\le   L^{-(\frac{d}{2}+2)}\bnorm{ \varphi}_{k+1,B^*},
\end{equation}
we get
\begin{equation}
\label{difflinearterm}
\bigl| P_1(\varphi)\bigr|  \le L^{-(\frac{d}{2}+2)}
2^{2^d} (3^d-1)^{2^d}  (2^{d+2}-1) 
\norm{K}_{k,r} \bnorm{ \varphi}_{k+1,B^*}.
\end{equation}

Similarly for the quadratic part. First, we prove the bound 
\begin{equation}
\label{E:P2}
\abs{P_2(\varphi,\varphi)} \le  2^{2^d+1} (3^d-1)^{2^d}\norm{K}_{k,r}\bnorm{\varphi}_{k,B^*}^2.
\end{equation}
While deriving it, the bound \eqref{E:ellxi} is replaced by
\begin{equation}
\label{E:Qxi}
\abs{Q_\xi(\varphi,\varphi)} \le  \sum_{\heap{X\in\Scal_k}{X\supset B}}\frac{1}{|X|_k} \abs{K(X,\xi)}^{k,X,r}   \bnorm{\varphi}_{k,B^*}^2.
\end{equation}
For its proof we consider the linear function 
\begin{equation}
\varphi_{\bold s}(x)=L^{-dk} (2^{d+2}-3)^{-d} \sum_i  s_i  x_i
\end{equation}
with the slope $s_i=s_i(\varphi)$ and
\begin{multline}
\label{E:|phis|}
h \bnorm{\varphi_{\bold s}}_{k,X}  =\\= L^{\frac{dk}2}\max_{x\in X^*} \abs{\nabla\varphi_{\bold s}(x)}
\le L^{-\frac{dk}2}(2^{d+2}-3)^{-d}\abs{\bold s} 
 =L^{-\frac{dk}2}(2^{d+2}-3)^{-d}\abs{\bold s}\le   h \bnorm{\varphi}_{k,B^*}
\end{multline}
yielding
\begin{multline}
\label{E:Qxix}
\abs{Q_\xi(\varphi,\varphi)}=\abs{Q_\xi(\varphi_{\bold s},\varphi_{\bold s})}\le 
\sum_{\heap{X\in\Scal_k}{X\supset B}}\frac{1}{|X|_k} \abs{\tfrac12 D^2 K(X,\xi)(\varphi_{\bold s},\varphi_{\bold s})}\le \\ \le\sum_{\heap{X\in\Scal_k}{X\supset B}}\frac{1}{|X|_k} \abs{K(X,\xi)}^{k,X,r} 
\bnorm{\varphi_{\bold s}}_{k,X}^2
 \le  
\sum_{\heap{X\in\Scal_k}{X\supset B}}\frac{1}{|X|_k} \abs{K(X,\xi)}^{k,X,r} 
\bnorm{\varphi}_{k,B^*}^2.
\end{multline}

Validity of \eqref{E:P2} for all $\varphi$, implies $\abs{P_2(\varphi,\psi)} \le  2^{2^d+2} (3^d-1)^{2^d}\norm{K}_{k,r}\bnorm{\varphi}_{k,B^*}\bnorm{\psi}_{k,B^*}$ for all $\varphi$ and $\psi$.
Taking now into account that $P_2(\varphi_1,\varphi_1)=0$ for any affine function $\varphi_1$, we  rewrite $P_2(\varphi,\varphi)=2P_2(\varphi,\varphi-\varphi_1)-P_2(\varphi-\varphi_1,\varphi-\varphi_1)$ to get
\begin{equation}
\label{T2est4}
\abs{P_2(\varphi,\varphi)}
\le   
 2^{2^d+1} (3^d-1)^{2^d}\norm{K}_{k,r}^{(\mathsf{A})} \bnorm{\varphi- \varphi_1}_{k,B^*} (4\bnorm{\varphi}_{k,B^*} +\bnorm{\varphi- \varphi_1}_{k,B^*}).
\end{equation}
Applying further  \eqref{T2est2}, we get
\begin{equation}
\label{T2est6}
\abs{P_2(\varphi,\varphi)} \le \bigl(4 L^{-(d+1)} +L^{-(d+2)} \bigr) 2^{2^d+1} (3^d-1)^{2^d}\norm{K}_{k,r}^{(\mathsf{A})}\bnorm{ \varphi}_{k+1,B^*}^2.
\end{equation}
Finally, combining \eqref{difflinearterm}  and \eqref{T2est6},  we get
\begin{multline}
\label{E:|R(B)|}
\abs{\bigl(T_2-\Pi_2\bigr) R(B,\varphi)}\le\\
\le 2^{2^d} (3^d-1)^{2^d} \bigl((2^{d+2}-1) L^{-(\frac{d}{2}+2)}+(8 L^{-(d+1)} +2L^{-(d+2)} )\bnorm{ \varphi}_{k+1,B^*}\bigr)\bnorm{ \varphi}_{k+1,B^*}\norm{K}_{k,r}^{(\mathsf{A})}.
\end{multline}

For the first and second  the derivatives, we first notice that 
\begin{equation}
D\bigl(P_1(\varphi) + P_2(\varphi,\varphi)\bigr)(\dot{\varphi})=P_1(\dot{\varphi})+2P_2(\varphi,\dot{\varphi})
\end{equation}
and
\begin{equation}
D^2\bigl(P_1(\varphi) + P_2(\varphi,\varphi)\bigr)(\dot{\varphi}, \dot{\varphi})=2P_2(\dot{\varphi},\dot{\varphi}) 
\end{equation}
yielding with the help of \eqref{difflinearterm} and  \eqref{T2est6} 
\begin{multline}
\bigl|D\bigl(P_1(\varphi) + P_2(\varphi,\varphi)\bigr)\bigr|^{k+1,B^*}\le \\
\le 2^{2^d} (3^d-1)^{2^d} \bigl( (2^{d+2}-1) L^{-(\frac{d}{2}+2)}+ (16 L^{-(d+1)}+4L^{-(d+2)}) \bnorm{ \varphi}_{k+1,B^*}\bigr)\norm{K}_{k,r}^{(\mathsf{A})}
\end{multline}        
and, using again  \eqref{T2est6},
\begin{equation}
\begin{aligned}
\bigl|D^2\bigl(P_1(\varphi) + P_2(\varphi,\varphi)\bigr)\bigr|^{k+1,B^*}\le 2^{2^d} (3^d-1)^{2^d}  (8 L^{-(d+1)}+2L^{-(d+2)}) \norm{K}_{k,r}^{(\mathsf{A})}.
\end{aligned}
\end{equation}
Combining last two inequalities with \eqref{E:|R(B)|}, we get
\begin{multline}
\label{E:|R(B)|^{}}
\bnorm{\bigl(T_2-\Pi_2\bigr) R(B,\varphi)}^{k+1, B^*,r}
\le 2^{2^d} (3^d-1)^{2^d} \bigl((2^{d+2}-1) L^{-(\frac{d}{2}+2)}+\\
+(8 L^{-(d+1)}+2L^{-(d+2)}) (1+\bnorm{ \varphi}_{k+1,B^*} )\bigr)(1+\bnorm{ \varphi}_{k+1,B^*})\norm{K}_{k,r}^{(\mathsf{A})}.
\end{multline}
With $(1+u)^2\le 2 {\rm e}^{u^2}$ and \eqref{E:boundon|phi|},
we get
\begin{multline}
\label{E:G2bound-pr}
\norm{G_2(U)}_{k+1,U,r} \le\\
\le  2^{2^d+1} (3^d-1)^{2^d} (2L)^d \bigl((2^{d+2}-1) L^{-(\frac{d}{2}+2)}+(8 L^{-(d+1)}+2L^{-(d+2)}) \bigr)\norm{K}_{k,r}^{(\mathsf{A})}
\end{multline}
yielding the sought bound.
\qed

The proof of Lemma~\ref{L:contraction} is the finished by combining the claims of Lemma~\ref{L:large_set} and Lemma~\ref{L:small_set}.\qed

\section{Bounds on the operators ${\bA^{(\bq)}}^{-1}$ and $\bB^{(\bq)}$}
\label{S:proof of Tlin}

The bounds on operators $\bA^{-1}$ and $\bB$ are rather straightforward.

\begin{lemma}
\label{L:ABcontraction}
Let $\theta \in(\frac14,\frac34)$ and $\o\ge 2 (d^2 2^{2d+1} +1)$. Consider the  constant $h_1=h_1(d,\o)$, $\upkappa(d)$, $\mathsf{A}_0=\mathsf{A}_0(d,L)$ as chosen from Lemma~\ref{L:contraction}.
Then there exists $L_0(d)$ such that 
\begin{equation}
\norm{\bA^{\ssup{\bq}^{-1}}}_{0;0}\le\frac{1}{\sqrt{\theta}}
\end{equation}
and there exists $M=M(d)$ such that
\begin{equation}
\norm{\bB^{\ssup{\bq}}}_{r;0}\le M L^d
\end{equation}
for any $\norm{\bq}\le\tfrac12$, any $N\in\N$, $k=1,\dots, N$, $r=1,\dots,r_0$, and any $L\ge L_0$,  $h\ge L^{\upkappa}h_1$, and $\mathsf{A}\ge \mathsf{A}_0$.
\end{lemma}

\begin{proofsect}{Proof}

When expressed in the coordinates $ \dot{\lambda}, \dot{a},\dot{\bc},\dot{\bd} $ of $\dot{H}$, the linear map 
$\bA$ according to \eqref{E:A}
keeps $ \dot{a},\dot{\bc}$, and $\dot{\bd} $ unchanged and only shifts $\dot{\lambda}$ by 
$$
\frac12\sum_{x\in B}
\sum_{i,j=1}^d \dot{d}_{i,j} \nabla_i \nabla_j^*{\Ccal}^{(\bq)}_{k+1}(0).
$$
Hence, $\bA^{-1}$ only makes the opposite shift and thus 
\begin{multline}
\norm{\bA^{-1}\dot{H}}_{k,0}=\\=L^{dk}|\dot{\lambda}|+L^{\frac{dk}{2}}h\sum_{i=1}^d|\dot{a}_i|+L^{\frac{(d-2)}{2}k}h\sum_{i,j=1}^d|\dot{\bc}_{i,j}|+\frac{h^{2}}2\sum_{i,j=1}^d|\dot{\bd}_{i,j}| \\+\frac{L^{dk}}2\sum_{i,j=1}^d|\dot{\bd}_{i,j}|\big|\nabla_i\nabla_j^*{\Ccal}^{(\bq)}_{k+1}(0)\big|.
\end{multline}
Using 
\begin{equation}\label{matrixdest}
\frac12\sum_{i,j=1}^d|\dot{\bd}_{i,j}|\le \frac{1}{h^2}\norm{\dot{H}}_{k,0},
\end{equation}
we get
$$
\norm{\bA^{-1}\dot{H}}_{k,0}\le (1+c_{2,0}L^{\eta(d)}h^{-2})\norm{\dot{H}}_{k+1,0}
$$
using that $\max_{i,j=1}^d\bigl|\nabla_i\nabla_j^*{\Ccal}^{(\bq)}_{k+1}(0)\bigr|\le c_{2,0} L^{-kd}L^{\eta(d)}$ according to   Proposition~\ref{P:FRD}. Given that $h^2\ge L^{2\upkappa(d)}= L^{\eta(d)+d}$ we can get
\begin{equation}
1+c_{2,0}L^{\eta(d)}h^{-2}\le 1+c_{2,0}L^{-d}\le \theta^{-1/2}
\end{equation}
 once $L > \bigl(\frac{2c_{2,0} }{\log 4}\bigr)^{1/d}$.

For the second bound, using Lemma~\ref{L:Taylor}, the first inequality of \eqref{E:k:k+1<k+1} and Lemma~\ref{L:prop}(iv),
\begin{multline}
\norm{\bB K}_{k+1,0}\le \sum_{B\in\Bcal_k(B^\prime)}\bigl\Vert\Pi_2\sum_{\heap{X\in\Scal_k,}{X\supset B}}\frac{1}{|X|_k}(\bR_{k+1}K)(X)\bigr\Vert_{k+1,0}\le\\
\le \sum_{B\in\Bcal_k(B^\prime)}C\sum_{\heap{X\in\Scal_k,}{X\supset B}}\frac{1}{|X|_k}\norm{(\bR_{k+1}K)(X)}_{k:k+1,X,r}
\\ \le \sum_{B\in\Bcal_k(B^\prime)}\sum_{\heap{X\in\Scal_k,}{X\supset B}}\frac{C2^{|X|_k}}{|X|_k} \norm{K(X)}_{k,X,r}\le\\
\le\sum_{B\in\Bcal_k(B^\prime)}\sum_{\heap{X\in\Scal_k,}{X\supset B}}\frac{C2^{|X|_k}}{|X|_k} \norm{K_k}_k^{(\mathsf{A})}\le  L^d M \norm{K_k}_k^{(\mathsf{A})},
\end{multline}
for any $ B^\prime\in\Bcal_{k+1} $. 
Here the factor $L^d$ comes from the number of blocks $B\in\Bcal_k(B^\prime)$ and we included into $M=M(d)$ the constant $C=C(d)$ as well as  the bound on the number of short polymers containing a fixed block.
\qed
\end{proofsect}

\medskip

\noindent Lemma~\ref{L:contraction} in conjunction with the estimates above give the estimates \eqref{E:boundsPropPT} in Proposition~\ref{P:Tlin}.

\begin{proofsect}{Proof of Remark~\ref{R:Tnonlin}}

The smoothness of the operators with respect to the fine tuning parameter $ \bq $ follows for $ \bB^{\ssup{\bq}} $ and $ \bC^{\ssup{\bq}} $ with the corresponding bounds in Chapter~\ref{S:smooth}  and for $ \bA^{\ssup{\bq}} $ from the regularity of the finite range decomposition \eqref{E:fluctk}, i.e., \eqref{E:boundsOperatorswrt_q} follows with $ C=C(d,h,L,\o)>0 $ and $ r\ge 2\ell+3 $ and all $ \norm{\bq}\le \frac{1}{2} $.  \qed
\end{proofsect}

%% file: AKM-ch8-tuning-30thJune-2016.tex
\chapter{Fine Tuning of the Initial Conditions}\label{S:Final} 
Finally, we address the fine tuning Theorem~\ref{T:tildeq}.
First, in Section~\ref{S:hatZ}, we prove the smoothness  of  the map $\Fcal $ assigning a fixed point of the renormalisation map $\cbT$ to initial values $\Hcal$ and $\Kcal$.
Then we can specify the map $\Hscr $ that chooses the initial ideal Hamiltonian $\Hcal$ in a self-consistent way so that it is reproduced in the first component $H_0$ of $\Fcal $.  Its properties summarized in Theorem~\ref{T:tildeq}
are proven in Section~\ref{S:tildeq}.

\section{Properties of the map $\Fcal $}
\label{S:hatZ}

Considering the space $\bE$ with the norm $\norm{\cdot}_\zeta$ with $\zeta>0$  as defined in \eqref{E:||h} and the Banach space $\bY_{\!\!r}$ introduced in \eqref{E:Zr}
and \eqref{E:normZr},
we find a map  $\Fcal $ from a neighbourhood of origin in  $\bE\times \bM_0$ (with a shorthand $\bM_0=M_0(\Bcal_0,\cbX)$) to $ \bY_{\!\!r}$
so that $\cbT(\Fcal (\Kcal,\Hcal),\Kcal,\Hcal)=\Fcal (\Kcal,\Hcal)$ with the following smoothness properties.
\begin{prop}
\label{P:hatZ}
Let  $d=2,3$, $\omega\ge 2(d^2 2^{2d+1}+1)$,   $r_0\ge  9$, and $2m+2\le r_0$  be fixed and let
$L_0$, $h_0(L)$,  $\mathsf{A}_0(L)$, $ M>0 $ (see \eqref{E:normZr}), and $ \theta\in (1/4,3/4) $ be the constants from Propositions~\ref{P:Tnonlin} and \ref{P:Tlin}. Then there exist constants $\alpha=\alpha(M,\theta)\ge 1$ and $\eta=\eta(\theta) \in (0,1)$  determining the norm of the spaces $\bY_{\!\!r}$, $r=r_0,r_0-2,\dots,r_0-2m$  and, for any $L\ge L_0$, $h\ge h_0(L)$, and  $\mathsf{A}\ge \mathsf{A}_0(L)$, a constant $\zeta=\zeta(h)$ determining the norm $\norm{\cdot}_\zeta$ on $\bE$ and  constants $\widehat \rho, \widehat\rho_1,\widehat\rho_2>0$  so that there exists a unique function $\Fcal \colon B_{\bE\times\bM_0}(\widehat\rho_1,\widehat\rho_2)\to B_{\bY_{\!\!r_0}}(\widehat \rho)$ solving the equation  $\cbT(\Fcal (\Kcal,\Hcal),\Kcal,\Hcal)=\Fcal (\Kcal,\Hcal)$ (see \eqref{E:TKqZ}).
Moreover, 
\begin{equation}
\label{E:FcalinC*}
\Fcal\in \widetilde C^m(B_{\bE\times \bM_0}(\widehat\rho_1,\widehat\rho_2),\bY)
\end{equation}
with bounds on derivatives that are uniform in $N$, i.e.,
there is $ \widehat{C}$ such that
\begin{equation}
\label{E:estsupderiv}
\norm{D^j_{\Kcal}D^{\ell}_{\Hcal}\Fcal (\Kcal,\Hcal)(\dot{\Kcal},\dots,\dot{\Kcal},  \dot{\Hcal},\dots, \dot{\Hcal})}_{\bY_{\!\!r_0-2\ell }}\le \widehat{C}\norm{\dot{\Hcal}}^{\ell}_{0}\norm{\dot{\Kcal}}_{\zeta}^j,
\end{equation}
for all $ (\Kcal,\Hcal)\in B_{\bE\times \bM_0}(\widehat\rho_1,\widehat\rho_2) $ and all $ \ell,j\in \N_0$  with $ \ell+j\le n \le m$.
\end{prop}
The proof of Proposition~\ref{P:hatZ} is  based on Theorem~\ref{T:implicit} applied in conjunction with Propositions~\ref{P:Tnonlin}  and~\ref{P:Tlin}. Here, the map $\cbT: \bY\times\bE\times M_0  \to\bY$ plays the role of the map $F$ and  the sequence of  spaces $\bY=\bY_{\!\!r_0}\embed \bY_{\!\!r_0-2}\embed\dots \embed \bY_{\!\!r_0-2m}$, $2m<r_0$,  the role of the sequence $\bX_n$, $n=m,m-1,\dots,0$. Using 
$ \Ocal_\rho:= B_{\bY}(\rho)$, 
$\Wcal_{\rho}:=B_{\bE}(\rho)=\{\Kcal\in \bE: \norm{\Kcal}_\zeta\le \rho\}$, 
and $\Vcal_{\rho}:=\{\Hcal\in \bM_0:  \norm{\Hcal}_0\le\rho\}$,
we just have to verify the assumptions of Theorem~\ref{T:implicit}, that is we need to prove the  following claim.
\begin{lemma} 
\label{L:verifyimplicit} 
Let $L, h$, and $\mathsf{A}$ be constants as in Proposition~\ref{P:hatZ} and  let $\theta\in(1/4,3/4)$ and $M>0$ be the constants from
Proposition~\ref{P:Tlin}. Then there exist parameters $\alpha$ and $\eta$ of the norms in $\bY_{\!\!r}$ depending only on $\theta$ and $M$,  constants  $\rho>0$, and $\zeta$ depending on $h$ and $\mathsf{A}$, so that:

\noindent
(i) \  $ \cbT\in \widetilde C^m(\Ocal_\rho\times\Wcal_\rho\times\Vcal_\rho,\bY)$ with the bounds on corresponding derivatives that are uniform in $N$,

\noindent
(ii) \  $ \cbT(0,0,\Hcal)=0 $ for all $ \Hcal\in \Vcal_{\rho} $, and

\noindent
(iii) \
$\norm*{D_1 \cbT(\by,0, \Hcal) \bigr\vert_{\by=0}}_{\Lcal(\bY_{\!\!r},\bY_{\!\!r})}\le \theta$
for all $ \Hcal\in \Vcal_{\rho} $ and $r=r_0,r_0-2, \dots, r_0-2m$.
\end{lemma}

\begin{proofsect}{Proof}
Let us recall the definition of the map $ \cbT$. The $2N$ coordinates of the image
\begin{equation}
\cbT(\by,\Kcal,\Hcal)=  \overline  \by=(\overline H_0,\overline H_1,\overline K_1,\dots,\overline H_{N-1},\overline K_{N-1},\overline K_N)
\end{equation}
are defined  by
\begin{equation}
\begin{aligned}
&\overline H_k=\bigl(\bA_k^{(\Hcal)}\bigr)^{-1}\bigl(H_{k+1}-{\bB}^{(\Hcal)}_{k}K_k\bigr)\ \text{ and }       \\
& \overline K_{k+1}=S_k(H_k,K_k,\Hcal),
\end{aligned}
\end{equation}
where  we set  $H_{N}=0$ and 
\begin{equation}
\label{E:KKcal}
K_0(X,\varphi):=
\exp\Bigl\{-\sum_{x\in X}\Hcal(x,\varphi)\Bigr\} \prod_{x\in X}\Kcal(\nabla\varphi(x)) 
\end{equation}
with  $ \Kcal\in\bE $.  
Notice that 
$\bA_{k}^{\Hcal}, \bB^{\Hcal}_{k}$, and $S_k(H_k,K_k,\Hcal)$ depend on $\Hcal$ only through  the coefficient of its quadratic term 
$\bq=\bq(\Hcal)$.
We will also use a shorthand
\begin{equation}
\label{E:KKcalq}
K_0(X,\varphi)=:K_0^{\ssup{\Kcal,\Hcal}}(X,\varphi)=
 \prod_{x\in X}\Kcal_{0}^{(\Kcal,\Hcal)}(x,\varphi)
\end{equation}
with 
\begin{equation}
\Kcal_{0}^{(\Kcal,\Hcal)}(x,\varphi) =\exp\bigl\{-\Hcal(x,\varphi)\bigr\}\Kcal(\nabla\varphi(x)).
\end{equation}
Here we explicitly invoke the dependence of the map $S_k$ on $k$ in contradistinction to Chapter~\ref{S:smooth}, where the index $k$ was omitted.
Notice that the only two coordinates of $\overline \by$ that depend on $\Kcal$ (through $K_0$) are 
$\overline H_0=\bigl(\bA_0^{(\Hcal)}\bigr)^{-1}\bigl(H_{1}-{\bB}^{(\Hcal)}_{0}K_0\bigr)$ and $\overline K_{1}=S_0(H_0,K_0,\Hcal)$.
\smallskip

\noindent (i)
The fact that  $ \cbT\in \widetilde C^m(\Ocal_\rho\times\Wcal_\rho\times\Vcal_\rho,\bY) $
 follows from Propositions~\ref{P:Tnonlin}  and \ref{P:Tlin}. We will treat separately the coordinates $\overline K_{k+1}$, $k=1,2,\dots, N-1$, the coordinates 
 $\overline H_{k}$, $k=1,2,\dots, N-1$, and finally, the coordinates $\overline H_{0}$ and $\overline K_{1}$ that depend on   $\Kcal$.
 
 Reinstating  the dependence on $k$, we denote more explicitly the sequence of normed spaces $\bM_{k,r}=\{M(\Pcal^{\com}_k, \cbX): \norm{\cdot}_{k,r}^{(\mathsf{A})}<\infty\}$, $r=r_0,r_0-2,\dots,r_0-2m$,  
 as well as $\bM_{k,0}=(M_0(\Bcal_k, \cbX), \norm{\cdot}_{k,0})$. Then the claim of Proposition~\ref{P:Tnonlin} is that the maping $S_k: \Ucal_{k,\rho}\times\Vcal_{1/2}\to  \bM_{k+1}=\bM_{k+1,r_0}$
 belongs to $\widetilde C^m(\Ucal_{k,\rho}\times\Vcal_{1/2}, \bM_{k+1})$ for all $k=1,2,\dots, N-1$.
 Here, 
 $$\Ucal_{k,\rho}=\{(H,K) \in \bM_{k,0}\times  \bM_{k,r_0} \colon
 \norm{ H}_{k,0} < {\rho}, \norm{ K}_{k,r_0}^{(\mathsf{A})} < {\rho} \}$$

 For the coordinates $\overline H_{k}$, $k=1,2,\dots, N-1$, we first observe that the defining map
$ \overline H_k= (\bA_{k}^{(\Hcal)})^{-1}\bigl(H_{k+1}-{\bB}_{k}^{(\Hcal)}K_k\bigr)$ is linear in $H_{k+1}$ and $K_k$
and that it does not depend on $\Kcal$.
Consider thus the map
 \begin{equation}
G:(\by,\Hcal)\mapsto (\bA_{k}^{(\Hcal)})^{-1}\bigl(H_{k+1}-{\bB}_{k}^{(\Hcal)}K_k\bigr)
\end{equation}
 and verify that $G\in \widetilde C^m(\bY\times\Vcal_{\rho},\bM_{k,0})$. 

First, we will address the smoothness  of the term ${\bB}_{k}^{(\Hcal)}K_k$.
 Comparing the formula \eqref{E:B} with \eqref{E:defR2}, we see that
\begin{equation}
{\bB}^{(\Hcal)}_{k}K_k(B^{\prime},\varphi)= -\sum_{B\in\Bcal(B^{\prime})} R_2(0,K_k,\bq(\Hcal)),
\end{equation}
obtaining the needed smoothness relying on the fact that $R_2\in \widetilde C^m(\Ucal_{k,\rho}\times\Vcal_{\rho},\bM_{k,0})$ (see Lemma~\ref{L:R_2}) and the fact that the projection $\Hcal\mapsto \bq(\Hcal)$ is a  linear mapping.
%For proving that
%$G\in \widetilde C^m(\bY\times\Vcal_{\rho})$,  we 
%can apply Lemma~\ref{L:linear}.
%Hence, it suffices  to verify two conditions: 
%
% \noindent
% (a) for any $\by\in\bY$, the mapping $\bq\mapsto G(\by,\bq)$ is in $C_*^\ell(\Vcal_\rho,\bY_{\!\!r_0-2m+2\ell})$
% for any $\ell\le m$ and
% 
% \noindent
% (b)  for any $\bq_0\in\Vcal_\rho$, there exist constants $\delta,C>0$ such that 
% \begin{equation}
%\norm{D_{\bq}^\ell G((\by,\bq),\dot{\bq}^\ell)}_{\bY_{\!\!r}}\le C \max(\norm{H_{k+1}}_{k+1,0} ,\norm{K_{k}}_{k,r-2\ell})\norm{\dot{\bq}}^{\ell}
%\end{equation}
%for any $0\le\ell\le m, r=r_0, \dots, r_0-2m+2\ell$, and any $(\by,\bq,\dot{\bq})\in \bY\times B_\delta(\bq_0)\times \R_{\rm sym}^{d\times d}$.
% 
% Analyzing the map $G(\by,\bq)= (\bA_{k}^{(\bq)})^{-1}\bigl(H_{k+1}-{\bB}_{k}^{(\bq)}K_k\bigr)$, we first recall that the operator 
% $(\bA_{k}^{(\bq)})^{-1}$ has a very simple explicit expression. 

Denoting $H=H_{k+1}-{\bB}_{k}^{(\Hcal)}K_k\in\bM_{k+1,0}$ and rewriting it
  in terms of the coordinates $ \lambda, a,\bc,\bd $ we see that the linear operator 
$(\bA_{k}^{(\Hcal)})^{-1}$   only shifts the coordinate $\lambda$ by 
\begin{equation}
\label{E:Ashift}
-\frac12\sum_{x\in B}
\sum_{i,j=1}^d \bd_{i,j} \nabla_i \nabla_j^*{\Ccal}^{(\bq(\Hcal))}_{k+1}(0),
\end{equation}
keeping the other coordinates unchanged (cf. the proof of Lemma~\ref{L:ABcontraction}).
The derivatives of this shift can be estimated  by finite range decomposition bound
\eqref{E:fluctk}  yielding
\begin{equation}
\sup_{\norm{\Hcal}_{0}\le \frac{1}{2}}\bigl\vert( D^{\ell}\nabla_i \nabla_j^*{\mathcal C}^{(\bq(\Hcal))}_{k+1})(0)(\dot{\Hcal},\dots,\dot{\Hcal})\bigr\vert\le c_{2,\ell} L^{-kd}
L^{\upeta(2,d)}\norm{\dot{\Hcal}}_{0}^\ell
\end{equation}
where we used that
\begin{equation}
\label{E:matrixdest}
\frac12\sum_{i,j=1}^d|\bd_{i,j}|\le \frac{1}{h^2}\norm{H}_{k+1,0}
\end{equation}
according to \eqref{normideal}. 
Hence
\begin{equation}\label{estAinverse}
\begin{aligned}
\norm{D^{\ell}((\bA_{k}^{(\bq(\Hcal))})^{-1}H)(\dot{\Hcal},\dots,\dot{\Hcal}) }_{k,0}&=\norm{D^{\ell} G(\by,\Hcal)(\dot{\Hcal},\ldots,\dot{\Hcal})}_{k,0}\\
&\le c_{2,\ell}L^{\upeta(2,d)}h^{-2}\norm{H}_{k+1,0}\norm{\dot{\Hcal}}_{0}^\ell,
\end{aligned}
\end{equation}
for $ \norm{\Hcal}_0\le \frac{1}{2} $ and $ \by\in\bY $.  Actually, in \cite{AKM09b} it is shown that  $\nabla_i \nabla_j^*{\Ccal}^{(\bq)}_{k+1}(0)$ is analytic in $\bq$.

\medskip

%As a result to verify the conditions (a) and (b) it remains to study the 
%  

%In addition, the bounds \eqref{E:boundsOperatorswrt_q} from Remark~\ref{R:Tnonlin} on derivatives of the operators ${\bA}_{k}^{(\bq)}$, ${\bB}_{k}^{(\bq)}$, and ${\bC}_{k}^{(\bq)}$ with respect to $\bq$  are also valid.

%Choosing thus $ \rho_3=\rho $ and $\rho_2<1/2$,  we can immediately  conclude that the $ 2(N-1) $ coordinates with no dependence on $\Kcal$ are $C^3-$maps from $ B_{\R_{\rm sym}^{d\times d}}(\rho_2)\times B_{\bZ_{r}}(\rho_3)$ to the corresponding subspace of $\bZ_{r-6}$ with the norm
% \eqref{E:normZr} corresponding to a particular coordinate. In addition, these maps satisfy the bounds (iv) with the continuity (v) (there is no dependence on $\Kcal$). 

Finally, we consider the coordinates $ \overline{H}_0$ and $ \overline{K}_1$.
Their derivatives with respect to $\Kcal$ have to be evaluated by composing the derivatives of $ \overline{H}_0$ and $ \overline{K}_1$ with respect to $K_0$ with the derivatives of $K_0$ with respect to $\Kcal$. 
%Let us first attend to the coordinate $\overline K_{1}=S(H_0,K_0,\bq)$. 
%Using $\bM_{0,0}$ and $\bM_{1,r}$ for the subspace of $\bY_{\!\!r}$  corresponding to the coordinate $H_0$ and $K_1$, respectively, and 
We first deal with the coordinate $ \overline{K}_1 $ which can be viewed as a composition of maps

\begin{equation}
F:\bM_{0}\times \bE\times \bM_0\to \bM_{0,0}\times\bM_{0,r_0} \text{ and }
S_0:(\bM_{0}\times\bM_{0,r_0})\times \bM_0\to \bM_{1,r_0}.
\end{equation}
Indeed, with 
\begin{equation}
F(H_0,\Kcal,\Hcal)=(H_0,K_0^{(\Kcal,\Hcal)})  
\end{equation}
we get
\begin{equation}
 \overline{K}_1=  S_0\diamond  F, \ \ \text{i.e.,}\ \  \overline{K}_1(H_0,\Kcal,\Hcal)=S_0(F(H_0,\Kcal,\Hcal),\Hcal).
\end{equation}
Here, $K_0^{(\Kcal,\Hcal)}$ is the polymer
defined in \eqref{E:KKcalq}, where we explicitly
denoted the dependence on $\Kcal$ and $\Hcal$.

Now, we apply the Chain Rule according to Theorem~\ref{T:fullchain} jointly with Remark~\ref{D:remDerivative} providing bounds on derivatives that are uniform in $N$. The needed condition  $S_0\in \widetilde C^m(\Ucal_{0,\rho}\times\Vcal_{1/2},\bM_{1})$ is just the corresponding claim \eqref{E:SinC}  from Proposition~\ref{P:Tnonlin}. 
For the map $F$, there is no grading on the domain space $\bM_{0}\times \bE\times \bM_0$, and we will actually show  that
$F\in C_*^m(\Ucal_{0,\rho}\times \Wcal_\rho\times\Vcal_\rho, \bM_{0}\times\bM_{0,r_0})$.  Indeed, choosing a suitable parameter $\zeta$ and $\rho$,
both depending on $h$,  we will prove that 
the derivative $D^j D^\ell K_0^{(\Kcal,\Hcal)}(\dot{\Kcal}^j, \dot{\Hcal}^\ell)$ exists and 
\begin{equation}
\label{E:derKKq}
\norm*{D^j D^\ell K_0^{(\Kcal,\Hcal)}(\dot{\Kcal}^j, \dot{\Hcal}^\ell)}_{0,r}\le C_1 \norm{\dot{\Kcal}}_\zeta^j \norm{\dot{\Hcal}}_{0}^\ell
\end{equation}
for any $j,\ell\le m+1$ with $C_1=C_1(h, \mathsf{A}, m)$, and thus also
\begin{equation}
\label{E:limK'toK}
\lim_{(\Kcal',\Hcal')\to (\Kcal,\Hcal)}\Bigl\Vert D^j D^\ell K_0^{(\Kcal,\Hcal)}(\dot{\Kcal}^j, \dot{\Hcal}^\ell)
-D^j D^\ell K_0^{(\Kcal,\Hcal)}(\dot{\Kcal}^j, \dot{\Hcal}^\ell)\Bigr\Vert_{0,r}=0
\end{equation}
for any $j,\ell \le m$ and  any $(H_0,\Kcal,\Hcal)\in \Ucal_{0,\rho}\times\Wcal_\rho\times\Vcal_\rho$.

Indeed, in view of the product form in 
%\alteight{ \eqref{E:KKcal} and \eqref{E:KKcalq}, we have
%\begin{multline}
%\label{E:DKKq}
%D^j D^\ell K_0^{(\Kcal,\Hcal)}(\dot{\Kcal}^j, \dot{\Hcal}^\ell) =\exp\bigl\{-\sum_{x\in X}
%\Hcal(x,\varphi)\bigr\}\times\\
%\times 2^{-\ell}\sum_{I,\abs{I}=\ell}\prod_{y\in \supp I}
%(- \dot{\Hcal}(y, \varphi))^{k_{y}}  
%\sum_{\heap{Y\subset X}{\abs{Y}=j}} \prod_{y\in Y} \dot{\Kcal}(\nabla \varphi(y))\prod_{y\in X\setminus Y}\Kcal(\nabla\varphi(y))=\\=2^{-\ell}
% \sum_{I,\abs{I}=\ell}\sum_{\heap{Y\subset X}{\abs{Y}=j}}\prod_{y\in \supp I} 
% (-\dot{\Hcal}(y, \varphi))^{k_{y}}
% \prod_{y\in Y} \dot{\Kcal}^{(\Hcal)}(y,\varphi)
% \prod_{y\in X\setminus Y}\Kcal^{(\Hcal)}(y,\varphi).
%\end{multline}
%}
\eqref{E:KKcal} and \eqref{E:KKcalq}, we first have 
\begin{equation}
\begin{aligned}
D^\ell K_0(X,\varphi)(\dot{\Hcal},\ldots,\dot{\Hcal})&=\sum_{\heap{k\in\N_0^X\colon }{\sum_{x\in X}k_x=\ell}}\frac{(-1)^\ell\ell!}{\prod_{x\in X}k_x!}\prod_{x\in X}\Big(\dot{\Hcal}(x,\varphi)^{k_x}{\rm e}^{-\Hcal(x,\varphi)}\Kcal(\nabla\varphi(x))\Big),
\end{aligned}
\end{equation}
and thus
\begin{equation}\label{combbound}
\begin{aligned}
%\label{E:DKKq}
& D^j D^\ell  K_0^{(\Kcal,\Hcal)}(\dot{\Kcal}^j, \dot{\Hcal}^\ell) =\sum_{\heap{k\in\N_0^X\colon}{ \sum_{x\in X}k_x=\ell}} \sum_{\heap{Y\subset X}{\abs{Y}=j}}  \frac{(-1)^\ell\ell!}{\prod_{x\in X}k_x!}\prod_{x\in X}\Big(\dot{\Hcal}(x,\varphi)^{k_x}{\rm e}^{-\Hcal(x,\varphi)}\Big)\\
&\quad \times  \prod_{y\in Y} \dot{\Kcal}(\nabla \varphi(y))\prod_{y\in X\setminus Y}\Kcal(\nabla\varphi(y))\\
&=\sum_{\heap{k\in\N_0^X\colon }{\sum_{x\in X}k_x=\ell}} \sum_{\heap{Y\subset X}{\abs{Y}=j}}  \frac{(-1)^\ell\ell!}{\prod_{x\in X}k_x!}\prod_{x\in X} \Big(\dot{\Hcal}(x,\varphi)^{k_x}\Big)
\prod_{x\in Y} {\dot{\Kcal}}_{0}^{(\Kcal,\Hcal)}(x,\varphi)
 \prod_{x\in X\setminus Y}\Kcal_{0}^{(\Kcal,\Hcal)}(x,\varphi).
\end{aligned}
\end{equation}

Here, we use the shorthand $\dot{\Kcal}_{0}^{(\Kcal,\Hcal)}(x,\varphi)=\exp\bigl\{-\Hcal(x,\varphi)\bigr\} \dot{\Kcal}(\nabla \varphi(x))$.
Observing  that, in the case $k=0$, the unit blocks are actually single sites, $\Bcal_k(\L_N)=\L_N$, we can  apply the claim (iia) of Lemma~\ref{L:prop} to get
\begin{multline}\label{productKcal}
\Bigl\Vert  \prod_{y\in Y}   \dot{\Kcal}_{0}^{(\Hcal, \dot{ \Hcal},k_y)}(y,\varphi)\prod_{y\in X\setminus Y}
\Kcal_{0}^{(\Hcal, \dot{ \Hcal},k_{y})}(y,\varphi)\Bigr\Vert_{0,X,r}\\\
\le  
\prod_{y\in Y}  \tnorm{ \dot{\Kcal}_{0}^{(\Hcal, \dot{ \Hcal},k_{y})}}_{0,\{y\}} \prod_{y\in X\setminus Y}
\tnorm{\Kcal_{0}^{(\Hcal, \dot{ \Hcal},k_{y})}}_{0,\{y\}}.
\end{multline}
Here we introduced the shorthands
\begin{equation}
\Kcal_{0}^{(\Hcal, \dot{ \Hcal},k_y)}(y,\varphi)=   - \dot{\Hcal}(y, \varphi)^{k_y} \Kcal_{0}^{(\Kcal,\Hcal)}(y,\varphi)
\end{equation}
 and 
 \begin{equation}
 \dot{\Kcal}_{0}^{(\Hcal, \dot{ \Hcal},k_y)}(y,\varphi)=   - \dot{\Hcal}(y, \varphi)^{k_y} \dot{\Kcal}_{0}^{(\Kcal,\Hcal)}(y,\varphi).
 \end{equation}
Further, using definitions \eqref{E:tnorm} and \eqref{E:bnormup}, 
\begin{equation}
\tnorm{\Kcal_{0}^{(\Hcal, \dot{ \Hcal},k_{y}y)}}_{0,\{y\}}=\sup_{\varphi} \bnorm{\Kcal_{0}^{(\Hcal, \dot{ \Hcal},k_y)}(y,\varphi)}_{0,\{y\},r_0}
\exp\{-G_{0,y}(\varphi)\}
\end{equation}
with the weight function $G_{0,y}(\varphi)$ defined in \eqref{E:Gkxdef}
and
\begin{equation}
\bnorm{\Kcal_{0}^{(\Hcal, \dot{ \Hcal},k_{y})}(y,\varphi)}_{0,\{y\},r_0}
=\sum_{r=0}^{r_0}  \frac{1}{r !}  \sup_{\bnorm{\dot{\varphi}}_{0,\{y\}}\le 1}
  \bigl\vert D^r \Kcal_{0}^{(\Hcal, \dot{ \Hcal},k_{y})}(y,\varphi)(\dot{\varphi}, \dots, \dot{\varphi})\bigr\vert.
\end{equation}
Using the definition \eqref{E:normphikX}, we can bound
\begin{equation}
\bnorm{\dot{\varphi}}_{0,\{y\}}= \max_{1\le s \le3}\sup_{w\in \{y\}^*}
\frac1h \bigl\vert\nabla^s\dot{\varphi}(w)\bigr\vert\ge \max\bigl( \frac1h \abs{\nabla\dot{\varphi}(y)},\frac{1}{h}\abs{\nabla^2\dot{\varphi}(y)}\big).
\end{equation}
Now
\begin{equation}
 \sup_{\bnorm{\dot{\varphi}}_{0,\{y\}}\le 1}
  \bigl\vert D^r \Kcal_{0}^{(\Hcal, \dot{ \Hcal},k_{y})}(y,\varphi)(\dot{\varphi}, \dots, \dot{\varphi})\bigr\vert
  \le \sup_{\bnorm{\dot{\varphi}}_{0,\{y\}}\le 1}\abs*{\frac{\d^r \Kcal_{0}^{(\Hcal, \dot{ \Hcal},k_{y})}(y,\varphi+t\dot{\varphi})}{\d t^r}\Bigr|_{t=0}}.
\end{equation}
Defining $v=\nabla\varphi(y),\ w=\nabla^2 \varphi(y)$, and $z=\frac{1}{h}\big(|v|^2+|w|^2\big)^1/2 $ we notice that
$$
\frac{\d^r \Kcal_{0}^{(\Hcal, \dot{ \Hcal},k_{y})}(y,\varphi+t\dot{\varphi})}{\d t^r}
$$
is a sum of  terms
 of the form
\begin{multline}
(\dot{\lambda} +\dot{a} v +\tfrac12\langle \dot{ \bq}v, v\rangle+\dot{c} w)^{i_0}
(\dot{a}\dot{v}+\langle \dot{ \bq}v, \dot{v}\rangle+\dot{c}\dot{w})^{i_1}
\langle \dot{ \bq} \dot{v},  \dot{v}\rangle^{i_2}
(a\dot{v}+ \langle \bq  v,\dot{v}\rangle+z\dot{w})^{j_1} \langle \bq \dot{v}, \dot{v}\rangle^{j_2}\times\\
\times
\exp\{- (\lambda+av+\tfrac12\langle \bq v, v\rangle+cw)\} \frac{\d^s \Kcal(v+t\dot{v})}{\d t^s}\Bigr|_{t=0}
\end{multline}
such that $i_0+i_1+i_2 =k_{y}$ and $i_1+2i_2+j_1+2j_2+s=r$.
Using the definition of the norm $\norm{\Hcal}_{0}$ and the fact that 
$\frac1h\max(\abs{\dot{v}},\abs{\dot{w}})\le \bnorm{\dot{\varphi}}_{0,\{y\}}\le 1$, 
the absolute value of the prefactor above can  be bounded by 
$$
2^{i_1+i_2+j_1+j_2}\norm{\dot{\Hcal}}^{j_1+j_2}_{0}\big(1+z\big)^{2i_0+i_1+j_1}
$$

Now assume that
\begin{equation}\label{Hidealbound}
\norm{\Hcal}_{0}\le\widetilde{\rho}\le 1.
\end{equation}
Since $ k_{y}\le m+1 $ and $ j_1\le m+1 $ we have
\begin{equation}
(1+z)^{2i_0+i_1+j_1}\le (1+z)^{4(m+1)}\le \big(1+\frac{16(m+1)}{\widetilde{\rho}}\big)^{2(m+1)}\exp\{\widetilde{\rho}z^2\}.
\end{equation}
In the last inequality we used that for $ a>0 $, $ z\ge 0 $,
\begin{equation}
(1+z)^a\le\big(1+\frac{2a}{\widetilde{\rho}}\big)^{a/2}\exp\{\widetilde{\rho}z^2\}
\end{equation}
To see this observe that for $ a>0 $ the maximum of the function
\begin{equation}
t\mapsto (1+t)^a\exp\{-\widetilde{\rho}t^2\}
\end{equation}
for $ t\ge 0 $ is attained at 
$$t=\overline{t}=\frac{1}{2}\big(\sqrt{1+\frac{2a}{\widetilde{\rho}}}-1\big)
$$ and is bounded by
$$
(1+\overline{t})^a\le (1+2\overline{t})^a=\big(1+\frac{2a}{\widetilde{\rho}}\big)^{a/2}.
$$

As a result, 
there exists a constant $\overline C(r_0)$ so that for $ |\dot{\varphi}|\le 1 $ and hence $|\dot{v}|\le h $, we have 
\begin{equation}
\begin{aligned}
&\abs*{\frac{\d^r \Kcal_{0}^{(\Hcal, \dot{ \Hcal},k_{y})}(\varphi+t\dot{\varphi})}{\d t^r}\Bigr|_{t=0}}\le \overline C(r_0)  \bigl(1+\tfrac{16(m+1)}{\widetilde\rho}\bigr)^{2(m+1)}  \norm{\dot{\Hcal}}^{k_{y}}_{0} \times\\
&\times
\exp\{\widetilde\rho \abs{z}^2\}  \biggl(\sum_{s=0}^{r_0}\abs*{\frac{\d^s \Kcal(v+t\dot{v})}{\d t^s}\Bigr|_{t=0}}\biggr)\\
&\le \overline{C}(r_0)\big(1+\tfrac{16(m+1)}{\widetilde{\rho}}\big)^{2(m+1)}\norm{\dot{\Hcal}}_{0}^{k_{y}}\exp\{\widetilde{\rho}z^2\}\sum_{\abs{\mathbf{\alpha}}\le r_0}h^{\abs{\mathbf{\alpha}}}\vert\partial_v^{\alpha}\Kcal(v)\vert
\end{aligned}
\end{equation}
for any  $\norm{\Hcal}_{0}\le \widetilde\rho$, and any $r\le r_0$.
Finally, choosing 
\begin{equation}\label{zetabound}
\zeta\ge h
\end{equation}
 and taking into account  that 
\begin{equation}
G_{0,y}(\varphi)\ge \frac1{h^2} |\nabla\varphi(y)|^2+\frac{1}{h^2}|\nabla^2\varphi(y)|^2=z^2
\end{equation}
and the definition \eqref{E:||h} of the norm $ \norm{\Kcal}_\zeta $ and using $ \abs{v}\le hz $ we get 

\begin{equation}
\begin{aligned}
\tnorm{\Kcal_{0}^{\ssup{\Hcal,\dot{\Hcal},k_{y}}}}_{0,\{y\}}\le\widetilde{C}\norm{\dot{\Hcal}}_{0}^{k_{y}}\,\sup_{z\ge 0}\Big(\exp\{(\widetilde{\rho}-1)z^2\}\exp\{\zeta^{-2}h^2z^2\}\norm{\Kcal}_{\zeta}\Big)
\end{aligned}
\end{equation}
with 
\begin{equation}
\widetilde{C}=\widetilde{C}(r_0,m,h,\widetilde{\rho})=\overline{C}(r_0)\big(1+\frac{16(m+1)}{\widetilde{\rho}}\big)^{2(m+1)}.
\end{equation}
The same estimate holds for $ \dot{\Kcal}^{\ssup{\Hcal,\dot{\Hcal},k_{y}}} $ if we replace $ \norm{\Kcal}_\zeta $ on the right hand side by $ \norm{\dot{\Kcal}}_\zeta $. The exponential term can be controlled if for given $h$ we choose $ \zeta $ and $ \widetilde{\rho} $ such that
\begin{equation}\label{rhozetah}
\frac{h^2}{\zeta^2}+\widetilde{\rho}\le 1.
\end{equation}
In particular we may take 
\begin{equation}
\widetilde{\rho}=\frac{1}{2}\quad\mbox{and}\quad \zeta=\sqrt{2}h.
\end{equation}
Note that \eqref{rhozetah} implies \eqref{zetabound} and \eqref{Hidealbound}.

\medskip

Summarising, we get,
\begin{multline}\label{estKcalder}
\Bigl\Vert  \prod_{y\in Y}  \dot{\Kcal}_{0}^{(\Hcal, \dot{ \Hcal},k_{y})}(\nabla \varphi(y)) \prod_{y\in X\setminus Y}\Kcal_{0}^{(\Hcal, \dot{ \Hcal},k_{y})}(y,\varphi)\Bigr\Vert_{0,X,r}\le\\
\le {\widetilde C}^{\abs{X}} \norm{\dot{\Hcal}}^{\ell}_{0}
\norm{\dot{\Kcal}}_\zeta^j \norm{\Kcal}_\zeta^{\abs{X}-j}.  
\end{multline}

Since $ \ell\le m+1 $ the sum in \eqref{combbound}  over $k\in\N_0^X$ with $ \sum_{x\in X}k_x=\ell $ involves at most $ (m+2)^{\abs{X}} $ terms. The sum over $ Y$ involves at most $ 2^{\abs{X}} $ terms. The counting terms with the factorial in \eqref{combbound} are bounded by $ (m+1)! $. Thus \eqref{combbound} and \eqref{estKcalder} give
\begin{equation}
\begin{aligned}
\norm{D_1^jD_2^\ell & K_0(X,\Kcal,\Hcal,\dot{\Kcal},\ldots,\dot{\Kcal},\dot{\Hcal},\ldots,\dot{\Hcal})}_{0,r}\\
&\le (m+1)!(2(m+2))^{\abs{X}}\widetilde{C}^{\abs{X}}\norm{\Kcal}_{\zeta}^{X-j}\norm{\dot{\Kcal}}_{\zeta}^j\norm{\dot{\Hcal}}_{0}^\ell.
\end{aligned}
\end{equation}
Thus with $ \zeta=\sqrt{2}h $ we have for all $ \Kcal\in B_{\bE}(\rho_1) $ with 
\begin{equation}
 \rho_1=\rho_1(\mathsf{A})=\big(2(m+2)(m+1)!\mathsf{A}\widetilde{C}\big)^{-1}
\end{equation}
and all $ \Hcal\in B_{\bM_0}(\widetilde{\rho}) $ with $ \widetilde{\rho}=\frac{1}{2} $,
\begin{equation}
\begin{aligned}
\Gamma_{\mathsf{A}}(X)&  \norm{D_1^jD_2^\ell  K_0(X,\Kcal,\Hcal,\dot{\Kcal},\ldots,\dot{\Kcal},\dot{\Hcal},\ldots,\dot{\Hcal})}_{0,r}\\
&\le C_1\norm{\dot{\Kcal}}_{\zeta}^j\norm{\Hcal}_{0}^\ell
\end{aligned}
\end{equation}
with 
\begin{equation}
C_1=C_1(\mathsf{A},m)=\big((m+1)!(2(m+2)\widetilde{C}\mathsf{A}\big)^{m+1}.
\end{equation}

Finally, for the coordinate $\overline H_0=(\bA^{(\Hcal)})^{-1}_{0}\bigl(H_{1}-{\bB}^{(\Hcal)}_{0}K_0\bigr)$, we can again apply the Chain Rule according to Theorem~\ref{T:fullchain}. The image coordinate $\overline H_0$  is obtained as a composition of maps 
\begin{equation}
F:\bM_{1,0}\times\bE\times\bM_0\to \bM_{1,0}\times \bM_{0,r_0} \text{ and }
G:(\bM_{1,0}\times\bM_{0,r_0})\times\bM_0\to \bM_{0,r_0}
\end{equation}
with 
\begin{equation}
F(H_1,\Kcal,\Hcal)=(H_1, K_0^{(\Kcal,\Hcal)})  \text{ and }
G((H_1,K_0),\Hcal)=(\bA^{(\Hcal)}_{0})^{-1}\bigl(H_{1}-{\bB}^{(\Hcal)}_{0}K_0\bigr)
\end{equation}
yielding $\overline H_0=G\diamond F$. Both needed conditions, $G\in \widetilde C^m(\bY\times\Vcal_{\rho},\bM_{0,r_0})$ as well as $F\in C_*^m(\Ucal_{1,\rho}\times \Wcal_\rho\times\Vcal_\rho,\bM_{1,0}\times \bM_{0,r_0})$
  have been already proven.

%
%
%Attending first to the norm $\norm{K_0}_{0,r}$, indeed, there exists a constant $\overline C$ such that, choosing  $\zeta$ sufficiently large and $\rho_2$ sufficiently small in dependence on already chosen $h$ as well as $\norm{\Kcal}_{\zeta}$ sufficiently small in dependence on $A$, we have
%\begin{equation}
%\norm{K_0}_{0,r}\le \overline C\norm{\Kcal}_{\zeta}.
%\end{equation}
%
%\begin{equation}
%\norm{K_0(X)}_{0,X,r}\le  \tnorm{\Kcal^{(\bq)}}_{0,\{x\}}^{\abs{X}}.
%\end{equation}
%
%Assuming further that 
%\begin{equation}
%e\overline C(r_0)\norm{\Kcal}_\zeta\le A^{-\frac{2^d}{2^d-1}}
%\end{equation}
%implying 
%\begin{equation}
%(e\overline C(r_0)\norm{\Kcal}_\zeta)^{\abs{X}} \Gamma_A(X)\le e\overline C(r_0)\norm{\Kcal}_\zeta
%\end{equation}
%for any $X$, 
%yields the final claim 
%\begin{equation}
%\norm{K_0}_{0,r}=\sup_{X\in \Pcal_0^{\text{\rm c}} }\norm{K(X)}_{0,X,r} \Gamma_{0,A}(X)\le e\overline C(r_0)\norm{\Kcal}_\zeta.
%\end{equation}
%
%

\bigskip

\noindent (ii)
This  is an immediate consequence of the definition of the map $ \cbT$  and the fact that $S(0,0,\Hcal)=0$ (cf. \eqref{E:K'}).
\bigskip

\noindent (iii)  Using that $K_0=0$ for $\Kcal=0$ and that
$\frac{\partial S_{k}}{\partial H_k}(0,0,\Hcal)= \frac{\partial S_{k}}{\partial K_k} (0,0,\Hcal)=0$,
we can compute the derivatives of $\overline \by=\cbT(\by,0,\Hcal)$ at $\Hcal=0$:
\begin{eqnarray}
\begin{aligned}
\frac{\partial\overline{H}_k}{\partial H_j}&=\begin{cases} \bA_{k}^{-1} & \mbox{ if } j=k+1,\  j=0,\ldots,N-2 \\ 0 & \mbox{ otherwise,} \end{cases}
\\
\frac{\partial \overline{H}_k}{\partial K_j}&=\begin{cases} -\bA_k^{-1}\bB_k & \mbox{ if } j=k,\\0 & \mbox{ otherwise},
\end{cases}
\end{aligned}
\end{eqnarray}
and
\begin{eqnarray}
\begin{aligned}
\frac{\partial\overline{K}_{k+1}}{\partial H_j}&=0,
\\
\frac{\partial \overline{K}_{k+1}}{\partial K_j}&=\begin{cases} \bC_k & \mbox{ if } j=k\not= 0,\\
0 & \mbox{ otherwise,}
\end{cases}
\end{aligned}
\end{eqnarray}
for $ k,j=0,\ldots,N-1 $. 

Consider now a vector $\by\in \bY_{\!\!r}$ with  $ \norm{\by}_{\bY_{\!\!r}}\le 1 $ and its image  $ \overline \by$ under the map $\frac{\partial\cbT(\by,0, \Hcal)}{\partial \by}\bigr\vert_{\by=0}$,
\begin{equation}
 \overline \by=\frac{\partial\cbT(\by,0, \Hcal)}{\partial \by}\Bigr\vert_{\by=0} \, \by.
\end{equation}
Since  $ \norm{\by}_{\bY_{\!\!r}}\le 1 $, we have $ \norm{{H}_k^{(\by)}}_{k,0}\le \eta^k $, $k=0,\ldots, N-1 $, and $ \norm{{K}_{k}^{(\by)}}_{k,r}\le \frac{\eta^k}{\alpha} $, $ k=1,\ldots,N $, for the coordinates ${H}_k^{(\by)}, {K}_{k}^{(\by)}$ of the vector $\by$.
Using  ${H}_k^{(\overline y)}, {K}_{k}^{(\overline y)}$,   for the coordinates of the image  $\overline y$, we get
$$
\begin{aligned}
\norm{H^{(\overline \by)}_0}_{k,0}&\le \norm{\bA_0^{-1}}\eta;\\
\norm{H^{(\overline \by)}_k}_{k,0}&\le  \norm{\bA_k^{-1}}\eta^{k+1}+\norm{\bA_k^{-1}}\norm{B_k}\frac{\eta^k}{\alpha}\le \frac{\eta^k}{\sqrt{\theta}}(\eta+\frac{M}{\alpha}), k=1,\ldots,N-2;\\
\norm{H^{(\overline \by)}_{N-1}}_{N-1,0} &\le  \norm{\bA_{N-1}^{-1}}\norm{\bB_{N-1}}\frac{\eta^{N-1}}{\alpha}  \le \frac{\eta^{N-1}M}{\alpha\sqrt{\theta}};\\
\norm{K^{(\overline \by)}_{1}}_{k,r}&=0;\\
 \norm{K^{(\overline \by)}_{k}}_{k,r}&\le  \norm{\bC_{k-1}}\frac{\eta^{k}}{\alpha}\le\theta\frac{\eta^{k}}{\alpha}, k=2,\ldots,N.
\end{aligned}
$$
As a result,
$$
\norm{\overline \by}_{\bY_{\!\!r}}\le \big(\frac{1}{\sqrt{\theta}}(\eta+\frac{M}{\alpha})\big)\vee \frac{\theta}{\eta}.
$$
It suffices to choose the parameters $ \eta $ and $ \alpha $ so that  $ \eta+M/\alpha\le \theta^{1/2} $ ($\theta<\eta<\theta^{1/2}$),  yielding
\begin{equation}
\Big\Vert\frac{\partial \cbT(\by,0,\Hcal)}{\partial \by}\Big|_{\by=0}\Big\Vert_{\Lcal(\bZ_s,\bZ_s)}\le  \theta <1, \quad s=r_0,r_0-2,\ldots,r_0-6.
\end{equation}
\qed
\end{proofsect}
\begin{proofsect}{Proof of Proposition~\ref{P:hatZ}}
Having thus, in Lemma~\ref{L:verifyimplicit}, verified the assumptions \eqref{E:ift_regularF0}-\eqref{E:iftderivative} of Theorem~\ref{T:implicit}
for the map $\cbT$ in the role of $F$, there exist constants $\widehat\rho_1$,  $\widehat\rho_2$, and $\widehat\rho $ 
 depending (through $\rho$ in Lemma~\ref{L:verifyimplicit})  on $h$ and $\mathsf{A}$ and  $\widehat C $, depending (through $ C=C(L,h,\mathsf{A}) $ in Proposition~\ref{P:Tnonlin}) on $L,h, $ and $ \mathsf{A}$, and the map 
\begin{equation}
\Fcal \colon B_{\bE\times \bM_0}(\widehat\rho_1,\widehat\rho_2)\to B_{\bY_{\!\!r_{0}}}(\widehat\rho)
\end{equation}
(in the role of $f$) so that $\cbT(\Fcal (\Kcal,\Hcal),\Kcal,\Hcal)=\Fcal (\Kcal,\Hcal)$ for any $$(\Kcal,\Hcal)\in 
B_{\bE\times\bM_0}(\widehat\rho_1,\widehat\rho_2),$$ and
\begin{equation}
\Fcal\in \widetilde C^m(B_{\bE\times\bM_0}(\widehat\rho_1,\widehat\rho_2),\bY), 
\end{equation}
%\begin{equation}
%\Fcal \in C^0(\bU;\bY_{\!\!r})\cap C^1(\bU;\bZ_{r-2})\cap C^2(\bU;\bZ_{r-4})\cap C^3(\bU;\bZ_{r-6}),
%\end{equation}
%\begin{equation}%\label{E:derivlkj}
%\frac{\partial^{j+\ell}\Fcal }{\partial\Kcal^j\partial\bq^\ell}\in C^0(\bU;\bZ_{r-2\ell})\ \mbox{ if } j+\ell\le 3,
%\end{equation}
%{\color{red} would not be better here $D_{\Kcal}^j  D_{\bq}^\ell \Fcal (\Kcal,\bq)\in C^0(\bU;\Lcal(\bE^{\otimes j}\otimes\
%\bigl(\R^{d\times d}_{\rm sym}\bigr)^{\otimes \ell}; \bZ_{r-2\ell})$?}
%and
satisfying \eqref{E:estsupderiv}
whenever $ (\Kcal,\Hcal)\in B_{\bE\times\bM_0}(\widehat\rho_1,\widehat\rho_2)$ and  $j, \ell\in\N_0$   with $\ell+j\le m $. Here, the estimates \eqref{E:estsupderiv} follow from the bounds \eqref{eq:E_estimate_tilde}.
\qed
\end{proofsect}	
	
\section{Properties of the map $\Hscr$}
\label{S:tildeq}
Using our results in the previous section we finally obtain a map $\Hscr $ mapping a neighbourhood of the origin in $ \bE $ to $ \bM_0 $ so that $ \cbT(\Fcal(\Kcal,\Hscr(\Kcal)),\Kcal,\Hscr(\Kcal))=\Fcal(\Kcal,\Hscr(\Kcal)) $ and $ \Pi(\Fcal(\Kcal,\Hscr(\Kcal)))=\Hscr(\Kcal)$. This requires another application of the implicit function theorem, this time for the composition of the projection $ \Pi $ with the map $ \Fcal $ in Proposition~\ref{P:hatZ}. We write $ \Gcal:=\Pi\circ\Fcal $ in the following.  The projection $ \Pi\colon \bY_{r_0-2n} \to\bM_0 $ is a bounded linear  mapping for any $ 0\le n\le m $. 
Using Proposition~\ref{P:hatZ} we obtain, in particular,  that $ \Gcal\in C_*^m(B_{\bE\times\bM_0}(\widehat{\rho}_1,\widehat{\rho}_2),\bM_0)$. 
Note that $ \Fcal(0,\Hcal)=0 $ because $ \cbT(0,0,\Hcal)=0 $ for all $ \Hcal\in \Vcal_\rho $ (see (ii) in Lemma~\ref{L:verifyimplicit}), and thus $ \Gcal(0,\Hcal)=0 $ and $D_{\Hcal}\Gcal(0,0)=0 $.
Therefore, by standard implicit function theorem, there exists a $C^m_*$-map
$ \Hscr\colon B_{\bE}(\rho_1) \to B_{\bM_0}(\rho_2) $ with a suitable $\rho_1\le \hat\rho_1$ and $\rho_2=\hat\rho_2$
such that
$\Gcal(\Kcal,\Hscr(\Kcal))=\Hscr(\Kcal)$.

%% file: AKM-appA-sobolev-30thJune-2016.tex
\chapter{Discrete Sobolev Estimates}\label{appSobolev}
%% New SM 29.6. 2016
For the convenience of the reader we recall  a discrete version of the Sobolev inequality. 
Discrete Sobolev inequalities are classical, see, e.g.,  Sobolev's original work \cite{Sob40}.
%% End new
Let $ B_n=[0,n]^d\cap\Z^d $, and for $ p>0 $ define the norm 
\begin{equation}
\norm{f}_p=\norm{f}_{p,B_n}=\Big(\sum_{x\in B_n}|f(x)|^p\Big)^{1/p}
\end{equation}
for any function $ f\colon B_n\to\R $.
\lsm[Bcna]{$B_n$}{$=[0,n]^d\cap\Z^d$}%

\begin{prop}\label{Sobolovpropi-iv}
For every $p\ge 1$ and $m,M\in \N$ there exists a constant $\mathfrak{C}=\mathfrak{C}(p,M,m)$ such that:
\begin{enumerate}
 \item[(i)] If $1\le p\le d$, $\frac{1}{p^*}=\frac{1}{p}-\frac{1}{d} $, and $ q\le p^*$,  $q<\infty $, then
\begin{equation}
n^{-\frac{d}{q}}\norm{f}_q\le \mathfrak{C}n^{-\frac{d}{2}}\norm{f}_2+\mathfrak{C}n^{1-\frac{d}{p}}\norm{\nabla f}_p.
\end{equation}
\lsm[Cyay]{$\mathfrak{C}$}{$=\mathfrak{C}(p,M,m)$,  the constant from the discrete Sobolev estimates, e.g., \\
$\max_{x\in B_n}\abs{f(x)}\le \mathfrak{C}n^{-\frac{d}{2}}\sum_{k=0}^M\norm{(n\nabla)^kf}_2$}%

\item[(ii)] If $ p> d $, then
\begin{equation}
\big|f(x)-f(y)\big| \le \mathfrak{C}n^{1-\frac{d}{p}}\norm{\nabla f}_p \qquad\mbox{ for all } x,y\in B_n.
\end{equation}

\item[(iii)] If $ m\in\N$, $1\le p\le \frac{d}{m}$,  $\frac{1}{p_m}=\frac{1}{p}-\frac{m}{d}$, and  $q\le p_m$, $q<\infty $, then
\begin{equation}
n^{-\frac{d}{q}}\norm{f}_q\le \mathfrak{C}n^{-\frac{d}{2}} \sum_{k=0}^{M-1}\norm{(n\nabla)^kf}_2 +\mathfrak{C}n^{-\frac{d}{p}}\norm{(n\nabla)^Mf}_p.
\end{equation}

\item[(iv)] If $ M=\lfloor\frac{d+2}2\rfloor$,  the integer value of $\frac{d+2}2$, then
\begin{equation}
\max_{x\in B_n}|f(x)|\le \mathfrak{C}n^{-\frac{d}{2}}\sum_{k=0}^M\norm{(n\nabla)^kf}_2.
\end{equation}
\end{enumerate}

\end{prop}

\begin{remark}
\item[(i)]
In the proof of (iv) we actually get
\begin{equation}
\max_{x\in B_n}|f(x)|   \le 
(n+1)^{-\frac{d}{2}}\sum_{x\in B_n}\abs{f(x)}^2+
\mathfrak{C}n^{-\frac{d}{2}}\sum_{k=1}^M\norm{(n\nabla)^kf}_2.
\end{equation}
\item[(ii)]
As written, the higher derivatives on the RHS of (i)-(iv) require the values of $f$ outside $B_n$.
If one traces the dependence more carefully then one sees that $(\nabla_1^{\alpha_1}\dots \nabla_d^{\alpha_d} f)(x)$ is only needed
for $x$ such that $x+\alpha_1 e_1+\dots +\alpha_d e_d\in B_n$, so that only the values of $f$ inside $B_n$ are needed.
\end{remark}

The proof may be reduced to the continuous case by interpolation. Let $ n=1 $, $B_1=\{0,1\}^d$, $f: B_1\to \R_+$, and let $ \widetilde f $ be the interpolation of $ f $ which is affine in each coordinate direction, i.e., $ \widetilde f $ is the unique function of the form
\begin{equation}
\label{E:tildef}
\widetilde f(x)=\prod_{i=1}^d(a_i x_i+b_i),\quad \widetilde f(x)=f(x) \quad\mbox{ for } x\in \{0,1\}^d.
\end{equation}

The Proposition~\ref{Sobolovpropi-iv} will be proven with help of the following Lemma.

\begin{lemma}\label{interpolation}\hfill
\begin{enumerate}
 \item[(i)] 
$
\frac{1}{(p+1)^d2^d}\sum_{x\in B_1}f^p(x) \le \int_{(0,1)^d}\widetilde f^p(x){\rm d} x\le \frac{1}{2^d}\sum_{x\in B_1} f^p(x).
$
\item[(ii)] 
$
\sup_{x\in (0,1)^d}|\partial_i\widetilde f(x)|\le \max_{x\in B_1,x_i=0}|f(x+e_i)-f(x)|\le \Bigl(\sum_{x\in B_1,x_i=0}|f(x+e_i)-f(x)|^p\Bigr)^{1/p}
$
for any $i=1,\dots, d$.
\end{enumerate}
\end{lemma}

\begin{proof}
(i) The integrand is a product of functions of one variable. Taking into account that 
\begin{equation}
\frac1{2^d}\sum_{x\in B_1} f^p(x)=\prod_{i=1}^d  \Bigl(\frac12(a_i+b_i)^p+\frac12 b_i^p\Bigr),
\end{equation}
it suffices to prove  the claim for $d=1$.
Considering thus a nonnegative function  on the interval $[0,1]$ of the form $ax+b$ and assuming w.l.o.g. that $a,b\ge 0$,
we get
\begin{equation}
\int_0^1 (ax +b)^p{\rm d} x= \sum_{k=0}^p \binom{p}{k} \frac{1}{k+1} a^k b^{p-k}\le b^p+\sum_{k=1}^p \binom{p}{k} \frac{1}{2} a^k b^{p-k}=
\frac12 b^p + \frac12 (a+b)^p.
\end{equation}
On the other hand, 
\begin{equation}
\begin{aligned}
\sum_{k=0}^p {p\choose k} \frac1{k+1} a^k b^{p-k}&\ge \frac1{p+1}\sum_{k=0}^p {p\choose k}  a^k b^{p-k}=
\frac1{p+1}(a+b)^p\\ & \ge \frac1{p+1}\Bigl( \frac12 (a+b)^p+\frac12 b^p\Bigr).
\end{aligned}
\end{equation}

(ii) For $\widetilde f$ of the form \eqref{E:tildef} we have $\partial_i\widetilde f(x)=a_i\prod_{j\neq i}^d(a_j x_j+b_j)$ while, on the other hand, 
we have $a_i\prod_{j\neq i}^d(a_j x_j+b_j)=\widetilde f(x+e_i)-\widetilde f(x)=f(x+e_i)-f(x)$ for any $x\in B_1$ such that $x_i=0$.
\end{proof}

\begin{proof}[Proof of Proposition~\ref{Sobolovpropi-iv}]

(i) and (ii)  follow from Lemma~\ref{interpolation} and the continuous embedding Theorem.

The claim (iii) follows from (i) by iteration.

To prove (iv), assume first that $d$ is odd and thus $M=\lfloor\frac{d+2}2\rfloor= \frac{d}{2}+\frac{1}{2}$. 
Let us apply (iii) with $p=2$, $m=M-1$, and
\begin{equation}
\frac1{p_{m}}=\frac12-\frac{M-1}d=\frac{d-(d-1)}{2d}=\frac1{2d}.
\end{equation}
Hence,
\begin{equation}
n^{-\frac{d}{2d}}\norm{\nabla f}_{2d}\le \mathfrak{C}n^{-\frac{d}{2}-1}  \sum_{k=1}^{M}\norm{(n\nabla)^kf}_2.
\end{equation}

Further, 
\begin{equation}
\big|f(x)-f(y)\big| \le \mathfrak{C}n^{1-\frac{d}{2d}}\norm{\nabla f}_{2d} =\mathfrak{C}n^{\frac{1}{2}}\norm{\nabla f}_{2d} 
\end{equation}
 for all $x,y\in B_n$ by (ii). Averaging over $y$ yields
\begin{equation}
\bigl|f(x)-(n+1)^{-d} \sum_{y\in B_n}f(y)\bigr| \le \mathfrak{C}n^{\frac{1}{2}}\norm{\nabla f}_{2d} .
\end{equation}
On the other hand,
\begin{equation}
\bigl|(n+1)^{-d} \sum_{y\in B_n}f(y)\bigr| \le (n+1)^{-d} \Bigl(\sum_{y\in B_n}f(y)^2\Bigr)^{1/2} \Bigl(\sum_{y\in B_n}1\Bigr)^{1/2}
\le (n+1)^{-d/2}\norm{f}_2
\end{equation}
yielding
\begin{equation}
|f(x)|\le \mathfrak{C}n^{\frac{1}{2}}\norm{\nabla f}_{2d} +(n+1)^{-d/2}\norm{f}_2
\end{equation}
for all $x\in B_n$. The assertion (iv) for odd $d$ follows.

Similarly for even $d$ when   $M=\lfloor\frac{d+2}2\rfloor=\frac{d}{2}+1$ and we use $m=M-2$ and $q=2d>p_m=d$.
\end{proof}

%% file: AKM-appB-parts-30thJune-2016.tex
\chapter{Integration by Parts and Estimates of the Boundary Terms}
%% New SM   29.6. 2016
For the convenience on the reader we spell out the estimates of the boundary terms in detail.
\bigskip

\section*{a)  $d=1$}\hfill

The forward and backward derivative are $ \partial v(x)=v(x+1)-v(x) $ and $ \partial^*v(x)=v(x-1)-v(x) $.

\begin{prop}[\textbf{Integration by parts}]\label{IBP1}
Let  $ g,v,u\colon\Z\to\R $ and $ m\in\N $. Then:
\begin{enumerate}
 \item[(i)] 
$$
\sum_{x=-m}^{m}g(x)\partial v(x)=\sum_{x=-m}^m \partial^*g(x)v(x) +g(m)v(m+1)-g(-m-1)v(-m).
$$
\smallskip

\item[(ii)]
$$
\sum_{x=-m}^m \partial u(x)\partial v(x)=\sum_{x=-m}^m(\partial^*\partial u)(x)v(x)+
\partial u(m)v(m+1) -\partial u(-m-1)v(-m).
$$
\end{enumerate}
\end{prop}

\begin{prop}[\textbf{Evaluation of the boundary terms}]\label{IBP2}
 There exist a constant $\mathfrak{c} <3\sqrt{2}$ such that for any $v\colon\Z\to\R $ and  any $ m\in\N $, $m>1$,  one has
\begin{equation}
v(-m)^2+v(m+1)^2\le \frac{\mathfrak{c}}{2m+1}\sum_{x=-m}^mv(x)^2+\mathfrak{c}(2m+1)\sum_{x=-m}^m\partial v(x)^2.
\end{equation}
\end{prop}
\begin{proof}
%Let $a=\min{v^2(-m),v^2(m+1)}$ and $b=\max{v^2(-m),v^2(m+1)}$.
Assume first that the number of those $x\in\{-m,\dots, m    \}$ for which $v(x)^2\ge \frac13\bigl(v(-m)^2+v(m+1)^2\bigr)$
is at least $\frac{2m+1}{\sqrt{2}}$. Then $\sum_{x=-m}^mv(x)^2\ge \frac1{3\sqrt{2}} (2m+1) \bigl(v(-m)^2+v(m+1)^2\bigr)$.
\lsm[cyax]{$\mathfrak{c}$}{$< 3\sqrt{2}$,  the constant from the bound \\
$v(-m)^2+v(m+1)^2\le \frac{\mathfrak{c}}{2m+1}\sum_{x=-m}^mv(x)^2+\mathfrak{c}(2m+1)\sum_{x=-m}^m\partial v(x)^2$}

On the other hand, if the number of such $x$'s is less then $\frac{2m+1}{\sqrt{2}}$, then there exists $x$ such that 
$\partial v(x)^2\ge \frac{\sqrt{2}}{6}\frac{v(-m)^2+v(m+1)^2}{2m+1}$, implying 
$$
\sum_{x=-m}^m\partial v(x)^2\ge \frac1{3\sqrt{2}}\frac{v(-m)^2+v(m+1)^2}{2m+1}.
$$
 Indeed, having assured the existence of $y$ and $z$ such $v(y)^2 < \frac13\bigl(v(-m)^2+v(m+1)^2\bigr)$ (the existence of such $y$ is obvious for $m>1$ implying that  $\bigl(1-\frac{1}{\sqrt{2}}\bigr) (2m+1)>1$) and  $v(z)^2 \ge \frac12\bigl(v(-m)^2+v(m+1)^2\bigr)$ (again, its existence follows since $\frac12\bigl(v(-m)^2+v(m+1)^2\bigr) \le \max\bigl\{v(-m)^2,v(m+1)^2\bigr\}$)
implying that the interval $\bigl[\frac13\bigl(v(-m)^2+v(m+1)^2\bigr),
\frac12 \bigl(v(-m)^2+v(m+1)^2\bigr)\bigr]$ has to be spanned within at most $\frac{2m+1}{\sqrt{2}}$ increments $\partial v(x)^2$.

In both cases,
\begin{equation}
\frac{1}{2m+1}\sum_{x=-m}^mv(x)^2+(2m+1)\sum_{x=-m}^m\partial v(x)^2\ge \frac1{3\sqrt{2}}\bigl(v(-m)^2+v(m+1)^2\bigr) 
\end{equation}
implying the claim.
\end{proof}

The combination of Proposition~\ref{IBP1} and \ref{IBP2} yields:

\begin{prop} Let  $ u,v\colon\Z\to\R $ and $ m\in\N $. With the constant $\mathfrak{c}$ from Proposition~\ref{IBP2} and any $\eta>0$, one has
\begin{multline}
\Big|\sum_{x=-m}^m\partial u(x)\partial v(x) \Big|\le \frac{1}{2}(2m+1)^2\frac{1}{\eta}\sum_{x=-m}^m\big|(\partial^*\partial u)(x)\big|^2 +\frac{1}{2}\frac{\eta}{(2m+1)^2}\sum_{x=-m}^mv(x)^2+\\
 +\frac{2m+1}{2\eta}\Big[\partial u(-m-1)^2+\partial u(m)^2\Big]+\frac{\mathfrak{c}\, \eta}{2}\Big[\frac{1}{(2m+1)^2}\sum_{x=-m}^mv(x)^2
+\sum_{x=-m}^m\partial v(x)^2\Big].
\end{multline}
\end{prop}

\section*{b)  Multidimensional case}
Let $ X\in\Pcal_k $ be a union of $k$-blocks.
Further, let $\partial^{\pm} X=\cup_{i=1}^d  \partial_i^{\pm} X$, where, for any $ i=1,\ldots, d $,
\begin{equation}
\partial_i^- X:=\{x\in\Z^d\colon x\notin X, x+e_i\in X\,\mbox{ or }\, x\in X, x+e_i\notin X\}
\end{equation}
and
\begin{equation}
\partial^+_i X= \partial^-_i X+ e_i:=\{x+e_i\colon x\in\partial^-_i X\}.
\end{equation}
Notice that $\partial^- X\cup \partial^+ X=\partial X$, the boundary  defined in \eqref{E:pX}.

\begin{lemma}
\label{L:boundaryest}
Let $B$ be a $k$-block and let $v: B\cup\p B\to \R$. Then, for any $i=1,\dots,d$, 
\begin{equation}
\label{E:estIBP+}
\sum_{x\in\partial_i^+ B}v(x)^2\le \mathfrak{c}\Big(\frac{1}{L^k}\sum_{x\in B} v(x)^2+L^k\sum_{x\in B}\abs{\nabla_i v(x)}^2\Big)
\end{equation}
and
\begin{equation}
\label{E:estIBP-}
\sum_{x\in\partial_i^- B}v(x)^2\le \mathfrak{c}\Big(\frac{1}{L^k}\sum_{x\in B} v(x)^2+L^k\sum_{x\in B}
\abs{\nabla_i^* v(x)}^2\Big),
\end{equation}
where $ c$ is the constant from Proposition~\ref{IBP2}.
\end{lemma}
\begin{proof}
Applying Proposition~\ref{IBP2}  to all lines in $ B$ that are parallel to $e_i$, we get
\eqref{E:estIBP+}. Similarly for \eqref{E:estIBP-}, when
considering the sites on these lines in the opposite order. 
\end{proof}

Notice that, using $\nabla_i^* v(x)= -\nabla_iv(x-e_i)$,  the last term in \eqref{E:estIBP-} can be actually replaced by
$L^k\sum_{x\in B-e_i} \abs{\nabla_i v(x)}^2$

To formulate the following immediate corollary of Lemma~\ref{L:boundaryest},
let,  for any $ X\in\Pcal_k $ and $\ell\in \N$, the neighbourhood $U_\ell(X)$ be defined iteratively with $U_1(X)=X\cup\partial X$
and $U_{\ell+1}(X)=U_\ell(X)\cup \partial U_\ell(X)$.

\begin{prop}\label{boundaryest}
Let  $ X\in\Pcal_k $ and  $u: U_4(X)\to \R$. With the constant $\mathfrak{c}$ from Proposition~\ref{IBP2}, 
\begin{enumerate}
 \item[(a)] 
$$
L^k\sum_{x\in\partial X}|\nabla v(x)|^2\le 2 \mathfrak{c}\Big(\sum_{x\in X}|\nabla v(x)|^2+L^{2k}\sum_{x\in U_1(X)}|\nabla^2 v(x)|^2\Big),
$$

\item[(b)] 
$$
L^{3k}\sum_{x\in\partial X}|\nabla^2 v(x)|^2\le 2 \mathfrak{c}\Big(L^{2k}\sum_{x\in X}|\nabla^2 v(x)|^2+L^{4k}\sum_{x\in U_1(X)}|\nabla^3 v(x)|^2\Big),
$$ 
and
\item[(c)] 
$$
L^{5k}\sum_{x\in\partial X}|\nabla^3 v(x)|^2\le 2 \mathfrak{c}\Big(L^{4k}\sum_{x\in X}|\nabla^3 v(x)|^2+L^{6k}\sum_{x\in U_1(X)}|\nabla^4 v(x)|^2\Big).
$$
\end{enumerate}
\end{prop}

\begin{proof}
Let $ B_1,\ldots, B_n $ denote the $k$-blocks contained in $ X $. 
Applying Lemma \ref{L:boundaryest}  to each  $ B_{\ell} $, $\ell=1,\dots,n$, $i=1,\dots, d$, observing that 
\begin{equation}
\partial X\subset \bigcup_{\ell=1}^n \partial  B_\ell,
\end{equation}
and summing over $i$, we get
\begin{equation}
L^k\sum_{x\in\partial X}|\nabla v(x)|^2\le  \mathfrak{c}\Big(2\sum_{x\in X}|\nabla v(x)|^2+L^{2k}\sum_{x\in X}\sum_{i=1}^d\bigl(|\nabla_i^2 v(x)|^2
+|\nabla_i^*\nabla_i v(x)|^2\bigr)\Big).
\end{equation}
Using 
\begin{equation}
\sum_{x\in X}\sum_{i=1}^d|\nabla_i^*\nabla_i v(x)|^2=  \sum_{x\in X-e_i}\sum_{i=1}^d|\nabla_i^2 v(x)|^2\le \sum_{x\in U_1(X)}|\nabla^2 v(x)|^2,
\end{equation}
we get the first claim. 

The second and the third claim  follow in a similar way.
 \end{proof}

Notice that the sums over $x\in U_1(X)$ on the right hand side of the bounds in Proposition~\ref{boundaryest} can be actually replaced by the sums over
$x\in (X\cup \partial^- X)\setminus(X\cap \partial^- X)$.

\begin{prop}\label{estmixedterm}
Let  $u,v: X\cup \partial  X\to \R$ and $ X\in\Pcal_k $. With the constant $\mathfrak{c}$ from Proposition~\ref{IBP2} and any $\eta>0$, we get
\begin{multline}
\bigl|\sum_{x\in X}\nabla u(x)\nabla v(x)\bigr|\le 
\frac{\eta(1+\mathfrak{c} d)}{2L^{2k}}\!\!\!\sum_{x\in X\cup\partial^- X}\!\!\! v(x)^2+
\frac{L^k}{2\eta}\sum_{x\in\partial^- X}|\nabla u(x)|^2   + \\ +
\frac{\mathfrak{c}\eta}2 \sum_{x\in X}|\nabla v(x)|^2+
\frac{L^{2k}}{2\eta}\sum_{x\in X\cup \p^-X}\abs{\nabla^2 u(x)}^2.
\end{multline}
\end{prop}
\begin{proof}

For any $x\in \partial_i^- X$, let $\epsilon_i(x)=+1$ if $x\in X$ and $\epsilon_i(x)=-1$ if $x\not\in X$.
By Proposition~\ref{IBP1}, for each $ i\in\{1,\ldots,d\} $,  we have 
\begin{equation}
\sum_{x\in X}\nabla_i u(x)\nabla_i v(x)=\sum_{x\in X}\nabla_i^*\nabla_i u(x)v(x) +\sum_{x\in\partial_i^- X}\epsilon_i(x)\nabla_i u(x)v(x+e_i).
\end{equation}
Summing over $ i=1,\dots, d $, we get
\begin{multline}
\label{E:nablauv}
\bigl|\sum_{x\in X}\nabla u(x)\nabla v(x)\bigr|\le \sum_{i=1}^d\sum_{x\in X-e_i}\abs{\nabla_i^2 u(x)v(x)} + \sum_{i=1}^d\sum_{x\in\partial_i^- X}\abs{\nabla_i u(x)v(x+e_i)}\le
\\ \le
 \frac{L^{2k}}{2\eta}\sum_{x\in X\cup \partial^- X}\abs{\nabla^2  u(x)}^2+\frac{\eta}{2 L^{2k}}\sum_{i=1}^d\sum_{x\in X-e_i}v(x)^2+\frac{L^k}{2\eta}\sum_{x\in\partial^- X}|\nabla u(x)|^2 \\ + \frac{\eta}{2L^k}\sum_{i=1}^d \sum_{x\in \partial_i^+ X} v(x)^2.
\end{multline}

Applying now Lemma \ref{L:boundaryest} on the last term, we get the claim.
 \end{proof}

\begin{lemma}
\label{L:6.20}
Let  $Y\subset X$, $ X,Y\in\Pcal_k $, and  $u: U_4(X)\to \R$. 
Then
\begin{equation}
\max_{x\in X} u(x)^2 \le \frac2{\abs{Y}} \sum_{x\in Y} u(x)^2 + 2 (\diam X)^2 \max_{x\in X} \abs{\nabla u(x)}^2.
\end{equation}
\end{lemma}

\begin{proof}
Cf. \cite[Lemma 6.20]{B07}.
Considering the shortest path from any $x\in X$ to $y\in Y$, we have
\begin{equation}
\abs{u(x)}\le \abs{u(y)} + \abs{x-y}_\infty \max_{z\in X} \abs{\nabla u(z)}.
\end{equation}
Using that $\abs{x-y}_\infty \le \diam X$ (with the diameter taken in  $\abs{\boldsymbol{\cdot}}_\infty$ metric on $\Z^d$), using the inequality $(a+b)^2\le 2a^2 +2b^2$, and averaging both sides over $Y$, we get
\begin{equation}
u(x)^2\le \frac2{\abs{Y}} \sum_{y\in Y} u(y)^2 + 2(\diam X)^2 \max_{z\in X} \abs{\nabla u(z)}^2
\end{equation}
yielding the claim.
\end{proof}

%% file: AKM-appC-gaussian-30thJune-2016.tex
\chapter{Gaussian Calculus}\label{app5}

Here we recall the formulae for the derivative of a Gaussian integral with respect to the covariance matrix. 
The arguments are classical, but we provide proofs for the convenience of the reader.
We begin with the first derivative. We will make the following general assumptions throughout this appendix.

{\it Let $V$ be a finite dimensional Euclidean vector space with scalar product $(\cdot, \cdot)$ and Lebesgue measure $\lambda$. Denote by $ \sym^{\ssup{+}}(V) $ and $ \sym^{\ssup{\ge}}(V) $ the set of positive definite respectively of positive semi-definite symmetric operators on $V$. For $\Cscr\in\sym^{\ssup{+}}(V) $ denote by $ \mu_{\Cscr} $ the Gaussian measure with covariance $\Cscr$.} Let $g\colon V\to\R $ be measurable and assume that there exists a $ \Bscr\in\sym^{\ssup{\ge}}(V) $ and a constant $ M\in\R $ such that
$$
|g(x)|\le M\ex^{\frac{1}{2}(\Bscr x,x)}\quad \mbox{ for all } x\in V.
$$
For $ \Cscr^{-1}>\Bscr $ define
\begin{equation}\label{Hdef}
H(\Cscr):=\int_V\;g(x)\,\mu_{\Cscr}(\d x)=\frac{1}{\det (2\pi\Cscr)^{1/2}}\int_V\,g(x)\ex^{-\frac{1}{2}(\Cscr^{-1} x,x)}\,\lambda(\d x).
\end{equation}
We first recall that $ H $ is real-analytic in the set $ \{\Cscr\in\sym^{\ssup{+}}(V)\colon \Cscr^{-1}>\Bscr\} $. In fact we will extend $H$ to a complex analytic function as follows. Let $ \widetilde{V} $ denote the complexification of $V$ with the canonical sesquilinear-form $ (\cdot,\cdot) $, let $ GL(\widetilde{V}) $ denote the set of all invertible $ \C$-linear maps from $ \widetilde{V} $ to itself and let
$$
\Ucal:=\{\Cscr\in GL(\widetilde{V})\colon \re(\Cscr^{-1}x,x)>(\Bscr x,x)\;\forall x\in V\setminus\{0\}\}.
$$
Define $H$ on $ \Ucal $ by the right hand side of \eqref{Hdef}.

\begin{lemma}\label{Clemma1}
\begin{enumerate}
\item[(i)] The map $ H\colon\Ucal\to\C $ is analytic and the derivative at $ \Cscr $ in direction $ \dot{\Cscr} $ reads as
\begin{equation}
DH(\Cscr,\dot{\Cscr})=\int_V\,g(x)\frac{1}{2}\big((\Cscr^{-1}\dot{\Cscr}\Cscr^{-1}x,x)-\tr(\Cscr^{-1}\dot{\Cscr})\big)\,\mu_{\Cscr}(\d x).
\end{equation}

\item[(ii)] Assume in addition that $ g $ is continuous and that there exists a continuous function $ w\colon V\to(0,\infty) $ such that
\begin{equation}\label{Clem1ass}
g(x+y)\le M\ex^{\frac{1}{2}(\Bscr x,x)}w(y), \quad x,y\in V.
\end{equation}
Define
\begin{equation}
\widetilde{H}(\Cscr)(y):=\int_V\,g(x+y)\,\mu_{\Cscr}(\d x)\quad\mbox{ for all } y\in V.
\end{equation}
Then $ \widetilde{H} $ is an analytic map from $ \Ucal $ to the space
$$
C_w^0:=\{h\in\Ccal^0(V)\colon \norm{h}_w<\infty\},
$$ where
$$
\norm{h}_w:=\sup_{y\in V}\frac{|h(y)|}{|w(y)|},
$$ and the derivative at $ \Cscr $ in direction $ \dot{\Cscr}\in GL(\widetilde{V}) $ is given as
$$
D\widetilde{H}(\Cscr,\dot{\Cscr})(y)=\int_V\,g(x+y)D_1f(\Cscr,x,\dot{\Cscr})\,\lambda(\d x),\quad y\in V,
$$ where 
$$
f(\Cscr,x):=\frac{\ex^{-\frac{1}{2}(\Cscr^{-1} x,x)}}{\det(2\pi\Cscr)^{1/2}}.
$$
\end{enumerate}
\end{lemma}

\begin{proof}
(i) Set
\begin{equation}
f(\Cscr,x):=\frac{\ex^{-\frac{1}{2}(\Cscr^{-1} x,x)}}{\det(2\pi\Cscr)^{1/2}}.
\end{equation}
Then for every $ x\in V $ the map $ \Cscr\mapsto f(\Cscr,x) $ is complex differentiable in $ \Ucal $, and (using Jacobi's formula for the derivative of determinants) we get that
\begin{equation}
D_1f(\Cscr,x,\dot{\Cscr})=\frac{1}{2}\big((\Cscr^{-1}\dot{\Cscr}\Cscr^{-1} x,x)-\tr(\Cscr^{-1}\dot{\Cscr})\big)f(\Cscr,x).
\end{equation}
In particular for each $\eps>0 $ there exists $ M^\prime>0 $ such that 
\begin{equation}\label{Clem1est1}
\big|D_1f(\Cscr,x,\dot{\Cscr})\big|\le M^\prime\ex^{\frac{1}{2}\eps|x|^2}\ex^{-\frac{1}{2}(\Cscr^{-1} x,x)}|\dot{\Cscr}|.
\end{equation}
Since $\re(\Cscr^{-1})>\Bscr $ and since $V$ is finite-dimensional we also have that $ \re(\Cscr^{-1})>\Bscr+\eps\Id $ and thus the function 
$$ 
g(x) \big|D_1f(\Cscr,x,\dot{\Cscr})\big|
$$ is integrable. Now for any $ \dot{\Cscr}\not= 0 $ we estimate
\begin{equation}
\begin{aligned}\label{Clem1}
\frac{1}{|\dot{\Cscr}|}&\Big|H(\Cscr+\dot{\Cscr})-H(\Cscr)-\int_V\,g(x)D_1f(\Cscr,x,\dot{\Cscr})\,\lambda(\d x)\Big|\\
\le & \int_V|g(x)|\Big|\frac{f(\Cscr+\dot{\Cscr})-f(\Cscr)-D_1f(\Cscr,x,\dot{\Cscr})}{|\dot{\Cscr}|}\Big|\,\lambda(\d x).
\end{aligned}
\end{equation}
For $\dot{\Cscr}\to 0 $ the integrand on the right hand side of \eqref{Clem1} goes to zero for every $ x\in V $. It remains to find an integrable majorant. We have
$$
f(\Cscr+\dot{\Cscr},x)-f(\Cscr,x)=\int_0^1\,D_1f(\Cscr+s\dot{\Cscr},x)\,\d s.
$$
Now for every $ \Cscr\in\Ucal $ and every $ \eps>0 $ there exist $ \delta>0 $ and $ M^{\prime\prime}>0 $ such that for all $ \widetilde{\Cscr}\in B_\delta(\Cscr) $ we have
$$
\big|D_1f(\widetilde{\Cscr},x,\dot{\Cscr})\big|\le M^{\prime\prime}\ex^{\frac{1}{2}\eps|x|^2}\ex^{-\frac{1}{2}(\Cscr^{-1} x,x)}|\dot{\Cscr}|.
$$
Hence for $ |\dot{\Cscr}|<\delta $ the integrand in \eqref{Clem1} is bounded by the integrable function
$$
|g(x)|(M^\prime+M^{\prime\prime})\ex^{\frac{1}{2}\eps|x|^2}\ex^{-\frac{1}{2}(\Cscr^{-1} x,x)}.
$$
Thus by the dominated convergence theorem the right hand side of \eqref{Clem1} goes to zero as $ \dot{\Cscr}\to 0 $. This concludes the proof of (i).

\noindent (ii) The continuity of the map $ y\mapsto \widetilde{H}(\Cscr)(y) $ follows directly from the dominated convergence theorem. Indeed, assume that $ y_k\to\overline{y} $ in $V$ as $ k\to\infty $. Using the continuity of $g$ we obtain
$$
g(x+y_k)f(\Cscr,x)\to g(x+\overline{y})f(\Cscr,x)\quad  \mbox{ for every } x\in V \mbox{ as } k\to\infty.
$$
Moreover, for $ |y_k-\overline{y}|<\delta $ we have 
$$
\big|g(x+y_k)f(\Cscr,x)\big|\le M\ex^{\frac{1}{2}(\Bscr x,x)}\big(\sup_{z\in B_\delta(\overline{y})}w(z)\big)f(\Cscr,x),
$$ and the right hand side is integrable. Hence
$$
\widetilde{H}(\Cscr)(y_k)\to \widetilde{H}(\Cscr)(\overline{y}) \quad \mbox{ as } k\to\infty
$$ by the dominated convergence theorem. To verify complex differentiability define first the linear map
$$
(L\dot{\Cscr})(y):=\int_V\,g(x+y)D_1f(\Cscr,x,\dot{\Cscr})\,\lambda(\d x).
$$ 
Then one sees as above that $ y\mapsto (L\dot{\Cscr})(y) $ is continuous. Moreover it follows from the bounds \eqref{Clem1ass} and \eqref{Clem1est1} that
$$
\norm{L\dot{\Cscr}}_w\le MM^\prime|\dot{\Cscr}|\int_V\,\ex^{\frac{1}{2}((\Bscr+\eps\Id-\Cscr^{-1})x,x)}\,\lambda(\d x)<\infty.
$$
Thus $ L $ is a bounded linear map  from $ GL(\widetilde{V}) $ to $ C^0_w(V) $. Finally we check differentiability. We have
$$
\begin{aligned}
&\Big|\widetilde{H}(\Cscr+\dot{\Cscr})(y)-\widetilde{H}(\Cscr)(y)-L\dot{\Cscr}(y)\Big|\\
\le &\int_V\,|g(x+y)|\big|f(\Cscr+\dot{\Cscr},x)-f(\Cscr,x)-D_1f(\Cscr,x,\dot{\Cscr})\big|\,\lambda(\d x)\\
\le & Mw(y)\int_V\,\ex^{\frac{1}{2}(\Bscr x,x)} \big|f(\Cscr+\dot{\Cscr},x)-f(\Cscr,x)-D_1f(\Cscr,x,\dot{\Cscr})\big|\,\lambda(\d x).\\
\end{aligned}
$$
Dividing by $ w(y)|\dot{\Cscr}| $ and taking the supremum over $ y $ we get
$$
\begin{aligned}
\norm{\widetilde{H}& (\Cscr+\dot{\Cscr})+\widetilde{H}(\Cscr)-L\dot{\Cscr}}_w\\&\le M\int_V\,\ex^{\frac{1}{2}(\Bscr x,x)}\frac{\big|f(\Cscr+\dot{\Cscr},x)-f(\Cscr,x)-D_1f(\Cscr,x,\dot{\Cscr})\big|}{|\dot{\Cscr}|}\,\lambda(\d x).
\end{aligned}
$$
Now as in (i) it follows from the dominated convergence theorem that the right hand side goes to zero as $ \dot{\Cscr}\to 0 $. Thus $ \widetilde{H} $ is complex differentiable at $ \Cscr $ with derivative $ D\widetilde{H}(\Cscr)=L $.
\end{proof}
We will apply Lemma~\ref{Clemma1} with $ \Cscr=\Cscr_k^{\ssup{\bq}} $, the covariance matrices which arise in the finite range decomposition (see Proposition~\ref{P:FRD}), and $ \Bscr=\varkappa\Bscr_k=2\overline{C}h^{-2}\Bscr_k $ where $ \Bscr_k $ is as in Lemma~\ref{L:spectrum}. Now an important point is that the finite range decomposition in Proposition~\ref{P:FRD} does not yield a bound on terms like $$ \tr\big(\Cscr^{\ssup{\bq}}_k\big)^{-1}D_{\bq}\dot{\Cscr}^{\ssup{\bq}}_k\dot{\bq} $$ which are independent of $k $ and $N$. 

In order to derive bounds on the derivatives of $ \bq\mapsto H(\Cscr^{\ssup{\bq}}_k) $ which are independent of $k$ and $N$ we now derive different expressions for the derivatives of $ H $ which do not involve $ \Cscr^{-1} $ but which require derivatives of $ g$. This leads to a loss of regularity when we consider the convolution operator $ g\mapsto \int\,g(\cdot+x)\,\mu_{\Cscr}(\d x) $ as an operator between function spaces and we shall see later how to deal with this loss of regularity. 

In the following we assume that
\begin{equation}
e_1,\ldots, e_{\dime(V)} \mbox{ is an orthonormal basis of } V.
\end{equation}

\begin{lemma} \label{L:der_cov1}
Let $ \Bscr\in\sym^{\ssup{\ge}}(V) $ and let  $g\in C^{2}(V)$ with
\begin{equation}  
\label{E:Cgrowr}
\sup_{x \in V}  \sum_{s=0}^{2} \abs*{D^s g(x)} {\rm e}^{- \frac12 (\Bscr x, x)} < \infty \, .
\end{equation}
Furthermore, let $ \Cscr\in\sym^{\ssup{+}}(V) $ be given with $  \Cscr^{-1}>\Bscr $. Let $ \dot{\Cscr}\in\sym^{\ssup{+}}(V) $ and define
$$
h(t):=\int_V\,g(x)\,\mu_{\Cscr+t\dot{\Cscr}}(\d x).
$$
Then $ h $ is a $C^1$-function on some interval $ (-a_0,a_0) $ and
\begin{equation}\label{E:C1}
h^\prime(t)=\int_V\,\big(Ag\big)(x)\,\mu_{\Cscr+t\dot{\Cscr}}(\d x),
\end{equation}
where
\begin{equation}\label{Clem2A}
Ag(x):=\frac{1}{2}\sum_{i,j=1}^{\dime(V)}\dot{\Ccal}_{i,j}D^2g(x,e_i,e_j),\quad \mbox{ with } \dot{\Ccal}_{i,j}:=(\dot{\Cscr}e_i,e_j).
\end{equation}

\end{lemma}

\begin{remark}
In coordinate free notation the map $ A $ in \eqref{Clem2A} can be written as 
$$
Ag(x)=\tr\big(\Hess(g(x))\dot{\Cscr}\big),
$$ where $ \Hess(g(x)) $ is the linear map $ V\to V $ defined by
$$
\big(\Hess(g(x)) a,b\big)=D^2g(x,a,b)\quad \mbox{ for all } a,b\in V.
$$
Sometimes it is more convenient to use an orthonormal basis of the complexification $ \widetilde{V} $ of $V$ to evaluate $ Ag $. If we extend $ \Hess(g(x)) $ as a $\C$-linear map and $ D^2 g(x,\cdot,\cdot) $ as a $ \C$-bilinear map, then
$$
\big(\Hess(g(x))a,b\big)=D^2g(x, a,\overline{b})\quad\mbox{ for all } a,b\in\widetilde{V}
$$ since the sesquilinear form $ (\cdot,\cdot) $ on $ \widetilde{V}\times\widetilde{V} $ is anti-linear in the second argument. If we also extend $ \dot{\Cscr} $ as a $\C$-linear map and if $ f_1\ldots,f_{\dime(V)} $ is an orthonormal basis of $ \widetilde{V}$, then
$$
\tr_V\big(\Hess(g(x))\dot{\Cscr}\big)=\tr_{\widetilde{V}}\big(\Hess(g(x))\dot{\Ccal}\big)=\sum_{i=1}^{\dime V}\big(\Hess(g(x))\dot{\Ccal}f_i,f_i\big).
$$
Hence
\begin{equation}\label{hess}
Ag(x)=\sum_{i=1}^{\dime(V)} D^2g(x,\dot{\Ccal}f_i,\overline{f}_i).
\end{equation}\hfill $\diamond$
\end{remark}

\begin{proof}

One can easily check that the definition of $A$ is independent of the choice of the orthonormal basis. The whole statement is invariant under isometries. Hence we may assume that  $ V = \R^n$ with the standard scalar product and  that $ e_1,\ldots, e_n $ is the standard basis. Furthermore, we write $ \Cscr(t):=\Cscr+t\dot{\Cscr} $ in the following.
The starting point is the formula for the Fourier transform of a Gaussian
\begin{equation}
\int_{\R^n} {\rm e}^{- i (\xi, x)} \, \mu_{\Cscr(t)}(\d x) = {\rm e}^{-\frac12 (\Cscr(t) \xi, \xi)} \, .
\end{equation}

By continuity of $t \mapsto \Cscr(t)$ we may assume that there is an $a_0 > 0$ and a $\delta > 0$ such that
for $t \in (-a_0, a_0)$ we have $\Bscr \le \Cscr^{-1}(t) - \delta \Id$ and $\Cscr(t) \ge \delta \Id$. 
From now on we consider $h(t)$ only on the interval $(-a_0, a_0)$. 

Now assume first that $g$ belong to the Schwartz class $\mathcal{S}(\R^n)$ of smooth and 
rapidly decreasing functions.
By Plancherel's formula we have
\begin{equation}
h(t) = \int_{\R^n} g(x) \,  \mu_{\Cscr(t)}(\d x) = 
 \frac{1}{(2 \pi)^n }  \int_{\R^n}   \hat g(\xi)  {\rm e}^{-\frac12 (\Cscr(t) \xi, \xi)} \, \d\xi \, .
 \end{equation}

Since $g \in \mathcal{S}(\R^n)$, the right hand side is differentiable with respect to $t$ and
the identity  $\widehat{ \partial_j g}(\xi) = i \xi_j \hat g(\xi)$ yields, with another application of
Plancherel's formula,
\begin{align*}
\dot{h}(t)&  =
 -  \frac12  \frac{1}{(2 \pi)^n} \int_{\R^n}    \hat g(\xi)   \sum_{j, k = 1}^n  \dot{\Ccal}_{jk} \xi_j  \xi_k  
 e^{-\frac12 (\Cscr(t) \xi, \xi)} \, \d\xi \\
& =   \frac12  \frac{1}{(2 \pi)^n}   \int_{\R^n}    \sum_{j, k = 1}^n  \dot{\Ccal}_{jk}(  \widehat{\partial_j \partial_k g}(\xi)
 e^{-\frac12 (\Cscr(t) \xi, \xi)} \, \d\xi \\
&= \frac{1}{2}  \int_{\R^n}     \sum_{j,k =1}^n  \dot{\Ccal}_{jk} (\partial_j \partial_k g)(x)  \, \mu_{\Cscr(t)}(\d x) 
= \frac{1}{2} \int_{\R^n}  \tr (\dot{\Ccal} D^2 g(x))  \, \mu_{\Cscr(t)}(\d x)  \, \\
&=\int_{\R^n}\,Ag(x)\,\mu_{\Cscr+t\dot{\Cscr}}(\d x).
\end{align*}
 
This proves assertion \eqref{E:C1} and \eqref{Clem2A} for $g \in \mathcal{S}(\R^n)$. For a general $g$ we use a cut-off and a convolution with a mollifier.
To do so we first rewrite the result for $g \in \mathcal{S}(\R^n)$ in the integral form
\begin{equation}   \label{eq:derivate_covariance_integrated}
\int_{\R^n}  g(x) \, \mu_{\Cscr(t)}(\d x) - \int_{\R^n} g(x) \, \mu_{\Cscr(0)}(\d x) = 
\int_0^t \frac{1}{2} \int_{\R^n}  \tr (\dot{\Ccal}(s) D^2 g(x))  \, \mu_{\Cscr(s)}(\d x)  \, \d s.
\end{equation}
Now, for $g \in C^2_{\rm c}(\R^n)$ consider the Gaussian measure $h_k(x) \d x$  on $\R$
with covariance $\frac1k$ and define $g_k := h_k \ast g \in \mathcal{S}(\R^n)$. Hence 
  \eqref{eq:derivate_covariance_integrated} holds for $g_k$ and we have a uniform convergence 
$g_k \to g$  and $D^2 g_k \to D^2 g$. Since $\Cscr(s) \ge \delta \Id$
we can pass to the limit using the dominated convergence theorem which proves
\eqref{eq:derivate_covariance_integrated} whenever $g \in C_{\rm c}^2(\R^n)$. 
Finally, for $g$  as in the lemma we let $\eta \in C_{\rm c}^\infty(\R^n)$ to be a cut-off function
that vanishes outside the unit ball $B(0,1)$ and equals $1$ in the ball $B(0, \frac12)$.
Let $g_k(x) = \varphi(\frac{x}{k}) g(x)$. Then $g_k \in C^2_{\rm c}(\R^n)$ with  $g_k \to g$ and $D^2 g_k \to D^2 g$ uniformly on 
compact subsets and 
\begin{equation}
\sup |g_k(x)| + \sup |D^2 g_k(x)| \leq C \sup \sum_{s=0}^2  |\nabla^s g(x)|.
\end{equation}
Since $\Cscr^{-1}(s) \ge \Bscr + \delta \Id$ we may pass to the limit by
the dominated convergence theorem. This shows that  \eqref{eq:derivate_covariance_integrated}
holds for all $g \in C^2(\R^n)$ which satisfy \eqref{E:Cgrowr} with $r=1$.
Finally continuity of $t \mapsto \Cscr(t)$, the bound   $\Bscr \le \Cscr^{-1}(s) - \delta \Id$ and the dominated convergence theorem imply
that $ s \mapsto \int_{\R^n} \tr (\dot{\Ccal} D^2 g(x)) \mu_{\Cscr(s)}(\d x)$ is continuous. 
This finishes the proof. 
\end{proof}

\begin{lemma}\label{Clemfinal}
Let $ \Bscr\in\sym^{\ssup{\ge}} $ and assume that $ g\in C^{2\ell}(V), \ell\in\N$, satisfies
\begin{equation}
\sup_{x\in V}\;\sum_{s=0}^{2\ell}\big|D^sg(x)\big|\ex^{-\frac{1}{2}(\Bscr x,x)}<\infty.
\end{equation}
Assume that $ \Cscr\in\sym^{\ssup{+}} $ with $ \Cscr^{-1}>\Bscr $. Then the function $H$ defined by \eqref{Hdef} satisfies
\begin{equation}
D^\ell H(\Cscr,\dot{\Cscr}_1,\ldots,\dot{\Cscr}_\ell)=\int_V\,\big(A_{\dot{\Cscr}_1}\cdots A_{\dot{\Cscr}_\ell} g\big)(x)\,\mu_{\Cscr}(\d x),
\end{equation}
where for $ f\in C^2(V) $ the operator $ A_{\dot{\Cscr}_i} $ is defined by
\begin{equation}
(A_{\dot{\Cscr}_i}f)(x)=\frac{1}{2}\sum_{i,j=1}^{\dime(V)}\dot{\Ccal}_{i,j}D^2f(x,e_i,e_j).
\end{equation}
\end{lemma}
\begin{proof}
Since we already know that $H$ is analytic in $ \Ucal $ it suffices to show the result for $ \dot{\Cscr}_1=\cdots=\dot{\Cscr}_\ell=\dot{\Cscr} $. The full result follows by polarization. It thus suffices to show that the function $ h $ in Lemma~\ref{L:der_cov1} satisfies
\begin{equation}\label{Clem3ind}
\frac{\d^k}{\d t^k}h(t)=\int_V\,\big(A^k g\big)(x)\,\mu_{\Cscr+t\dot{\Cscr}}(\d x)\quad \mbox{ for } 1\le k\le \ell,
\end{equation}
where $ A=A_{\dot{\Cscr}} $. We prove this by induction. The case $ k=1 $ is just Lemma~\ref{L:der_cov1}. Thus assume that $ k\le \ell-1 $ and \eqref{Clem3ind} holds for $ k $. Let $ \widetilde{g}:=A^k g $. Then $ \widetilde{g} $ satisfies the assumptions of Lemma~\ref{L:der_cov1}. Thus by the induction assumption and Lemma~\ref{L:der_cov1}, we obtain
$$
\begin{aligned}
\frac{\d^{k+1}}{\d t^{k+1}} h(t)&=\frac{\d}{\d t}\int_V\,\widetilde{g}(x)\mu_{\Cscr+t\dot{\Cscr}}(\d x)=\int_V\, \big(A\widetilde{g}\big)(x)\,\mu_{\Cscr+t\dot{\Cscr}}(\d x)\\&=\int_V\,\big(A^{k+1} g\big)(x)\,\mu_{\Cscr+t\dot{\Cscr}}(\d x).
\end{aligned}
$$

\end{proof}

We finally collect formulae for the derivatives up to the third order for a general dependence, that is, 
we now let $ (-\delta,\delta)\ni t\mapsto\Cscr(t)\in\sym^{\ssup{+}}(V) $ be a $ C^\ell $ map with $ \Cscr(0)^{-1}>\Bscr $ and let $g$ satisfies the assumptions of Lemma~\ref{Clemfinal}.  Then
\begin{equation}
\widetilde{h}(t):=\int_V\,g(x)\,\mu_{\Cscr(t)}(\d x)
\end{equation}
is a $ C^\ell $ map on some interval $ (-\delta^\prime,\delta^\prime) $ and the derivatives of $ \widetilde{h} $ can be computed by the chain rule. In particular we obtain the following formulae.
%\begin{alignat}1
\begin{align}
\dot{\widetilde{h}}(t)&=DH(\Cscr(t),\dot{\Cscr}(t))\label{1stder},\\
\ddot{\widetilde{h}}(t)&=D^2H(\Cscr(t),\dot{\Cscr}(t),\dot{\Cscr}(t))+DH(\Cscr(t),\ddot{\Cscr}(t))\label{2ndder},\\
\dddot{\widetilde{h}}(t)&=D^3H(\Cscr(t),\dot{\Cscr}(t),\dot{\Cscr}(t),\dot{\Cscr}(t))+3D^2H(\Cscr(t),\dot{\Cscr}(t),\ddot{\Cscr}(t))\label{3rdder} \\\nonumber &\quad  +DH(\Cscr(t),\ddot{\Cscr}(t)).
\end{align}
%\end{alignat}

In general $ D^k\widetilde{h}(t) $ is a sum of terms of the form
\begin{equation}
D^\ell H(\Cscr(t), A_1,\ldots,A_k)
\end{equation}
with
\begin{equation}
A_i=D^{j_1}\Cscr(t) \quad \mbox{ and } \sum_{i=1}^\ell j_i=k.
\end{equation}

%% file: AKM-appD-chain-30thJune-2016.tex
\chapter{Chain Rules}\label{appChain} 

Here we formulate and prove a chain rule with loss of regularity for a composition of two maps.
It turns out that proving the needed claims as well as checking their assumptions in particular cases is much simpler 
when formulated in terms of higher order one-dimensional directional derivatives and the related Peano derivatives.
We first review their properties and the mutual relations.\footnote{The present 
version of this Appendix is based on notes written by David Preiss.
He has not only provided a suitable framework for smoothness, in terms of classes $C^m_*$ and ${\widetilde C}^m$
introduced below, with  particularly clear proofs of chain rule with loss of regularity, but he has also shown (Theorem~\ref{T:*=Ham}) that functions from  $C^m_*$ have continuous, multilinear, and symmetric  directional derivatives.
Nevertheless, all deficiencies of the present Appendix are the author's fault.}

\section{Motivation}
Before we enter into the precise statement of the setting and the results we consider 
a simple  example how loss of regularity can easily arise even for seemingly innocuous maps
and we sketch the key calculation in the proof of the
 main result.
 % Theorem \ref{T:fullchain}.
Consider the space
$C^k(S^1)$ of $2 \pi$-periodic $k$ times continuously differentiable functions and the map
$F : C^k(S^1) \times \R \to C^k(S^1)$
defined by 
$$F(y,p)(t) = \sin (y( t -  p)).$$
It is easy to see that $F$ is continuous and  that the map $y \mapsto F(y,p)$ is smooth (in fact real-analytic)  as a map from $C^k(S^1)$ to  itself. For a fixed $y \in C^k(S^1) \setminus C^{k+1}(S)$ the map $p \mapsto F(y,p)$ is, however, 
not differentiable as a map from $\R$ to $C^k(S^1)$. 
It is only differentiable as a map from $\R$ to $C^{k-1}(S^1)$ and we have
$$ \frac{\partial}{\partial p} F(y,p)(\cdot )   = - \cos y(\cdot -p)   \, \,   y'(\cdot - p).  $$
Similarly $p \mapsto F(y,p)$ is a $C^l$ map to $C^{k-l}$ for $l \le k$. 
Thus each derivative with respect to $p$ leads to loss of one derivative in $y$. 
A similar phenomenon occurs if we use formula \eqref{E:C1} to compute the derivative
of the convolution maps $G(g, \Cscr) :=    g \ast  \mu_{\Cscr}$ with respect to the covariance $\Cscr$.
Our renormalisation step involves a composition of several maps of this type and one might think
that this leads to a multiple loss of regularity. The main result of this appendix, Theorem \ref{T:fullchain} below, 
shows that this is not the case. 
The behaviour of the composed map  %$H(x,p) := F(G(x,p),p)$ 
is no worse
than the behaviour of the individual maps.
% $F$ and $G$. 

To state the result informally consider scales of of Banach spaces $\bX_m \subset \bX_{m-1} \subset \ldots \subset \bX_0$, 
$\bY_m \subset \ldots \subset \bY_0$ and $\bZ_m \subset \ldots \subset \bZ_0$ as well as a Banach space $\bP$ and
maps
$$ G:  \bX_m \times \bP \to \bY_m, \qquad F: \bY_m \times \bP \to \bZ_m$$
and the composed map
$$ H(x,p) := F(G(x,p), p). $$
Informally, the assumptions on $F$ and $G$ are  that these maps are well-behaved with respect to the first argument,
but  each derivative with respect to the second argument  leads
to a loss of order one in the scale of Banach spaces, i.e., that  for all $0 \le n \le m-l$
\begin{equation}  \label{E:DlossF}  D_1^j D_2^l F(y,p):  \bY^j_{n+l} \times \bP^l \to \bZ_n  \quad \mbox{is bounded}
\end{equation}
and 
\begin{equation}  \label{E:DlossG}  D_1^j D_2^l G(x,p):  \bX^j_{n+l} \times \bP^l \to \bY_n  \quad \mbox{is bounded.}
\end{equation}
Then we want to show that
\begin{equation}  \label{E:DlossH}  D_1^j D_2^l H(y,p):  \bX^j_{n+l} \times \bP^l \to \bZ_n  \quad \mbox{is bounded.}
\end{equation}

If we assume that all natural expressions make sense this can be seen as follows. From the chain rule we deduce
inductively that $D^l_2 H(x,p,\dot{p}^l) := D^l_2 H(x,p, \dot{p}, \ldots, \dot{p})$ is a weighted sum of the terms
$$
D^k_1 D^i_2 F(G(x,p),\,  p,  \, D_2^{l_1} G(x,p, \dot{p}^{l_1}), \ldots, \,   D_2^{l_k} G(x,p, \dot{p}^{l_k} ), \, \dot{p}^i)
$$
with $k \ge 0$ and $ i + \sum_{s=1}^k l_s  = l$.
Another application of the chain rule shows that $D_1^j D_2^l H(x,p, \dot{x}^j, \dot{p}^l)$ is a weighted sum of the terms
$$ D_1^{k+\bar k} D_2^i F(G(x,p),\,  p, \,  D_1^{\bar j_1} G(x, p, \dot{x}^{\bar j_1}), 
% \ldots, D^{\bar j_m} G(x, p, \dot{x}^{\bar j_m}),
\ldots, \, D_1^{j_k} D_2^{l_k} G(x,p,  \dot{x}^{j_k}, \dot{p}^{l_k}), \,  \dot{p}^i)
$$
with $\bar j_r \ge 1$, $j_s \ge 0$, $l_s \ge 1$ and 
$$ \sum_{r=1}^{\bar k} \bar j_r + \sum_{s=1}^k j_s = j, \quad  i + \sum_{s=1}^k l_s =l.$$
In particular we have $l_s \le l-i$ and hence 
$$ D_1^{j_s} D_2^{l_s} G:  \bX_{n+l}^{j_s} \times \bP^{l_s} \to \bY_{n +l -(l-i)} = \bY_{n+i} \quad \mbox{is bounded}.$$
Moreover 
$$ D_1^{k+\bar k } D_2^i F : \bY_{n+i}^{k+\bar k} \times  \bP^i  \to \bZ_n  \quad \mbox{is bounded}.$$
Thus $\| D_1^j D_2^l H(x,p, \dot{x}^j, \dot{p}^l)\|_{\bZ_n} $ is bounded in terms of $\| \dot{x} \|_{\bX_{n+l}}^j$
and $\| \dot{p} \|_{\bP}^l$. By polarization we get the desired assertion  \eqref{E:DlossH}.
The main point in the proof of Theorem \ref{T:fullchain}  is to give a precise definition of the
informal assumptions  \eqref{E:DlossF}  and  \eqref{E:DlossG} and to show that under these assumptions
all the operations performed above make sense.

\section{Derivatives and their relations}

\subsection*{Directional derivatives}
\begin{definition} 
\label{D:directional}
Let $\bX$ and $\bY$  be  normed linear spaces, $ \Ucal\subset \bX$ open and 
$G: \Ucal\to \bY$ be a function.
{\em Directional derivatives} of $G$ at $x\in\Ucal$ in directions $\dot{x}_1,\dots,\dot{x}_j\in \bX$ are defined by
\begin{equation}
\label{E:directional}
D^j G(x,\dot{x}_1,\dots,\dot{x}_j) = \frac{\d}{\d t_j}\dots\frac{\d}{\d t_1} G(x+\sum t_k\dot{x}_k)\Big|_{t_1=\ldots=t_j=0}.
\end{equation}
\end{definition}

We will use the shorthand $D^jG(x,\dot{x}^j)=D^jG(x,\underbrace{\dot{x}, \dots,\dot{x}}_{j} )$, and later, similarly, 
$$
D^{j} G(x,\dot{x}_1^{j_1},\dots, \dot{x}_k^{j_k})=D^{j} G(x,\underbrace{\dot{x}_1,\dots,\dot{x}_1}_{j_1},\dots, 
\underbrace{\dot{x}_k,\dots,\dot{x}_k}_{j_k})
$$
with $j=\sum_{s=1}^k j_s$.

\begin{definition}
We use  $C_*^m(\Ucal ,\bY)$ to denote the set of   continuous functions $G: \Ucal\to \bY$  such that for each $j\le m$ and $\dot x\in X$, the derivative $D^j G(x,\dot{x}^j)$ exists and 
the map $(x,\dot{x})\in \Ucal\times \bX \to D^jG(x,\dot{x}^j)\in \bY$ is continuous. 
\end{definition}

\begin{remark}
The star  $*$ is added just to indicate that this is not the standard class  $C^m$
of $m$-differentiable functions. 
Also, this definition is formally much weaker than that by
Hamilton \cite{Ham82} who takes 
$G$ to be $m$-times differentiable if $D^mf: \Ucal\times\underbrace{\bX\times\dots\times \bX}_{m}\to \bY$ exists and  is continuous (jointly as a function on the product space). However,  Theorem~\ref{T:*=Ham} below shows that it actually yields the same space.
Note that  for $X= \R$ it follows directly from the definition of $C^m_*(\Ucal, \bY)$ that  $C^m_*(\Ucal ,\bY) = C^m(\Ucal ,\bY)$. We will see  in Proposition \ref{P:rel} that this identity holds whenever $X$
is finite dimensional.  \hfill $\diamond$
\end{remark}
  
In proofs, especially when  proving chain rules, it is often useful to rely on the notion of  Peano derivatives.
\begin{definition}
The Peano derivatives $G^{(n)}(x,\dot{x})$ of a function $G$ at $x$ in direction $\dot{x}$ are defined inductively by
\begin{equation}
\label{E:Peano}
G^{(n)}(x,\dot{x})=n!\lim_{t\to0}\frac{G(x+t\dot{x})-\sum_{j=0}^{n-1} \frac{G^{(j)}(x,\dot{x})}{j!} t^j}{t^{n}}
\end{equation}
whenever the derivative exists. Equivalently,
\begin{equation}
\label{E:Peanoequiv}
\Bigl\| G(x+t\dot{x})-\sum_{j=0}^n \frac{G^{(j)}(x,\dot{x})}{j!} t^j\Bigr\|_{Y} = o(t^n)\text{ as $t\to 0$.}
\end{equation}
\end{definition}

\begin{lemma}\label{triv}
We notice the following obvious properties of these derivatives.
\begin{enumerate}[\rm(a)] 
\item
$G^{(0)}(x,\dot{x})$ exists iff $G$ is continuous  at $x$ in direction $\dot{x}$; then $G^{(0)}(x,\dot{x})=G(x)$.
\item
$G^{(n)}(x,t\dot{x})=t^n G^{(n)}(x,\dot{x})$.
\end{enumerate}
\end{lemma}

We show that $C_*^n(\Ucal ,\bY)$ can be equivalently defined using the Peano derivatives.

\begin{lemma}\label{L:Prem2}
Suppose $G$ is $m$-times Peano differentiable at every point of the line segment $[x,x+\dot{x}]$
in the direction of $\dot{x}$. Then for any  $0\le j\le n\le m$,
\[\Bigl\|G^{(j)}(x+\dot{x},\dot x) -
\sum_{i=0}^{n-j} \frac{G^{(j+i)}(x,\dot{x})}{i!}\Bigr\|_{\bY}
\le 
\sup_{0\le \tau \le 1} \Bigl\|\frac{G^{(n)}(x+\tau\dot{x},\dot{x})-G^{(n)}(x,\dot{x})}
{(n-j)!}\Bigr\|_{\bY}.
\]
\end{lemma}
\begin{proof} 
The case $j=n$ is obvious. When $j<n$,
$\bX=\bY=\R$ and $\dot x=1$,
the inequality follows immediately from the mean value statement of
\cite[Theorem 2(ii)]{Oli54}.
To prove the general case, find $y^*\in \bY^*$ realizing 
the norm on the left and use the special case for the map
$t\in\R \to y^*G(x + t \dot x)\in\R$.
\end{proof}

\begin{prop}\label{P:*=P}
$G\in C^m_*(\Ucal,\bY)$ iff $G^{(n)}(x,\dot{x})$, $n\le m$ exist
and are continuous on $\Ucal\times \bX$.
Moreover, for such $G$, $D^n G(x,\dot{x}^n)=G^{(n)}(x,\dot{x})$
on $\Ucal\times \bX$ for $n\le m$.
\end{prop}

\begin{proof}
If $G\in C^m_*(\Ucal ,\bY)$ and the segment $[x,x+\dot{x}]\subset \Ucal$, 
then  the function $(-\epsilon,1+\epsilon)\ni t\mapsto G(x+t\dot{x})\in \bY$ is $m$-times continuously differentiable,
and, in view  of \cite[8.14.3 and 8.14, Problem 5]{Die60},
\begin{equation}
\label{E:D=Peano}
\Bigl\|G(x+t\dot{x})-\sum_{j=0}^n \frac{D^j G(x,\dot{x}^j)}{j!} t^k\Bigr\|_{\bY} = o(t^n)\text{ as } t\to 0,
\end{equation}
for each $n\le m$, yielding $G^{(j)}(x,\dot{x})=D^j G(x,\dot{x}^j)$, $j=0,1,\dots,m$.

For the opposite implication, suppose $G^{(m)}$ exists and is continuous on $\Ucal\times \bX$.
Given any $(x,\dot{x})\in \Ucal\times \bX$, for small enough $|t|$ we may use 
Lemma~\ref{L:Prem2} with $n=m$ and $t\dot x$ instead of $\dot x$ to infer that
for each $0\le j < n=j+1\le m$,
\[\Bigl\|G^{(j)}(x+t\dot{x},\dot x) -
\sum_{i=0}^{1} G^{(j+i)}(x,\dot{x})t^j\Bigr\|_{\bY}
= o(t)\text{ as $t\to 0$,}
\]
which says that 
$\frac{d}{d t}G^{(j)}(x+t\dot{x},\dot x)\big|_{t=0} = G^{(j+1)}(x,\dot x)$.
Hence $D^n G(x,\dot x^n)$ exists and equals to $G^{(n)}(x,\dot{x})$
for every $(x,\dot{x})\in \Ucal\times \bX$ and $0\le n\le m$.
Since $G^{(n)}$ are continuous, $G\in C^m_*(\Ucal ,\bY)$.
\end{proof}

We also show that in the presence of continuity 
it suffices to require the existence of the Peano derivatives in a rather weak sense.

\begin{lemma}\label{L:pc}
Suppose $G\colon\Ucal\to \bY$ and $g_j:\Ucal\times \bX \to \bY$, $0\le j\le m$, 
are continuous functions
such that for a weak$^*$ dense set of $y^*\in \bY^*$, $y^*\circ G$ is $m$-times Peano differentiable on $\Ucal$
with its $j$th Peano derivative being $y^*\circ g_j$.
Then $G\in C^m_*(\Ucal, \bY)$ and $D^k G(x,\dot x^j)=G^{(j)}(x,\dot x) = g_j(x,\dot x)$. 
\end{lemma}

\begin{proof}
For the $y^*$ for which the assumption holds, Proposition~\ref{P:*=P} shows that
$y^*\circ G\in C^m_*(\Ucal,\R)$ and $D^j(y^*\circ G)(x,\dot{x}^j)=y^*\circ g_j(x,\dot{x})$.
Hence, whenever the segment $[x,x+t\dot x]$ is contained in $\Ucal$,
\[y^*\Bigl(G(x+t\dot{x})-\sum_{j=0}^m \frac{g_j(x,\dot{x})}{j!} t^j\Bigr)
=\frac{1}{m!}\int_0^t (t-s)^m y^*\bigl(g_m(x+s\dot x,\dot{x})-g_m(x,\dot{x})\bigr)\,\d s.
\]
The function $s\in [0,t]\to (t-s)^m(g_m(x+s\dot x,\dot{x})-g_m(x,\dot{x}))$ is continuous, hence its
Riemann integral, say $I$, exists as an element of the completion of $\bY$. But since
by the above 
$y^*(I)=y^*\bigl( G(x+t\dot{x})-\sum_{j=0}^m \frac{g_j(x,\dot{x})}{j!} t^j\bigr)$ for 
a weak$^*$ dense set of $y^*\in \bY^*$,
\[G(x+t\dot{x})-\sum_{j=0}^m \frac{g_j(x,\dot{x})}{j!} t^j
=\frac{1}{m!}\int_0^t (t-s)^m \bigl(g_m(x+s\dot x,\dot{x})-g_m(x,\dot{x})\bigr)\,\d s.
\]
Since $g_m$ is continuous, $G$ is $m$ times Peano differentiable at
every $x\in\Ucal$ as a mapping of $\Ucal$ to $\bY$,
with continuous $G^{(j)}(x,\dot x)=g_j(x,\dot{x})$. So the statement follows from
Proposition~\ref{P:*=P}.
\end{proof}
 
The previous Lemma will be used in the situation when 
$G:\Ucal \to \bY$ and $\bY\embed \bV$ (meaning $\bY$ is a linear subspace of $\bV$
and $\|\cdot\|_{\bV}\le\|\cdot\|_{\bY}$) to require differentiability
for the map $G:\Ucal \to \bV$ only. 

\begin{cor}\label{C:pcc}
Suppose $\bY\embed \bV$ and $G\colon\Ucal\to \bY$ is $m$ times Peano differentiable
when considered as a map to $\bV$ and such that each function $G^{(j)}(x,\dot x)$, $0\le j\le m$,
has values in $\bY$ and is continuous as a map of $\Ucal \times \bX$ to $\bY$.
Then $G\in C^m_*(\Ucal, \bY)$ and $D^j G(x,\dot x^j)=G^{(j)}(x,\dot x)$. 
\end{cor}

\begin{proof}
Since $\bV^*$ is weak$^*$ dense in $\bY^*$,
Lemma~\ref{L:pc} is applicable with 
$$g_j(x,\dot x)=G^{(j)}(x,\dot x).$$ 
\end{proof}

\subsection*{Multilinearity and symmetry of derivatives}

\begin{theorem}\label{T:*=Ham}
$\bX$, $\bY$ be normed linear spaces with $\Ucal\subset  \bX$ open, and let $G\in C^m_*(\Ucal ,\bY)$. 
Then, for every $1\le j\le m$, the directional derivative 
$D^j G(x,\dot{x}_1,\dots,\dot{x}_j)$ exists for all $x\in \Ucal$ and $\dot{x}_1,\dots,\dot{x}_j\in \bX$.

Moreover, it is a continuous, symmetric, $j$-linear map
in the variables $\dot{x}_1,\dots,\dot{x}_j$ and 
$D^j G \in C^{m-j}_*(\Ucal\times \bX^j,\bY)$.
\end{theorem}

The main idea is to get information on the map $s \mapsto G^{(j)}(x + s v, \dot{x}, \ldots, \dot{x})$ by writing
$$ G(x + s(v + t \dot{x})) =  G(x + sv + st \dot{x})  $$
and using Peano differentiability of $G$ at $x$ on the left hand side and Peano differentiability at $x + sv$ on the
right hand side. A key tool is the following polynomial interpolation lemma.
Theorem \ref{T:*=Ham} will then be a consequence of Proposition~\ref{P:*=multilin} below.

\begin{lemma}\label{L:poly}
For any   $j=0,\dots, m$, let $\varPhi_j:(-s_0,s_0)\to \bX$ be bounded and 
$\varPsi_j:\R\to \bX$. Suppose that
\begin{equation}
\label{E:poly}
\sum_{j=0}^{m} s^j(\varPsi_j(t) - \varPhi_j(s)t^j) = o(s^m)
\text{ as $s\to 0$}
\end{equation}
for every $t\in\R$. Then for each $j=0,\dots, m$:
\begin{enumerate}[{\rm(a)}]
\item\label{E:polya}
The function $\varPsi_j$ is a polynomial of degree at most $j$ and
\item\label{E:polyb}
there exists a polynomial $p_j:\R\to \bX$ of degree at most $m-j$
such that $\varPhi_j(s)=p_j(s)+o(s^{m-j})$ as $s\to 0$.
\item\label{E:polyc}
Moreover, if $\widehat\varPhi_j,\widehat\varPsi_j$ also satisfy \eqref{E:poly} then%
\footnote{For $p(s)=\sum_{\ell=0}^n  p_\ell s^\ell$ we define $\norm{p}_\poly=\max_{\ell=0,\dots,n} \abs{p_\ell}$.}
\[
\norm{\widehat\varPhi_j-\varPhi_j}_\poly\le C\limsup_{s\to 0}\sup_{t\in(0,1)}
\Bigl\|\sum_{j=0}^{m} s^j(\widehat\varPsi_j-\varPsi_j(t))\Bigr\|.
\]
\end{enumerate}
\end{lemma}

\begin{proof}
Fix different $t_0,\dots,t_m\in(0,1)$ and let  $q_j$ be
the corresponding Lagrange basis 
polynomials, $q_j(t_k)=\delta_{k,j}$.
Then for every $t\in\R$,
\begin{multline} 
\label{E:psiphi}
\sum_{j=0}^{m} s^j\Bigl(\varPsi_j(t) - \sum_{k=0}^m\varPsi_j(t_k)q_k(t)\Bigr)=\\
=\sum_{j=0}^{m} s^j(\varPsi_j(t) - \varPhi_j(s)t^j) 
-\sum_{k=0}^m q_k(t)\sum_{j=0}^{m} s^j(\varPsi_j(t_k) - \varPhi_j(s)t_k^j)
=o(s^m),
\end{multline}
implying that  $\varPsi_j(t)-\sum_{k=0}^m\varPsi_j(t_k)q_k(t)=0$ for each $j=0,1,\dots,m$ and thus each $\varPsi_j(t)$ is a polynomial of degree at most $m$. Only now we use that  $\varPhi_j$ are bounded, yielding from \eqref{E:poly} that
$\sum_{k=0}^{j} s^k(\varPsi_k(t) - \varPhi_k(s)t^k) = o(s^j)$
for every $j=0,\dots,m$, and the above argument with $j$ instead of $m$
shows that $\varPsi_j$ has degree at most $j$. 

For \eqref{E:polyb}, let $0\le \ell\le m$ and find $a_k$ so that
$\sum_{i=0}^m a_kt_k^j = \delta_{j,\ell}$. 
By the degree estimate on $\varPsi_j$, $\sum_{k=0}^m a_k\varPsi_j(t_k)=0$ for $j<\ell$.
Hence
\begin{equation}
\label{E:pk}
\varPhi_\ell(s)-
\sum_{j=0}^{m-\ell} s^j\sum_{k=0}^m a_k\varPsi_{j+\ell}(t_k) = 
- s^{-\ell}\sum_{k=0}^m a_k\sum_{j=0}^{m} s^j(\varPhi_j(s)t_k^j-\varPsi_j(t_k)) 
=o(s^{m-\ell}).
\end{equation}

For \eqref{E:polyc}, we just notice that, in view of \eqref{E:pk}, the coefficients of $p_k(s)$ are linear combinations 
(with fixed coefficients) of the values $\varPsi_{j+k}(t_k)$ with $t_k\in(0,1)$.

\qedhere
\end{proof}

\begin{prop}\label{P:*=multilin}
Let $G\in C^m_*(\Ucal,\bY)$. Then for every $1\le j\le m$, the directional derivative
$D^j G(x,\dot{x}_1,\dots,\dot{x}_j)$ exists for all $x\in \Ucal$ and $\dot{x}_1,\dots,\dot{x}_j\in \bX$, it
is  symmetric and $j$-linear 
in the variables $\dot{x}_1,\dots,\dot{x}_j$, and
$D^j G\in C^{m-j}_*(\Ucal\times \bX^j,\bY)$. 
\end{prop}

\begin{proof}
We show that $f(x,\dot{x}):= G^{(1)}(x,\dot{x})$ belongs to $C_*^{m-1}(\Ucal\times \bX, \bY)$
and is linear in $\dot{x}$. Used recursively, this shows
that for each $1\le j \le m$,
$(x,\dot{x}_1,\dots,\dot{x}_j)\to D^j G(x,\dot{x}_1,\dots,\dot{x}_j)$
is $j$-linear in $\dot{x}_1,\dots,\dot{x}_j$ and belongs to
$C_*^{m-j}(\Ucal\times \bX^{j}, \bY)$. 
Recall that by Proposition~\ref{P:*=P}, $G$ is $m$-times Peano differentiable and $G^{(j)}(x,\dot{x})=D^j G(x,\dot{x}^j)$
for $j\le m$, $x\in\Ucal$, and $\dot{x}\in\bX$.

Fix $x,\dot{x},v\in \bX$ and denote
$\varPhi_j(s)=G^{(j)}(x+sv,\dot{x})/j!$ and $\varPsi_j(t)=G^{(j)}(x,v+t\dot{x})/j!$.
By definition, for each $t\in\R$,
$G(x+s(v+t\dot{x}))=\sum_{j=0}^m \varPsi_j(t)s^j +o(s^m)$.
Also, by Lemma~\ref{L:Prem2},
\begin{multline} 
\| G((x+sv)+st\dot{x})-\sum_{j=0}^m\varPhi_j(s)(st)^j\|\le\\
\le (st)^m\sup_{0\le\tau\le 1} \| G^{(m)}(x+sv+\tau st\dot{x},u)-G^{(m)}(x+sv)\|=o(s^m).
\end{multline}
Hence
$\sum_{j=0}^m s^j(\varPsi_j(t)-\varPhi_j(s)t^j)=o(s^m)$
and we see from Lemma~\ref{L:poly}\eqref{E:polya} that
$G^{(1)}(x,v+t\dot{x})=a+bt$ for some $a,b$. For $t=0$ we get $a=G^{(1)}(x,v)$
and by continuity, $b=\lim_{t\to\infty} G^{(1)}(x,v/t+\dot{x})=G^{(1)}(x,\dot{x})$.
Hence $G^{(1)}(x,v+\dot{x})=G^{(1)}(x,v)+G^{(1)}(x,\dot{x})$, and we infer that
$f(x,\dot{x})=G^{(1)}(x,\dot{x})$ is linear in the second variable.

By Lemma~\ref{L:poly}\eqref{E:polyb}, for each fixed $x,\dot{x}$ the function $g_{\dot{x}}(x)=f(x,\dot{x})$ has the Peano derivative $g_{\dot{x}}^{(j)}(x,v)$,
$j=1,\dots,m-1$. 
Moreover, continuity of Peano derivatives $G^{(n)}$ and Lemma~\ref{L:poly}\eqref{E:polyc} imply that $(x,\dot{x},v)\to g_{\dot{x}}^{(j)}(x,v)$
is continuous on $\Ucal\times\bX^2$. Since $f(x,\dot{x})$ is linear in $\dot{x}$,
\begin{equation}
f ((x,\dot{x})+t(u,\dot{u})) -  f ((x,\dot{x})) = g_{\dot{x}}(x+tu)-g_{\dot{x}}(x)+tg_{\dot{u}}(x+tu),
\end{equation}
showing that $f$ is $m-1$ times continuously Peano differentiable. Hence $f$ belongs to $C_*^{m-1}(\Ucal\times\bX,\bY) $ by 
Proposition~\ref{P:*=P}. 

Symmetry of the directional derivatives follows from the following lemma. 
\end{proof}

\begin{lemma}\label{p1}
Let $G \colon\Ucal \to\bY$ and  fix (not necessarily distinct) $\dot{x}_1,\dots,\dot{x}_k\in \bX$. Suppose that the directional derivative $x\in \Ucal \to D^j G(x,\dot{x}_1^{j_1},\dots, \dot{x}_k^{j_k})$ exists
 and is continuous
whenever $j:=j_1+\cdots + j_k \le m$.
Then for any $t_1,\dots,t_k\in\R$,
\begin{equation}\label{E:multiexp}
G^{(j)}(x, \sum_{s=1}^k t_s\dot{x}_s)=  j!\sum_{j_1+\dots+j_k=j}
D^{j} G(x,\dot{x}_1^{j_1},\dots, \dot{x}_k^{j_k})\frac{t_1^{j_1}\dots t_k^{j_k}}{j_1!\cdots j_k!}.
\end{equation}
In particular, 
$D^k G(x,(\sum_{s=1}^k t_s\dot{x}_s)^k) =G^{(k)}(x,\sum_{s=1}^k t_s\dot{x}_s) $ exists
and
\begin{equation}
D^k G(x,\dot{x}_1,\dots, \dot{x}_k) =D^k G(x,\dot{x}_{\pi(1)},\dots, \dot{x}_{\pi(k)})
\end{equation}
for  every permutation $\pi$ of \, $\{1,\dots,k\}$.
\end{lemma}

\begin{proof}
Expanding recursively and estimating errors by Lemma~\ref{L:Prem2},
we get
\begin{equation}
 G(x+t\sum t_s \dot{x}_s) = 
\sum_{j:=j_1+\dots+j_k\le m}
D^{j} G(x,\dot{x}_1^{j_1},\dots ,\dot{x}_k^{j_k})\frac{t_1^{j_1}\dots t_k^{j_k}}{j_1!\cdots j_k!}t^j +o(t^m),
\end{equation} which shows \eqref{E:multiexp}.
Since the right hand side of  \eqref{E:multiexp} is continuous in $x$, Proposition~\ref{P:*=P} used separately on each line in the direction 
$\sum_{s=1}^k t_s\dot{x}_s$ implies that the iterated derivative $D^kG(x,( \sum_{s=1}^k t_s\dot{x}_s)^k)$ exists and equals 
$G^{(k)}(x, \sum_{s=1}^k t_s\dot{x}_s)$.

Using the equality \eqref{E:multiexp} with $\sum_{s=1}^k t_s\dot{x}_s$ replaced by $\sum_{s=1}^kt_{\pi(s)}\dot{x}_{\pi(s)}$ gives the same left hand side.
Since the right side is a polynomial, the 
coefficients in front of $t_1\cdots t_k$ are equal, giving the last statement.
\end{proof}

\begin{remark}\label{R:p1order}
Notice that the order of directions in the recursive expansion can be chosen.
As a result, the assumption can be narrowed, say in the case of two directions $\{\dot x_1, \dot x_2\}$, to the assumption 
that the directional derivative $x\in \Ucal \to D^j G(x,\dot{x}_1^{j_1},\dot x_2^{j_2}, \dot{x}_1^{j_3})$ exists
 and is continuous
whenever $j:=j_1+j_2 + j_3 \le m$ and $j_3\in\{0,1\}$. \hfill $\diamond$
\end{remark}

The following Corollary is a useful criterion for 
proving that a given function on a product space belongs to $C^m_*$.
It involves partial derivatives which are defined 
and denoted in the standard way. In particular, 
$D^{j}_1D^\ell_2 G((x,p),\dot{p}^\ell,\dot{x}^j)=
D^{j+\ell}G((x,p), (0,\dot{p})^\ell,(\dot{x},0)^j)$.

\begin{cor}\label{C:2var}
Suppose $G:\Ocal\subset \bX\times \bP\to \bY$, $m\in\N$, 
and for each $j+\ell\le m$, the derivative
$(x,p,\dot{x},\dot{p})\to D^{j}_1D^\ell_2 G((x,p),\dot{p}^\ell,\dot{x}^j)$
exists and is continuous on  $ \Ocal\times \bX\times \bP$.
Then $G\in C^m_*(\Ocal ,\bY)$.
\end{cor}

\begin{proof}
Lemma~\ref{p1} shows that for each $j\le m$ the Peano derivative 
\begin{align*}
G^{(j)}((x,p),(\dot x,\dot p)) 
&= D^j G((x,p),((\dot x,0)+(0,\dot p))^j)=\\
&= \sum_{k=0}^j {j \choose k} D^j G((x,p),(0,\dot p)^k,(\dot x,0)^{j-k})=\\
&= \sum_{k=0}^j {j \choose k} D^{j-k}_1D^{k}_2 G((x,p),\dot p^k,\dot x^{j-k})
\end{align*}
exists and is continuous. Hence $G\in C^m_*(\Ocal ,Y)$ by Proposition~\ref{P:*=P}.
\end{proof}
\begin{remark}
\label{R:p1order+}
Notice that in view of Remark~\ref{R:p1order}, there is also a flexibility in the demanded order of partial derivatives in the condition in the Corollary. \hfill $\diamond$
\end{remark}

\subsection*{Relation to usual derivatives}

\begin{prop}\label{P:rel} Using $C^m(\Ucal ,\bY)$ to denote the usual spaces of Fr\'echet differentiable functions (with operator norms on multilinear forms from $L_m(\bX,\bY)$) and $m\ge 0$, we have
\[C^{m}(\Ucal,\bY)=\bigl\{G\in C^m_*(\Ucal,\bY): D^m G\in C(\Ucal,L_m(\bX,\bY))\bigr\}\supset\nobreak C^{m+1}_*(\Ucal,\bY).\]
If $\bX$ is finite dimensional then $C^m(\Ucal, \bY) = C^m_*(\Ucal, \bY)$.
\end{prop}

\begin{proof}
We first show the inclusion
 \begin{equation}  \label{eq:peano_second_inclusion}
\{ G \in C^m_*(\Ucal, \bY) : D^m G \in C(\Ucal, L_m(\bX, \bY)  \} \supset C_*^{m+1}.
\end{equation}
Let $G\in C^{m+1}_*(\Ucal,\bY)$. Given $x\in \Ucal$ find
$\delta>0$ with 
\begin{equation}
\|D^{m+1}G(x+\dot{x},\dot{x}_1,\dots,\dot{x}_{m+1})\|\le 1\
\text{ whenever }\  \max\{\|\dot{x}\|,\|\dot{x}_i\|\}\le \delta.
\end{equation}
Hence for $\|\dot{x}\|<\eps \delta^{m+1}$ and $\max_i\|\dot{x}_i\|\le 1$,
\begin{multline}
\|D^{m} G(x+\dot{x},\dot{x}_1,\dots,\dot{x}_{m})-D^{m} G(x,\dot{x}_1,\dots,\dot{x}_{m})\|=\\
 = \delta^{-m}
 \|D^{m} G(x+\dot{x},\delta \dot{x}_1,\dots,\delta \dot{x}_{m})-D^{m} G(x,\delta \dot{x}_1,\dots,\delta \dot{x}_{m})\|\le\\
\le \delta^{-m-1}
 \sup_{0<t<1} \|D^{m+1} G(x+t \dot{x},\delta \dot{x}_1,\dots,\delta \dot{x}_{m}, \delta \dot{x} /\|\dot{x}\|)\|\, \|\dot{x}\|
<\eps,
\end{multline}
yielding the inclusion.

Now we show  by induction that 
\begin{equation}  \label{eq:peano_frechet}
C^m(\Ucal, \bY) \supset \{ G \in C^m_*(\Ucal, \bY) : D^m G \in C(\Ucal, L_m(\bX, \bY) \} 
\end{equation}
since the other inclusion is obvious. 
For $m=1$ the inclusion follows from the linearity of the derivative $DG(x,\cdot)$, Proposition~\ref{P:*=P}   and Lemma~\ref{L:Prem2} applied with
$n=1$ and $j=0$.
Now assume that (\ref{eq:peano_frechet}) holds for $m-1$ and let $G \in C_*^m(\Ucal, \bY)$ 
with $D^m G \in C(\Ucal, L_m(\bX, \bY)$. By (\ref{eq:peano_second_inclusion}) applied with $m-1$ instead of $m$
we have $D^{m-1} G \in  C(\Ucal, L_{m-1}(\bX, \bY))$ and thus by induction assumption
$G \in C^{m-1}(\Ucal, \bY)$.

Define the maps $F : \Ucal \to L_{m-1}(\bX,\bY)$ and $K: \Ucal \to L(\bX, L_{m-1}(\bX,\bY))$  by
\begin{alignat}1
F(x)(\dot x_1,  \ldots, \dot x_{m-1}) &:= D^{m-1} G(x, \dot x_1, \ldots \dot x_{m-1}),  \\
K(x)(\dot x_m)(\dot x_1,  \ldots, \dot x_{m-1})  &:= D^m G (x, \dot x_1, \ldots \dot x_{m}).
\end{alignat}
Our aim is to show that $F$ is Fr\'echet differentiable at $x \in \Ucal$ and its Fr\'echet derivative
 agrees with  $K$. Then $F \in C^1(\Ucal, L_{m-1}(\bX,\bY))$ and thus $G \in C^m(\Ucal,\bY)$. 

For a  fixed  $\dot x_1, \ldots, \dot x_{m-1} \in \bX$,  let
$\varPhi(t) := F(x + t\dot x_m)(\dot x_1, \ldots, \dot x_{m-1})$ and assume that $[x, x + \dot x_m] \subset \Ucal$. Since $G \in C^m_*(\Ucal,\bY)$,
the function  $\varPhi$ is in  $C^1( (-\eps, 1 + \eps), \bY)$ and by Lemma~\ref{L:Prem2},
\begin{multline}
 \| \varPhi(1) - \varPhi(0) - \varPhi'(0) \|_{\bY}
\le \sup_{\tau \in (0,1)} \| \varPhi'(\tau ) - \varPhi'(0) \|_{\bY}  \le \\
 \le \sup_{\tau \in (0,1)} \| D^m G(x + \tau \dot x_m) - D^m G(x)\|_{L_m(\bX,\bY)} \|\dot x_1\| \ldots \| \dot x_m\|.
\end{multline}
Now $\varPhi'(0) = K(x) (\dot x_m)(\dot x_1,  \ldots, \dot x_{m-1}) $ and taking the supremum over all 
$$\dot x_1, \ldots,\dot x_{m-1}$$ 
with $\| \dot x_i \| \le 1$ we get
\begin{multline}
\| F( x + \dot x_m) - F(x) - K(x)(\dot x_m) \|_{L_{m-1}(\bX,\bY)}\le \\
\le \sup_{\tau \in (0,1)} \| D^m G(x + \tau \dot x_m) - D^m G(x)\|_{L_m(\bX,\bY)} \|\dot x_m\|.
\end{multline}
It follows from the continuity of $D^m G$ (as a map with values in $L_m(\bX,\bY))$  that $F$ is Fr\'echet differentiable with derivative $K$. 

Finally assume that $\bX$ is finite dimensional and let $G \in C^m_*(\Ucal, \bY)$.  By multilinearity of $D^m G(x, \cdot)$
and polarization we see that 
$$\| D^m G(x) - D^m G(x')\|_{L_m(\bX, \bY)} \le C(m) \sup_{v \in X : \|v \| =1} 
\|D^m G(x,v^m) - D^m G(x', v^m)  \|_{\bY}.$$
Since $(x,v) \to D^m G(x,v^m)$ is continuous and $\{ v \in  \bX : \|v\|=1 \}$ is compact it follows that
$D^m G \in C(\Ucal, L_m(\bX, \bY)$.  This finishes the proof of the proposition.

\end{proof}

\section{Chain rule with a loss of regularity}

Here we consider the chain rule showing that
$F\circ G\in C^m_*(\Ucal,\bZ)$ in the  situation when
$G:\Ucal\to \Ycal$, $F:\Ycal\to \bZ$,
where $\Ucal$ and $\Ycal$ are open subsets of $\bX$ and $\bY$, respectively,
and $G\in C^m_*(\Ucal,\bV)$ 
for some $\bY\embed \bV$ (meaning, as above, that $\bY$ is a linear subspace of $\bV$
and $\norm{\cdot}_{\bV}\le \norm{\cdot}_{\bY}$). This generalizes the chain rule
of \cite[Theorem 3.6.4]{Ham82} where $\bV=\bY$ and 
$F$ is assumed to belong to $C^m_*(\Ycal,\bZ)$.
In our situation, although
$F\circ G$ obviously makes sense, expressions such as 
$DF(G(x), DG(x,\dot x))$ may not, since derivatives of $G$
belong to $\bV$ and so not to the domain of the derivative of $F$.
So for the chain rule to hold, a natural assumptions are
that $\bY$ is dense in $\bV$ and $D^j F$ has a continuous extension 
from $\Ycal\times \bY^j$ to $\Ycal\times \bV^j$. (The density of $\bY$ in $\bV$ 
is not really needed, but is
convenient since it guarantees that the extension is unique and
$j$-multilinear 
in the last variables.) 

\begin{definition}\label{D:CVm}
We use  $C^m_{\bV}(\Ycal,\bZ)$ to denote the space of maps $F:\Ycal\subset \bY\to \bZ$ such that for any $j\le m$, the derivative  $D^j F$ exists and can be extended to a continuous map  $D^j_{\bV} F$ of $\Ycal\times \bV^j$ to $\bZ$ (with a slight abuse of notation we usually  skip the subscript $\bV$ from $D^j_{\bV}$). 
\end{definition}

\begin{remark}\label{R:CVm}\hfill
\begin{enumerate}[\rm(a)]
\item\label{rcmv:0}
For $j=0$ this requires only that $F\colon\Ycal\to \bZ$ be continuous.
\item\label{rcmv:V=Y}
Proposition~\ref{P:*=P} and the polarization formula show that it suffices
to extend the maps  $(y,\dot y) \in \Ycal\times \bY \to D^j F(y,\dot y^j)$
to continuous maps defined on $\Ycal\times \bV$. 
\item\label{rcmv:prev}
By Proposition~\ref{P:*=P}, $C^m_{\bV}(\Ycal,\bZ)\subset C^m_*(\Ycal,\bZ)$ with
equality when $\bV=\bY$. \hfill $\diamond$
\end{enumerate}
\end{remark}

\begin{lemma}\label{L:D^jF}
Let  $F\in C^m_{\bV}(\Ycal,\bZ)$ and $j\le m$.
Then $D^j_V F\in C^{m-j}_{{\bV}^{j+1}}(\Ycal\times \bV^j,\bZ)$.
\end{lemma}
\begin{proof}
By the polarization formula it suffices to show that 
$(y,v)\to \varPhi(y,v) := D_{\bV}^j F(y,v^j)$
belongs to $C^{m-j}_{\bV^{2}}(\Ycal\times \bV,\bZ)$.
Considering first $\varPhi$ as a map of $\Ycal\times \bY$ to $\bZ$
and using multilinearity of the derivative, we have
\begin{equation}\label{L:2.1}
D^k_1D^\ell_2\varPhi((y,v), \dot v^\ell,\dot y^k)
= j\cdots(j-\ell+1)D^{j+k} F(y,v^{j-\ell},\dot v^\ell,\dot y^k)
\end{equation}
for $\ell\le j$ and $k\le m-j$. Since these derivatives are zero
for $\ell>j$, we have $\varPhi\in C^{m-j}_{*}(\Ycal\times \bY,\bZ)$ by Corollary~\ref{C:2var}
and Theorem~\ref{T:*=Ham}.
Moreover, expressing $D^s\varPhi$, $0\le s\le m-j$, 
with the help of partial derivatives, we see
that these derivatives have continuous extensions to maps
$(\Ycal\times \bV)\times (\bV\times \bV)^s\to \bZ$
implying the statement.
\end{proof}

\begin{theorem} \label{T:C*Ctilde}
Suppose $\Ucal\subset \bX$ and $\Ycal\subset \bY$ are open, $\bY\embed \bV$, $G:\Ucal\to \bY$, $G(\Ucal)\subset \Ycal$,
$G\in C^m_{*}(\Ucal ,\bV)$, and $F:\Ycal\to \bZ$, $F\in C^m_{\bV}(\Ycal,\bZ)$. Then $F\circ G \in C^m_{*}(\Ucal ,\bZ)$
and $D^j(F\circ G)(x,\dot x^j)$ is a linear combination of terms
\begin{equation}\label{E:DVFGterms}
D_V^{k} F(G(x), D^{j_1} G(x,\dot x^{j_1}),\dots, D^{j_k} G(x,\dot x^{j_k}))
\end{equation}
where $j_s\ge 1$ and $\sum_{s=1}^k j_s=j$.
\end{theorem}

\begin{proof}  We will show existence and continuity of Peano derivatives of
$F\circ G$.
Let $x\in \Ucal$, $\dot{x}\in \bX$. For any $t$, working just on the segment 
$$I_t:=[G(x),G(x+t\dot{x})]\subset \bY$$
we have an estimate
\begin{multline}\label{E:FGt}
\Bigl\|F(G(x+t\dot{x}))-\sum_{s=0}^j \frac{D^s F(G(x),(G(x+t\dot{x})-G(x))^s)}{s!}\Bigr\|\\ 
\le
\sup_{y\in I_t}  
\Bigl\|\frac{D^j F(y,(G(x+t\dot{x})-G(x))^j)-D^{j}F(x,(G(x+t\dot{x})-G(x))^j)}{j!}\Bigr\|
\end{multline}
for any $j\le m$.
Here all derivatives of $F$ are applied to elements of $\bY$, so the extension
has not been used yet.
Since $(G(x+t\dot{x})-G(x))/t$ converge, in the norm $\|\cdot\|_{\bV}$, 
to $G'(x,\dot{x})$, $G'(x,\dot{x})\in\bV$ and, 
using continuity of the extended $D^j F$,
\[D^j F(y_t,((G(x+t\dot{x})-G(x))/t)^j)\to D^j F(x,G'(x,\dot{x})^j)
\text{ as $t\to 0$}\] 
whenever $y_t\in I_t$. Hence the right side 
of \eqref{E:FGt} is $o(t^j)$.
Since $x,\dot{x}$ are fixed, expanding $D^s F( G(x),(G(x+t\dot{x})-G(x))^s)$ 
is standard: $D^s F(y,\dot{y}_1,\dots,\dot{y}_s)$
has been extended to a continuous $s$-linear form on $\bV^s$,
into which one plugs a $C^j$ function $\R\to \bY\subset \bV$, 
namely $t\to G(x+t\dot{x})-G(x)$.

It follows that $F\circ G$ is $m$-times Peano differentiable with derivatives
given by the terms from the expansion
of $D^s F(G(x),(G(x+t\dot{x})-G(x))^s)$, giving \eqref{E:DVFGterms}.
These formulas show that $(F\circ G)^{(s)}$ is
continuous as a map  $\Ucal\times \bX\to \bZ$. Consequently,
$F\circ G \in C^m_{*}(\Ucal ,\bZ)$ by Proposition~\ref{P:*=P}.
\end{proof}

\section{Chain rule with parameter and a  graded loss of regularity}

In the chain rule of this section, the main point is that the inner and/or outer function depend on an additional parameter, the regularity of partial derivatives depends on the order of the derivative with respect to the parameter, and the resulting composition has the same regularity properties as the functions we are composing. In principle, this chain rule is very different from the one in Theorem~\ref{T:fullchain}, although we will reduce its proof to is.

\begin{prop}\label{P:parchain}
Suppose $\bP,\bQ,\bY,\bV$ are normed linear spaces, $\Pcal$, $\Qcal$ and $\Ycal$ are open
subsets of $\bP$, $\bQ$ and $\bY$, respectively,
$\bY=\bY_{\!\!m}\embed \bY_{\!\!m-1}\embed\dots\embed \bY_{\!\!0}$, $\varPhi\colon \Pcal\to \bY$ and 
$F:\Ycal\times\Qcal \to \bV$ are such that $\varPhi(\Pcal)\subset \Ycal$ and
for each $0\le \ell\le m$,
\begin{enumerate}[\rm(i)]
\item\label{p0.i}
$\varPhi\in C^{m-\ell}_*(\Pcal,\bY_{\!\!\ell})$;
\item\label{p0.ii}
for each $j\le m-\ell$, $D_1^jD_2^\ell F$ exists on $\Ycal\times\Qcal \times\bQ^\ell\times \bY^j $
and has a continuous 
extension to $\Ycal\times\Qcal\times \bQ^\ell\times \bY_{\!\!\ell}^j$.
\end{enumerate}
Then the map $\varPsi(p,q):=F(\varPhi(p),q)$ belongs to $C^m_*(\Pcal\times\Qcal,\bV)$ and
for each $j+\ell\le m$ the derivative $D_1^jD_2^\ell\varPsi((p,q),\dot q^\ell,\dot p^j)$ 
is a combination of terms
\begin{equation}\label{p0.e1}
D_1^{k}D_2^\ell F((\varPhi(p),q),\dot q^\ell,D^{j_1}\varPhi(p,\dot p^{j_1}),\dots, D^{j_k}\varPhi(p,\dot p^{j_k}))
\end{equation}
where $j_s\ge 1$, $\sum_{s=1}^k j_s=j$ and $D_1^{i}D_2^\ell F$ denotes the extension
from \eqref{p0.ii}.
\end{prop}

\begin{proof}
Clearly, $D_2^\ell\varPsi((p,q),\dot q^\ell)=D_2^\ell F((\varPhi(p),q),\dot q^\ell)$ exists for each $0\le \ell\le m$,
and with fixed $q$ and $\dot q$ it is a composition 
$f_{q,\dot q}\circ \varPhi$, where $f_{q,\dot q}(y)=D_2^\ell F((y,q),\dot q^\ell)$.
By (\ref{p0.i}), $\varPhi\in C^{m-\ell}_*(\Pcal,\bY_{\!\!\ell})$,
and by (\ref{p0.ii}),
$f_{q,\dot q}\in C^{m-\ell}_{\bY_{\!\!\ell}}(\Ycal,\bV)$. Hence
by Theorem~\ref{T:C*Ctilde}, the function $p\to D_2^\ell\varPsi((p,q),\dot q^\ell)$ belongs
to $C^{m-\ell}_*(\Pcal,\bV)$ and its $j$th derivative
is a combination of the terms specified in~\eqref{p0.e1}.

It remains to observe that
$(p,q)\to \bigl((\varPhi(p),q),\dot q,D^{j_1}\varPhi(p,\dot p^{j_1}),\dots, D^{j_k}\varPhi(p,\dot p^{j_k})\bigr)$
maps, by the condition $j_s\le j\le m-\ell$ and \eqref{p0.i}, 
$\Pcal\times\Qcal$ continuously to $(\Ycal\times\Qcal)\times \bQ\times \bY^k_{\!\!\ell}$
and this space is mapped by
$\bigl((y,q),\dot q,\dot y_1,\dots,\dot y_k\bigl)
\to D_1^{i}D_2^\ell F\bigl((y,q),\dot q^\ell,\dot y_1,\dots,\dot y_k\bigl)$ 
continuously to $\bV$ by \eqref{p0.ii}.
Hence each of the functions in \eqref{p0.e1} maps $\Pcal\times\Qcal$ 
continuously to~$\bV$,
implying that $\varPsi\in C^m_*(\Pcal\times\Qcal,\bV)$.
\end{proof}

\begin{cor}\label{C:parchaincor}
If, under the assumptions of Proposition~\ref{P:parchain} we are also given a function $\varUpsilon\in C^m_*(\Pcal,\bQ)$
with $\varUpsilon(\Pcal)\subset\Qcal$,
the map $\varTheta(p):=F(\varPhi(p),\varUpsilon(p))$ belongs to $C^m_*(\Pcal,\bV)$ and
for each $n\le m$, the derivative $D^n\varTheta(p,\dot p^n)$ 
is a combination of terms
\begin{equation*}
D_1^{i}D_2^k F\bigl((\varPhi(p),\varUpsilon(p)),D^{j_1}\varUpsilon(p,\dot p^{j_1}),\dots, D^{j_i}\varUpsilon(p,\dot p^{j_i}),
D^{\ell_1}\varPhi(p,\dot p^{\ell_1}),\dots, D^{\ell_k}\varPhi(p,\dot p^{\ell_k})\bigr)
\end{equation*}
where $j_s,\ell_s\ge 1$ and $\sum_{s=1}^i j_s+\sum_{s=1}^k \ell_s =n$.
\end{cor}

\begin{proof}
Observe that $\varTheta=\varPsi\circ\kappa$ where $\varPsi$ 
comes from Proposition~\ref{P:parchain}
and $\kappa:\Pcal\to \bP\times \bQ$ is $\kappa(p)=(p,\varUpsilon(p))$. 
Since $\kappa\in C^m_*(\Pcal,\bP\times \bQ)$,
$\kappa(\Pcal)\subset\Pcal\times\Qcal$
and $\varPsi\in C^m_*(\Pcal\times\Qcal,V)$, the statement follows from 
Theorem~\ref{T:C*Ctilde}.
\end{proof}

The following main chain rule is a `symmetric' version of the above, 
which is capable of being iterated.
It will be stated in the following situation. 
Let
$\bP$, $\bX = \bX_m\embed\dots\embed \bX_0$, 
$\bY = \bY_{\!\!m}\embed\dots\embed \bY_{\!\!0}$ and
$\bZ = \bZ_m\embed\dots\embed \bZ_0$
be normed linear spaces, $\Ucal\subset \bX$, $\Vcal\subset \bP$, and $\Ycal\subset \bY$  are open.
We will use
$\widetilde{\bX}_n$ to denote the closure of $\bX$ in $\bX_n$, and similarly for $\widetilde{\bY}_{\!\!n}$
and $\widetilde{\bZ}_n$. Also, we use $\bbX$ (and similarly $\bbY$ and $\bbZ$)  for the sequence
 $(\bX_m,\dots, \bX_0)$.

The class of functions we will consider may be informally described as those $G\colon\Ucal\times\Vcal\to \bY$
for which $D^j_1D^\ell_2G$ is a continuous map 
$\Ucal \times\Vcal\times \bP^\ell\times \widetilde{\bX}_{n}^j \to \bY_{\!\!n+\ell}$, i.e., $\ell$ derivatives
in the parameter $p \in \Vcal$ lead to a loss of regularity of order $\ell$ in the scale of Banach spaces.
Since this description has several interpretations, we give a rather detailed one as a formal definition.

\begin{definition}\label{D:tildeC}
For any $0\le k\le m$, we define $\widetilde C^k(\Ucal \times\Vcal,\bbX,\bbY)$ as the set of all maps $G:\Ucal\times\Vcal\to \bY$ such that
\begin{enumerate}[\rm (a)]
\item\label{tildeCa}
$G\in C^k_*(\Ucal \times\Vcal,\bY_{\!\!0})$.
\item\label{tildeCb}
For each $j+\ell\le k$, the function
$$
(x,p,\dot{x}_1,\dots, ,\dot{x}_j,\dot{p}_1,\dots, ,\dot{p}_\ell)\to D^j_1D_2^\ell G((x,p),\dot{p}_1,\dots, ,\dot{p}_\ell,\dot{x}_1,\dots, ,\dot{x}_j),
$$
which is by (a) defined as a map $\Ucal\times\Vcal \times \bX^j\times \bP^\ell\to  \bY_{\!\!0}$ has a (necessarily unique) extension to a continuous mapping 
$\Ucal\times\Vcal \times \widetilde{\bX}_{\ell}^j\times \bP^\ell\to  \bY_{\!\!0}$. This extension is also denoted $D^j_1D_2^\ell G$.
\item\label{tildeCc}
For each $0\le j \le k-\ell$ and each $0 \le n \le m-\ell$   the restriction of  $D^j_1D_2^\ell G$ (which has  been already extended by (b)) to $\Ucal\times\Vcal \times \widetilde{\bX}_{n+\ell}^j\times \bP^\ell$ has values in $\bY_{\!\!n}$ and is continuous as a mapping between these spaces.
\end{enumerate}
\end{definition}

Notice that, clearly, $\widetilde C^{i}(\Ucal \times\Vcal,\bbX,\bbY))\subset \widetilde C^{k}(\Ucal \times\Vcal,\bbX,\bbY)$ for $k\le i$.  For proving that $G\in \widetilde C^k(\Ucal \times\Vcal,\bbX,\bbY)$ the following simplification   of this definition is rather useful.

\begin{lemma}\label{L:tildeC}  Assume that $0 \le k \le m$. Then
$G\colon\Ucal\times\Vcal\to \bY$ belongs to 
$\widetilde C^k(\Ucal \times\Vcal,\bbX,\bbY)$ iff
\begin{enumerate}[\rm(i)]
\item\label{tildeCi}
as a map of $\Ucal\times\Vcal$ to $\bY_{\!\!0}$, 
$G$ has derivatives $D^j_1D^\ell_2 G((x,p),\dot p^\ell,\dot x^j)$ for all $j+\ell\le k$,
$(x,p)\in\Ucal\times\Vcal$, $\dot p\in \bP$ and $\dot x \in \bX$;
\item\label{tildeCii}
for $0\le j \le k- \ell$ and all $0 \le n \le m - \ell$ 
there is continuous map
$\varPsi_{j,\ell,n}\colon\Ucal\times\Vcal \times\widetilde{\bX}_{n+\ell} \times \bP\to \bY_{\!\!n}$
such that 
$D^j_1D^\ell_2 G((x,p),\dot p^\ell,\dot x^j) =\varPsi_{j,\ell,n}(x,p,\dot x,\dot p)$
for every $(x,p)\in\Ucal \times\Vcal$,
$\dot p\in \bP$, and $\dot x \in \bX$.
\end{enumerate}
\end{lemma}

\begin{proof}
If $G\in\widetilde C^k(\Ucal \times\Vcal,\bbX,\bbY)$, \eqref{tildeCi} and \eqref{tildeCii} are
obvious. For the opposite implication, assuming \eqref{tildeCi} and \eqref{tildeCii}
we see that for each $j+\ell\le k$,
$(x,p,\dot x,\dot p) \to D^j_1D^\ell_2 G((x,p),\dot p^\ell,\dot x^j)$
is a continuous map $\Ucal\times\Vcal\times \bX\times \bP\to \bY_{\!\!0}$.
Hence $G\in C^k_{*}(\Ucal\times\Vcal,\bY_{\!\!0})$ by 
Corollary~\ref{C:2var}, yielding \ref{D:tildeC}\eqref{tildeCa}. 
Lemma~\ref{p1} and the polarization formula
establish the function
\[(x,p,\dot x_1,\dots,\dot x_j, \dot p_1,\dots, \dot p_\ell)
\to D^j_1D^\ell_2 G((x,p),\dot p_1,\dots, \dot p_\ell,\dot x_1,\dots,\dot x_j)\]
as a combination of terms
\[(x,p,\dot x_1,\dots,\dot x_j, \dot p_1,\dots, \dot p_\ell)
\to D^j_1D^\ell_2 G((x,p),(\sum_{k\in I}\sigma_k\dot p_k)^\ell,(\sum_{k\in J}\tau_k\dot x_k)^j)\]
where $I\subset\{1,\dots,\ell\}$, $J\subset\{1,\dots,j\}$, and $\sigma_k,\tau_k=\pm 1$.
This shows that for each $0\le n\le m-\ell$, the derivative
$D^j_1D^\ell_2 G$ can be extended to a continuous map
$\widetilde\Psi_{j,\ell,n}$,
from $\Ucal\times\Vcal\times \widetilde \bX_{n+\ell}^j\times \bP^\ell$ to $\bY_{\!\!n}$.
With $n=0$ this shows  \ref{D:tildeC}\eqref{tildeCb}. For $0\le n\le m-\ell$
we see from $\bX=\bX_m\embed \bX_{n+\ell}\embed \bX_{\ell}$ that
both $\widetilde\Psi_{j,\ell,n}$ and the restriction of $\widetilde\Psi_{j,\ell,0}$
to $\Ucal\times\Vcal\times \widetilde \bX_{n+\ell}^j\times P^\ell$
are continuous as maps of $U:=\bigl(\Ucal\times\Vcal\times \widetilde \bX_{n+\ell}^j\times \bP^\ell,\|\cdot\|_{\bX_{\ell}}\bigr)$ to $\bY_{\!\!0}$. Since $\bX$ is dense in 
$(\widetilde \bX_{n+\ell},\|\cdot\|_{X_{n+\ell}})$, and so also in $(\widetilde \bX_{n+\ell},\|\cdot\|_{X_{\ell}})$, the maps
$\widetilde\Psi_{j,\ell,n}$ and 
$\widetilde\Psi_{j,\ell,0}$ coincide on a dense subset of~$U$, hence on
all of $U$, proving  \ref{D:tildeC}\eqref{tildeCc}.
\end{proof}

\begin{remark}\label{R:opposite}
Clearly, the claim remains true if one replaces  
$$D^j_1D^\ell_2 G((x,p),\dot p^\ell,\dot x^j)$$ 
with  the derivatives taken in the opposite order   (see Remark~\ref{R:p1order+}).  In the present and the following appendices, in the notation $\widetilde C^m(\Ucal\times\Vcal,\bbX,\bbY)$ we indicate, somehow pedantically but usefully for clarity in proofs,  the sequences $\X$, $\Y$ of Banach spaces. When using this notion in particular
applications, the sequences $\X$ and $\Y$ will be clear from the context and we will skip them from the notation writing just $\widetilde C^m(\Ucal\times\Vcal)$. \hfill $\diamond$ \end{remark}

For working with functions from $\widetilde C^m(\Ucal \times\Vcal,\bbX,\bbY)$ it is useful to know
that they have properties stronger than those given in the definition.

\begin{lemma}\label{p} 
Let $G\in\widetilde C^m(\Ucal\times\Vcal,\bbX,\bbY)$ and $0\le j,n\le m-\ell$.
Then 
\begin{enumerate}[\rm(1)]
\item\label{p.1}
for fixed $x\in\Ucal$, the map $p\to G(x,p)$ belongs to  $C^{\ell}_{*}(\Vcal,\widetilde \bY_{\!\!m-\ell})$;
\item\label{p.3}
for fixed $p\in\Vcal$ and $\dot p_1,\dots,\dot p_\ell\in \bP$, the (extended) map
\[(x,\dot x_1,\dots,\dot x_j)\to D_1^jD^\ell_2 G((x,p),\dot p_1,\dots,\dot p_\ell,\dot x_1,\dots,\dot x_j)\] 
belongs to 
$C^{m-\ell-j}_{\bX_{n+\ell}^{j+1}}(\Ucal\times \widetilde \bX_{n+\ell}^j,\widetilde \bY_{n})$.
\end{enumerate}
\end{lemma}

\begin{proof}
\eqref{p.1}
By Corollary~\ref{C:pcc} and  \ref{D:tildeC}\eqref{tildeCc} with $n=m-\ell$, the map  $p\to G(x,p)$ belongs to  
$C^{\ell}_{*}(\Pcal,\bY_{\!\!m-\ell})$. Hence the 
derivative $D^\ell_2 G$ is an iterated limit of elements of $\bY$ 
taken in the norm of $\bY_{\!\!m-\ell}$, and so it
belongs to $\widetilde \bY_{\!\!m-\ell}$. 

\eqref{p.3} By Lemma~\ref{L:D^jF} it suffices to show that
the function 
$$x\to D^\ell_2 G((x,p),\dot p_1,\dots,\dot p_\ell)$$ belongs to 
$C^{m-\ell}_{\bX_{\!n+\ell}}(\Ucal,\widetilde \bY_{\!\!n})$.
But this follows by the same argument as in the proof of \eqref{p.1}.
\end{proof}
%% The proof is a bit cryptic, but on the other hand the result is clear almost without proof.

\begin{remark}\label{le}
Since \eqref{p.3} puts the values of the (extended) derivatives into
the corresponding closures of $\bY$, $G$ belongs to $C^m(\Ucal\times\Vcal,\bbX,\bbY)$
iff and only if it belongs to this space when $\bX_{\!n}$ and $\bY_{\!\!n}$ are replaced by
$\widetilde \bX_{\!n}$  and $\widetilde \bY_{\!\!n}$, respectively. So, at least in proofs,
we may always assume that $\bX$ is dense in $\bX_{\!n}$ and $\bY$ in $\bY_{\!\!n}$. \hfill $\diamond$
\end{remark}

\begin{theorem}\label{T:fullchain}
Let $G\in \widetilde C^m(\Ucal\times\Vcal ,\bbX,\bbY)$, $G(\Ucal\times\Vcal)\subset \Ycal$,
$F\in \widetilde C^m(\Ycal\times\Vcal,\bbY,\bbZ)$ and define
$F\diamond G\colon \Ucal\times\Vcal \to \bZ$ by $F\diamond G(x,p):=F(G(x,p),p)$. Then 
$F\diamond G \in \widetilde C^m(\Ucal\times\Vcal ,\bbX,\bbZ)$.
\end{theorem}

\begin{proof}
By Remark \ref{le}, we may assume $\widetilde \bX_{\!n}=\bX_{\!n}$, and similarly for $\bY_{\!\!n}$ and $\bZ_n$. Set $ H:=F\diamond G $.
For fixed $x\in\Ucal$, the function $p\to H(x,p)$ is of the form of a composition
$F(\varPhi(p),\varUpsilon(p))$
where the outer function $F\colon\Ycal\times \Vcal\to \bZ$ 
and the inner functions $\varPhi(p)=G(x,p)$ and $\varUpsilon(p)=p$ satisfy the 
assumptions of Corollary~\ref{C:parchaincor} with $\bQ=\bP$, $\Qcal=\Pcal$ and $\bV=\bZ_0$.
Hence $p\to H(x,p)$ belongs to $C^m_*(\Pcal,\bZ_0)$ and
for each $\ell\le m$, the derivative $D_2^\ell H((x,p),\dot p^\ell)$ 
is a combination of terms
\begin{equation}   \label{r1.e1} 
D_1^{k}D_2^i  F\bigl((G(x,p),p),\dot p^{i},
D_2^{m_1} G((x,p),\dot p^{m_1}),\dots, D_2^{m_k} G((x,p),\dot p^{m_k})\bigr)   
\end{equation}
where $ m_s\ge 1$ and $i+\sum_{s=1}^k m_s =\ell$.

We now fix $p,\dot p$ and differentiate the function in~\eqref{r1.e1}
with respect to~$x$. 
We set
$$ K(x) := \bigl(G(x,p), D_2^{m_1}G((x,p),\dot p^{m_1}),\dots, D_2^{m_k} G((x,p),\dot p^{m_k})\bigr)$$
and
$$ L(y,\dot y_1,\dots,\dot y_k) = D_1^{k}D_2^i F\bigl((y,p),\dot p^{i},\dot y_1,\dots,\dot y_k\bigr). $$
Then the expression in  \eqref{r1.e1}  is  given by the composition $(L \circ K)(x)$.
Since $m_s \le l-i \le m-i$  we have $m - m_s \ge i$  and it follows from Lemma~\ref{p}~\eqref{p.3} (applied to the $s$-th component of $K$ with $n= m - m_s$) that 
$$ K \in C^{m-l}_*(\Ucal; \Ycal \times \bY_{i}^k). $$
% take $j=0$ and $\ell = \ell_s$ in the Lemma 
Application of Lemma~\ref{p}~\eqref{p.3} to $F$ yields that
$$ L \in C^{m-i-k}_{\bY_{\!\!i}^{k+1}}(\Ycal\times \bY_{\!\!i}^k,\bZ_0) \subset 
C^{m-\ell}_{\bY_{\!\!i}^{k+1}}(\Ycal\times \bY_{\!\!i}^k, \bZ_0 ).
$$
where the inclusion follows from the relation $l \ge i + k$.
Hence, Theorem~\ref{T:C*Ctilde} shows  $L \circ K \in C^{m-l}_*(\Ucal, \bZ_0)$ and  for each $j\le m-\ell$ the derivative of 
$D^j (L \circ K)$ (and hence the derivative $D^j_1 D^\ell_2 H$) exists and is given by  a sum of terms 
of the form 
\begin{equation}\label{r1.e2}
D^{k}_1D^{i}_2 F\Bigl((G(x,p),p),\dot p^i, D^{j_{1}}_1D^{\ell_1}_2 G((x,p),\dot{p}^{\ell_1},\dot{x}^{j_1}),
\dots,
D^{j_{k}}_1D^{\ell_k}_2 G((x,p),\dot{p}^{\ell_k},\dot{x}^{j_k})\Bigr)
\end{equation}
where $j_s+\ell_s\ge 1$, $i+\sum_{s=1}^k \ell_s =\ell$ and $\sum_{s=1}^k j_s = j$.

Finally, we rely on Lemma~\ref{p} once more. For any $s=1,\dots,k$, the map $(x,p,\dot x, \dot p) \to D^{j_{s}}_1D^{\ell_s}_2 G((x,p), \dot p^{\ell_s}, \dot x^{j_s})$ 
is a continuous map  from $\Ucal\times\Vcal\times \bX_{n_s+\ell_s} \times \bP $ to
%$\Ycal\times \bP \times \bY_{\!\! n_s}^k{j_s}\times \bP^{\ell_s}$ with any $n_s\le m-\ell_s$.
$\bY_{\!\! n_s}$ whenever $n_s\le m-\ell_s$.
Choosing $n_s= n+\ell-\ell_s$  for any fixed $n\le m-\ell$, we get 
a map $\Ucal\times\Vcal\times \bX_{n+\ell} \times \bP \to \bY_{\!\! n+\ell-\ell_s}$.
Using that $\ell_s\le\ell-i$, the derivatives have been extended so 
that the function of 
$(x,p,\dot x,\dot p)$ defined in \eqref{r1.e2} is a composition of continuous maps
\[\Ucal\times\Vcal\times \bX_{n+\ell} \times \bP \to 
\Ycal\times \bP \times \bP^i\times \bY_{\!\! n+i}^k\]
and
\[\Ycal\times \bP \times \bP^i  \times \bY_{\!\! n+i}^k\to \bZ_{n}.\]
Hence $(x,p,\dot x,\dot p)\to D_1^jD_2^\ell H((x,p),\dot p^\ell,\dot x^j)$ is continuous as a map of
$\Ucal\times\Vcal\times \bX_{n+\ell}\times \bP$ to $\bZ_{n}$ 
and we conclude from Lemma~\ref{L:tildeC} 
that $H\in \widetilde C^m(\Ucal\times\Vcal ,\bbX,\bbZ)$.
\end{proof}

\begin{remark}\label{D:remDerivative}
Let $p_0\in\Vcal$ and assume that $G(\Ucal \times B_\delta(p_0)) \subset   \Ycal$,
\begin{equation}
\| D_1^j D_2^\ell G ((x,p), \dot p^\ell, \dot x^j)  \|_{\bY_{\!\!n}}  \le C_1   \| \dot x\|_{\bX_{n + \ell}}^j \| \dot p \|^\ell
\end{equation}
for any   $(x,p,  \dot x, \dot p) \in \Ucal\times B_\delta(p_0) \times \bX_{n+\ell}\times \bP \text{ and any } 0\le j+\ell\le m, \, \, 
0 \le n \le m -l$ and
\begin{equation}
\| D_1^j D_2^\ell F ((y,p), \dot p^\ell, \dot y^j)  \|_{\bZ_n}  \le C_2   \| \dot y\|_{\bY_{n + \ell}}^j \| \dot p \|^\ell
\end{equation}
for any $(y,p,  \dot y, \dot p) \in \Ycal\times\ B_\delta(p_0) \times \bY_{\!\!n+\ell}\times \bP \text{ and any } 0\le j+\ell\le m, 
\, \, 0 \le n \le m -l$.
Then
\begin{equation}
\| D_1^j D_2^\ell H ((x,p), \dot p^\ell, \dot x^j)  \|_{\bZ_n}  \le C_3   \| \dot x\|_{\bX_{n + \ell}}^j \| \dot p \|^\ell
\end{equation}
for any $(x,p,  \dot x, \dot p) \in \Ucal\times B_\delta(p_0)\times \bX_{n+\ell}\times \bP \text{ and any } 0\le j+\ell\le m, 
\, \,  0 \le n \le m -l$,
where $C_3$ depends only on $C_1$, $C_2$ and $m$. 
In fact, since 
$D_1^j D_2^\ell H ((x,p), \dot p^\ell, \dot x^j)$ is a weighted sum of the terms in  \eqref{r1.e2} it is easy to see
that there  exists a constant 
$C(m)$ such that $C_3  \le C(m) \,  C_1 (1 + C_2^m)$. \hfill $\diamond$
\end{remark}

If we the introduce the norm
\begin{align} \label{E:D_norm_C_tilde_m}
\|G\|_{\widetilde C^m(\Ucal \times \Vcal, \bbX, \bbY)} := 
\inf \left\{
 M : \| D_1^j D_2^\ell G ((x,p), \dot p^\ell, \dot x^j)  \|_{\bY_{\!\!n}}  \le M   \| \dot x\|_{\bX_{n + \ell}}^j \| \dot p \|^\ell, 
\right.
\\
\left. \forall (x,p,  \dot x, \dot p) \in \Ucal\times \Vcal  \times \bX_{n+\ell}\times \bP \text{ and any } 0\le j+\ell\le m,  \, \, 0 \le n \le m-l
 \right\}  \nonumber
\end{align} 
then the remark implies  that $\|H\|$ can be controlled in terms of $\|F\|$ and $\|G\|$.

\section{A special case of a function $G$ that is linear in its first argument}

Here we discuss conditions assuring that  $G \in \widetilde C^m$ in a special case of linear dependence on the first variable:

\begin{lemma}
\label{L:linear} 
Let $G : \bX  \times \Vcal \to \bY$ and assume that: 
\begin{enumerate}[(i)]
\item \label{i:Llinear} For any  $p\in\Vcal$, the map $x\mapsto G(x,p)$ is linear.
\item  \label{ii:Llinear}  For any  $0 \le \ell \le m$ and any  $x\in\bX$, the map $p\mapsto G(x,p)$  is in $C_*^\ell(\Vcal, \bY_{\!\!m-\ell})$.
\item \label{iii:Llinear}  For  any $p_0 \in \Vcal$ there exists $\delta, C > 0$ such that
$$
\| D^\ell_2 G((x, p), \dot p^\ell)\|_{\bY_{\!\! n}} \le  C \| x\|_{\bX_{\!n+ \ell}} \| \dot p\|^\ell 
$$
for any $0 \le \ell \le m$, $0 \le n \le m-\ell$, and $(x,  p, \dot p) \in \bX \times B_\delta(p_0) \times \bP$.
\end{enumerate}
Then $G \in \widetilde C^m(\bX \times \Vcal,  \bbX,\bbY)$.
Moreover
\begin{equation}
\|G\|_{\widetilde C^m(B_R \times \Vcal, \bbX, \bbY)}  \le C(m) ( 1 + R) M',
\end{equation}
where 
\begin{align}
M':= \inf \left\{
 M : \| D_2^\ell G ((x,p), \dot p^\ell)  \|_{\bY_{\!\!n}}  \le M   \| \dot x\|_{\bX_{n + \ell}} \| \dot p \|^\ell, 
\right.
\nonumber \\
\left. \mbox{for any   $(x,p,  \dot x, \dot p) \in \bX \times \Vcal  \times \bP  \text{ and any } 0\le n +\ell\le m$}  \right\}
\end{align} 

\end{lemma}

\begin{proof} We will verify the conditions of Lemma~\ref{L:tildeC}.

The conditions (\ref{i:Llinear}) and  (\ref{ii:Llinear}) above imply the condition Lemma~\ref{L:tildeC}(\ref{tildeCi}).
Indeed, taking into account the linearity of $G$ in the first variable, the derivative $D_1 G((x,p),\dot x)$ exists and equals $G(\dot x,p)$ (with any norm $\norm{\cdot}_{\bY_{\!\! n}}$, $0\le n\le m$ (in particular, also $n=m-\ell$) on the target space $\bY$).
Thus $D^\ell_2 D_1 G(( x, p), \dot x, \dot p^\ell)= D^\ell_2 G((\dot x, p),  \dot p^\ell)$ and  $D^\ell_2 D_1^j G((\dot x, p),\dot x^j, \dot p^\ell)=0$ for $j\ge 2$.

Further, we show that the derivatives $(x,p,\dot p)\to D^\ell_2 G((x, p), \dot p^\ell)$ can be extended to continuous maps 
$\varPhi_{\ell,n}:  \widetilde \bX_{n+\ell}\times \Vcal\times\bP\to \bY_{\!\! n}$. 
Indeed, consider  fixed $p\in\Vcal, \dot p\in\bP$, $x\in \widetilde\bX_{n+\ell}$,    and a sequence $x_k\in \bX_{m}$ converging to $x$ in the norm of $\bX_{n+\ell}$,  $\norm{x_k-x}_{\bX_{n+\ell}}\to 0$. The derivative $D^\ell_2 G(( x_k, p), \dot p^\ell)$ 
belongs to  $\bY_{\!\! m-\ell}\embed \bY_{\!\! n}$ for each $x_k$, and in view of the bound (\ref{iii:Llinear})
we get 
\begin{equation}
\norm{D^\ell_2 G(( x_k, p), \dot p^\ell)-  D^\ell_2 G(( x_{k'}, p), \dot p^\ell)}_{\bY_{\!\! n}}\le C \norm{x_{k}-x_{k'} }_{\bX_{n+\ell}} \| \dot p\|^\ell,
\end{equation}
yielding the existence of the limit
$\varPhi_{\ell,n}(x,p,\dot p):= \lim_{k\to\infty} D^\ell_2 G(( x_k, p), \dot p^\ell)\in \bY_{\!\! n}$.
This also gives the continuity of the map $x\to \varPhi_{\ell,n}(x,p,\dot p)$. 
Combined with the continuity $(p,\dot p) \to   D^\ell_2 G(( x, p), \dot p^\ell)$ 
from the condition (\ref{ii:Llinear}), we get the continuity of $\varPhi_{\ell,n}$ as stated above.

To conclude, we introduce the continuous $\varPsi_{0,\ell,n}: \bX\times \Vcal\times  \bX_{\! n+\ell}   \times\bP\to \bY_{\!\! n}$ defined by 
$\varPsi_{0,\ell,n}(x,p,\dot x,\dot p) = \varPhi_{\ell,n}(p,x,\dot p)$ and $\varPsi_{1,\ell,n}: \bX\times \Vcal\times \bX_{\! n+\ell} \times\bP\to \bY_{\!\! n}$ 
defined by 
$\varPsi_{1,\ell,n}(x,p,\dot x,\dot p) = \varPhi_{\ell,n}(p,\dot x,\dot p)$.
For $j\ge 2$ we take $\varPsi_{j,\ell,n}(x,p,\dot x,\dot p)=0$.

The assumptions of Lemma~\ref{L:tildeC} are thus satisfied, allowing us to conclude that   $G \in \widetilde C^m(\bX \times \Vcal, \bbX,\bbY)$. 
\end{proof}

\section[Function not depending on the parameter]{A special case of  function $G$ not depending on the parameter $p$}

In applications of the chain rule it is convenient to also  consider the case of maps
that do not explicitly depend on the parameter $p$.
We get

\begin{lemma}\label{L:chainnoq}
Suppose that $G :\Ucal \times \Vcal \to \bY$ and $\tilde G: \Ucal \to \bY$
satisfy
\begin{equation}
G(x,p) = \tilde G(x)  \quad \forall (x,p) \in \Ucal \times \Vcal.
\end{equation}
Assume that 
\begin{enumerate}
\item $\tilde G \in C^m_*(\Ucal, \bY_m)$ and
\item for $1 \le \ell \le m$ the map $(x, \dot x) \mapsto D^\ell \tilde G(x, \dot{x}^\ell)$ can be extended to a continuous 
map from $\Ucal \times \bX_0$ to $\bY_0$ and for
$1 \le n \le m-1$ the restriction of this map to $\Ucal \times \bX_n$ is continuous as a map with values in $\bY_n$.
\end{enumerate}
Then $G \in \tilde C^m(\Ucal \times \Vcal, \bbX, \bbY).$
Moreover 
\begin{equation}
\|G\|_{\widetilde C^m(\Ucal \times \Vcal, \bbX, \bbY)} \le M'
\end{equation}
with 
\begin{equation}
M' = \inf\left\{ M :  \| D^j G(x, \dot{x}^\ell)\|_{\bY_n} \le M \| \dot{x}\|_{X_n}^l  \, \,   \forall (x, \dot x) \in \Ucal \times X_n
\, \, \forall  \, 0 \le n \le m \right\}.
\end{equation}
\end{lemma}

\begin{proof} 
First note that $D_2^\ell G = 0$ for $\ell \ne 0$. Let $\phi_{l,0} : \Ucal \times X_0 \to Y_0$ denote the extension
of $D^l G$ to $\Ucal \times X_0$ and let $\phi_{l,n}$ denote the restriction of $\phi_{l,0}$ to 
$\Ucal \times X_n$.    Set
\begin{equation}
\psi_{j,0,n}(x, p, \dot x, \dot p): = \phi_{l,n}(x, \dot x), \qquad \psi_{j,l,n}(x, p, \dot x, \dot p) = 0 \quad \text{if $l \ne 0$}.
\end{equation}
Then the assertion follows from Lemma \ref{L:tildeC}
\end{proof}

\section[Example and failure of implicit function theorem]{A map in $C^1\setminus C^1_*$ and   failure of the inverse functions theorem in $C^1_*$}

\begin{prop}  \label{P:D_C1*_minus_C1}
Let $H$ be an infinite dimensional separable  Hilbert space. Then there exists 
$G \in C^1_*(H,H) \cap C^\infty(H \setminus \{0\}, H)$  such that $G$ is not Fr\'{e}chet differentiable at zero.
%\setminus C^1(H,H)$.
Moreover the exists a function $F \in C^1_*(H,H) $ which satisfies 
$DF(0, \dot x) = \dot x$ but which is not invertible in any neighbourhood of $0$. 
\end{prop}

\begin{proof} Let $(e_k)_{k \in \N}$ be an orthonormal basis of $H$. We will construct $G$ as a convergent sum
\begin{equation}
G(x) = \sum_{k \in \N} G_k(x) e_k
\end{equation}
such that 
\begin{itemize}
\item $G_k \in C^\infty(H)$,
\item the support $\supp\,  G_k$ of $G_k$ is concentrated near $2^{-k} e_k$,
\item $\supp\,  G_k \cap \supp\,  G_l = \emptyset$ for $k \ne l$,
\item the gradients $\nabla G_k$ are uniformly bounded and converge weakly, but not strongly, 
to $0$ as $ k \to \infty$.
\end{itemize}

Specifically $G_k$ can be defined as follows.
Let $P_k$ denote the orthogonal projection of $H$ onto the subspace
\begin{equation}
X_k := \{ x \in H : (x, e_j) = 0 \quad \forall j \le k-1 \}.
\end{equation}
Let 
\begin{equation}
\varphi \in C_{\rm c}^\infty\bigl(( -\tfrac1{16}, \tfrac1{16})\bigr), \quad 0 \le \varphi \le 1, \quad \varphi(0) = 1,
\end{equation}
\begin{equation}
G_k(x) = 2^{-k}\varphi(\| 2^k P_k x - e_k\|^2) \, \prod_{j \le k-1} \varphi\left(2^{\frac{j+k}{2}}(x,e_j) \right).
\end{equation}
For $k = 0$ the product $\prod_{j \le k-1}$ is replaced by $1$. 
Clearly $G_k \in C^\infty(H)$. Moreover
\begin{multline}
\supp G_k \subset K_k := \Big\{  x :  |(x,e_j)| \le   \tfrac14 2^{- \frac{k+j}{2}} \text{ if } j \le k-1, 
  \\  \ \ \ \
 |(x,e_k) - 2^{-k}| \le \tfrac14 \text{ and } |P_{k+1} x| \le \tfrac14 2^{-k} \Big\}.
\end{multline}
We claim that
\begin{equation}  K_k \cap K_l = \emptyset \quad \text{if $k \ne l$.}
\end{equation}
To show this we may assume that $k < l$. If $x \in K_k \cap K_l$ then the definition of  $K_k$ implies that
$(x, e_k) \ge \frac34 2^{-k}$ while the definition of $K_l$ yields 
$| (x, e_k)|  \le \frac14 2^{-\frac{k+l}{2}}$. Since both inequalities cannot hold simulateneously we get $K_k \cap K_l= \emptyset$. 
Note also that
\begin{equation}
x \in K_k \quad \Longrightarrow \quad |x|^2  \le   \frac18 2^{-k}  + \frac{25}{16}   2^{-2k} +
\frac{1}{8} 2^{-2k}   \le 2^{-k+1} %\left( \frac18 + \frac{25}{16} + \frac16 \right) 2^{-k}.
\end{equation}
In particular if $x_0 \ne 0$ then the ball $B_{|x_0|/2}(x_0)$ intersects only finitely many of the sets $K_k$. 
Hence the sum $G = \sum_k G_k e_k$ is a finite sum in $B_{|x_0|/2}(x_0)$ and thus defines a $C^\infty$ map on that set.
Thus
\begin{equation}
G \in C^\infty(H \setminus \{0\}, H).
\end{equation}
Moreover $G_k(0) = 0$ and thus $G(0) = 0$. 

We now show that 
\begin{equation}  \label{E:D_derivG_exists}
\text{the directional derivative $D^1 G(0, \dot x)$ exists and equals $0$; and that}
\end{equation}
\begin{equation}  \label{E:D_derivG_cont}
\text{the map $(x, \dot x) \mapsto D^1 G(x, \dot x)$ is a continuous map from $H \times H$ to $H$.}
\end{equation}
To prove  \eqref{E:D_derivG_exists} we note that
$G_k(x) = 0$ if $|(x,e_k)| \le \frac12$ and $|G_k(x)| \le 1$ for all $x \in H$. Thus
\begin{equation}  \label{E:D_est_Gk}
|G_k(x)| \le 2 |(x, e_k)|.
\end{equation}
Since each function $G_k$ is in $C^\infty(H)$ it suffices to show that for each $\dot x \in H$
\begin{equation}  \label{E:D_diff_G}
\lim_{m \to \infty}   \limsup_{t \to 0} \frac{1}{t}   \Big|  \sum_{k \ge m} G_k(t \dot x)  e_k  \Big| = 0.
\end{equation}
Now by \eqref{E:D_est_Gk} and orthogonality 
\begin{equation}
\Big|  \sum_{k \ge m} G_k(t \dot x) e_k  \Big|^2 = \sum_{k \ge m} |G_k(t \dot x)|^2 
\le 4 t^2  \sum_{k \ge m} |(\dot x, e_k)|^2 = 4 t^2 |P_m \dot x|^2. 
\end{equation}
Thus
\begin{equation}
\limsup_{t \to 0} \frac{1}{t}  \Big|  \sum_{k \ge m} G_k(t \dot x)  e_k  \Big|  \le 2 |P_m \dot x|
\end{equation}
and the assertion  \eqref{E:D_diff_G} follows. 

To prove  \eqref{E:D_derivG_cont} it suffices to prove continuity at $(0, \dot x)$ since we already know that
$G \in C^\infty(H \setminus \{0\}, H)$. Thus we need to show
\begin{equation}  \label{E:D_G_C1star}
\lim_{(x,v) \to (0, \dot x)} D^1 G(x,v) = 0.
\end{equation}
Since $D^1 G$ is linear in the second argument  and since finite linear combinations
$\sum_{l=0}^M a_l e_l$ are dense in $H$ it suffices to establish the following two properties
\begin{equation} \label{E:D_bound_DG}
\|D^1G(x,v)\|  \le C \|v\|    \quad \forall (x,v) \in H \times H,
\end{equation}
\begin{equation}  \label{E:D_DG_weak_conv}
\lim_{x \to 0} D^1 G(x, e_m)  = 0  \quad \forall m \in \N.
\end{equation}
To prove the bound on $D^1 G$ note that (for $x \ne 0$)
\begin{align}
\nabla G_k(x) =  2 \varphi'(  \| 2^k P_k x - e_k \|^2)  (2 ^k P_k x - e_k)  \prod_{j \le k-1} \varphi\big( 2^{\frac{j+k}{2}} (x, e_j)\big) 
 \\
 +   \varphi(  \| 2^k P_k x - e_k \|^2)   \sum_{l \le k-1} 
       \varphi'\big( 2^{\frac{l+k}{2}} (x, e_l)\big)  2^{\frac{l-k}{2}} e_l
\prod_{ j \le k-1, j \ne l}   \varphi\big( 2^{\frac{j+k}{2}} (x, e_j)\big).  \nonumber
\end{align}
Since the vectors $e_1, \ldots, e_{k-1}, 2^k P_k x - e_k$ are orthogonal this yields, with $C' = \sup |\varphi'|^2$,
\begin{equation}
|\nabla G_k(x)|^2   \le  4 C' \frac14 + C' \sum_{l \le {k-1}}  2^{l-k}  \le 2 C'. 
\end{equation}
Since the $G_k$ have disjoint support  and since $D^1 G(0, v) = 0$ it follows that
\begin{equation}
\| D^1 G(x, v)    \| \le \sqrt 2  \sup |\varphi'|  \, \,   \| v\|     \quad  \forall (x,v) \in H \times H
\end{equation}
and thus  \eqref{E:D_bound_DG}.

To prove  \eqref{E:D_DG_weak_conv} note that $G_k(x) = 0$ if  $\| x \| \le \frac34 2^{-k}$. 
Thus for $\| x\| \le \frac34 2^{-m}$ we have
\begin{equation}
|D^1 G(x, e_m)| 
\begin{cases}
\le 2^{\frac{m -k}{2}}   & \text{if $x \in \supp\, G_k$ for some $k$,} \\
= 0 & \text{else.}
\end{cases}
\end{equation}
Now if $x \in \supp \, G_k$ and $x \to 0$ then $k \to \infty$. 
This implies  \eqref{E:D_DG_weak_conv}.

Thus we have shown that
\begin{equation}    \label{E:D_G_C1star_sum}
G \in C^1_*(H,H) \quad \text{with} \quad    D^1 G(0, \dot x) = 0 \quad \forall \dot x \in H.
\end{equation}
We finally show that $G$ is not Fr\'{e}chet differentiable at $0$. 
If $G$ was Fr\'{e}chet differentiable at $0$ the Fr\'{e}chet derivative $DG(0)$ 
would satisfy $DG(0) = 0$. Thus Fr\'{e}chet differentiability would give
\begin{equation}  \label{E:D_G_frechet}
\lim_{x \to 0} \frac{\| G(x)\|}{\|x \|} = 0. 
\end{equation}
On the other hand we have
\begin{equation}   \label{E:D_failure_Frechet}
G( 2^{-k} e_k) = G_k( 2^{-k} e_k) e_k= 2^{-k} e_k.
\end{equation}
Taking $k \to \infty$ we get a contradiction to \eqref{E:D_G_frechet}.

\medskip
To get a counterexample to the inverse function theorem in $C^1_*(H,H)$ set
\begin{equation}
F(x) := x - G(x).
\end{equation}
Then $F \in C^1_*(H,H)$ and by \eqref{E:D_G_C1star_sum}
\begin{equation}
  D^1 F(0, \dot x) = \dot x \quad \forall \dot x \in H.
\end{equation}
Now \eqref{E:D_failure_Frechet} imlies that
\begin{equation}
F( 2^{-k} e_k) = 0 = F(0)
\end{equation}
and hence there exists no neighbourhood of $0$ in which $F$ is invertible.
\end{proof}

%% file: AKM-appE-implicit-30thJune-2016.tex
\chapter{Implicit Function Theorem with Loss of  Regularity}\label{appIFT}  

Here we state and prove a version of the implicit function theorem which incorporates a loss of 
regularity and is tailored for the use in Chapters  \ref{S:initial conditions}  and \ref{S:Final}. 

We consider a function of three variables (rather than a function of two variables as in the standard
version of the implicit function theorem). The implicit function we are looking for expresses the first variable
as a function of the second and the third variable. The reason for this set-up is that the second and the
third variable play very different roles. Differentiation with the respect to the third variable
(which in our application is the renormalised coefficient in the difference operator) leads to 
a loss of regularity, while differentiation with respect to the second variable does not. 
This bad behaviour with respect to the third variable is partially compensated by the fact that
we know that $F(0,0,p) = 0$ for all values of the third variable in a neighbourhood of $0$  (and not  just for $p=0$) and that
we have uniform control of $D_1 F(0,0,p)$. 

\begin{theorem} 
\label{T:implicit}
Let $m \ge 2$. Let $\bX = \bX_m \hookrightarrow \ldots \hookrightarrow \bX_0$, $\bE$,  and $\bP$  be normed spaces,
with $\bbX=(\bX_m,\dots, \bX_0)$, $\E=(\bE,\dots,\bE)$, and $\bbX\times \E=(\bX_m\times\bE,\dots, \bX_0\times\bE)$. Further, let
 $\Ucal \subset \bX$,  $\Vcal \subset \bE$, and $\Wcal \subset \bP$  be open and assume
that $F \in \tilde C^m((\Ucal \times \Vcal) \times \Wcal; \bbX\times \E, \bbX)$, i.e., $F \in C^m_*(\Ucal \times \Vcal \times \Wcal, \bX_0)$,  for any   $j'+ j'' + \ell \le m$ the derivative
\begin{align}
&\hskip-1.3cm D_1^{j'} D_2^{j''} D_3^\ell F\   \text{can be extended to a continuous map}
\nonumber  \\ 
&\hskip-1.3cm \Ucal \times \Vcal \times \Wcal \times \bX_{\ell}^{j'}
\times \bE^{j''} \times \bP^\ell \to \bX_{0}  
\label{E:ift_regularF0}
\end{align}
and
\begin{align}
& \text{the restriction of $D_1^{j'} D_2^{j''} D_3^\ell F$ defines a continuous map}
\nonumber  \\ 
&\Ucal \times \Vcal \times \Wcal \times \bX_{n+\ell}^{j'}
\times \bE^{j''} \times \bP^\ell \to \bX_{n}    \text{ if  $0 \le n \le m-\ell$.}
\label{E:ift_regularFell}
\end{align}
Assume, moreover, that $(0, 0, 0) \in \Ucal \times \Vcal \times \Wcal$ and
\begin{equation}\label{E:ift_fix}
F(0,0, p) = 0   \text{ for all } p \in \Wcal,
\end{equation}
and, there exists   $\gamma \in (0,1) $  such that 
\begin{equation}\label{E:iftderivative}
  \|D_1 F(0,0, p) \|_{L(\bX_n, \bX_n)} \le \gamma  \text{ for any } n \le m \text{ and } p \in \Wcal.
\end{equation}
Then there exist  open subsets $\widetilde \Ucal \subset  \Ucal$, $\widetilde \Vcal \subset \Vcal$, and $\widetilde \Wcal \subset \Wcal$ 
with $0 \in \widetilde \Ucal$, $0 \in \widetilde \Vcal$, $0 \in \widetilde \Wcal$,
and a unique function $f:  \widetilde \Vcal \times \widetilde \Wcal \to \widetilde \Ucal$ such that 
\begin{equation}
F(f(\varpi, p), \varpi, p) = f(\varpi, p)    \text{ for any }  (\varpi, p) \in \widetilde \Vcal \times\widetilde \Wcal.
\end{equation}
Moreover $f \in \tilde C^m(\widetilde \Vcal \times  \widetilde \Wcal, \bX)$, i.e., 
\begin{equation}   \label{E:ift_peano}
 f \in C^n_*(\widetilde \Vcal \times  \widetilde \Wcal, \bX_{m-n})   \quad \hbox{for all $0 \le n \le m$}
 \end{equation}
and 
\begin{equation}   \label{E:ift_C_tilde}
  D_1^{j''} D_2^l f: \widetilde \Vcal \times  \widetilde \Wcal \times E^{j''} \times P^l  \to \bX_{m-l}  \quad  \hbox{is continuous}
  \end{equation}
for $j'' +l \le m$.

Finally if  $F(x, \varpi, p) = x$ and $(x, \varpi, p) \in \widetilde \Ucal \times  \widetilde \Vcal \times\widetilde \Wcal$ then $x = f(\varpi, p)$. 
The derivatives of $f$ are given by the usual formulae, 
see  \eqref{E:ift_formula_first}  for the first derivative and the inductive definitions \eqref{eq:ift_def1}  and \eqref{eq:ift_def2} for 
the higher derivatives. 

If 
$$ \| D_1^{j'} D_2^{j''} D_3^\ell F(x, \varpi, p, \dot{x}^{j'}, \dot{\varpi}^{j''}, \dot{p}^l)  \|_{X_n} \le
C_1 \| \dot{x}\|_{X_{n+l}}^{j'}  \,  \| \dot{\varpi}\|_E^{j''}  \,  \|\dot{p} \|_P^{\ell}. $$
for all $(x, \varpi, p) \in \Ucal \times \Vcal \times \Wcal$ and all $0 \le n \le m - \ell$, 
then there exists a constant $C_2 = C_2(C_1, \gamma, m)$ such that 
\begin{equation}   \label{eq:E_estimate_tilde}
 \| D_1^j D_2^\ell f(\varpi, p, \dot{\varpi}^j, \dot{p}^\ell) \|_{X_{m-l}} \le
C_2  \| \dot{\varpi}\|^j  \, \| \dot{p}\|^\ell 
\end{equation} 
for all $(\varpi, p) \in \widetilde \Vcal \times \widetilde \Wcal$.

\end{theorem}

The examples in Proposition  \ref{P:D_C1*_minus_C1} shows that the inverse function theorem  (and hence the implicit function theorem) in general does not 
hold in $C^1_*$, even when there is no loss of regularity. This is why we assume $m \ge 2$ in 
Theorem \ref{T:implicit}.

\begin{remark}
The usual implicit function theorem also holds in the $C^m_*$ spaces instead of the $C^m$ spaces
as long as $m \ge 2$. More specifically, let $\Ucal \subset \bX$, $\Vcal \subset \bE$ and  assume that 
$F \in C^m_*(\Ucal \times \Vcal, \bX)$ with $F(0,0)= 0$ and $\| D_1 F(0,0)\| \le \gamma < 1$. 
Then there exist $\widetilde \Ucal \subset \Ucal$ and $\widetilde \Vcal \subset \Vcal$ 
and $f \in C^m_*(\widetilde \Vcal, \bX)$ with $f(\widetilde \Vcal) \subset \widetilde \Ucal$ such that
$F(f(\varpi), \varpi) = f(\varpi)$ for all $\varpi \in \widetilde \Vcal$.
This follows directly from Theorem \ref{T:implicit}. Indeed, it suffices to consider the situation 
where $\bX_m = \ldots = \bX_0 = \bX$ and to   extend $F$ trivially to a function 
on $\Ucal \times \Vcal \times \bP$ which is  independent of the third argument.
Then $F$ satisfies all the hypothesis of Theorem  \ref{T:implicit} and the conclusion of the theorem
gives the desired assertion. \hfill $ \diamond $
\end{remark}

\begin{remark}
Let $\hat \Ucal = \Ucal \times \Vcal$, $\hat \bX_\ell = \bX_\ell \times \bE$.
Then, strictly speaking, the definition of $\widetilde C^m((\Ucal \times \Vcal) \times \Wcal, \bbX\times\E,\bbX)$
requires that
\begin{alignat}1
&D^j_{(x, \varpi)} D_p^\ell F  \quad \text{can be extended to a continuous map}
\nonumber  \\ \
&\hat \Ucal  \times \Wcal \times {\hat \bX}_{n+l}^{j'}
 \times \bP^\ell  \to \bX_{n}  \quad \text{if  $0 \le n \le \ell-m$ and $j+ \ell \le m$.}
\end{alignat}
In view of Corollary~\ref{C:2var} this is equivalent to (\ref{E:ift_regularFell}). \hfill $\diamond$
\end{remark}

\begin{proof}
\item[\bf{Step 1.}]  Prelimary estimates.

%\newE{Using the continuity of $D_1{j'} D_2^{j''} D_3^\ell$ %and the fact that $D_1{j'} D_2^{j''} D_3^\ell(0,0) = 0$
We claim that there exist
subsets  $\widetilde \Ucal \subset \Ucal, \widetilde \Vcal \subset  \Vcal, \widetilde \Wcal \subset  \Wcal$ that are balls around $0$
and a constant $M$ such that 
the following estimates hold:
\begin{equation}
 \label{eq:ift_est1}
\| D_1^{j'} D_2^{j''} D_3^\ell F((x,\varpi, p), \dot x^{j'}, \dot \varpi^{j''}, \dot p) \|_{\bX_{n+\ell}}   \le M \| \dot x\|^{j'}_{\bX_n} \| \dot \varpi\|_{\bE}^{j''} 
\norm{\dot p}_{\bP}^{\ell}
\end{equation}
for all  $(x,\varpi, p) \in \widetilde\Ucal \times \widetilde\Vcal \times \widetilde\Wcal$, all 
 $\dot x\in  \bX  , \dot\varpi\in \bE , \dot p \in \bP$, and all $j' + j'' +\ell=2$, $0 \le n +\ell \le m$,
\begin{equation}
\label{eq:ift_est2} 
\| D_2 F((x, \varpi, p), \dot \varpi) \|_{\bX_m} \le M  \| \dot \varpi\|_{\bE}
\quad \text{ for all }  (x,\varpi, p) \in \widetilde\Ucal \times \widetilde\Vcal \times \widetilde\Wcal, 
\end{equation}
\begin{equation}
\label{eq:ift_est3}
\|F(0, \varpi, p)  \|_{\bX_m}  \le  M \| \varpi\|_{\bE} 
\quad \text{ for all }  (\varpi, p) \in  \widetilde\Vcal \times \widetilde\Wcal,  \text{ and }
\end{equation}

\begin{equation}
\label{eq:ift_est4}
\|D_1 F(x,\varpi, p)\|_{L(\bX_n, \bX_n)} \le \tfrac{1+\gamma}2
\quad \text{ for all }   (x,\varpi, p) \in \widetilde\Ucal \times \widetilde\Vcal \times \widetilde\Wcal, 0 \le n \le m.
\end{equation}
Indeed, 
 using the joint continuity in  \eqref{E:ift_regularFell}  at $(x,\varpi, p) = 0$ and $(\dot x, \dot{\varpi}, \dot p) = 0$ 
  we see that for  $\eps =1$ there exists a $\delta \in (0,1]$
such that 
$$
\| D_1^{j'} D_2^{j''} D_3^\ell F((x,\varpi, p), \dot x^{j'}, \dot \varpi^{j''}, \dot p) \|_{\bX_{n+\ell}} < 1
$$  
if  
$\max(\| \dot x\|_{\bX_n}, \| \dot \varpi\|_{\bE}, 
\norm{\dot p}_{\bP}) <  \delta$ and  $\max(\| x \|_{\bX},  \| \varpi \|_{\bE},  \|p\|_{\bP}) < \delta$. 
By the multilinearity of   $D_1^{j'} D_2^{j''} D_3^\ell$ this implies \eqref{eq:ift_est1} if  $M \ge \delta^{-2}$. 
Similarly we see that  \eqref{eq:ift_est2} holds.
Now \eqref{eq:ift_est3} follows from \eqref{eq:ift_est2}, the assumption 
$F(0,0, p) = 0$ and Lemma~\ref{L:Prem2}. Finally \eqref{eq:ift_est4} follows from the assumption
$\|D_1 F(0,0, p)\|_{L(\bX_n, \bX_n)} \le \gamma$ and \eqref{eq:ift_est1} (applied with $\ell = 0$) provided that the radius of
$\widetilde\Ucal$ and $\widetilde\Vcal$ is chosen sufficiently small.

\item[\bf{Step 2.}] Existence, uniqueness and continuity of  $f$.

First, observe that, according to \eqref{E:ift_regularFell}, the derivative $D_1F$ defines a continuous map $D_1F: \widetilde\Ucal\times\widetilde\Vcal\times\widetilde\Wcal\times\bX_m\to\bX_m$.
Taking into account the inequality \eqref{eq:ift_est4} and, possibly, shrinking the diameters of balls $ \widetilde\Ucal$, $\widetilde\Vcal$, and $\widetilde\Wcal$, we have
\begin{equation}
\norm{F(x_1,\varpi, p)-F(x_2,\varpi, p)}_{\bX_m}\le \tfrac{1+\gamma}2   \norm{x_1-x_2}_{\bX_m}
\end{equation}
for any $x_1,x_2\in \widetilde\Ucal$ and any $\varpi\in\widetilde\Vcal$ and $p\in\widetilde\Wcal$.
Employing now the Banach fixed point theorem \cite[(10.1.1)]{Die60}
(and possibly shrinking $\widetilde \Vcal$ and $\widetilde \Wcal$ further) we get the existence of a unique map
$f: \widetilde\Vcal \times  \widetilde\Wcal \to  \widetilde\Ucal $ such that 
 $F(f(\varpi, p), \varpi, p) = f(\varpi, p)$ for any  $(\varpi, p) \in   \widetilde\Vcal \times  \widetilde\Wcal$; moreover, $f\in C^0(\widetilde\Vcal\times\widetilde\Wcal, \bX_m)$.
 %% Keep: continuity follow 
 %% from the estimate $\| f(\varpi', p') - f(\varpi, p)\| \le \frac{2}{1- \gamma} \| f(x, \varpi', p') - f(x, \varpi, p) \|$ where $x = f( \varpi, p)$ etc.

 \item[\bf{Step 3.}] Differentiability of $f$, i.e., $f\in C_*^1(\widetilde\Vcal\times\widetilde\Wcal,\bX_{m-1})$. 
 
 Using the characterisation in terms of Peano derivatives, Proposition~\ref{P:*=P}, we need to find a continuous function $f^{(1)}:(\widetilde\Vcal\times\widetilde\Wcal)\times(\bE\times\bP) \to\bX_{m-1}$ so that, for any $\varpi\times p\in \widetilde\Vcal\times \widetilde\Wcal$ and $\dot{\varpi}\times \dot p\in \bE\times \bP$, we have
 \begin{equation}  \label{eq:ift_diff1}
\lim_{t \to 0} \Bigl\| \frac{\xi(t)}{t} - f^{(1)} \Bigr\|_{\bX_{m-1}} = 0
\end{equation}
 with
 \begin{equation}  
\xi(t) := f(\varpi + t \dot \varpi, p + t \dot p) - f(\varpi, p).
\end{equation}
Introducing
\begin{equation}
G(x,\varpi, p) := F(x,\varpi, p)-x,
\end{equation}
the function $f$ is defined by
\begin{equation}
\label{E:G=0}
G(f(\varpi, p), \varpi, p) = 0 \text{ for all }  (\varpi,p )\in  \widetilde\Vcal \times  \widetilde\Wcal.
\end{equation}
Differentiating now formally the equation 
\begin{equation}
\label{E:Gt=0}
G(f((\varpi, p)+t(\dot\varpi, \dot p)), \varpi+t \varpi, p+ t\dot p)) = 0  
\end{equation}
with respect to $t$ and setting
\begin{equation}
R_1^{(1)} := D_2G((x,\varpi,p), \dot\varpi)+ D_3G((x,\varpi,p), \dot p)
\end{equation}

we expect that
\begin{equation}
\label{E:deff1}
f^{(1)}((\varpi, p),(\dot\varpi, \dot p))=  -D_1 G(x,\varpi,p)^{-1} R_1^{(1)}
%(D_2G((x,\varpi,p), \dot\varpi)+ D_3G((x,\varpi,p), \dot p))
\end{equation}
with $x=f(\varpi,p)$.

The mapping $D_1 G(x, \varpi, p): \bX_n\to \bX_n$ is bounded and invertible for any $n\le m$ since,
according to \eqref{eq:ift_est4}, 
\begin{equation}
\|D_1 G(x, \varpi, p) -\1 \|_{L(\bX_n, \bX_n)} \le \frac{1+\gamma}{2}<1
\end{equation}
and thus
\begin{equation}  \label{eq:ift_D1G_inv}
\|  D_1 G(x, \varpi, p)^{-1} \|_{L(\bX_n, \bX_n)} \le \frac{2}{1-\gamma}
\end{equation}
for any $ (x, \varpi,p )\in  \widetilde\Ucal \times  \widetilde\Vcal \times  \widetilde\Wcal$. Hence, the function 
$f^{(1)}$ introduced by \eqref{E:deff1}  is well defined.

To verify the claim \eqref{eq:ift_diff1}, we recall that $\xi$ is continuous (with values in $X_m$) and 
use  the first assertion in Lemma \ref{p} with $l=1$ and Lemma \ref{L:Prem2} to estimate
\begin{multline}
\| \underbrace{G(x + \xi(t), \varpi + t \dot{\varpi}, p + t \dot p)}_{= 0} - G(x + \xi(t), \varpi + t \dot{\varpi}, p) -
D_3 G(x + \xi(t), \varpi + t \dot{\varpi}, p, t \dot p) \|_{X_{m-1}}  \\  \ \ \ \
\le
t \, \sup_{\tau \in [0, 1]} \| D_3 G(x + \xi(t), \varpi + t \dot{\varpi}, p + \tau t \dot{p} ,  \dot p) - D_3 G(x + \xi(t), \varpi + t \dot{\varpi}, p, \dot p) \|_{X_{m-1}}  
\\    \ \ \ \
= o(t).  \hfill
\end{multline}
Similarly, using the second assertion in  Lemma \ref{p}  and Lemma \ref{L:Prem2} we get
\begin{multline}
\| G(x + \xi(t), \varpi + t \dot{\varpi}, p) - \underbrace{G(x,\varpi,p)}_{=0}  - 
D_1 G(x, \varpi,p, \xi(t)) - D_2 G(x, \varpi, p, t \dot{\varpi}) \|_{X_{m-1}} \\ \ \ \ \
=  o(t) + o(\| \xi (t)\|_{X_{m-1}}). \hfill
\end{multline}
Combining these two estimate we deduce that
\begin{equation}
\| D_1 G(x, \varpi, p) \xi(t) +  t R^{(1)}_1 \|_{X_{m-1}} \le o(t) + o(\| \xi(t) \|_{X_{m-1}})
\end{equation}.
and using \eqref{eq:ift_D1G_inv} and the definition of $f^{(1)}$ it follows that
\begin{equation}
\| \xi(t) - t f^{(1)} \|_{X_{m-1}} =  o(t) + o(\| \xi(t) \|_{X_{m-1}}).
\end{equation}
This implies first that $ \| \xi(t) \|_{X_{m-1}} \le Ct$ for small  $|t|$ and then division by $t$ yields the desired assertion
\eqref{eq:ift_diff1}.

We finally show that 
\begin{equation}   \label{E:ift_formula_first}
f^{(1)}( (\varpi, p), (\dot{\varpi}, \dot p)) =  -D_1 G(x,\varpi,p)^{-1} (D_2G((x,\varpi,p), \dot\varpi)+ D_3G((x,\varpi,p), \dot p))
\end{equation}
defines a continuos map from $\widetilde \Vcal \times \widetilde \Wcal \times  \times E \times \bP$ to $X_{m-1}$.  
Together with \eqref{eq:ift_diff1} this show that $f \in C^1_*(\widetilde \Vcal \times \widetilde \Wcal; \bX_{m-1})$.
Clearly the map 
\begin{equation}
(\varpi, p), (\dot{\varpi}, \dot p)) \mapsto D_2G((x,\varpi,p), \dot\varpi)+ D_3G((x,\varpi,p), \dot p)
\end{equation}
has the desired continuity properties. 

It thus suffices to verify the following continuity property of $D_1 G^{-1}$
for any $n$ with $0 \le n \le m$:
%The claim \eqref{eq:ift_diff1} then follows from the following continuity property of $D_1 G^{-1}$
%which holds for any $n$ with $0 \le n \le m$:
\begin{multline}  \label{eq:ift_continuity_inverse}
\text{Whenever } (x_j, \varpi_j, p_j, y_j ) \to (x, \varpi, p, y)    \text{ in  }\widetilde \Ucal 
 \times \widetilde \Vcal \times \widetilde \Wcal \times \bX_n  \\
\text{then }   D_1G(x_j, \varpi_j, p_j)^{-1} y_j \to D_1G(x, \varpi, p)^{-1} y \text{ in } \bX_n.
\end{multline}
  This would be obvious if were able to assume that $(x,\varpi, p) \to D_1 G(x, \varpi, p)$ is continuous as a map
  with values in $L(\bX_n, \bX_n)$. However, we only have continuity of $(x,\varpi, p, \dot x) \to D_1 G((x, \varpi, p), \dot x)$
 as a map from 
 $\widetilde \Ucal  \times \widetilde \Vcal \times \widetilde \Wcal \times \bX_n$ to $\bX_n$. 
 To show that  \eqref{eq:ift_continuity_inverse} holds under this weaker assumption let
 $ z := D_1G(x, \varpi, p)^{-1} y$   and   $ z_j := D_1G(x_j, \varpi_j, p_j)^{-1} y_j$.
 Then 
 \begin{equation}
D_1 G((x_j, \varpi_j, p_j), z_j - z) = (y_j - y) - (D_1 G((x_j, \varpi_j, p_j), z) - y) \to 0 \quad \text{in $\bX_n$}.
\end{equation}
Since $\| D_1G(x_j, \varpi_j, p_j)^{-1} \|_{L(\bX_n, \bX_n)} \le  2/(1-\gamma)$ it follows that 
$z_j \to z$ in $\bX_n$.

% \item[\bf{Step 4.}] The first four steps are 'as before' with some minor change in notation (see also the continuity argument for the inverse
% at the end of Step 5). Also we use
%Lemma D.4 instead of the integral formula to estimate the remainder of the Taylor expansion. 
 
\item[\bf{Step 4.}] Higher Peano derivatives and proof of  \eqref{E:ift_peano}.

Let $ 2 \le k \le m$. Employing   Proposition~\ref{P:*=P} again, we will prove that $f \in C^k_*(\widetilde \Vcal \times \widetilde \Wcal, \bX_{m-k})$
by showing that $f: \widetilde \Vcal \times \widetilde \Wcal \to \bX_{m-k}$ has continuous Peano derivatives up to order $k$.
As before  $(\varpi, p) \in \widetilde \Vcal \times \widetilde \Wcal$ and for sufficiently small $t$ let
%\begin{equation}
$\xi(t) := f(s + t \dot \varpi, p + t \dot p) - f(\varpi, p)$.
%\end{equation}
We will show by induction  that $\xi(t)$ is Peano differentiable at $0$ and that the Peano derivatives up to order $k$
can be computed by expanding the identity 
\begin{equation}
0 = G( x + \xi(t), \varpi + t \dot \varpi, p + t \dot p)), \quad \text{where $x = f(\varpi, p)$,}
\end{equation}
 to order $k$ in $t$. 
 
 Define $f^{(1)}$ by   \eqref{E:deff1}. For $k \ge 2$
 define inductively $R_k=R_k(t)=R_k(t, \varpi,p,\dot\varpi,\dot p)$ and $f^{(k)}=f^{(k)}( \varpi,p,\dot\varpi,\dot p)$ as follows,
 \begin{multline}  \label{eq:ift_def1}
 R_k(t) :=  \\\sum_{\heap{j' + j'' + \ell \le k}{ j'' + \ell \ge 1}}  \!\!\!\!\tfrac{1}{j'! \, j''! \, \ell!}D_1^{j'} D_2^{j''} D_3^\ell G
 \left( (x, \varpi, p),  \Biggl(\sum_{q=1}^{k-\ell-j''}  \frac{f^{(q)}}{q!} t^q    \Biggr)^{j'}, \dot \varpi^{j''}, \dot p^\ell   )  t^{j'' + \ell} \right) +  \\
 + \sum_{2 \le j' \le k}  \!\!\!\!\tfrac{1}{j'! }D_1^{j'}G \left((x, \varpi, p),   \Bigl(\sum_{q=1}^{k-1}  \frac{f^{(q)}}{q!} t^q    \Bigr)^{j'} \right).
 \end{multline}
Note that $R_k$ is a polynomial in $t$. We use $R_k^{(j)}$ to denote its $j$-th order derivative at $t=0$,
i.e.,  $R_k^{(j)}/ j!$ is the coefficient of $t^{j}$ in the polynomial $R_k$. Also, notice that  in the right hand side
of the equation above, only   terms $f^{(q)}$ of the order $q\le k-1$ occur.
Note also that $R_k(t)$ contains all the terms of order $t^j$ with $j \le k$ of the joint Taylor expansion of $G$ and $\xi(t)$
except for the term $D_1 G(x, \varpi, p, \xi(t))$. Thus looking on the coefficients of $t^k$ it is natural to define
\begin{equation}  \label{eq:ift_def2}
f^{(k)} :=  - D_1 G(x, \varpi, p)^{-1} R_k^{(k)},
\end{equation}
i.e., $f^{(k)}$ is the unique solution of the linear equation $D_1 G(x, \varpi, p, \dot x) + R_k^{(k)} = 0$
(we will see  below that $R_k^{(k)} \in \bX_{m-k}$ and that this equation has indeed a unique solution in $\bX_{m-k}$).

For $k \le m$, we will prove by induction that 
\begin{equation}  \label{eq:ift_induction1}
f^{(k)} \in \bX_{m-k}
\end{equation}
and that $f^{(k)}$ is the sought Peano derivative since
\begin{equation}   \label{eq:ift_induction2}
\Bigl\|\xi(t) - \sum_{q=1}^k \frac{f^{(q)}}{q!} t^q \Bigr\|_{\bX_{m-k}} = o(t^{k}).
\end{equation}

For $k=1$ the definitions of  $R_1^{(1)}$ and $f^{(1)}$ agree with 
those given in Step 3. The claims  \eqref{eq:ift_induction1} and \eqref{eq:ift_induction2} for $k=1$
were also established in Step 3. 

Assume now that \eqref{eq:ift_induction1} and \eqref{eq:ift_induction2} hold for $k-1$ and that $k \le m$.
Then it is easy to see that for all $t$ we have $R_k(t) \in \bX_{m-k}$ and in particular  $R_k^{(k)} \in \bX_{m-k}$. Indeed, if $\ell+j''\ge 1$ then 
$\sum_{q=1}^{k-\ell-j''}  \frac{f^{(q)}}{q!} t^q \in \bX_{m-k+\ell}$ and, since 
\begin{equation} D_1^{j'} D_2^{j''} D_3^{\ell} G \quad 
\text{maps  \quad   $\Ucal \times \Vcal \times \Wcal \times \bX_{m-k+\ell}^{j'} \times \bE^{j''} \times \bP^\ell$ to  $\bX_{m-k}$},
\end{equation}
the first sum in the definition of  $R_k(t)$ is in $\bX_{m-k}$.
If $\ell=j''=0$, then $\sum_{q=1}^{k-1}  \frac{f^{(q)}}{q!} t^q \in \bX_{m-k+1}$ which is mapped by $D_1^{j'}G(x,\varpi,p)$ into $\bX_{m-k+1}$
implying that the second sum in the definition of   $R_k(t)$  is contained in $\bX_{m-k+1} \subset \bX_{m-k}$.
 We have seen in Step 3  that the map $\dot x \mapsto D_1 G((x,\varpi, p), \dot x)$ is bounded and invertible
as a map from $\bX_{n}$ to $\bX_n$ for all $0 \le n \le m$. Hence, the definition \eqref{eq:ift_def2}  implies that $f^{(k)}$ is well defined and lies in $\bX_{m-k}$.

To prove  \eqref{eq:ift_induction2},  we first define 
 \begin{multline}
 \widetilde R_k(t) :=   \sum_{\heap{j' + j'' + \ell\le k}{ j'' + \ell \ge 1} }   \!\!\!\!\tfrac{1}{j'! \, j''! \, \ell!}D_1^{j'} D_2^{j''} D_3^\ell
 ( (x, \varpi, p),  \xi(t)^{j'}, \dot\varpi^{j''}, \dot p^\ell   )  t^{j'' + \ell}  +\\
  + \sum_{2 \le j' \le k}  \!\!\!\!\tfrac{1}{j'! }D_1^{j'}G ((x, \varpi, p),   \xi(t)^{j'} ).
 \end{multline}
Similar to the estimate for the first derivative, it  follows from Lemma~\ref{p},  Lemma~\ref{L:Prem2} and Proposition~\ref{P:*=P} (c.f. also  Lemma~\ref{p1}) that
\begin{equation}
 \Bigl\|\underbrace{G(x + \xi(t), \varpi + t \dot \varpi, p + t \dot p)}_{=0} - \underbrace{G(x, \varpi, p)}_{=0}  
 - D_1 G((x,\varpi, p), \xi(t)) - \widetilde R_k(t)\Bigr\|_{\bX_{m-k}} \le
 \end{equation}
 \begin{multline}
 \nonumber
  \le  \sup_{\tau \in [0,1]} 
\Bigl\|\sum_{j'' + \ell =k}  \tfrac{1}{j''! \, \ell!} \Bigl(D_2^{j''} D_3^\ell G((x + \tau \xi(t), \varpi + \tau t \dot \varpi, p + \tau t\dot p), \dot\varpi^{j''}, \dot p^{\ell})-\\
 - D_2^{j''} D_3^\ell G((x, \varpi, p), \dot\varpi^{j''}, \dot p^{\ell})\Bigr)
\Bigr\|_{\bX_{m-k}} t^k
 \end{multline}
 \begin{multline}
  \nonumber
\!\!\!\!\!\!+  \sup_{\tau \in [0,1]} 
\Bigl\|   
\!\!\!\sum_{\heap{j' + j'' + \ell =k}{ j' \ge 1}} \!\!\!\!\tfrac{1}{j'! \, j''! \, \ell!}   \Bigl(
 D_1^{j'} D_2^{j''} D_3^\ell G((x, + \tau \xi(t), \varpi + \tau t\dot \varpi, p + \tau t\dot p),  
  (\tfrac{\xi(t)}{ t})^{j'}, \dot\varpi^{j''}, \dot p^{\ell}) -\\
  -      \
 D_1^{j'} D_2^{j''} D_3^\ell G((x, \varpi, p),  (\tfrac{\xi(t)}{ t})^{j'}, \dot\varpi^{j''}, \dot p^{\ell})\Bigr)
  \Bigr\|_{\bX_{m-k}} t^k 
\end{multline}
The first term on the right hand side is $o(t^k)$ since $D_2^{j''} D_3^\ell G$ is continuous in all of its arguments 
and since $\xi(t)\to 0$ in $\bX_m$.
For the second term we use that $\ell \le k-1$ since $j'  \ge 1$ and that, as proven in the Step 3,  the function $\xi(t)/ t$ converges to $f^{(1)}$ in $\bX_{m-1}$. 
As a result, observing that $D_1^{j'} D_2^{j''} D_3^\ell$ is a continuous map from 
$\Ucal \times \Vcal \times \Wcal \times \bX_{m-1}^{j'} \times E^{j''} \times \bP^\ell$ to $\bX_{m-1-\ell} \hookrightarrow \bX_{m-k}$,
 the second term is also $o(t^k)$.
In summary,
\begin{equation}
\| D_1 G((x,\varpi, p), \xi(t))  +  \widetilde R_k(t) \|_{\bX_{m-k}}  = o(t^{k}).
\end{equation}

Combining the induction assumption, 
\begin{equation}  \label{eq:E_step4_induction}
\Bigl\|  \xi(t) - \sum_{q=1}^{k-j''- \ell} \frac{f^{(q)}}{q!} t^q \Bigr\|_{\bX_{m-k+\ell+j''}} = o(t^{k-j''-\ell})
\end{equation}
valid for  any $j'' + \ell \ge 1$ with the estimate
 $\|\sum_{q=1}^{k-j''- \ell} \frac{f^{(q)}}{q!} t^q\|_{\bX_{m-k+ \ell + j''}} \le 3C t$
 which follows from  \eqref{eq:E_step4_induction}
 and the bound $\norm{\xi(t)}_{X_{m-1}}\le  C t$ proven in Step 3,
we can evaluate every term occurring in $R_k-\widetilde R_k$. 
Namely, we bound
\begin{equation}
\Bigl\| D_1^{j_1'+j_2'} D_2^{j''} D_3^\ell G
 ( (x, \varpi, p),  \Bigl(\sum_{q=1}^{k-\ell-j''}  \tfrac{f^{(q)}}{q!} t^q    -\xi(t)\Bigr)^{j'_1}, \xi(t)^{j_2'},\dot \varpi^{j''}, \dot p^\ell   )  t^{j'' + \ell} \Bigr\|_{\bX_{m-k}} = o(t^{k}).
\end{equation}
Here we took into account that the difference $R_k-\widetilde R_k$ contains only terms  with $j'_1\ge 1$  implying that $o((t^{k-j''-\ell})^{j'_1}) t^{j'' + \ell}t^{j'_2}=o(t^k)$
since $(k-j''-\ell)j'_1+j''+\ell +j'_2\ge k+ (j_1'-1)(k-j''-\ell)+j_2'\ge k$. Similarly for the remaining terms, 
\begin{equation}
\Bigl\| D_1^{j_1'+j_2'}  G
 ( (x, \varpi, p),  \Bigl(\sum_{q=1}^{k-1}  \tfrac{f^{(q)}}{q!} t^q    -\xi(t)\Bigr)^{j'_1}, \xi(t)^{j_2'} )\Bigr\|_{\bX_{m-k}} = o(t^{k})
\end{equation}
since $j_1'\ge 1$ and $j'_1+j'_2\ge 2$ and thus
 $o((t^{k-1})^{j'_1}) t^{j'_2}=o(t^{k})t^{(k-1)(j'_1-1)+j'_2-1}=o(t^k)$.

As a result, we can conclude that
\begin{equation}
\|  R_k(t) -\widetilde R_k(t)\|_{\bX_{m-k}} = o(t^k)
\end{equation}
and thus
\begin{equation}  \label{eq:ift_one}
\| D_1 G((x,\varpi, p), \xi(t))  +  R_k(t) \|_{\bX_{m-k}}  = o(t^k).
\end{equation}

Moreover one can easily check that for any  $q\le k$  
\begin{equation}
\| R_{q}(t) - R_k(t) \|_{\bX_{m-k}} = o(t^{q})
\end{equation}
and thus the  derivatives  of order $q$ at $0$ satisfy
$R_{q}^{(q)} = R_k^{(q)}$.
Now the definition of $f^{(q)}$ for $q \le k$ implies that
\begin{equation}
D_1 G((x,\varpi, p), f^{(q)}) = - R_{q}^{(q)} = - R_k^{(q)}.
\end{equation}
Thus
\begin{equation}
\|D_1 G((x, \varpi, p), \sum_{q=1}^{k} \frac{f^{(q)}}{q!}t^q) + R_k(t)  \|_{\bX_{m-k}} = o(t^k)
\end{equation}
since $R_k$ is a polynomial with values in $\bX_{m-k}$. 
Comparison with \eqref{eq:ift_one} yields
\begin{equation}
\|D_1 G((x, \varpi, p), \xi(t) - \sum_{q=1}^{k} \frac{f^{(q)}}{q!}t^q)  \|_{\bX_{m-k}} = o(t^k)
\end{equation}
and this implies  the claim \eqref{eq:ift_induction2} since $\dot x \mapsto G((x,\varpi, p), \dot x)$ is
a bounded and invertible map from $\bX_{m-k}$ to itself. 

We have thus shown that  for any $n \le m$ the map $f: \Vcal \times \Wcal \to \bX_{m-n}$ has Peano derivatives  for any $k\le n$
given by
\begin{equation}
f^{(k)}((\varpi, p), (\dot \varpi, \dot p)) = f^{(k)},
\end{equation}
where $f^{(k)}$ is inductively defined by \eqref{eq:ift_def1} and \eqref{eq:ift_def2} with $x = f(\varpi, p)$. 
It follows by induction that the maps
\begin{alignat}1
(\varpi, p, (\dot \varpi, \dot p)) &\mapsto R_k^{(k)}, \\
(\varpi, p, (\dot \varpi, \dot p)) &\mapsto f^{(k)}  \label{eq:ift_continuity}
\end{alignat}
are continuous as maps from $\widetilde \Vcal \times \widetilde \Wcal \times E \times \bP$ to $\bX_{m-n}$
(here we use again \eqref{eq:ift_continuity_inverse}).

Thus $f^{(n)}$ exists and is continuous on $(\widetilde \Vcal \times \widetilde \Wcal, \bX_{m-n})$. By Proposition~\ref{P:*=P}, the existence and continuity of Peano derivatives $f^{(n)}$ thus 
finally implies  that $f \in C_*^n(\widetilde \Vcal \times \widetilde \Wcal, \bX_{m-n})$ for all $n \le m$.

\item[\bf{Step 5.}] Improved estimates for $D^j_1 D^{\ell}_2 f$ and proof of  \eqref{E:ift_C_tilde}.\\
For $j=0$ there is nothing to show since
$D^l_2 f(\varpi, p,  \dot{p}^\ell) = f^{(l)}(\varpi, p, 0, \dot{p})$ and thus
 \eqref{E:ift_C_tilde} follows from  \eqref{E:ift_peano}.
For $j \ge 1$ set 
$$ n := j + \ell$$ and  note  that 
\begin{equation}  \label{eq:E_ift_expand}
\frac{1}{n!} f^{(n)}(\varpi, p, \dot{\varpi}, s \dot{p})  = 
\sum_{l=0}^n  s^l   \frac{1}{j!} \frac{1}{\ell!}  D_1^j D_2^l f(\varpi, p, \dot{\varpi}^j, \dot{p}^\ell)
\end{equation}

Thus, up to a constant factor, $D_1^{j} D_2^\ell  f$ is  given by the coefficient of $s^l$ in the polynomial 
$ s \mapsto   f^{(n)}(\varpi, p, \dot{\varpi}, s \dot{p})$. 
Using this observation we will now prove \eqref{E:ift_C_tilde} by induction over $n$. 

For $n=1$ the assertion follows directly from  \eqref{E:ift_formula_first}. 

Assume the assertion has been shown for $j +l \le n-1$ (where $n \le m$). We will show 
the assertion for $j+ l = n$. In view of 
\eqref{eq:ift_def2} it suffices to show the following: 
If 
$ R^{(n)}_{n,l}(\varpi, p, \dot{\varpi}, \dot{p})$ is the coefficient of $s^l$ in  the polynomial
$$ h(s):=  R^{(n)}_n(\varpi, p, \dot{\varpi}, s \dot{p}) $$
then 
$$ R^{(n)}_{n,l}: \widetilde \Vcal \times \widetilde \Wcal \times E \times P \to X_{m-l} \quad \hbox{is continuous.} $$
To see this note that 
$h(s)$
is a weighted sum of terms of the form
$$  D^{j'}_1 D_2^{j''} D_3^{\ell'} F(x,\varpi, p, f^{(q_1)},  \ldots, f^{(q_{j'})}, \dot{\varpi}^{j''},  \dot{p}^{\ell'}) \, \, s^{\ell'}$$
with $f^{(q_i)} = f^{(q_i)}(\varpi, p, \dot{\varpi}, s \dot{p})$ and
terms of the form 
$$ D^{j'}_1 F(x,\varpi, p, f^{(q_1)},  \ldots, f^{(q_{j'})}).$$
Using  \eqref{eq:E_ift_expand} we see that $R^{(n)}_{n,l}$ is a weighted sum of terms
$$T_1 :=  D^{j'}_1 D_2^{j''} D_3^{\ell'} F(x,\varpi, p, D_1^{a_1} D_2^{\ell_1} f , \ldots, D_1^{a_{j'}} D_2^{\ell_{j'} } f,  \dot{\varpi}^{j''},  \dot{p}^{\ell'}) 
\quad \hbox{with $\ell_i \le \ell - \ell'$} $$
and 
of terms 
$$ 
T_2:= D^{j'}_1  F(x,\varpi, p, D_1^{a_1} D_2^{\ell_1} f , \ldots, D_1^{a_{j'}} D_2^{\ell_{j'} } f)  \quad   \hbox{with $q_i \le \ell$} 
$$
where 
$$ D_1^{a_i} D_2^{\ell_i} f =  D_1^{a_i} D_2^{\ell_i} f(\varpi, p, \dot{\varpi}^{a_i}, \dot{p}^{l_i}).$$
Now by induction assumption 
$$  D_1^{a_i} D_2^{\ell_i} f:  \widetilde \Vcal \times \widetilde \Wcal \times E^{a_i} \times P^{l_i} \to X_{m-(\ell- \ell')}$$
is continuous if $\ell_i \le \ell - \ell'$.
Thus $T_1: \widetilde \Vcal \times \widetilde \Wcal \times E \times P \to X_{m-\ell}$ is continuous.
Similarly one shows continuity of $T_2$.

\item[\bf{Step 6.}] Proof of   \eqref{eq:E_estimate_tilde}.\\
This is proved by induction over $n = j +l$ very similar to Step 5.

\end{proof}

%% file: AKM-appF-geometry-30thJune-2016.tex
\chapter{Geometry of Course Graining}\label{appF} 

We will use two  combinatorial lemmas  (Lemma 6.15 and 6.16 from \cite{B07}) proven by Brydges that are for completeness summarised below.

\begin{lemma}
\label{L:Xk+1largeandall}
Let $X\in \Pcal_k^{\com}\setminus \Scal_k$. Then
\begin{equation}
\label{E:Xk+1large}
\abs{X}_k\ge (1+2\upalpha(d))\abs{\overline{X}}_{k+1}\ \text{ with }\ \upalpha(d)=\tfrac1{(1+2^d)(1+6^d)}.
\end{equation}
For any $X\in \Pcal_k$ we have
\begin{equation}
\label{E:Xk+1all}
\abs{X}_k\ge (1+\upalpha(d))\abs{\overline{X}}_{k+1} - (1+\upalpha(d)) 2^{d+1} \abs{\Ccal(X)} \ \text{ with }\ \upalpha(d)=\tfrac1{(1+2^d)(1+6^d)}.
\end{equation}
\lsm[agx]{$\upalpha(d)$}{$=\tfrac1{(1+2^d)(1+6^d)}$ from the bound $\abs{X}_k\ge (1+\upalpha(d))\abs{\overline{X}}_{k+1} - (1+\upalpha(d)) 2^{d+1} \abs{\Ccal(X)}$ for any $X\in \Pcal_k$}% 
\end{lemma}

\begin{lemma} 
\label{L:Peierls}
There exist $\updelta=\updelta(d,L)<1$ such that 
\begin{equation}
\label{E:Peierls}
\sum_{\begin{subarray}{c}  X\in  \Pcal_k^{\rm c} \setminus \Scal_k \\  \overline X= U \end{subarray}}  \updelta^{|X|_k}\le 1
\end{equation}
for any $k\in\N$ and any $U\in\Pcal_{k+1}^{\rm c}$.
\end{lemma}
\begin{proof}
For any $X$ contributing to the sum we have $\abs{X}_k\ge (1+2\upalpha(d))\abs{\overline{X}}_{k+1}$ and thus
\begin{equation}
\sum_{\begin{subarray}{c}  X\in  \Pcal_k^{\rm c} \setminus \Scal_k \\  \overline X= U \end{subarray}}  \updelta^{|X|_k}\le
2^{L^d\abs{U}_{k+1}} \updelta^{(1+2\upalpha(d))\abs{U}_{k+1}}\le 1
\end{equation}
once $\updelta\le 2^{-\frac{L^d}{1+2\upalpha(d)}}$.
\end{proof}

%% file: AKM-biblio-30thJune-2016.tex
%\bibliographystyle{amsalpha}

%% file: AKM.bbl
\begin{thebibliography}{ABCD11}

\bibitem[AF05]{AF05}  R.A.~Adams and J.J.F.~Fournier,
\textit{Sobolev Spaces}, Academic Press 2nd ed., Elsevier, (2005).


\bibitem[AKM13]{AKM09b}  S.~Adams, R.~Koteck\'{y} and  S.~M\"uller,
\textit{Finite range decomposition for families of gradient Gaussian measures},  
Journal of Functional Analysis \textbf{264},  169--206   (2013). 

 
 \bibitem[AR67]{AR67} R. Abraham, J. Robbin, \textit{Transversal Mappings and Flows},
 Benjamin, New York, Amsterdam (1967).
 
 
  \bibitem[BSTW15]{BSTW15} R.~Bauerschmidt, G.~Slade, A.~Tomberg \& B.C.~Wallace, B. C, 
 \textit{Finite-order correlation length for 4-dimensional weakly self-avoiding walk 
 and $|\varphi|^4$ spins},  arXiv:1511. 02790v1  (2015).

 \bibitem[BK07]{BK07}
M.~Biskup and  R.~Koteck{\'y},
\textit{Phase coexistence of gradient Gibbs states},
 Probability Theory and Related Fields
\textbf{139},  1--39 (2007).

  
  \bibitem[Bry09]{B07} D.C.~Brydges,
\textit{Lectures on Renormalisation group},
In {\it Statistical Mechanics} IAS/Park City Mathematics Series, ed. S.~Sheffield and T.~Spencer, (2009).

 

 \bibitem[BBS15b]{BBS15b}
R.~Bauerschmidt, D.C.~Brydges \& G.~Slade,
\textit{Logarithmic correction for the susceptibility of the 4-dimensional 
weakly self-avoiding walk: a renormalisation group analysis},
Communications  in Mathematical Physics \textbf{337}, 817--877 (2015).


%% to be added
 \bibitem[BBS15a]{BBS15a}
R.~Bauerschmidt, D.C.~Brydges \& G.~Slade,
\textit{A renormalisation group method. III. Perturbative analysis},
Journal of Statistical Physics \textbf{159}, 492--529 (2015).

 

\bibitem[BS15a]{BS15a} D.C.~Brydges and G.~Slade,
 \textit{A renormalisation group method. I. Gaussian integration and normed algebras},
  J. Stat. Phys., \textbf{159}, 421--460, (2015). 


\bibitem[BS15b]{BS15b} D.C.~Brydges and  G. Slade, \textit{A renormalisation group method. II. Approximation by local polynomials},
 Journal of Statistical Physics, \textbf{159}, 461--491 (2015).
 
 
 
 \bibitem[BS15c]{BS15c} D.C.~Brydges and G.~Slade, \textit{A renormalisation group method. IV. Stability analysis}, 
J. Stat. Phys., \textbf{159}, 530--588, (2015). 

\bibitem[BS15d]{BS15d} D.C.~Brydges and G.~Slade, \textit{A renormalisation group method. V. A single renormalisation group step}, 
J. Stat. Phys., \textbf{159}, 589--667, (2015). 



\bibitem[BT06]{BT06}  D.C.~Brydges and A.~Talarczyk,
\textit{Finite range decompositions of positive-definite functions},
 Journal of Functional Analysis {\bf 236}, 682--711 (2006).

\bibitem[BGM04]{BGM04}  D.C.~Brydges, G.~Guadagni, and P.K.~Mitter,
\textit{Finite Range Decomposition of Gaussian Processes},
JSP \textbf{115}, Nos. 1/2, 415--449 (2004).

%\smallskip
%\bibitem[BLM87]{BLM87} {\sc J.~Bricmont, J.L.~Lebowitz} and {\sc C.~Maes},
%\newblock Percolation in Strongly Correlated Systems: The Massless Gaussian Field,
%\newblock JSP {\bf 48}, Nos. 5/6, 1249--1268 (1987).

\bibitem[BY90]{BY90} D.C.~Brydges and H.T.~Yau,
	\textit{Grad $ \varphi $ Perturbations of Massless Gaussian Fields},
 Commun. Math. Phys. \textbf{129}, 351--392 (1990).
 
  \bibitem[Buc16]{BU16} S.~Buchholz, \textit{Finite Range Decomposition for Gaussian Measures with Improved Regularity},
 arXiv:1603.06685v1 (2016).


%\smallskip
%\bibitem[Car81]{Carr81}{\sc J.~Carr},
%\newblock \textit{Applications of centre manifold theory}, 
%\newblock Volume 35 of Applied Mathematical Sciences, Springer-Verlag, New York (1981).

\bibitem[CD12]{CD09} C.~Cotar and J.-D.~Deuschel,
\textit{Decay of covariances, uniqueness of ergodic component and scaling limit for a class of gradient  systems with non-convex potential}, 
Annals de l'Institut Henri Poincar\'e \textbf{48}, 819--853, (2012).

\bibitem[CDM09]{CDM09} C.~Cotar, J.-D.~Deuschel, and  S.~M\"uller, 
\textit{Strict Convexity of the Free Energy for a Class of Non-Convex Gradient Models},
 Commun. Math. Phys. \textbf{286}, 359--376 (2009).

%\smallskip
%\bibitem[DZ98]{DZ98}
%{\sc A.~Dembo} and {\sc O.~Zeitouni},
%\newblock {\it Large Deviations Techniques and Applications,}
%\newblock 2nd ed., Springer, New York (1998).
%
%\bibitem[Day77]{Day77} S.~Dayal,
%\textit{A converse of Taylor's theorem for functions on Banach spaces},
%Proc. Amer. Math. Soc. \textbf{65},  no. 2, 265--273 (1977).

\bibitem[Die60]{Die60} J.~Dieudonn\'e,
\textit{Foundations of Modern Analysis},
Academic Press,  (1960).

\bibitem[DGI00]{DGI00}  J.D.~Deuschel, G.~Giacomin, and  D.~Ioffe,
 \textit{Large deviations and concentration properties for $\nabla\varphi$ interface models},
 Probab. Theory Related Fields {\bf 117}, 49--111 (2000).

%\smallskip
%\bibitem[DS01]{DS01} 
%{\sc J.-D.~Deuschel} and {\sc D.W.~Stroock},
%\newblock {\it Large Deviations},
%\newblock AMS Chelsea Publishing, American Mathematical Society (2001).

%\smallskip
%\bibitem[DR78]{DR78}
%{\sc H.~Doss} and {\sc G.~Royer},
%\newblock Processus de diffusion associe aux measure de Gibbs sur $\R^d$,
%\newblock Probab. Theory Relat. Fields {\bf 46}, 107-124 (1978).


\bibitem[Fed]{F} 
H. Federer,
\textit{Geometric Measure Theory},
Springer Verlag, Section 3.1.11, (1996).

\bibitem[FP81]{FP} J.~Fr\"ohlich and  C.~Pfister,
\textit{On the absence of spontaneous symmetry breaking and of crystalline ordering in two-dimensional systems},
Commun.  Mathematical Physics \textbf{81}, 277--298 (1981).

%\smallskip 
%\bibitem[Fun05]{Fun05} {\sc T.~Funaki},
%\newblock Stochastic Interface Models,
%\newblock In: Lectures on Probability Theory and Statistics, Ecole d'Et\'{e} de Probabilit'{e}s de Saint-Flour XXXIII - 2003 (ed. J. Picard), 103--274, Lect. Notes Math. {\bf 1869}, Springer (2005).


\bibitem[FS97]{FS97}  T.~Funaki and  H.~Spohn,
\textit{Motion by Mean Curvature from the Ginzburg-Landau $
\nabla \varphi $ Interface Model}, 
 Communications in Mathematical Physics {\bf 185}, 1--36 (1997).

%\smallskip
%\bibitem[Ge88]{Geo88} 
%{\sc H.-O.~Georgii},
%\newblock {\it Gibbs Measures and Phase Transitions},
%\newblock Walter de Gruyter, Berlin (1988).

%\smallskip
%\bibitem[dH00]{dH00} {\sc F.~den Hollander},
%\newblock {\it Large Deviations},
%\newblock Fields Institute Monographs, AMS (2000).


\bibitem[GJ87]{GlimmJaffe}  J.~Glimm and  A.~Jaffe, 
\textit{Quantum Physics - a Functional Integral Point of View},
2nd ed., Springer New York (1987).


\bibitem[Ham82]{Ham82}  R.S. Hamilton,
 \textit{The Inverse Function Theorem of Nash and Moser}, 
 Bulletin of the American Mathematical Society Volume \textbf{7}, Number 1, 65--222 (1982).


\bibitem[Har]{H}  M.~Hardy,
\textit{Combinatorics of Partial Derivatives},
   The Electronic Journal of Combinatorics   {\bf 13}  \# R1 (2006).
   
   
\bibitem[Hil16]{Hi16} S. Hilger, \textit{Scaling limit and convergence of smoothed covariance for gradient models with non-convex potential},
 arXiv:1603.04703v1 (2016).

%\smallskip
%\bibitem[ISV02]{ISV02}  {\sc D.~Ioffe, S.~Shlosman} and  {\sc Y.~Velenik},
%\newblock 2D models of statistical physics with continuous symmetry: the case of singular interactions,
%\newblock   Commun. Math. Phys.  {\bf 226}, 433-454 (2002).
 

\bibitem[Lev06]{L06}
E.~Levy, 
\textit{Why do partitions occur in Fa\`a di Bruno's chain rule for higher
  derivatives?}, 
  Arxiv preprint math/0602183 (2006).
  

\bibitem[Oli54]{Oli54}
H. Oliver, \textit{The exact Peano derivative}, 
Trans. Amer. Math. Soc. 76, 444--456 (1954).


\bibitem[Sob40]{Sob40}
S. Sobolev \textit{Sur l'evaluation de quelques sommes pour une fonction definie sur un r\'eseau.} (Russian, French summary)
   % FJournal = {{Izvestiya Akademii Nauk SSSR. Seriya Matematicheskaya}},
    Izv. Akad. Nauk SSSR, Ser. Mat.,
  %  ISSN = {0373-2436},
    \textbf{4}, 5--16 (1940).
 %   Publisher = {Izdatel'stvo Nauka, Moskva},
    %Language = {Russian},
   % Zbl = {0025.39802}
    %% End new



\end{thebibliography}
